%% file: Mark_Bugden.tex
\documentclass[a4paper,12pt,twoside,openright]{report}

\input{preamble}

\begin{document}
\pagenumbering{gobble}

\include{titlepage}

\include{chapter1_introduction}
\include{chapter2}

\include{chapter3}
\include{chapter4}

\include{chapter5}
\include{chapter6}

\include{appendix}

\include{bibliography}
\end{document}

%% file: preamble.tex
\usepackage{float}
\usepackage{amsmath}
\usepackage{amssymb}
\usepackage{amsfonts}
\usepackage{amsthm}
\usepackage[dvipsnames]{xcolor}
\usepackage{braket}
\usepackage{graphicx}
\usepackage[mathscr]{euscript}
\usepackage[margin=2cm]{geometry}
\usepackage{mathtools}
\usepackage{appendix}
\usepackage{tikz}
\usepackage{tikz-cd}
\usepackage[most]{tcolorbox}
\usepackage{dsfont}
\usepackage{subcaption}

\usepackage[customcolors]{hf-tikz}

\hfsetfillcolor{blue!5}
\hfsetbordercolor{blue!40}

\DeclareFontFamily{U}{wncy}{}
\DeclareFontShape{U}{wncy}{m}{n}{<->wncyr10}{}
\DeclareSymbolFont{mcy}{U}{wncy}{m}{n}
\DeclareMathSymbol{\Sh}{\mathord}{mcy}{"58}

\usetikzlibrary{positioning}

\tikzset{node distance=1in, auto}

\usetikzlibrary{arrows}

\newtcbox{\mybox}{colback=red!5!white,
	colframe=red!75!black}

\usepackage[intoc]{nomentbl}
\makenomenclature

\renewcommand{\nompreamble}{The following symbols are used in the present work:}
\usepackage{hyperref}
\hypersetup{
	colorlinks = true,
	citecolor = blue
}

\newcommand{\blanknonumber}{\newpage\thispagestyle{empty}}

\newcommand{\pr}{\partial}     
\newcommand{\NN}{\mathbb{N}}   
\newcommand{\ZZ}{\mathbb{Z}}   
\newcommand{\CC}{\mathbb{C}}   
\newcommand{\RR}{\mathbb{R}}   
\newcommand{\HH}{\mathbb{H}}
\newcommand{\GG}{\mathbb{G}}
\DeclareMathOperator{\Tr}{Tr}  
\DeclareMathOperator{\dd}{d\!}   
\newcommand{\Lie}{\mathcal{L}}
\DeclareMathOperator{\ddd}{\textrm{d}}
\newcommand{\de}{\delta_{\epsilon}}

\newcommand{\fg}{\mathfrak{g}}
\newcommand{\fd}{\mathfrak{d}}
\newcommand{\sG}{\mathsf{G}}
\newcommand{\sg}{\mathsf{g}}
\newcommand{\sh}{\mathsf{h}}
\newcommand{\sn}{\mathsf{n}}
\newcommand{\sH}{\mathsf{H}}

\newcommand{\sD}{\mathsf{D}}
\newcommand{\sN}{\mathsf{N}}
\newcommand{\scD}{\mathscr{D}}

\newcommand{\cA}{\mathcal{A}}
\newcommand{\cF}{\mathcal{F}}
\newcommand{\cS}{\mathcal{S}}
\newcommand{\CA}{$C^{\ast}$}

\newcommand{\fpa}{\frac{1}{4 \pi}}

\theoremstyle{plain}
\newtheorem{theorem}{Theorem}[section]
\newtheorem{lemma}[theorem]{Lemma}
\newtheorem{corollary}[theorem]{Corollary}

\theoremstyle{definition}
\newtheorem{definition}[theorem]{Definition}

\theoremstyle{remark}

\newcommand{\chapquotenoauthor}[2]{\begin{quotation} \textit{#1} \end{quotation} \begin{flushright} - #2\end{flushright} }

\newcommand{\chapquote}[3]{\begin{quotation} \textit{#1} \end{quotation} \begin{flushright} - #2, \textit{#3}\end{flushright} }

%% file: titlepage.tex
\title{T-duality}
\author{Mark Bugden}
\date{2018}

\begin{titlepage}

\begin{center}
\includegraphics[]{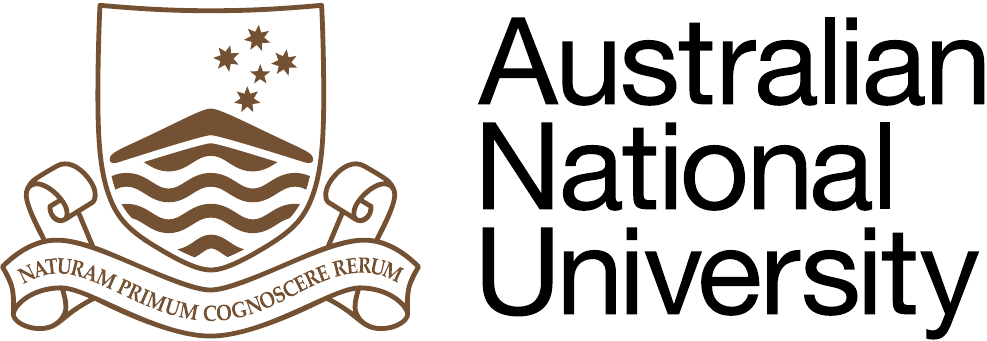}\\[2.5cm]
\end{center}
\begin{center}
{ \huge \textbf{\sc A Tour of T-duality}}\\
{\large \textbf{\sc Geometric and Topological Aspects of T-dualities }}\\[1.2cm]
\textmd{By}\\[0.7cm]
{\huge{Mark Bugden}\\[1.2cm]}

\rule{\textwidth}{2pt}\\[\baselineskip]

Supervisor: Professor Peter Bouwknegt\\[1cm]
This thesis submitted for the degree of \\
\large{\sc Doctor of Philosophy}\\
\textmd{\sc Mathematical Sciences Institute}\\
\textmd{\sc August 2018}\\[0.25cm]

\rule{\textwidth}{2pt}\\[\baselineskip]

\end{center}
\end{titlepage}

\blanknonumber\ \blanknonumber

\begin{abstract}
The primary focus of this thesis is to investigate the mathematical and physical properties of spaces that are related by T-duality and its generalisations. In string theory, T-duality is a relationship between two \emph{a priori} different string backgrounds which nevertheless behave identically from a physical point of view. These backgrounds can have different geometries, different fluxes, and even be topologically distinct manifolds. T-duality is a uniquely `stringy' phenomenon, since it does not occur in a theory of point particles, and together with other dualities has been incredibly useful in elucidating the nature of string theory and M-theory. \\

There exist various generalisations of the usual T-duality, some of which are still putative, and none of which are fully understood. Some of these dualities are inspired by mathematics and some are inspired by physics. These generalisations include non-abelian T-duality, Poisson-Lie T-duality, non-isometric T-duality, and spherical T-duality. In this thesis we review T-duality and its various generalisations, studying the geometric, topological, and physical properties of spaces related by these dualities. 
\end{abstract}

\blanknonumber\ \blanknonumber

\newpage %




\mbox{} %
\vspace*{0.25\textheight} \\ %
\centerline{\bf Declaration} %
\mbox{}
\\
The work in this thesis is, except where specifically acknowledged, the original work of the author. Some of the work in this thesis, and in particular some of the results in Chapter \ref{chptr:Ch5}, have appeared in \cite{BBKW}, written together with P. Bouwknegt, C. Klim\v{c}\'{i}k, and K. Wright. After the thesis was submitted, but before corrections were published, the author has uploaded a new paper based, in part, on the results of several sections of this thesis \cite{B19}.

In addition, during the preparation of this thesis the author has published two other papers, \cite{B18} and \cite{BP18}, the results of which are not included in this thesis.

\bigskip

The research contributing to this thesis was undertaken at the Mathematical Sciences Institute (MSI) at the Australian National University (ANU), and was supported by an Australian Postgraduate Award and an ANU Supplementary Scholarship. 

\newpage

\blanknonumber\ \blanknonumber

\newpage

\mbox{} %
\vspace*{0.15\textheight} \\ %
\centerline{\bf Acknowledgements} %
\mbox{}
\\
The MSI has been a wonderful place to spend four years, and I have been incredibly lucky with all the opportunities with which I have been afforded as a result of staying here. In addition to the ANU, I would also like to thank the Institut Henri Poincar\'{e} for their hospitality during my stay in Paris.

I would like to express my heartfelt gratitude to all the people who have helped and supported me throughout the years it took me to produce this thesis.

To my supervisor Peter Bouwknegt, for your guidance and support. Your mystical aura of understanding, as well as your patience for my na\"{i}ve exuberance, has been instrumental in the completion of this thesis. 

To my family, and in particular my mum and dad, who have been supportive of my goals, no matter how bizarre they may seem. 

To all my friends, local and nonlocal, and especially to the members (past and present) of the \emph{MSI Beerz Crew}, for being the glue keeping me together. 

To Sabrina, for keeping my default happiness at ludicrously elevated levels. You have set the standards for other humans far too high.  

To my collaborators, Kyle, Ctirad, and Claudio. It has been an absolute pleasure to work with you. Thank you for making work seem like play. To Kyle, in particular, for being my sounding board and sanity tester.\\

The extent of my gratitude and appreciation for all of you cannot be overstated.

\begin{center}
	\bf{Thank you.}
\end{center}

\newpage

\blanknonumber\ \blanknonumber

\newpage

\pagenumbering{roman}
{\hypersetup{linkcolor=black}
	\tableofcontents %
}
\blanknonumber\ \blanknonumber

\mbox{}
\vspace*{0.15\textheight} \\ %
\chapquote{My father made no reproach in his letters, and only took notice of my silence by enquiring into my occupations more particularly than before. Winter, spring, and summer passed away during my labours; but I did not watch the blossom or the expanding leaves - sights which before always yielded me supreme delight - so deeply was I engrossed in my occupation. The leaves of that year had withered before my work drew to a close; and now every day showed me more plainly how well I had succeeded. But my enthusiasm was checked by anxiety, and I appeared rather like one doomed by slavery to toil in the mines, or any other unwholesome trade than an artist occupied by his favourite employment. Every night I was oppressed by a slow fever, and I became nervous to a most painful degree; the fall of a leaf startled me, and I shunned my fellow-creatures as if I had been guilty of a crime. Sometimes I grew alarmed at the wreck I perceived that I had become; the energy of my purpose alone sustained me: my labours would soon end, and I believed that exercise and amusement would then drive away incipient disease; and I promised myself both of these when my creation should be complete. }{Mary Shelley}{Frankenstein; or, The Modern Prometheus}
\mbox{}

\blanknonumber

%% file: chapter1_introduction.tex
\chapter{Introduction}
\label{chptr:Ch1} \pagenumbering{arabic}

\section{Dualities}
In theoretical physics, a duality refers to two different descriptions of the same physical phenomenon. One of the simplest examples of such a duality occurs in classical electromagnetism. Maxwell's equations in a vacuum (or in the presence of electric and magnetic sources) exhibit a symmetry under the interchange of the electric and magnetic fields. More specifically, the equations determining the physics\footnote{Note that we have set the speed of light $c = 1$.}
\begin{equation*}
\begin{aligned}[c]
\nabla \cdot \vec{E} &= 0 \\[1em]
\nabla \cdot \vec{B} &= 0 
\end{aligned}
\qquad \qquad 
\begin{aligned}[c]
- \nabla \times \vec{E} &=  \frac{\pr \vec{B}}{\pr t}\\[1em]
\nabla \times \vec{B} &= \frac{\pr \vec{E}}{\pr t}
\end{aligned}
\end{equation*}
are invariant under the transformation
\begin{subequations}
\label{maxwellsymmetry}
\begin{align}
\vec{E} &\longmapsto \vec{B} \\[1em]
\vec{B} &\longmapsto -\vec{E}.
\end{align}
\end{subequations}
That is, if there is a given configuration $(\vec{E},\vec{B})$ of electric and magnetic fields which solves Maxwell's equations, then the configuration $(\vec{B},-\vec{E})$ will also solve the equations. What we call an `electric field' and what we call a `magnetic field' are therefore simply just conventions, since every magnetic field has an equivalent description as an electric field, and vice versa. Let us see how the symmetry transformation (\ref{maxwellsymmetry}) behaves in terms of the action. Recall that we can package the electric and magnetic fields into the field strength tensor:
\begin{align*}
F &=
\left( 
\begin{matrix}
0 & E_x & E_y & E_z \\
-E_x & 0 & -B_z & B_y \\
-E_y & B_z & 0 & -B_x \\
-E_z & -B_y & B_x & 0
\end{matrix}
\right).
\end{align*}
That is, we have 
\begin{align*}
E &= E_x \dd t \wedge \dd x + E_y \dd t \wedge \dd y + E_z \dd t \wedge \dd z \\[1em]
B &= B_x \dd z \wedge \dd y + B_y \dd x \wedge \dd z + B_z \dd y \wedge \dd x,
\end{align*}
and 
\begin{align*}
F = E + B.
\end{align*}
The action is then
\begin{align*}
S_{EM} = - \frac{1}{4} \int  F \wedge \star F.
\end{align*}
A simple computation, however, reveals that $\star F$ is given by
\begin{align*}
\star F &=
\left( 
\begin{matrix}
0 & -B_x & -B_y & -B_z \\
B_x & 0 & -E_z & E_y \\
B_y & E_z & 0 & -E_x \\
B_z & -E_y & E_x & 0
\end{matrix}
\right).
\end{align*}
It follows that the transformation (\ref{maxwellsymmetry}) acts on $F$ by
\begin{align*}
F &\longmapsto F ' = \star F,
\end{align*}
and so the action transforms by
\begin{align*}
S_{EM}' \longmapsto -\frac{1}{4} &\int \star F \wedge \star \left( \star F \right) \\
= \frac{1}{4} &\int F \wedge \star F \\
= & - S_{EM}.
\end{align*}
The action only changes by a global sign, and so the symmetry transformation (\ref{maxwellsymmetry}) leaves its variation, and therefore the equations of motion, invariant. This phrasing of electromagnetic duality is a wonderful starting point to discuss the dualities that will appear in this thesis. As we shall see, T-duality can be similarly described as a transformation of fields leaving an action invariant. 

There are many different dualities in physics: examples include $AdS$/CFT \cite{M97}, electromagnetic duality, Kramers-Wannier duality \cite{KW41}, Montonen-Olive duality \cite{MO77}, and the string dualities: S,T and U-dualities. Whilst dualities are certainly mathematical curiosities, they can also be incredibly useful. Consider, for example, S-duality, which is a generalisation of electromagnetic duality. In string theory, S-duality acts by inverting the coupling constant. If the value of the coupling constant is small (much less than 1), then the system is \emph{weakly coupled} and we can describe the system using perturbation theory. When the coupling constant is large, the system is \emph{strongly coupled}, and we can no longer trust the perturbative exapansion. Strongly coupled systems are therefore usually much more difficult to understand. Since S-duality maps between strongly coupled and weakly coupled theories, it can be used to help understand strongly coupled theories by first dualising to a weakly coupled theory, and then studying that theory using perturbation theory. 
\section{Target space duality}
\label{sec:TargetSpaceDuality}
T-duality is a duality which emerged from string theory in the early 1980's \cite{BGS82,KY84,SS86,GRV}.\footnote{The ``T" either stands for ``Torus" or ``Target", depending on whom you ask.} At its core, it is a statement that there are two different string backgrounds which behave identically as far as physically observable quantities are concerned. Let's discuss how T-duality first appears by studying the closed Bosonic string moving in a target space of the form $M = \RR^{1,24} \times S^1$.\footnote{In this section we follow closely the wonderfully readable lecture notes in \cite{Tcamb}. This content is widely known, however, and is covered in most of the standard string theory textbooks.} That is, we want to study maps $X: \Sigma \to M$ with components $X (\tau, \sigma) = (X^0,X^1,\dots,X^{24},X^{25})$, where $\Sigma$ is a two-dimensional world sheet of a string with coordinates $( \tau,\sigma)$. We want to impose the constraint 
\begin{align*}
X^{25} (\tau , \sigma + 2\pi) = X^{25}(\tau, \sigma) + 2 \pi m R,
\end{align*}
where $m \in \ZZ$ is the \emph{winding number} of the string, and $R \in \RR_{>0}$ is the radius of the circular coordinate. The string equations of motion can be obtained by extremising the Polyakov action:
\begin{align}
\label{Polyakov}
S = \frac{1}{4 \pi \alpha'}\int_{\Sigma} \dd \sigma \dd \tau \left(  h^{\alpha \beta} \sqrt{-h} \, g_{\mu \nu}  \pr_{\alpha} X^{\mu} \pr_{\beta} X^{\nu} \right).
\end{align}
The solution to the equations of motion can be written in terms of left-movers and right-movers:
\begin{align*}
X^{\mu}_L &= \tfrac{1}{2} x^{\mu} + \tfrac{1}{2} \alpha' p^{\mu}_L (\tau + \sigma) + \textrm{oscillator terms} \\[1em]
X^{\mu}_R &= \tfrac{1}{2} x^{\mu} + \tfrac{1}{2} \alpha' p^{\mu}_R (\tau - \sigma) + \textrm{oscillator terms}.
\end{align*}
The only difference for the periodic coordinate $X^{25}$ is in the expression for $p_L$ and $p_R$:
\begin{align*}
p^{25}_L &= \frac{n}{R} + \frac{mR}{\alpha'} \\[1em]
p^{25}_R &= \frac{n}{R} - \frac{mR}{\alpha'}, 
\end{align*}
where $n \in \ZZ$ is the integer associated to the quantised momentum in the compact direction. The spectrum of the theory is easy to calculate, and is given by
\begin{align}
\label{masssquared}
M^2 &= - p_{\mu} p^{\mu} = \frac{n^2}{R^2} + \frac{m^2 R^2}{(\alpha')^2} + \frac{2}{\alpha'} \left( N + \widetilde{N} - 2 \right).
\end{align}
The first term corresponds to a contribution of the momentum to the mass, which is just the familiar ``kinetic energy'' term. To understand the origin of the second term, we recall that strings have an intrinsic tension $T = \frac{1}{2\pi \alpha'}$ encouraging the string to contract. Stretching a string increases its mass/energy, and since a string wrapping $m$ times around a circle of radius $R$ has a minimum length of $l = 2 \pi m R$, the minimum mass of such a string is given by
\begin{align*}
M = lT = \frac{m R}{\alpha'},
\end{align*}
and so the second term is simply the mass contribution that comes from the winding of the string around the compact direction. 

The expression for the spectrum (\ref{masssquared}) has a very curious feature. If we make the following transformation
\begin{align*}
R &\longmapsto \frac{\alpha'}{R} \\[1em]
n &\longmapsto m \\[1em]
m &\longmapsto n,
\end{align*}
then the spectrum remains invariant. That is, the masses of the particles observed for a string moving in $M$ are exactly the same as the masses observed for a string moving in some other background $M'$, where $M' = \RR^{1,24} \times S^1$ and the $S^1$ has radius $\frac{\alpha'}{R}$. This invariance is the essence of T-duality. 

Note that this duality is an inherently `stringy' feature - although the quantisation for a particle on $M$ exhibits the same quantised momentum $n$, the winding modes determined by $m$ only exist for a string, since a particle cannot wind around the compact direction. Indeed, this interchange of momentum and winding is, at least thematically, representative of the features of `\emph{stringy geometry}'. Stringy geometry is simply the observation that strings behave differently to particles when it comes to geometry, and indeed that geometry is not the relevant structure to use when discussing the physics of strings, since there are inequivalent geometries which give rise to the same physics. We will see in Chapter \ref{chptr:Ch2} and beyond that under T-duality, the geometry is intermixed with the $B$-field, a stringy version of electromagnetism. Other structures, such as complex and symplectic structures, are also interchanged under T-duality. 

These interchanges under duality suggest that we should consider generalised structures incorporating both the original structures, under which T-duality acts by an automorphism. For the geometry and $B$-field we will discuss in Chapter \ref{chptr:Ch2}, such a generalised structure goes by the name of generalised geometry. Complex and symplectic structures can be incorporated in this framework as generalised complex structures. A different approach, which aims to unify momentum and winding modes into a single geometric framework goes by the name of Double Field Theory. 

Let us briefly mention that this simple duality has far-reaching consequences. The first thing to note is that this is a full symmetry of the entire conformal field theory, and not just a symmetry of the spectrum. Even more, the symmetry also holds for a theory of open strings. For open strings, we need to specify boundary conditions for the endpoints of the string - either Neumann or Dirichlet. If the boundary conditions are Dirichlet, then there is some submanifold of spacetime to which the endpoints of the string are confined. A simple calculation shows, however, that in the presence of such boundary conditions, momentum of the string is not conserved. For some time it was assumed, therefore, that Dirichlet boundary conditions were unphysical. This is where T-duality appears. T-duality acts on open strings by interchanging the boundary conditions - Neumann boundary conditions get swapped with Dirichlet boundary conditions. This means that, from a physical perspective, there is no distinction between a theory with entirely Neumann boundary conditions and a theory with entirely Dirichlet boundary conditions - our theory must incorporate both. The loss of momentum coming from Dirichlet boundary conditions now has to be understood, and it was soon realised that the submanifold on which the string ends carries away the momentum, and is itself a dynamical object, now termed a D-brane. 

Another facet of the importance of T-duality in string theory is the fact that in 10 dimensions, there are 5 consistent superstring theories: Type I, Type IIA and Type IIB, Heterotic-O and Heterotic-E. The discovery that various dualities related these theories to each other led Witten to the remarkable suggestion that these 5 string theories were simply different regimes of an underlying 11-dimensional theory, mysteriously referred to as M-theory \cite{Wit95}.

Finally, we remark that mathematicians also have plenty of reasons to be interested in T-duality. In the context of algebraic geometry, mirror symmetry is the conjectural relation between two Calabi-Yau varieties which, when formulated as the target spaces of an $N=2$ supersymmetric 2d conformal field theory, have equivalent CFTs. The Hodge numbers $h^{1,1}$ and $h^{1,2}$ of a variety $X$ and its dual $\widehat{X}$ are interchanged under Mirror symmetry. The SYZ conjecture says that if a Calabi-Yau variety has a mirror pair, then it admits a $\mathbb{T}^3$ fibration, and mirror symmetry corresponds to T-duality along these fibers. In the context of $C^{\ast}$-algebras, T-duality can be thought of as a map between two algebras inducing an isomorphism of their K-theories, see Section \ref{Cstar} for more details. Of course, we have already mentioned generalised geometry, which has as much interest to mathematicians as to physicists.

T-duality has a plethora of generalisations, and one of the aims of this thesis is to provide an overview of these generalisations. In particular, we are interested in how the geometry and topology of spaces behave under these T-dualities. 

\section{Structure of thesis}
Each chapter of this thesis is focussed around one of the generalisations of T-duality. Care has been taken to ensure that examples are included for each of the different forms of T-duality. 

Chapter \ref{chptr:Ch2} is centred on the standard form of T-duality, also sometimes called abelian T-duality. We provide a detailed review of T-duality \`{a} la Buscher, and discuss the geometric implications of this procedure. A new result in this section is related to Sasaki-Einstein geometry - that a trivial circle bundle is Sasaki-Einstein iff it is one-dimensional. We then discuss the topology of the Buscher procedure, and mathematical aspects of topological T-duality. 

Chapter \ref{chptr:Ch3} deals with non-abelian T-duality. We introduce and review relevant aspects of the non-abelian T-duality procedure. We then discuss the open problem of understanding the global aspects of non-abelian T-duality, reviewing existing work from a physical perspective in this direction, and providing comments from a more mathematical perspective. 

Chapter \ref{chptr:Ch4} is a small digression on a generalisation of non-abelian T-duality known as Poisson-Lie T-duality. We introduce Poisson-Lie T-duality and discuss its relation to other forms of T-duality.

Chapter \ref{chptr:Ch5} discusses a relatively new generalisation of T-Duality known as non-isometric T-duality. After introducing and discussing this new duality, we give the proof, first published in \cite{BBKW}, that this duality is equivalent to the standard notion of non-abelian T-duality. We then introduce a generalisation of the original proposal, and discuss the relation between the generalisation and Poisson-Lie T-duality. 

Chapter \ref{chptr:Ch6} discusses a final generalisation of T-duality called spherical T-duality. This mathematical generalisation of topological T-duality is a promising candidate for new duality in M-theory/11-dimensional supergravity. We introduce and review this putative new duality, and discuss its potential geometric and physical aspects.

Following Chapter 6 there are several appendices containing supplementary information which is occasionally referenced in the body of the text. 

%% file: chapter2.tex
\chapter{Abelian T-duality}
\label{chptr:Ch2} 


\section{Geometry}
\label{sec:AbelGeometry}


\subsection{Buscher rules}
\label{BuscherRules}
The $R \to 1/R$ transformation rule of toroidal compactifications has a generalisation to curved string backgrounds possessing an abelian group of isometries. This generalisation, due to Buscher \cite{Bus87,Bus88}, uses a gauging procedure which will feature prominently throughout this thesis.

Our starting point is the string non-linear sigma model, described by maps $X: \Sigma \to M$, where $(\Sigma,h_{\alpha \beta})$ is a two-dimensional Lorentzian worldsheet, and $(M,g_{\mu \nu})$ is a (pseudo)-Riemannian manifold together with a B-field, a locally-defined two-form gauge field. The action is
\begin{align}
\label{stringNLSM}
S = \frac{1}{4 \pi \alpha'}\int_{\Sigma} \dd \sigma \dd \tau \left(  h^{\alpha \beta} \sqrt{-h} \, g_{\mu \nu}  \pr_{\alpha} X^{\mu} \pr_{\beta} X^{\nu} + \epsilon^{\alpha \beta} B_{\mu \nu} \pr_{\alpha} X^{\mu} \pr_{\beta} X^{\nu}  + \alpha' \sqrt{-h} \, \Phi \mathcal{R}^{(2)}\right)
\end{align}
where $\Phi$ is the dilaton, and $\mathcal{R}^{(2)}$ is the Ricci scalar of the worldsheet. We will, for the moment, ignore the contribution from the dilaton (returning to it in Section \ref{subsec:dilatonAbel}). We will also set $\alpha' = 1$ in what follows, unless we are talking about quantum aspects (in which case $\alpha'$ becomes relevant).  We note that the action may be written succinctly as
\begin{align}
\label{NLSMsimple}
S &=  \fpa \int_{\Sigma} g_{\mu \nu} \dd X^{\mu} \wedge \star \dd X^{\nu} + B_{\mu \nu} \dd X^{\mu} \wedge \dd X^{\nu},
\end{align}
where $\star$ is the Hodge dual on the worldsheet,\footnote{$\star^2 = 1$ on one-forms, since our worldsheet is Lorentzian.} and the fields are assumed to be pulled back to the worldsheet. We will be pedantic in referring to (\ref{stringNLSM}) as the string non-linear sigma model, and (\ref{NLSMsimple}) as a non-linear sigma model.\\
Suppose now that we have a vector field $v = v^i \pr_i$ on $M$, generating the following global symmetry:
\begin{align*}
\de X^i &= v^i \epsilon, 
\end{align*}
for a constant parameter $\epsilon$. The variation of the action (\ref{NLSMsimple}) under this symmetry is
\begin{align}
\de S &= \fpa \int_{\Sigma} \de (g_{ij}) \dd X^i \wedge \star \dd X^j + g_{ij}\, \de (\dd X^i) \wedge \star \dd X^j + g_{ij} \dd X^i \wedge \de ( \star \dd X^j) \notag \\
& \quad + \de (B_{ij}) \dd X^i \wedge \dd X^j + B_{ij}\, \de (\dd X^i) \wedge \dd X^j + B_{ij} \dd X^i \wedge \de ( \dd X^j) \notag \\[1em]
&= \fpa \int_{\Sigma} \epsilon \left[ v^k \pr_k g_{ij} + g_{kj} \pr_i v^k + g_{ik} \pr_j v^k \right] \dd X^i \wedge \star \dd X^j \notag \\
& \quad + \epsilon \left[ v^k \pr_k B_{ij} + B_{kj} \pr_i v^k + B_{ik} \pr_j v^k \right] \dd X^i \wedge  \dd X^j \notag \\[1em]
\label{variationactionabelian}
&= \fpa \int_{\Sigma} \epsilon (\Lie_{v} g)_{ij} \dd X^i \wedge \star \dd X^j + \epsilon (\Lie_{v} B)_{ij} \dd X^i \wedge \dd X^j.
\end{align}
It follows that the action is invariant under the symmetry generated by the vector field $v$ if  $\Lie_v g = 0$ and $\Lie_v B = 0$. A vector field for which $\Lie_v g = 0$ is known as a Killing vector, and the flow generated by it is a one-parameter group of diffeomorphisms preserving the metric - that is, a one-parameter group of isometries.

We now assume that our spacetime has at least one continuous isometry. The infinitesimal generator of this isometry is a Killing vector, and we will work in coordinates $\{X^{\mu}\}=\{X^i,\theta\}$ adapted to this Killing vector.\footnote{See Appendix \ref{Adapt} for a discussion of adapted coordinates} This means that in these coordinates the Killing vector is $\pr_{\theta}$, and the isometry generated by this Killing vector is given by translation of the coordinate $\theta \to \theta + \epsilon$. Since $\pr_{\theta}$ is a Killing vector, the Lie derivative of the metric vanishes, $\mathcal{L}_{\pr_{\theta}} g = 0$, and we will also assume that $\mathcal{L}_{\pr_{\theta}}B=0$. In the system of adapted coordinates we are using, the infinitesimal variation of the coordinates is
\begin{subequations}
\begin{align*}
\delta_{\epsilon} X^i &= 0 \\
\delta_{\epsilon} \theta &= \epsilon.
\end{align*}
\end{subequations}
The symmetry generated by the Killing vector is a global symmetry, but we can promote it to a local symmetry by gauging.\footnote{We say that the translation $\theta \to \theta + \epsilon$ is a \emph{global symmetry} because the action is invariant under this translation for constant $\epsilon$. It is a \emph{local symmetry} if the parameter $\epsilon$ is allowed to depend on the worldsheet coordinates, i.e. $\epsilon = \epsilon(\sigma, \tau)$} To gauge the symmetry, we introduce an abelian gauge field $\mathcal{A}$, and minimally couple it to the field $\theta$ by the replacement
\begin{subequations}
\begin{align*}
\dd X^i &\to \mathscr{D} X^i = \dd X^i \\
\dd \theta &\to \mathscr{D} \theta = \dd \theta - \mathcal{A}.
\end{align*}
\end{subequations}
The minimally-coupled action,
\begin{align}
S_{MC} =  \fpa \int_{\Sigma} g_{\mu \nu} \mathscr{D} X^{\mu} \wedge \star \mathscr{D} X^{\nu} + B_{\mu \nu} \mathscr{D} X^{\mu} \wedge \mathscr{D} X^{\nu} ,
\end{align}
 is now invariant under a local symmetry, provided the infinitesimal variation of the gauge field is
\begin{align*}
\delta_{\epsilon}\mathcal{A} = \dd \epsilon.
\end{align*}
In addition to the gauge field we add another term, $\frac{1}{2 \pi} \int_{\Sigma} \mathcal{F} \hat{\theta} $, to the action. The auxilliary field, $\hat{\theta}$, is an additional scalar field, and $\mathcal{F} = \dd \mathcal{A}$ is the field strength of the gauge field. This extra term is added to the action so that the gauged model reduces to the original model. To see why this is true, observe that the equations of motion for the auxilliary field force the field strength to vanish, $\mathcal{F} = 0$, which implies that the gauge field must be pure gauge,\footnote{This is true locally, but care needs to be taken here when discussing global properties of T-duality. See Section \ref{subsec:FibTop} for a discussion of this point in topogically non-trivial worldsheets} i.e. $\mathcal{A} = \dd \chi$. Thus when integrating out the auxilliary field,\footnote{Integrating out a field means to solve for the equations of motion of this field, and then substitute the solution back into the action.} the second term vanishes (since $\mathcal{F}=0$), and we may choose a convenient gauge so that $\mathcal{A} = 0$, thereby recovering the original model. Since the minimally coupled action is independently gauge-invariant (that is, invariant under the local symmetry transformations of $\{X^i,\theta,\cA\}$), and the additional term is also gauge invariant (provided we specify $\de \hat{\theta} = 0$), the entire gauged action is invariant.

In summary, the gauged action
\begin{align}
\label{abeliangaugedaction}
S_{G} &=  \fpa \int_{\Sigma} g_{\mu \nu} \mathscr{D} X^{\mu} \wedge \star \mathscr{D} X^{\nu} + B_{\mu \nu} \mathscr{D} X^{\mu} \wedge \mathscr{D} X^{\nu} \notag\\
& \qquad + \frac{1}{2 \pi} \int_{\Sigma} \mathcal{F} \hat{\theta}
\end{align} 
is invariant under the following local gauge transformations
\begin{subequations}
\label{abeliangauge}
\begin{align}
\delta_{\epsilon} X^i &= 0 \\[1em]
\delta_{\epsilon} \theta &= \epsilon \\[1em]
\delta_{\epsilon} \mathcal{A} &= \dd \epsilon \\[1em]
\delta_{\epsilon} \hat{\theta} &= 0.
\end{align}
\end{subequations}
We have seen that starting from the gauged action, we can recover the original model by integrating out the auxilliary variable and then gauge fixing. The Buscher procedure hinges on the observation that we could instead integrate out the gauge fields first and then gauge fix. One can verify that the model obtained by this procedure is given in terms of coordinates $\{\hat{X}^{\mu}\} = \{X^i,\hat{\theta}\}$ with the dual action
\begin{align*}
\widehat{S} = \fpa \int_{\Sigma} \widehat{g}_{\mu \nu} \dd \hat{X}^{\mu} \wedge \star \dd \hat{X}^{\nu} + \widehat{B}_{\mu \nu} \dd \hat{X}^{\mu} \wedge \dd \hat{X}^{\nu},
\end{align*}
where the new fields $\{\widehat{g}_{\mu \nu}, \widehat{B}_{\mu \nu}\}$ are given in terms of the old fields $\{g_{\mu \nu}, B_{\mu \nu}\}$ by the following transformation rules:
\begin{subequations}
\label{Buscher}
\begin{align}
\widehat{g}_{\hat{\theta} \hat{\theta}} &= \frac{1}{g_{\theta \theta}}\\[1em]
\widehat{g}_{i\hat{\theta}} &= \frac{B_{i \theta}}{g_{\theta \theta}}\\[1em]
\widehat{g}_{ij} &= g_{ij} - \frac{1}{g_{\theta \theta}}\left( g_{i \theta} g_{j \theta} - B_{i \theta} B_{j \theta} \right) \\[1em]
\widehat{B}_{i \hat{\theta}} &= \frac{g_{i \theta}}{g_{\theta \theta}}\\[1em]
\widehat{B}_{ij} &= B_{ij} - \frac{1}{g_{\theta \theta}}\left( B_{i \theta} g_{j \theta} -  g_{i \theta} B_{j \theta} \right).
\end{align}
\end{subequations}
These are the famous Buscher rules for (abelian) T-duality \cite{Bus87,Bus88}. Note that the dual of the dual recovers the original space. Notice also that the metric and the B-field are mixed under this transformation - the geometry is intertwined with the gauge field. This is the hallmark of ``stringy geometry".  We note also that if we begin with a flat metric on a cylinder,
\begin{align*}
\dd s^2 = \sum_{i} (\dd X^i)^2  + R^2 \dd \theta^2,
\end{align*}
that is, if $g_{ij} = \delta_{ij}$ and $g_{\theta \theta} = R^2$, then the Buscher rules give $\widehat{g}_{ij} = g_{ij}$ and $g_{\hat{\theta} \hat{\theta}} = R^{-2}$, establishing the $R \to 1/R$ relationship we found earlier. Modifications of the Buscher procedure allow for more general forms of T-duality:
\begin{itemize}
	\item The generalisation to multiple abelian isometries is straightforward. We discuss this in section \ref{subsec:MultipleAbelian}, where we explicitly compute the Buscher rules for $n$ abelian isometries. The possibility of global obstructions to this procedure are discussed in Section \ref{sec:torusbundles}.
	\item The generalisation to gauging with respect to multiple, non-commuting Killing vectors is known as non-abelian T-duality. This is discussed in detail in Chapter \ref{chptr:Ch3}.
	\item Attempting to gauge a model without requiring the strictness of isometries leads to a generalisation known as non-isometric T-duality. This generalisation forms the content of Chapter \ref{chptr:Ch5}.
\end{itemize}


\subsection{The dilaton}
\label{subsec:dilatonAbel}
Thus far, we have neglected the  contribution of the dilaton, and worked only with the non-linear sigma model (\ref{NLSMsimple}). Non-linear sigma models have interest outside of string theory, so this is fine, but if we are interested in string theory then we need to include the dilaton in our discussion of T-duality. On a flat worldsheet, the dilaton term in the action (\ref{stringNLSM}) vanishes, so this is only relevant for curved worldsheets. 
The first thing we note about (\ref{stringNLSM}) is that it is not conformally invariant, even at a classical level. This is fixed by noting that the dilaton term in (\ref{stringNLSM}) appears at the $\alpha'$ level - the failure of conformal invariance in the dilaton is compensated by the one-loop contribution from the metric and the $B$-field. We will discuss this in more detail in section \ref{subsec:betafunctions}, but for now let us finish our discussion on the transformation of the dilaton under T-duality. 

Buscher found \cite{Bus87} that T-duality maps a conformally invariant theory to another conformally invariant theory, \emph{only} if the dilaton transforms in a specific way. In particular, he found that the transformation
\begin{align}
\label{dilatontransformation}
\hat{\Phi}  = \Phi - \tfrac{1}{2} \log g_{\theta \theta}
\end{align}
should supplement the tranformation rules (\ref{Buscher}), in order for T-duality to preserve conformal invariance at the one-loop level. This shift is related in \cite{Bus88} to a functional determinant resulting from elimination of the first-order gauge field.


\subsection{The closed string spectrum}
In Section \ref{sec:TargetSpaceDuality}, we discussed the spectrum for the theory of a closed string moving on a background with a circular direction. We showed that the spectrum was invariant under the $R \to 1/R$ transformation, provided we also interchanged the momentum and winding modes. For $D$-dimensional toroidal backgrounds, the mass formula \eqref{masssquared} can be written as
\begin{align}
M^2 &= Z^T \mathbb{G} Z + (N + \tilde{N}-2),
\end{align}
where $Z$ is the $2D$-dimensional column vector called the generalised momentum
\begin{align}
Z^M &= \left( 
\begin{matrix}
m_i \\
n^i
\end{matrix}
\right),
\end{align}
and $\mathbb{G}$ is the $2D \times 2D$ generalised metric
\begin{align}
\mathbb{G} (g,B) &= \left( 
\begin{matrix}
g - B g^{-1} B & B g^{-1} \\
-g^{-1} B & g^{-1}
\end{matrix}
\right).
\end{align}
This expression for the mass is invariant under an $SO(D,D;\mathbb{Z})$ group of transformations generated by the following transformations:
\begin{itemize}
	\item \textbf{Diffeomorphisms:} If $A \in GL(D;\mathbb{Z})$, then one can change the basis for the compactification lattice $\Lambda$ by $A\Lambda A^T$. This acts on the generalised metric through 
	\begin{align}
	\left( 
	\begin{matrix}
	A & 0 \\
	0 & A^{-T}
	\end{matrix}
	\right).
	\end{align}
	\item \textbf{B-shifts:} If $\Theta$ is an antisymmetric matrix with integer entries, then one can use it to shift the $B$-field, acting on the generalised metric as
	\begin{align}
	\left(
	\begin{matrix}
	\mathds{1} & \Theta \\
	0 & \mathds{1}
	\end{matrix}
	\right)
	\end{align}
	\item \textbf{Factorised dualities:} This is the $\mathbb{Z}_2$ duality corresponding to the $R \to 1/R$ transformation for a single circular direction. It acts on $\mathbb{G}$ as 
	\begin{align}
	\left(
	\begin{matrix}
	\mathds{1}-e_i & e_i \\
	e_i & \mathds{1}-e_i
	\end{matrix}
	\right),
	\end{align}
	where $e_i$ is the $D \times D$ matrix with 1 in the $(i,i)$-th entry, and zeroes elsewhere.
\end{itemize}
Using this formalism, one can show that the group $SO(D,D;\mathbb{Z})$ acts as a canonical transformation on the phase space of the system (that is, that the duality acts on the oscillators in a way that preserves the commutators).\footnote{See \cite{GPR} for more details.} 

\subsection{Examples: abelian Buscher rules}
\label{subsec:abelgeoexamples}
In order to familiarise the reader with the nature of T-duality, we include some simple examples. These are mostly well-studied examples, and although they aren't all honest string backgrounds, they provide useful toy models to study the properties of T-duality. 

\subsubsection{$S^3$ with no flux}
\label{S3Buscher}
Our first example is a simple and well-studied one: the three sphere with the round metric and no flux. We shall use Hopf coordinates $(\eta, \xi_1, \xi_2)$ for $S^3$, related to complex coordinates $(z_1,z_2) \in \mathbb{C}^2$ by 
\begin{align*}
z_1 &= e^{\frac{i(\xi_1 + \xi_2)}{2}} \sin \eta\\
z_2 &=e^{\frac{i(\xi_1 - \xi_2)}{2}} \cos \eta,
\end{align*}
or real coordinates $(x_1,x_2,x_3,x_4) \in \mathbb{R}^4$ by 
\begin{align*}
x_1 &= \cos \left( \frac{\xi_1 + \xi_2}{2} \right) \sin \eta \\
x_2 &= \sin \left( \frac{\xi_1 + \xi_2}{2} \right) \sin \eta \\
x_3 &= \cos \left( \frac{\xi_1 - \xi_2}{2} \right) \cos \eta \\
x_4 &= \sin \left( \frac{\xi_1 - \xi_2}{2} \right) \cos \eta.
\end{align*}
Here $\eta$ lies in the range $[0,\frac{\pi}{2}]$,  $\xi_2$ lies in the range $[0, 4\pi]$ and $\xi_1$ runs over the range $[0,2\pi]$. These coordinates realise $S^3$ as an embedded submanifold of $\mathbb{R}^4$, and the flat metric on $\mathbb{R}^4$ induces a Riemannian metric on $S^3$. This metric is just the usual round metric on $S^3$, and is given in these Hopf coordinates by
\begin{align*}
\dd s^2 = \dd \eta^2 + \frac{1}{4} \Big( \dd \xi_1^2 + \dd \xi_2^2 - 2 \cos(2 \eta) \dd \xi_1 \dd \xi_2 \Big).
\end{align*} 
Note that with this normalisation, the radius of the $S^3$ is one. We also suppose that the $B$-field vanishes. Since the metric is independent of the $(\xi_1,\xi_2)$ coordinates, it is clear that the corresponding vectors $\pr_{\xi_1}$ and $\pr_{\xi_2}$ are Killing vectors. In addition, since $B = 0$, we have $\Lie_{\pr_{\xi_1}} B = \Lie_{\pr_{\xi_2}}B =  0$. We now perform T-duality along the Hopf direction, parameterised by the $\xi_1$ coordinate. The Buscher rules give us the following dual metric and dual $B$-field:
\begin{subequations}
\begin{align*}
\widehat{\dd s^2} & = \dd \eta^2 + \frac{1}{4} \sin^2 (2 \eta) \dd \xi_2^2 + 4 \dd \widehat{\xi_1}^2 \\[1em]
\widehat{B} &=  \frac{1}{2} \cos(2 \eta) \dd \xi_2 \wedge \dd \widehat{\xi_1}.
\end{align*}
\end{subequations}
The dilaton acquires a shift under this duality, and as we will see later, this metric is just a product metric on $S^2 \times S^1$. Note that the $H$-flux is given by
\begin{align*}
\widehat{H} &= \dd \widehat{B} \\[1em]
&= \sin (2 \eta) \dd \eta \wedge \dd \widehat{\xi}_1 \wedge \dd \xi_2 \\[1em]
&= \dd V_{S^2 \times S^1}
\end{align*}


\subsubsection{$\mathbb{T}^3$ with $H$-flux}
\label{T3Buscher}
Our second example is also a simple, well-studied example: the three torus with $H$-flux. We will use cartesian coordinates $(x,y,z)$ for $\mathbb{T}^3$, with periodic identifications of the coordinates $x \sim x+1$, $y \sim y+1$, and $z \sim z+1$. The metric is simply the flat metric,
\begin{align*}
\dd s^2 = \dd x^2 + \dd y^2 + \dd z^2,
\end{align*}
and we wish to choose a $B$-field such that $H = \dd B$ is non-trivial in cohomology.\footnote{This will be explained further in \ref{subsec:TopCircBund}} Explicitly, let us take
\begin{align*}
B = - x \dd y \wedge \dd z,
\end{align*}
so that 
\begin{align*}
H = \dd B = - \dd x \wedge \dd y \wedge \dd z.
\end{align*}
If $B$ was a globally defined form, then $H$ would be exact, and therefore trivial in cohomology. It is easy to see, however, that $B$ cannot be globally defined on the torus since the transformation $x \sim x+1$ does not leave $B$ invariant. Like the Dirac monopole of electromagnetism, the $B$-field potential is defined on open patches, and glued together on the overlaps using gauge transformations.  We note that 
\begin{align*}
\Lie_{\pr_z} g = \Lie_{\pr_z} B = 0,
\end{align*}
and so we proceed to perform T-duality along the $z$ coordinate. Applying the Buscher rules results in the following dual metric and $B$-field:
\begin{subequations}
\begin{align}
\label{ffluxmetric}
\widehat{g} &= \dd x^2 + \dd y^2 + (\dd \hat{z} - x \dd y)^2 \\[1em]
\label{ffluxB}
\widehat{B} &= 0.
\end{align}
\end{subequations}
We began with a flat metric and a non-trivial $B$-field, and we obtain a dual geometry which is non-flat,\footnote{The Ricci scalar curvature of this metric is $\mathcal{R} = -\frac{1}{2}$.} together with a vanishing $B$-field. This dual model is known as the $f$-flux background, and provides us with another clear example of how the gauge field and the geometry intermix under T-duality. 

This example is often studied in the T-duality literature because it is quite simple, but exhibits a lot of the interesting features of T-duality. The three torus directions of the original model provide, in principle, three different isometries to gauge, and therefore three different T-dualities to perform. Indeed, the dual metric we have obtained is the first of a series of dualities:
\begin{align}
T_{xyz} \stackrel{\pr_z}{\longleftrightarrow} f_{xy} \!^z \stackrel{\pr_y}{\longleftrightarrow} Q_x \!^{yz} \stackrel{\pr_x}{\longleftrightarrow} R^{xyz}.
\end{align}
A quick glance at (\ref{ffluxmetric}) and (\ref{ffluxB}) is enough to see that the dual fields are independent of $y$, and therefore $\pr_y$ is a Killing vector. That is, the dual metric and $B$-field that we obtained after a T-duality transformation of the three torus with flux still retains the residual isometry generated by $\pr_y$. We can use this isometry to perform a T-duality along the $y$ coordinate. Applying the Buscher rules, we obtain the so-called $Q$-flux background:
\begin{subequations}
\label{Qflux}
\begin{align}
\label{Qfluxmetric}
\widehat{\widehat{g}} &= \dd x^2 + \frac{1}{1+x^2} \left( \dd \hat{y}^2 + \dd \hat{z}^2 \right) \\[1em]
\label{QfluxBfield}
\widehat{\widehat{B}} &= \frac{2x}{1+x^2} \dd \hat{y} \wedge \dd \hat{z}.
\end{align}
\end{subequations}
This background will appear in several different areas of this thesis, so we refrain from talking about the properties of it here. We note, however, that a na\"{i}ve attempt to perform a third T-duality along the $x$ coordinate runs into a problem; the vector $\pr_x$ is no longer a Killing vector for the metric. A quick calculation shows that
\begin{align}
\Lie_{\pr_x} \widehat{\widehat{g}} &= -\frac{2x}{(1+x^2)^2} \left( \dd \hat{y}^2 + \dd \hat{z}^2 \right) \\[1em]
\Lie_{\pr_x} \widehat{\widehat{B}} &=  \frac{2 (1-x^2)}{(1+x^2)^2} \dd \hat{y} \wedge \dd \hat{z}. 
\end{align}
Despite this, the putative dual appears often in the literature, particularly in the context of double field theory. We will discuss this chain of dualities more in Section \ref{sec:torusbundles}, Section \ref{Cstar}, and Section \ref{NATDcommentsTT}.


\subsubsection{The time-dual of Schwarzschild}

Let us now look at an example of T-duality on a metric which is very familiar from general relativity - the Schwarzschild solution. The metric is
\begin{align}
\label{SchMetric}
\dd s^2 = - \left( 1- \frac{2M}{r} \right) \dd t^2 + \left(1-\frac{2M}{r} \right)^{-1} \dd r^2 + r^2 \dd \Omega^2,
\end{align}
with no $B$-field. This metric is Ricci flat (for $r \not= 0$), and therefore a vacuum solution of the Einstein Field equations in 4 dimensions. We can promote this solution to a solution of Type II supergravity by taking the trivial product with 6 additional flat directions, and including a constant dilaton (which we take to be zero). The metric (\ref{SchMetric}) is independent of the coordinate $t$, and so $\pr_t$ is a Killing vector. Applying the Buscher rules, we obtain the \emph{time-dual of the Schwarzschild background}:
\begin{align}
\widehat{\dd s^2} = -\left(1-\frac{2M}{r} \right)^{-1} \dd \hat{t}^2 + \left(1-\frac{2M}{r} \right)^{-1} \dd r^2 + r^2 \dd \Omega^2.
\end{align}
This duality is along a time-like, non-compact direction, so is distinctly different to the other T-dualities we have considered. We can see that this metric has curvature singularities at $r = 0$ and $r = 2M$ by computing the Ricci scalar curvature:
\begin{align}
\widehat{\mathcal{R}} = \frac{-4M^2}{r^4 \left( 1- \frac{2M}{r}\right)}.
\end{align}
These are naked singularities, so this metric is perhaps not overly interesting from a general relativity point of view. Furthermore, it doesn't solve the vacuum Einstein field equations since $\widehat{\mathcal{R}}_{\mu \nu} \not=0$. It should, however, solve the 10 dimensional supergravity equations of motion once the appropriate transformation of the dilaton is implemented. The dual dilaton is given by $\widehat{\Phi} = -\frac{1}{2} \log \left( -1+\frac{2M}{r} \right)$,\footnote{see Section \ref{subsec:dilatonAbel}} and a quick calculation shows that this solves the relevant SUGRA equations of motion: $\mathcal{R}_{\mu \nu} + 2 \nabla_{\mu} \nabla_{\nu} \Phi = 0$ and $ (\nabla \Phi)^2 - \frac{1}{2} \nabla^2 \Phi = 0$.

\subsubsection{A example with fixed points: $S^2$}
\label{subsec:S2abelian}
The round metric on $S^2$ can be written using the usual spherical polar coordinates as
\begin{align*}
\dd s^2 &= \dd \theta^2 + \sin^2 \theta \dd \phi^2.
\end{align*}
This metric is independent of the coordinate $\phi$, and so $\pr_{\phi}$ is a Killing vector for this metric. The group action generated by this vector, however, has fixed points. This can be seen by computing the norm of $\pr_{\phi}$ using the metric:
\begin{align}
|\pr_{\phi}| &= \sqrt{g(\pr_{\phi},\pr_{\phi}) } \\
&= \sin^2 \theta.
\end{align}
The norm therefore vanishes at $\theta = 0$ and $\theta = \pi$, corresponding to the north and south poles. Geometrically, the orbits of the Killing vector define a circle fibration for $S^2$; the points at which the norm of the Killing vector vanish correspond to points at which the circle fibers degenerate. 

We will assume that $B = 0$. The T-dual metric follows from the Buscher rules, and is given by
\begin{align}
\label{S2dual}
\widehat{\dd s}^2 &= \dd \theta^2 + \frac{1}{\sin^2 \theta} \dd \hat{\phi}^2.
\end{align}
This space has curvature singularities at $\theta = 0$ and $\theta = \pi$, as can be verified by computing the Ricci scalar:
\begin{align*}
\widehat{\mathcal{R}} &= 2 \frac{\cos^2 \theta + 1}{\cos^2 \theta -1}
\end{align*}
This is totally expected, however, since T-duality acts on the fibers by inverting the radius, and at the poles the radius shrinks to zero. The space and its dual are shown in Figure \ref{S2abelian}. Note that the dual $B$-field is also zero. 
\begin{figure}[h!]
	\centering
	\begin{subfigure}{0.5\textwidth}
		\centering
		\includegraphics[width =0.6\linewidth]{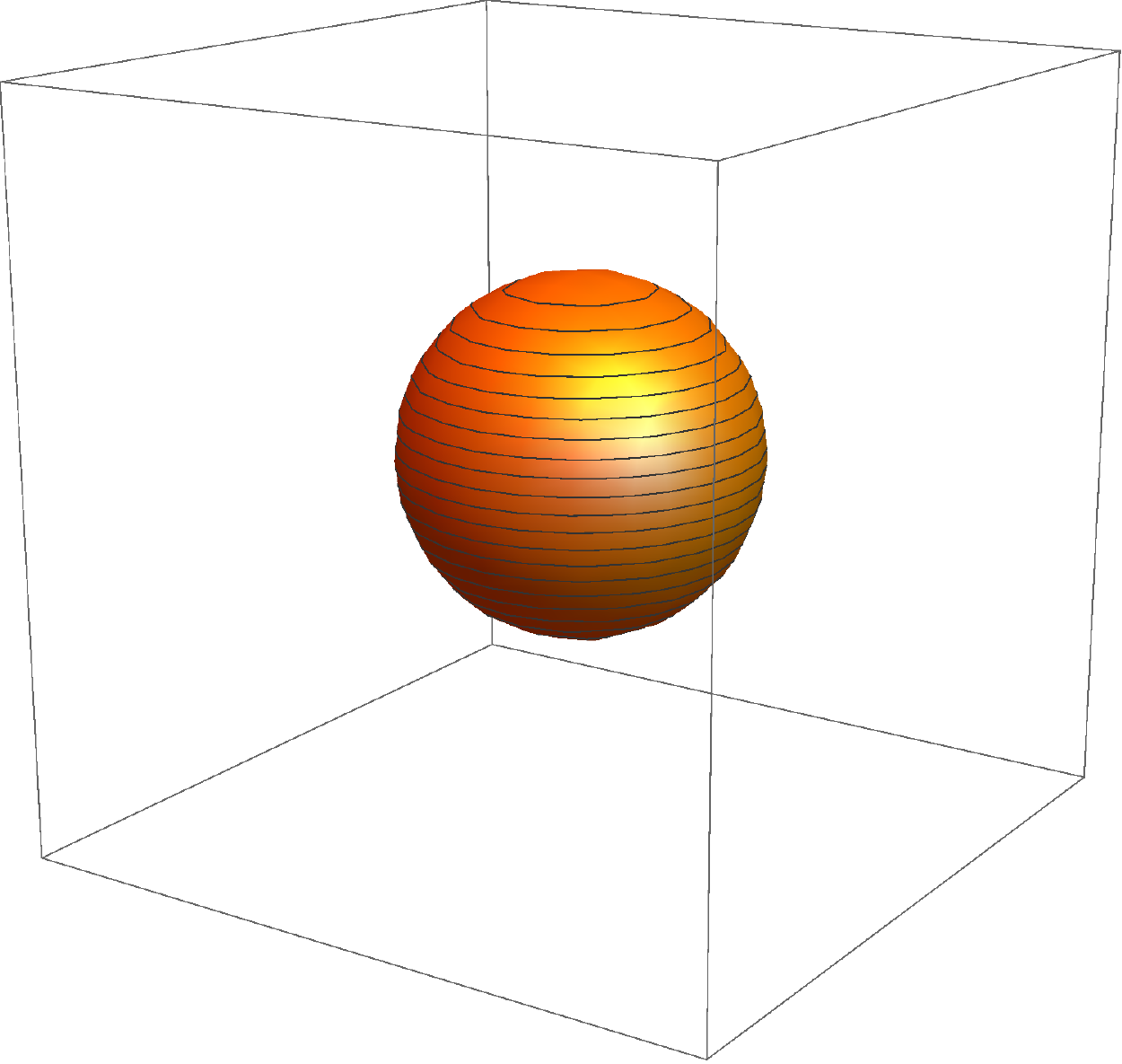}
	\end{subfigure}%
	\begin{subfigure}{0.5\textwidth}
		\centering
		\includegraphics[width=0.6\textwidth]{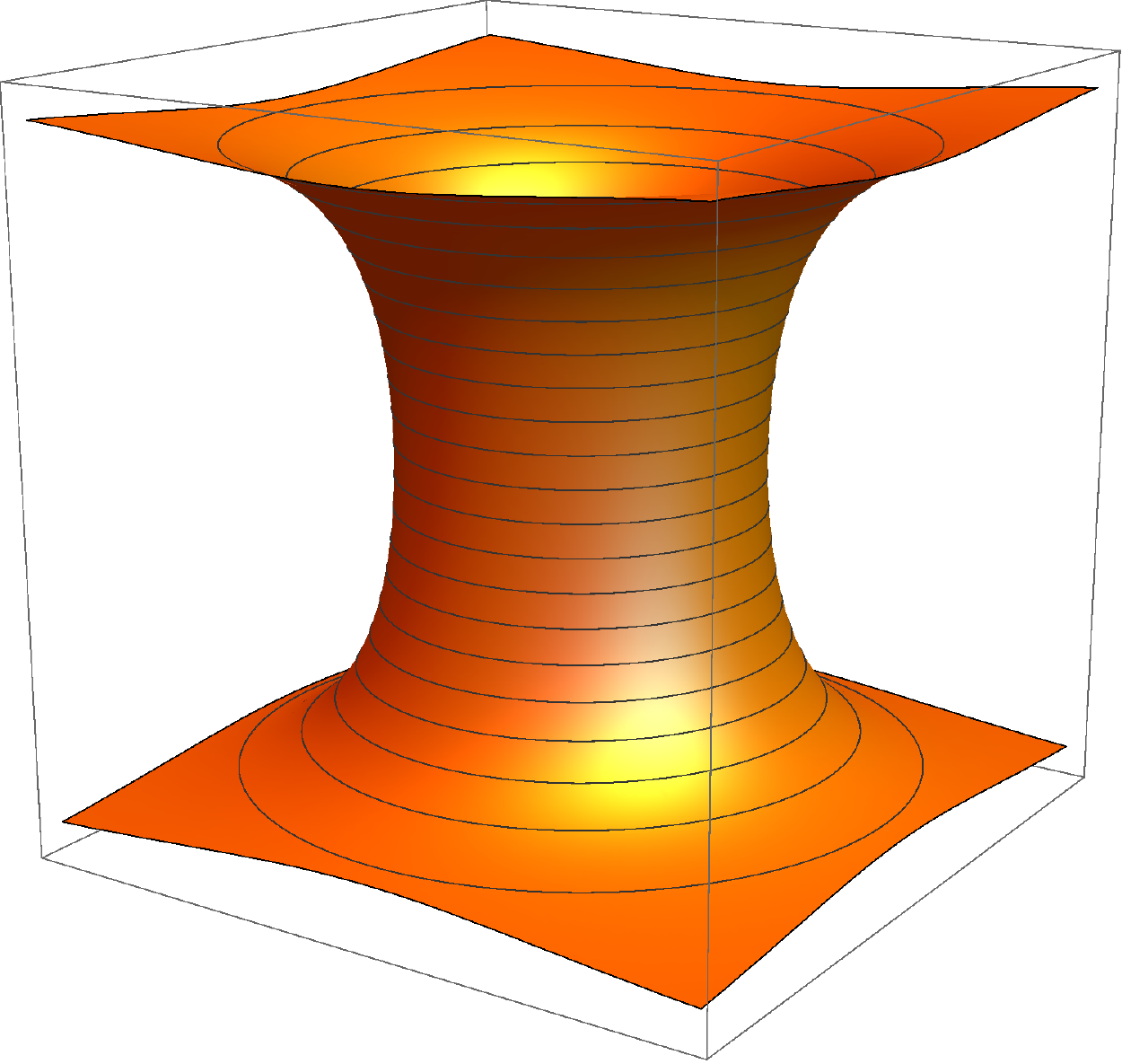}
	\end{subfigure}
	\caption{The round metric for $S^2$ has a $U(1)$ isometry corresponding to rotation around the $z$-axis. The action has fixed points at the north and south pole, which correspond to singularities in the dual space.}
	\label{S2abelian}
\end{figure}

A puzzle: The dual space seems to have a non-trivial fundamental group. In particular, $\pi_1(\widehat{E}) = \ZZ$, so the dual space should have integer winding modes, in addition to momentum modes. The sphere, however, has a trivial fundamental group: $\pi_1(S^2) = 0$. So although $S^2$ can have momentum modes, it cannot have winding modes.  If T-duality acts by interchanging momentum and winding modes, what happens to the momentum modes on $\widehat{E}$ when you perform a T-duality to get to $S^2$?

There are three possible resolutions to this paradox. The first possible resolution is that the argument which states that winding modes and momentum modes get interchanged under T-duality may be flawed - it may only hold for spacetimes of the form $M \times S^1$. The second possible resolution is that $S^2$ is not a valid string background, and that perhaps the interchange of momentum and winding only holds for valid string theory backgrounds. The final possible resolution hinges on the fact that a $B$-field is turned on in the dual space through this duality. When there is a non-zero $B$-field, the momenta conjugate to a circular coordinate couple to the $B$-field through the winding modes.\footnote{See, for example, Exercise 17.4 in \cite{Zwie}.} It is plausible that when there are no winding modes in the original space, the momentum/$B$-field coupling occurs in such a way that the observed momenta also vanish. 

\subsection{The effect on the curvature}
\label{sec:curv}
The Buscher rules give a transformation of the metric and the $B$-field. As we saw in Section \ref{subsec:abelgeoexamples}, quantities like the scalar curvature are not invariant under T-duality. The induced transformation on various geometric quantities is something we care about, so we include them here. The decomposition of various quantities used in this section follows, with some small changes, the notation of \cite{Haa96}, where they were used to derive consistency conditions for quantum corrections to T-duality.\footnote{This is discussed more in Section \ref{subsec:betafunctions}.} We assume we have an abelian isometry, and work in coordinates, $\{ X^{\mu} \} = \{X^i, \theta \}$, adapted to this isometry (so that the Killing vector is $k = \pr_{\theta}$). The metric decomposes in these coordinates as
\begin{align*}
g_{\mu \nu} = \begin{pmatrix}
g_{ij}   & g_{i \theta} \\
g_{\theta j} & g_{\theta \theta}
\end{pmatrix}
\end{align*}
Since $k$ is a Killing vector, we have $\Lie_k g = 0$, which implies that none of the components of this matrix depend on $\theta$. We now decompose the metric \`{a} la Kaluza Klein:
\begin{align}
\label{abelianKK}
g_{\mu \nu} = \begin{pmatrix}
\bar{g}_{ij} + e^{2 \sigma} A_i A_j  & e^{2 \sigma} A_i \\
e^{2 \sigma} A_j & e^{2 \sigma}
\end{pmatrix}
\end{align}	
corresponding to the following relabelling of fields: $g_{ij} = \bar{g}_{ij} + e^{2 \sigma} A_i A_j$, $g_{i \theta} = e^{2 \sigma} A_i$, and $g_{\theta \theta} = e^{2 \sigma}$. It follows that none of these quantities can depend on $\theta$ either.
In line element form, we have 
\begin{align*}
ds^2 = \bar{g}_{ij} \dd X^i \dd X^j + e^{2 \sigma} (\dd \theta + A_i \dd X^i)^2
\end{align*}
The $B$-field also decomposes, as
\begin{align}
\label{abelianKKB}
B_{\mu \nu} = \begin{pmatrix}
\overline{B}_{ij}  & B_i \\
-B_i & 0
\end{pmatrix}
\end{align}
In form notation, we have
\begin{align*}
B = \tfrac{1}{2} \overline{B}_{ij} \dd X^i \wedge \dd X^j +  B_i \dd X^i \wedge \dd \theta
\end{align*}
Applying the Buscher rules to the metric (\ref{abelianKK}) and B-field (\ref{abelianKKB}) gives us the following T-dual metric and B-field:
\begin{align*}
\hat{g}_{\mu \nu} &= \begin{pmatrix}
\bar{g}_{ij} + e^{-2 \sigma} B_i B_j  & e^{-2\sigma} B_i \\
e^{-2\sigma} B_j & e^{-2\sigma}
\end{pmatrix} \\[1em]
\widehat{B}_{\mu \nu} &= \begin{pmatrix}
\overline{B}_{ij}+B_i A_j -  A_i B_j  & A_i \\
-A_i & 0
\end{pmatrix}
\end{align*}
In line element and differential form notation, we have 
\begin{align*}
\widehat{ds}^2 &= \bar{g}_{ij} \dd X^i \dd X^j + e^{-2 \sigma} (\dd \hat{\theta} + B_i \dd X^i)^2\\[1em]
\widehat{B} &= \tfrac{1}{2}(\overline{B}_{ij} + B_i A_j - A_i B_j) \dd X^i \wedge \dd X^j +  A_i \dd X^i \wedge \dd \hat{\theta} .
\end{align*}
It's now easy to see that with this decomposition, the T-duality transformation rules correspond to the following simple map on the fields:
\begin{align}
\label{simpleBuscher}
\left( \begin{array}{cc}
A_i  \\
B_i \\
\sigma \\
\bar{g}_{ij} \\
\overline{B}_{ij}
\end{array}
\right)
\to
\left( \begin{array}{cc}
B_i \\
A_i \\
- \sigma \\
\bar{g}_{ij} \\
\overline{B}_{ij} + B_i A_j - A_i B_j
\end{array} \right).
\end{align}
By writing geometric quantities, like the Ricci scalar curvature, in terms of the fields appearing in (\ref{simpleBuscher}), the corresponding geometric quantities for the T-dual are then obtained by this simple transformation rule.
\begin{lemma}[Haagenson \cite{Haa96}]
	\label{KaluzaGeometry}
The metric (\ref{abelianKK}) and B-field (\ref{abelianKKB}) have the following geometric data:
\begin{itemize}
	\item Metric Determinant 
		\begin{align*} 
		\det (g) = e^{2 \sigma} \det{\bar{g}} 
		\end{align*}
	\item Inverse Metric
		\begin{align*}
		{g}^{\mu \nu} = \begin{pmatrix}
		 \bar{g}^{ij}& -A^i  \\
		 -A^j & e^{-2\sigma} + A_i A^i 
		\end{pmatrix}
		\end{align*}
	\item Ricci Tensor 
		\begin{align*} 
		\mathcal{R}_{0 0} &= -e^{2\sigma} \left[ \square \sigma + \nabla_i \sigma \nabla^i \sigma - \frac{e^{2\sigma}}{4} F_{ij} F^{ij} \right] \\[1em]
		\mathcal{R}_{i 0} &= A_i \mathcal{R}_{0 0} + 3e^{2\sigma} F_{ij} \nabla^j \sigma  + \frac{e^{2\sigma}}{2} \bar{\nabla}^j F_{ij} \\[1em]
		\mathcal{R}_{ij} &= \overline{\mathcal{R}}_{ij} +  A_i \mathcal{R}_{0 j} +A_j \mathcal{R}_{0 i} - A_i A_j \mathcal{R}_{0 0} - \bar{\nabla}_i \bar{\nabla}_j \sigma -  \bar{\nabla}_i \sigma  \bar{\nabla}_j \sigma - \frac{e^{2\sigma}}{2} F_{ik} F_j\,^k 
		\end{align*} 
	\item Torsion
		\begin{align*}
		H_{0 i j} &= - \pr_i B_j + \pr_j B_i = -G_{ij} \\[1em]
		H_{ijk} &= \pr_i \overline{B}_{jk}+\pr_j \overline{B}_{ki}+\pr_k \overline{B}_{ij}
		\end{align*}
		with other components vanishing. 
\end{itemize}
For ease of computation, we include the following:
\begin{align*}
H_{0 \mu \nu} H_0^{\mu \nu} &= G_{ij} G^{ij} \\[1em]
H_{0 \mu \nu} H_i^{\mu \nu} &= -2 G_{ij} G^{jk} A_k - H_{ijk} G^{jk} \\[1em]
H_{i \mu \nu} H_j^{\mu \nu} &= 2 \left( e^{-2\sigma} +A_m A^m \right) G_i\,^k G_{jk} - 2A^k A^m G_{ik} G_{jm} \\ &\quad + 2\left( H_{kmi} G_{j}\,^k A^m +H_{kmj} G_{i}\,^k A^m \right)+ H_{ikm} H_j\,^{km},
\end{align*}
where barred quantities refer to the metric $\bar{g}$, and on decomposed tensors all indices are raised and lowered with $\bar{g}_{ij}$ and its inverse.
\end{lemma}
\begin{lemma}
	\label{dualRicciScalar}
Let $\mathcal{R}$ be the Ricci scalar curvature of the original model. The Ricci scalar curvature of the T-dual model, $\widehat{\mathcal{R}}$, is given by
\begin{align*}
\widehat{\mathcal{R}} = \mathcal{R} + \frac{1}{4} e^{2 \sigma} \left( F_{ij} F^{ij} - e^{-4 \sigma} G_{ij}G^{ij} \right) + 4 \square \sigma.
\end{align*}
\end{lemma}
\begin{proof}
From Lemma \ref{KaluzaGeometry}, we can compute the Ricci scalar curvature of the original model in terms of the reduced quantities appearing in (\ref{simpleBuscher}). It is given by
\begin{align*}
\mathcal{R} = \overline{\mathcal{R}} -  \frac{1}{4} e^{2 \sigma} F_{ij} F^{ij} - 2 \square \sigma - 2 \bar{\nabla}^i \sigma \bar{\nabla}_i \sigma.
\end{align*}
The simple transformation properties of the reduced quantities then gives us 
\begin{align*}
\widehat{\mathcal{R}} = \overline{\mathcal{R}} - \frac{1}{4} e^{-2\sigma} G_{ij} G^{ij} + 2 \square \sigma - 2 \bar{\nabla}^i \sigma \bar{\nabla}_i \sigma,
\end{align*}
and the result follows from comparing these expressions.
\end{proof}


It is worth pausing for a moment to mention a few comments on the geometric interpretation of the quantities appearing here. Lemma \ref{dualRicciScalar} gives us a prescription for calculating the Ricci scalar of the dual, provided we know the Ricci scalar of the original space. In particular, if there is no $B$-field, and we normalise the length of the fibers, the Ricci scalar of the dual space will have the form $\widehat{\mathcal{R}} = \mathcal{R} + \frac{1}{4}F_{ij} F^{ij}$. On a Euclidean background, the quantity $F_{ij} F^{ij}$ is positive, so the Ricci scalar is non-decreasing in this situation.

This result has interesting consequences, since the existence of metrics with prescribed curvatures provides restrictions on the topology of the underlying manifold. In two dimensions, this is just a fancy way of restating the Gauss Bonnet theorem, but in other dimensions it becomes more interesting. Explicitly for dimensions 2 and 3 we have the following results:

\begin{theorem}[Gauss-Bonnet]
Let $(M,g)$ be a closed, two-dimensional Riemannian manifold with Ricci scalar curvature $\mathcal{R}$. Then
\begin{align}
\int_M \mathcal{R} = 4 \pi \chi (M),
\end{align}
where $\chi (M)$ is the Euler characteristic of $M$.
\end{theorem}
In particular, if a closed 2-manifold $M$ admits a metric a positive scalar curvature, then $\chi (M) >0$, and so $M$ must be either $S^2$ or $\mathbb{RP}^1$

\begin{theorem}[Schoen-Yau-Gromov-Lawson-Perelman]
\label{SchoenYau}
Let $M$ be a closed, orientable 3-manifold. Then $M$ admits a metric of positive scalar curvature if and only if it is a connected sum of spherical 3-manifolds and copies of $S^1 \times S^2$.
\end{theorem}



\subsection{Beta functions and generalised Ricci flow}
\label{subsec:betafunctions}

In string theory, conformal invariance is a crucial property of the theory. In a general quantum field theory, the failure of conformal invariance is measured by the trace of the energy-momentum tensor, also called the Weyl anomaly. The theory is locally conformally invariant if and only if the trace vanishes. The operator expression for the Weyl anomaly is given by
\begin{align*}
4 \pi \alpha' T^a_a &=  \pr^{\alpha} X^{\mu} \pr_{\alpha} X^{\nu} \bar{\beta}^g_{\mu \nu}   +\frac{\epsilon^{\alpha \beta}}{\sqrt{-h}}  \pr_{\alpha} X^{\mu} \pr_{\beta} X^{\nu} \bar{\beta}^B_{\mu \nu}   + \alpha' \mathcal{R}^{(2)}  \bar{\beta}^{\Phi}.
\end{align*}
The functions $\{\bar{\beta}^g_{\mu \nu}, \, {\bar{\beta}}^B_{\mu \nu} ,\, {\bar{\beta}}^{\Phi} \}$ are the Weyl anomaly coefficients. They are related to the $\beta$-functions of the quantum field theory:\footnote{See, for example,  \cite{Haa97,CT,CFMP}.}
\begin{align}
\label{Weylcoefficients}
\bar{\beta}^g_{\mu \nu} &= \beta^g_{\mu \nu} + 2 \alpha' \nabla_{\mu} \nabla_{\nu} \Phi  + \nabla_{(\mu} W_{\nu)}\\[1em]
\bar{\beta}^B_{\mu \nu} &= \beta^B_{\mu \nu} + \alpha' H_{\mu \nu}^{\ \ \lambda} \nabla_{\lambda} \Phi + \frac{1}{2} H_{\mu \nu}^{\ \ \lambda} W_{\lambda} + \nabla_{[ \,\mu} L_{\nu ]} \\[1em]
\bar{\beta}^{\Phi} &= \beta^{\Phi} + \alpha' (\nabla \Phi)^2 +\frac{1}{2}W^{\lambda} \nabla_{\lambda} \Phi.
\end{align}
The $L$ and $W$ functions vanish at the one-loop level, and we will ignore them in what follows. The $\beta$-functions $\{\beta^g_{\mu \nu}, \, \beta^B_{\mu \nu},\, \beta^{\phi}\}$ are the Renormalisation Group (RG) $\beta$-functions, determining how the couplings depend on a renormalisation cutoff scale $\mu$:
\begin{subequations}
\begin{align*}
\beta^g_{\mu \nu} &\equiv \mu \frac{\dd}{\dd \mu} g_{\mu \nu} \\[1em]
\beta^B_{\mu \nu} &\equiv \mu \frac{\dd}{\dd \mu} B_{\mu \nu} \\[1em]
\beta^{\Phi} &\equiv \mu \frac{\dd}{\dd \mu} \Phi.
\end{align*}
\end{subequations}
These functions determine the coefficients of a vector field on the (infinite dimensional) space of fields, the flow of which we refer to as the renormalisation group flow (RG flow). A conformal theory is a fixed point of this flow - that is, the quantum theory is conformal only if the beta functions vanish. For the string sigma model, the beta functions are given at the one-loop level by
\begin{subequations}
\label{betafunctions}
\begin{align}
\label{gbetafunction}
\beta^g_{\mu \nu} &= \alpha' \left( \mathcal{R}_{\mu \nu} - \frac{1}{4} H_{\mu \lambda \sigma} H_{\nu}^{\ \lambda \sigma} \right)  \\[1em]
\label{Bbetafunction}
\beta^B_{\mu \nu} &= -\frac{\alpha'}{2} \nabla^{\lambda} H_{\lambda \mu \nu} \\[1em]
\label{dilatonbetafunction}
\beta^{\Phi} &=  -\frac{\alpha'}{2} \left( \nabla^2 \Phi + \frac{1}{12} H_{\alpha \beta \gamma} H^{\alpha \beta \gamma}  \right).
\end{align}
\end{subequations}
Combining (\ref{betafunctions}) and (\ref{Weylcoefficients}), we get the expression for the Weyl anomaly coefficients at one-loop order:
\begin{subequations}
\label{Weyanomalyoneloop}
\begin{align}
\label{gWeylanomalyoneloop}
\bar{\beta}^g_{\mu \nu} &= \alpha' \left( \mathcal{R}_{\mu \nu} - \frac{1}{4} H_{\mu \lambda \sigma} H_{\nu}^{\ \lambda \sigma} + 2\nabla_{\mu} \nabla_{\nu} \Phi \right)  \\[1em]
\bar{\beta}^B_{\mu \nu} &= \alpha' \left( -\frac{1}{2} \nabla^{\lambda} H_{\lambda \mu \nu} + H_{\mu \nu \lambda} \nabla^{\lambda} \Phi \right) \\[1em]
\bar{\beta}^{\Phi} &= \alpha' \left( (\nabla \Phi)^2 -\frac{1}{2} \nabla^2 \Phi - \frac{1}{24} H_{\alpha \beta \gamma} H^{\alpha \beta \gamma}  \right).
\end{align}
\end{subequations}
Let us look at the one-loop Weyl anomaly coefficients for a moment, and consider the case where $B$ and $\Phi$ are both vanishing. Then (\ref{gWeylanomalyoneloop}) reduces to $\beta^g_{\mu \nu} = \mathcal{R}_{\mu \nu}$. The (one-loop) RG flow for this model is\footnote{The factor of $-2$ comes from rescaling $t$ and is not an essential feature. We choose this scaling here to make clearer contact with the literature on Ricci flow.}
\begin{align}
\label{RicciFlow}
\pr_t \, g_{\mu \nu}(t) = -2 \mathcal{R}_{\mu \nu}(t).
\end{align}
This is exactly the Ricci flow studied by mathematicians in the context of geometric analysis. It was first introduced by Hamilton in \cite{H82} to study the properties of three manifolds with positive scalar curvature. The dilaton version of this flow,
\begin{align*}
\pr_t \, g_{\mu \nu} = -2 \left(\mathcal{R}_{\mu \nu} + 2 \nabla_{\mu} \nabla_{\nu} \Phi \right),
\end{align*}
was later used by Perelman to prove the Thurston geometrization conjecture and the Poincar\'{e} conjecture \cite{P02,P03a,P03b}. 

To get a feeling for the effect the Ricci flow has on a Riemannian manifold, let us consider the simplest non-trivial example: the round sphere in $n$ dimensions. The round metric for a sphere of radius $r$ is simply $g_{\mu \nu} = r^2 g^{round}_{\mu \nu}$, where $g^{round}$ is the round metric for a sphere of radius 1. The Ricci tensor is given by $\mathcal{R}_{\mu \nu} = (n-1)g^{round}_{\mu \nu}$, since the unit sphere in any dimension is an Einstein manifold and the Ricci tensor is invariant under uniform scalings of the metric, and so the differential equation (\ref{RicciFlow}) has the form:
\begin{align}
\label{Ricciflowsphere}
\pr_t\, \left( r^2 g^{round} \right) = -2\left( n-1 \right) g^{round}.
\end{align}
Let us suppose that our solution has the form $g_{\mu \nu}(t) = r(t)^2 g^{round}_{\mu \nu}$. That is, the only dependence on $t$ is through the radius $r$. Once we find that this is a solution, then we can appeal to the proven uniqueness results to say this is the only solution. Substituting this ansatz into (\ref{Ricciflowsphere}) and integrating gives us
\begin{align*}
r^2 = r_0^2 - 2(n-1)t,
\end{align*}
with $r_0$ the radius of the sphere at time $t=0$. It follows that for increasing $t$ the sphere decreases in size, and shrinks to a point in finite time.\footnote{Note that the parameter $t$ refers to momentum, not time, when we are in the context of beta functions and quantum field theory.} 

In order that T-duality preserve conformal invariance, which is essential in string theory, it must map fixed points of the RG flow to fixed points. On a more general level, however, one could ask if T-duality is consistent with RG flow - that is, if we are given a solution $(g_t,\,B_t,\,\Phi_t)$ of (\ref{betafunctions}) with initial data $(g,\,B,\,\Phi)$, does the one-parameter family $(\hat{g}_t,\,\widehat{B}_t,\,\widehat{\Phi}_t)$ also satisfy (\ref{betafunctions}) with initial data $(\hat{g},\,\widehat{B},\,\widehat{\Phi})$?\footnote{The one-parameter family $(\hat{g}_t,\,\widehat{B}_t,\,\widehat{\Phi}_t)$ is defined by applying the T-duality transformations (\ref{Buscher}) and (\ref{dilatontransformation}) to $(g_t,\,B_t,\,\Phi_t)$ for each $t$.} Phrased another way, this question asks whether T-duality commutes with the RG flow. This question was first studied by Haagenson in \cite{Haa96},\footnote{This question was later studied in a more mathematical framework, and in the context of generalized geometry in \cite{Stre17}.} who derived a set of consistency conditions that the one-loop beta functions must satisfy in order for T-duality to be consistent with the RG flow. It is a remarkable fact that these consistency conditions are not only satisfied by the one-loop beta functions, (\ref{betafunctions}), but they are stringent enough to \emph{derive} the one-loop beta functions. That is, the only one-loop beta functions (up to a global constant) which are consistent with T-duality and the RG flow, are the ones given by (\ref{betafunctions}). 

The compatibility of T-duality and RG flow at the two-loop level has also been studied in, for example, \cite{HaaO}.


\subsection{T-duality and Sasaki-Einstein manifolds}
\label{subsec:SasakiEinstein}
Let us take a small digression here to discuss a nice result that follows from the Kaluza-Klein decomposition of the metric we have used in Section \ref{sec:curv}. 

Most of the toy models studied in this thesis are not actual string backgrounds - they are simplified versions of string backgrounds concocted in order to study particular aspects of string theory, such as how T-duality affects the topology of the spacetime. Authentic string backgrounds must satisfy a variety of mathematical conditions, coming from various physical assumptions. Assuming a string compactification of the form $M_4 \times K$, where $M_4$ is the four-dimensional Minkowski background, and $K$ is a compact six-dimensional background, the now-famous paper \cite{CHSW} asserts that $K$ must have holonomy contained in $SU(3)$, that is, $K$ must be a particular type of K\"{a}hler manifold known as a Calabi-Yau manifold. Depending on the dimension, the presence of fluxes, and the amount of supersymmetry preserved, these conditions can be weakened, so that one wants to consider more generalised types of manifolds such as nearly-K\"{a}hler, Sasaki-Einstein, 3-Sasakian, or weak $G_2$-manifolds \cite{AFHS}. From another perspective, the AdS/CFT correspondence is a relation between anti-de Sitter supergravity compactifications and certain conformal field theories on the boundary \cite{M97}. Arguably one of the most influential ideas in string theory, the AdS/CFT correspondence has generated significant interest in Sasaki-Einstein manifolds, where they arise as a large class of examples of the form $AdS_5 \times L$, where $L$ is a Sasaki-Einstein 5-manifold and the dual theory is a four-dimensional $\mathcal{N}=1$ superconformal field theory \cite{K98,KW98,AFHS,MRP}. Such backgrounds have been studied from a non-abelian T-duality perspective in \cite{ST14}.

In order to define a Sasakian manifold, we require the notion of a cone metric. Given a Riemannian manifold $(M,g)$, we define the \emph{cone metric} as
\begin{align*}
\widetilde{g}  = \dd r^2 + r^2 g. 
\end{align*}
This metric is defined on the cone manifold $C(M) = M \times \RR^{+}$, with the coordinate $r$ parametrising the $\RR^{+}$. We say that $(M,g)$ is a Sasakian manifold if the cone metric $\widetilde{g}$ is K\"{a}hler. Since a K\"{a}hler manifold is even-dimensional, a Sasakian manifold must be odd-dimensional. If, in addition, the cone metric is Ricci-flat, then we call $(M,g)$ a Sasaki-Einstein manifold. If the cone of a Sasakian manifold is hyper-K\"{a}hler, then we refer to $(M,g)$ as a 3-Sasakian manifold. 
The canonical examples of Sasaki manifolds are the odd-dimensional round spheres. The cone of $S^{2n-1}$ is $\CC^n$, equipped with the flat metric and standard complex structure.  

Two important objects for Sasakian manifolds, from the geometric and topological perspective, are the homothetic vector field, and the associated Reeb vector fields. The homothetic vector field is $r \pr_r$, and the associated Reeb vector field $\xi$ is given by 
\begin{align*}
\xi = -J(r \pr_r),
\end{align*}
where $J$ is the complex structure on the cone (which always exists since the cone is K\"{a}hler). The Reeb vector field defines a nowhere vanishing vector field on the cone, and hence defines a foliation of $C(M)$ into the orbits of $\xi$, called the Reeb foliation. If the orbits close, then $\xi$ induces a $U(1)$ action on $C(M)$. If this action is free, the quotient space is a K\"{a}hler manifold and we refer to $g$ as a regular Sasakian metric. Otherwise the quotient is a K\"{a}hler orbifold and $g$ is referred to as a quasi-regular metric. The Reeb vector field is also interesting from a physics perspective. The symmetry generated by the Reeb vector field corresponds, in the dual field theory, to R-symmetry. In addition, the volume of the Sasaki-Einstein manifold corresponds to the central charge, $a$, of the dual CFT, and volume minimisation on the gravity side corresponds to $a$-maximisation on the field theory side \cite{MSY06,MSY08}. Because of the interest in string theory for Sasaki-Einstein manifolds, it is natural to ask how they interact with T-duality. Answering that question is the primary purpose of this section.  

It is a well-known result, following from a simple calculation, that the Ricci scalar of the cone metric is related to the Ricci scalar of the base. 
\begin{lemma}
\label{ConeRicci}
Let $(M,g)$ be a Sasakian manifold of dimension $2D-1$, and let $X^M = (X^i,X^r)$ denote coordinates on the cone. The Ricci curvature tensor of the cone metric, $\widetilde{R}_{MN}$, decomposes as
\begin{align*}
\widetilde{R}_{MN} = \left( 
\begin{matrix}
\widetilde{R}_{ij} & \widetilde{R}_{i r} \\
\widetilde{R}_{rj} & \widetilde{R}_{rr}
\end{matrix}
\right),
\end{align*}
where the components can be written in terms of the base metric, $g_{ij}$, and the Ricci curvature tensor on the base, $R_{ij}$:
\begin{align*}
\widetilde{R}_{ij} &= R_{ij} - 2(D-1) g_{ij} \\[1em]
\widetilde{R}_{ir}  &= 0 \\[1em]
\widetilde{R}_{rj} &= 0 \\[1em]
\widetilde{R}_{rr} &= 0. 
\end{align*}
\end{lemma}
\begin{proof}
This proof reduces to a straightforward calculation. The metric components decompose as 
\begin{align*}
\widetilde{g}_{MN} = \left( 
\begin{matrix}
\widetilde{g}_{ij} & \widetilde{g}_{i r} \\
\widetilde{g}_{rj} & \widetilde{g}_{rr}
\end{matrix}
\right),
\end{align*}
where 
\begin{align*}
\widetilde{g}_{ij} &= r^2 g_{ij} \\[1em]
\widetilde{g}_{ir}  &= 0 \\[1em]
\widetilde{g}_{rj} &= 0 \\[1em]
\widetilde{g}_{rr} &= 1. 
\end{align*}
We can then calculate the Christoffel symbols of the cone in terms of the Christoffel symbols of the base:
\begin{align*}
\widetilde{\Gamma}^i_{jk} &= \Gamma^i_{jk} \\[1em]
 \widetilde{\Gamma}^i_{rj} = \widetilde{\Gamma}^i_{jr} &= \frac{1}{r} \delta^i_j  \\[1em]
\widetilde{\Gamma}^r_{jk} &= -r g_{jk} \\[1em]
 \widetilde{\Gamma}^r_{ri} = \widetilde{\Gamma}^r_{ir} &= 0 \\[1em]
\widetilde{\Gamma}^r_{rr} &= 0. 
\end{align*}
The result then follows from the formula for the Ricci curvature tensor:
\begin{align*}
\widetilde{R}_{MN} = \pr_K \widetilde{\Gamma}^K_{NM} - \pr_{N} \widetilde{\Gamma}^K_{KM} + \widetilde{\Gamma}^K_{KL} \widetilde{\Gamma}^L_{NM} - \widetilde{\Gamma}^K_{NL} \widetilde{\Gamma}^L_{KM}.
\end{align*}
\end{proof}
\begin{corollary}
Let $(M,g)$ be a Sasaki-Einstein manifold of dimension $2D-1$. Then $M$ is an Einstein manifold with $\lambda = 2(D-1)$.
\end{corollary}
\begin{proof}
Since a Sasaki-Einstein manifold is a Riemannian manifold whose cone is Ricci flat (and K\"{a}hler), this follows immediately from Lemma \ref{ConeRicci}.
\end{proof}
Suppose now we have a manifold, $E$, with metric and $B$-field, $(g,B)$ which we can T-dualise, giving us a dual metric and $B$-field $(\hat{g},\widehat{B})$ on the dual manifold $\widehat{E}$. On both $E$ and $\widehat{E}$, we can construct the cone manifold and cone metric. A natural question to ask is whether the cone metrics of the original space and the dual space are related by a T-duality transformation. That is, is there a T-duality map which makes the following diagram commute?
\begin{equation}
\begin{tikzpicture}[baseline=(current  bounding  box.center)]
\node (CE) {$C(E)$};
\node (CEh) [right=2cm of CE] {$C(\widehat{E})$};
\node (E) [below=2cm of CE] {$E$};
\node (Eh) [below=1.9cm of CEh] {$\widehat{E}$};

\draw[<->,dashed] (CE) to node {$T-duality?$} (CEh);
\draw[->] (E) to node {$ $} (CE);
\draw[<->] (E) to node [swap] {$T-duality$} (Eh);
\draw[->] (Eh) to node {$ $} (CEh);

\end{tikzpicture}
\end{equation}
A few moments thought reveals that the answer to this question is no. The component of the metric corresponding to the isometry direction can't possibly be the same in $C(E)$ and $C(\widehat{E})$.  On $E$, the component along the isometry direction is denoted by $g_{\theta \theta}$, and the Buscher rules tell us that the corresponding component in $\widehat{E}$ is $\frac{1}{g_{\theta \theta}}$. If we conify this we get $\frac{r^2}{g_{\theta \theta}}$. On the other hand, if we conify first we get the component $r^2 g_{\theta \theta}$ in $C(E)$. Performing a T-duality on this will give us $\frac{1}{r^2 g_{\theta \theta}}$.

Since we have a decomposition of the geometric quantities for a manifold with circular dimensions \`{a} la Kaluza-Klein, as well as a natural decomposition of the geometric quantities for Sasakian manifolds, we can ask what we can say about manifolds that have both structures. To that end, suppose that $(E,g)$ is a Sasakian manifold of dimension $2D-1$, and suppose also that $E$ is the total space of a circle bundle $\pi:E \to M$.\footnote{This notation will be described more in Section \ref{subsec:TopCircBund}.} The spaces fit into the following diagram:
\begin{center}
	\begin{tikzpicture}
	\node (CE) {$C(E)$};
	\node (E) [below=1cm of CE] {$E$};
	\node (M) [below=1cm of E] {$M$};
	\node (dCE) [right=2cm of CE] {$\textrm{dim}(C(E)) = 2D$};
	\node (dE) [below=1cm of dCE] {$\textrm{dim}(E) = 2D-1$};
	\node (dM) [below=1cm of dE] {$\textrm{dim}(M) = 2D-2$};
	\draw[->] (CE) to node  {$ $} (E);
	\draw[->] (E) to node [right] {$ $} (M);
	\end{tikzpicture}
\end{center}
Choose coordinates $X^{\mu} = (X^i, X^{\theta})$ for $E$, and suppose that the metric on $E$ has the Kaluza-Klein decomposition of (\ref{abelianKK}):
\begin{align*}
g_{\mu \nu} = \left( \begin{matrix}
\bar{g}_{ij}+ A_i A_j & A_i \\
A_j & 1
\end{matrix} \right).
\end{align*}
Note that we have rescaled the fibers to have constant length 1 (that is, set $\sigma = 0$). This will affect the geometry, but not the topology of the space. The metric on the cone then decomposes, in the coordinates $X^M = (X^i, X^\theta, X^r)$, as
\begin{align}
\label{conecirclemetric}
\widetilde{g}_{MN} = \left( 
\begin{matrix}
r^2 \bar{g}_{ij} + r^2 A_i A_j & r^2 A_i & 0  \\
r^2 A_j & r^2 &0 \\
0 & 0 & 1
\end{matrix}
\right).
\end{align}
The Ricci curvature tensor of the cone, $\widetilde{R}_{MN}$, decomposes from Lemma \ref{ConeRicci} into geometric quantities on $E$. Geometric quantities on $E$, however, decompose by Lemma \ref{KaluzaGeometry} into geometric quantities on $M$, so we should be able to write the Ricci curvature tensor of the cone entirely in terms of geometric quantities on the base. This is the content of the next lemma. 
\begin{lemma}
\label{sasakiRicci}
Let $(E,g)$ be a Sasakian manifold, and suppose that $E$ is also the total space of a circle bundle $\pi:E \to M$.  Then the cone metric (\ref{conecirclemetric}) has the following Ricci curvature tensor:
\begin{align*}
\widetilde{R}_{MN} = \left( \begin{matrix}
\widetilde{R}_{ij} & \widetilde{R}_{i \theta} & \widetilde{R}_{i r} \\
\widetilde{R}_{\theta j} & \widetilde{R}_{\theta \theta} & \widetilde{R}_{\theta r} \\
\widetilde{R}_{r j} & \widetilde{R}_{r \theta} & \widetilde{R}_{rr}
\end{matrix} \right),
\end{align*}
where 
\begin{subequations}
\label{sasakicircleconegeometry}
\begin{align}
\widetilde{R}_{ij} &= \tfrac{1}{4} F_{mn} F^{mn} - 2(D-1)\\[1em]
\widetilde{R}_{\theta i} = \widetilde{R}_{i \theta} &= \tfrac{1}{4} A_i F_{mn}F^{mn} + \tfrac{1}{2} \nabla^{m} F_{im} - 2(D-1)A_i \\[1em]
\widetilde{R}_{\theta \theta} &= \overline{R} + \tfrac{1}{4} A_i A_j F_{mn} F^{mn} + \tfrac{1}{2} A_i \nabla^m F_{jm} +\tfrac{1}{2}A_j \nabla^m F_{im} \notag \\
&\qquad  - \tfrac{1}{2} F_{ik} F_j \!^k - 2(D-1) \bar{g}_{ij} - 2 (D-1) A_i A_j.
\end{align}
\end{subequations}
If, in addition, $(E,g)$ is Sasaki-Einstein, then we have the following identities:
\begin{subequations}
\label{sasakieinsteinconegeometry}
\begin{align}
\label{interesting}
\tfrac{1}{4} F_{mn} F^{mn} &= 2 (D-1) \\[1em]
\nabla^m F_{im} &= 0 \\[1em]
\overline{R}_{ij} &= \tfrac{1}{2} F_{ik} F_j \!^k + 2(D-1) \bar{g}_{ij}.
\end{align}
\end{subequations}
\end{lemma}
\begin{proof}
Equations (\ref{sasakicircleconegeometry}) follow by combining the results of Lemma \ref{ConeRicci} and Lemma \ref{KaluzaGeometry}. Equations (\ref{sasakieinsteinconegeometry}) then follow directly from (\ref{sasakicircleconegeometry}) by setting the left hand sides to zero, as is required for a Sasaki-Einstein manifold. 
\end{proof}
This lemma has the following interesting corollary.
\begin{corollary}
\label{notSasakiEinstein}
Let $E = M \times S^1$, equipped with a  product metric $g = \bar{g} + \dd \theta^2$. Then $(E,g)$ is a Sasaki-Einstein manifold if and only if $M$ is a point.
\end{corollary}
\begin{proof}
If $M$ is a point, then $E = S^1$, which is Sasaki-Einstein. For the converse direction, suppose that $E = M \times S^1$ where the product metric is Sasaki-Einstein, and apply the results of Lemma \ref{sasakiRicci}. For this metric, we have $A_i = 0$, and so $F_{ij} = 0$. It follows from (\ref{interesting}) that $D=1$, and therefore $\textrm{dim}(E) = 1$. 
\end{proof}
A direct result of this corollary is that T-duality does not preserve the property of being a Sasaki-Einstein manifold. As we will see in Section \ref{subsec:TopCircBund}, the T-dual of a spacetime with no flux is always a trivial bundle. That is, if we begin with a Sasaki-Einstein manifold (of dimension $D>1$) with no flux, the T-dual will be a trivial bundle which, by Corollary \ref{notSasakiEinstein}, cannot be Sasaki-Einstein. The T-duality between $S^3$ with no flux (which is Sasaki-Einstein) and $S^2 \times S^1$ with one unit of flux (which is not Sasaki-Einstein) provides a concrete realisation of this statement.\footnote{This T-duality example is discussed from the geometric point of view in Section \ref{S3Buscher}, and from a topological point of view in \ref{S3Topology}.}
\subsection{The Ramond-Ramond fluxes}
\label{RRfluxes}
The transformation rules that we have discussed so far have related only to the massless bosonic fields of type II string theory - the so-called NS-NS fields. In order to discuss T-duality as a symmetry of the full string theory, one should include a discussion of the RR fluxes. This was first done for abelian T-duality in \cite{BHO95} from a target space perspective, and later given a worldsheet derivation in \cite{CLPS,KR00,BJ03}. There are also approaches using the pure spinor formalism \cite{BPT08}, and as a canonical transformation \cite{SST10}. We will not include a thorough discussion of the RR fluxes in this chapter - we wish to only recall some basic facts about the fluxes which will be pertinent to our discussion in Section \ref{TopKtheory}. A description of how the RR fields transform under T-duality will can be found in \cite{ST11}, where the transformation of the RR fluxes under non-abelian T-duality is given (which includes, as a special case, abelian T-duality). 

Recall that the RR fluxes are objects to which a D$p$-brane couples. They are differential $(p+2)$-forms, $F_{p+2}$. For type IIA supergravity $p$ is even, and for type IIB supergravity $p$ is odd. In terms of potentials, we have
\begin{align*}
F_{p+2} &= \dd C_{p+1} + H \wedge C_{p-1}.
\end{align*}
These are not closed, but satisfy the Bianchi identities
\begin{align*}
\dd F_{p+2}  - H \wedge F_{p} &= 0. 
\end{align*}

\subsubsection{D-branes, cohomology, and K-theory}
When the $H$-flux is zero, the RR fluxes are closed, and therefore determine classes in the (de-Rham) cohomology of spacetime. In fact, quantum consistency requires these to live in cohomology with integral coefficients. When the $H$-flux is non-zero, however, the fluxes are not closed, and therefore do not determine classes in ordinary cohomology. We can assemble the fluxes into a polyform, $G$,
\begin{align*}
G &= \sum_{p} F_{p+2}
\end{align*}
which is closed under the twisted differential, $d_H$:
\begin{align*}
d_H G :=  \dd G -  H \wedge G = 0.
\end{align*}
Since $H$ is a closed 3-form, we have
\begin{align*}
(d_H)^2 \alpha &= (\dd - H\wedge)(\dd \alpha - H \wedge \alpha) \\[1em]
&= \dd\, (\dd \alpha) - \dd\, (H \wedge \alpha) - H \wedge \dd \alpha + H \wedge H \wedge \alpha \\[1em]
&= -\dd H \wedge \alpha + H \wedge \dd \alpha - H \wedge \dd \alpha \\[1em]
&= 0.
\end{align*}
The twisted differential is nilpotent, and determines its own cohomology theory on the $\ZZ_2$-graded de Rham complex $\Omega^{\bullet}$.  

RR fluxes are sourced by D-branes, and so it would seem that D-brane charge is therefore classified by the (twisted) cohomology of spacetime. On the other hand, it has been argued that D-brane charge should actually live in the K-theory of spacetime, rather than the cohomology \cite{MM97,Wit98}. This argument, following from Sen's conjecture that all brane configurations can be obtained from stacks of space-filling $D_9$ and anti-$D_9$ branes by tachyon condensation, proceeds as follows. Consider a system of $n$ coincident $D_9$-branes, and $m$ coincident anti-$D_9$-branes. The zero mass states in the spectrum of the configuration of $D_9$-branes at low energies is a set of interacting quantum fields which give rise to a $U(n)$ gauge theory. We label such a configuration by $(\mathscr{E},\mathscr{F})$, where $\mathscr{E}$ is a $U(n)$ complex gauge bundle carried by the D-branes and $\mathscr{F}$ is a $U(m)$ complex gauge bundle carried by the anti-D-branes. The invariance of D-brane charge under tachyon condensation leads us to regard as equivalent the configurations $(\mathscr{E},\mathscr{F})$ and $(\mathscr{E} \oplus \mathscr{H}, \mathscr{F} \oplus \mathscr{H})$, since the configuration $(\mathscr{H},\mathscr{H})$ can annihilate to the vacuum. Thus the allowed configurations, and therefore the conserved D-brane charges, are classified by pairs of (complex) vector bundles over the spacetime $M$, modulo the above equivalence relation. This is the very definition of the (topological) K-theory of the spacetime $M$. 

In the presence of a non-trivial $H$-flux, this classification was shown to be inconsistent, and instead the appropriate generalised cohomology theory was \emph{twisted} K-theory \cite{Kap99} (see also \cite{Wit98}). If T-duality is an honest symmetry of string theory, one would expect that T-dual spaces have the same sets of D-brane charges. We expect, therefore, that T-duality should determine an isomorphism in twisted K-theory. Such an isomorphism was one of the main achievements of the topological description of T-duality \cite{BEM04}. In Section \ref{sec:torusbundles} and Section \ref{Cstar} we shall see that performing multiple T-dualities can take us outside the realm of manifolds - we encounter non-commutative or even non-associative spaces! Such spaces have natural descriptions in terms of algebras, and the appropriate generalisation of K-theory to this setting is algebraic K-theory. As we shall see, topological T-duality in this context also includes a statement of the algebraic K-theories. 

We finally mention that there are arguments that twisted K-theory is not, in fact, the correct generalised cohomology theory to describe D-brane charge. This argument is based on the lack of covariance with S-duality \cite{DMW} (see also \cite{Evs06} for a very readable review). We will not address this open problem.


\subsection{Multiple abelian T-dualities}
\label{subsec:MultipleAbelian}
The Buscher rules generalise straightforwardly to multiple abelian isometries. We find it instructive to explicitly compute these transformation rules, following \cite{B}. We take the action of (\ref{stringNLSM}), and use conformal invariance of the sigma model to choose a flat metric $h_{\alpha \beta} = \eta_{\alpha \beta}$ for the worldsheet. Switching to complex coordinates $(z,\bar{z})$, instead of $(\sigma, \tau)$, and defining $E_{MN} = g_{MN} + B_{MN}$, we have 
\begin{align}
\label{multabel}
S = \fpa \int_{\Sigma} \ddd^2 \!z \, \, E_{MN} \, \pr X^M \bar{\pr} X^N.
\end{align}
As before, we will ignore the dilaton in what follows. We now assume that the action has a $U(1)^n$ isometry, that is, there are $n$ commuting Killing vectors. We choose coordinates $(X^{M}) = (X^{\mu}, X^{m})$ adapted to these Killing vectors, and decompose $E$ as
\begin{align*}
E_{MN} = \begin{pmatrix}
E_{\mu \nu}   & E_{\mu n} \\
E_{m \nu} & E_{m n}
\end{pmatrix}
\end{align*} 
The infinitesimal variations of the coordinates are given by
\begin{align*}
\de X^{\mu} &= 0 \\
\de X^{m} &= \epsilon^m.
\end{align*}
As before, we can promote this global symmetry to a local one by gauging. We introduce abelian gauge fields $\mathcal{A}^m = A^m \dd z + \bar{A}^m \dd \bar{z}$, and minimally couple them to the fields $X^m$ by introducing the gauge covariant derivatives $\mathscr{D}X^m = \dd X^m - \mathcal{A}^m$. Since $\dd X^m = \pr X^m \dd z + \bar{\pr} X^m \dd \bar{z}$, this corresponds to the replacement:
\begin{align*}
\pr X^m &\to DX^m = \pr X^m - A^m \\
\bar{\pr} X^m &\to \bar{D} X^m = \bar{\pr} X^m - \bar{A}^m.
\end{align*}
The $U(1)^n$ curvature is $\mathcal{F}^m = \dd \mathcal{A}^m = \tfrac{1}{2} \left( \pr \bar{A}^m - \bar{\pr} A^m \right) \textrm{d}^2 z$, and introducing the term $\frac{1}{2\pi} \int_{\Sigma} \widehat{X}_m \mathcal{F}^m $ to the minimally coupled action, we obtain the gauged action
\begin{align}
\label{multabelGauged}
S_G = \fpa &\int_{\Sigma} \ddd^2 \! z \bigg[ E_{\mu \nu} \pr X^{\mu} \bar{\pr} X^{\nu} + E_{\mu n} \pr X^{\mu} \bar{D} X^n + E_{m \nu} D X^m \bar{\pr} X^{\nu}   \nonumber \\
 & \qquad \qquad + E_{mn} D X^m \bar{D} X^n + \widehat{X}_m \left( {\pr} \bar{A}^m - \bar{\pr} A^m \right) \bigg]. 
\end{align}
The local gauge transformations leaving this action invariant are 
\begin{align*}
\delta_{\epsilon} X^{\mu} &= 0 \\
\delta_{\epsilon} X^m &= \epsilon^m\\
\delta_{\epsilon} \cA^m &= \dd \epsilon^m \\
\delta_{\epsilon} \widehat{X} &= 0,
\end{align*}
corresponding to the following transformations:
\begin{align*}
X^{\mu} &\to X^{\mu} \\
X^m &\to X^m + \epsilon^m \\
A^m &\to A^m + \pr \epsilon^m \\
\bar{A}^m &\to \bar{A}^m + \bar{\pr} \epsilon^m.
\end{align*}
We first note that we can integrate out the auxilliary coordinates in the action (\ref{multabelGauged}) and gauge fix to obtain the original action (\ref{multabel}). Explicitly, the Euler-Lagrange equations for $\widehat{X}_m$ are 
\begin{align*}
\frac{\pr \mathcal{L}}{\pr \widehat{X}_m} = \bar{\pr} A^m - \pr \bar{A}^m = 0.
\end{align*}
Solving this for $\mathcal{A}^m = A^m \dd z + \bar{A}^m \dd \bar{z}$ gives $\mathcal{A}^m = \dd \Lambda^m = \pr \Lambda^m \dd z + \bar{\pr} \Lambda^m \dd \bar{z}$, or in other words, $A^m = \pr \Lambda^m$ and $\bar{A}^m = \bar{\pr} \Lambda^m$. We then substitute this into the gauged action and use gauge invariance to set $A^m$ and $\bar{A}^m$ to zero, thus recovering the original model (\ref{multabel}).
On the other hand, we can integrate out the gauge fields $A^m$ and $\bar{A}^m$. The Euler-Lagrange equation for $A^m$ gives an equation containing $\bar{A}^m$, and similarly, the Euler-Lagrange equation for $\bar{A}^m$ gives an equation containing $A^m$. Solving these equations for the corresponding variable gives
\begin{align*}
A^m &= \pr X^m + E_{\mu n} (E^{-1})_{nm} \pr X^{\mu} + (E^{-1})_{nm} \pr \widehat{X}_n \\[1em]
\bar{A}^n &= \bar{\pr} X^n + (E^{-1})_{ns} E_{s \nu} \bar{\pr} X^{\nu} - (E^{-1})_{ns} \bar{\pr} \widehat{X}_s.
\end{align*}
These can then be inserted into the gauged action (\ref{multabelGauged}). With a bit of algebra, integration by parts, and some patience, one arrives at the following action:
\begin{align*}
\widehat{S} = \fpa \int_{\Sigma} \ddd^2 \!z &\bigg[ (E_{\mu \nu} - E_{\mu m} (E^{-1})_{mn} E_{n \nu}) \pr X^{\mu} \bar{\pr} X^{\nu} + E_{\mu m} (E^{-1})_{mn} \pr X^{\mu} \bar{\pr} \widehat{X}_n  \\
 &\quad - \, (E^{-1})_{mn} E_{n \nu} \pr \widehat{X}_m \bar{\pr} X^{\nu} + (E^{-1})_{mn} \pr \widehat{X}_m \bar{\pr} \widehat{X}_n \bigg].
\end{align*}
This new action has an interpretation of a non-linear sigma model defined on the coordinates $(X^{\mu}, \widehat{X}_m)$, with new fields $\widehat{E}_{MN}$, expressed in terms of the old fields $E_{MN}$ as
\begin{align}
\label{multipleBuscher}
\left( \begin{array}{cc}
\widehat{E}_{\mu \nu} & \widehat{E}_{\mu n} \\[1em]
\widehat{E}_{m \nu} & \widehat{E}_{mn}
\end{array}
\right)
=
\left( \begin{array}{cc}
E_{\mu \nu} - E_{\mu m} (E^{-1})_{mn} E_{n \nu} & E_{\mu m} (E^{-1})_{mn} \\[1em]
-(E^{-1})_{mn} E_{n \nu} & (E^{-1})_{mn}
\end{array} \right).
\end{align}
To obtain the dual metric and the dual $B$-field, we can simply extract the symmetric and antisymmetric components of $\widehat{E}_{MN}$:
\begin{align}
\label{gBfromE}
\widehat{g}_{MN} &= \tfrac{1}{2} \left( \widehat{E}_{MN} + \widehat{E}_{NM} \right) \\[1em]
\widehat{B}_{MN} &= \tfrac{1}{2} \left( \widehat{E}_{MN} - \widehat{E}_{NM} \right).
\end{align}
The transformation rules (\ref{multipleBuscher}) are the Buscher rules for multiple abelian isometries, and taking $n=1$ recovers the original Buscher rules (\ref{Buscher}) for a single abelian isometry. We note that the expression $\widehat{E}_{\mu \nu} = E_{\mu \nu} - E_{\mu m}\left( E^{-1} \right)_{mn} E_{n \nu}$ is simply the Schur complement, $E / E_{mn}$, of the original matrix, where the Schur complement, $M/D$ of a block matrix 
\begin{align*}
M &= 
\left( \begin{array}{cc}
A & B \\[1em]
C & D
\end{array}
\right)
\end{align*} 
is defined by 
\begin{align*}
M / D &= A - BD^{-1}C.
\end{align*}

\section{Topology}
\label{sec:AbelTopology}

The Buscher rules are an inherently local set of transformation rules. They describe transformation rules for the metric and $B$-field written in a particular set of local coordinates. Given this local description, it is interesting to ask if we can say anything about this procedure from a global, i.e. topological, perspective. In particular, what is the topology of the dual space? Since the dual space is comprised of the base space $M$, together with the fibers described by the Lagrange multiplier coordinates, there are two components to this question: what is the topology of the fibers, and how are these fibers patched together globally over the base?

The answer to the first question is included in the next section. The short summary is that in string theory, the topology of the fibers must be preserved under abelian T-duality. The second part of the question is answered in Section \ref{subsec:TopCircBund}. 


\subsection{Topology of the fibers}
\label{subsec:FibTop}
The Buscher procedure relies on the existence of a Killing vector field $k$ on the original spacetime $M$. The flow of this vector field generates a one-parameter group of diffeomorphisms of $M$, and since $k$ is Killing, these diffeomorphisms are isometries. The orbit of a single point under this group of diffeomorphisms is a one-dimensional, immersed submanifold - either $\mathbb{R}$ or $S^1$. From a classical perspective, there is nothing which constrains these orbits. Indeed, as a solution generating technique in supergravity, one can dualise along non-compact directions (see for example, Section \ref{subsec:abelgeoexamples}). We shall see, however, that higher genus considerations constrain the periodicities of the dual coordinates. This argument, which we will now reproduce, was first done in \cite{RV}, although \cite{AABL} and \cite{Tphd} also have very readable reviews. 

In string theory we usually want our original spacetime to be compact, so we often assume that the orbit of each point is compact - that is, the orbit of each point is homeomorphic to $S^1$. This is equivalent to asking that the $\mathbb{R}$ action of the isometry on $M$ descends to a $U(1)$ action. In practical terms this means that in the adapted coordinates $\{X^i, \, \theta\}$ we have chosen, in which the Killing vector $k$ is simply $k = \pr_{\theta}$, the coordinate $\theta$ is periodic. 

Let us now be more careful about topological considerations. In Section \ref{BuscherRules} we carried out the Buscher procedure to gauge the isometries of a non-linear sigma model. Once we had the gauged model, the equation of motion for the Lagrange multiplier forced the curvature $\mathcal{F} = \dd \mathcal{A}$ of the gauge field to vanish. We concluded from this that $\mathcal{A} = \dd \chi$, and so there was an appropriate gauge transformation setting $\cA$ to zero, thus recovering the original model. The key observation here is that although that argument always holds locally, this conclusion is invalid globally in a non-trivial topology. More precisely, the cohomology of a manifold determines whether a differential form on the manifold can be closed but not exact. If $\cA$ is a $U(1)$ gauge field on $\Sigma$ with vanishing curvature, then $\cA$ determines a class $[\cA]$ in $H^1(\Sigma,\ZZ)$. When the worldsheet is spherical, this cohomology vanishes, and so $\cA = \dd \chi$ is true globally. On a genus $g$ worldsheet, $\Sigma_g$, we have $H^1 (\Sigma_g,\ZZ) \cong \ZZ^{2g}$, and so we need to be more careful.

After performing the Buscher procedure, we obtain a new spacetime with dual coordinate $\hat{\theta}$ parameterising dual fibers. These dual fibers are also one dimensional, and we can now ask whether they are compact ($S^1$), or noncompact ($\mathbb{R}$). To answer this, we now consider the standard Buscher procedure on a genus 1 worldsheet.\footnote{The extension to genus $g$ worldsheets is straightforward.} A genus 1 worldsheet has two non-trivial homology cycles, which we label $(a,b)$. The gauged action (\ref{abeliangaugedaction}) contains the Lagrange multiplier term
\begin{align*}
\frac{1}{2\pi} \int_{\Sigma} \mathcal{F} \hat{\theta},
\end{align*}
which, after an integration by parts, becomes
\begin{align}
\label{Lagrange}
\frac{1}{2\pi} \int_{\Sigma} \dd \hat{\theta} \wedge \cA.
\end{align}
Integrating out the Lagrange multiplier term forces the gauge field to be flat, but it could have nontrivial holonomies around the cycles $(a,b)$. This would cause problems with the duality procedure, since the dual theory would then have twisted sectors. 

We now consider the variation of the term (\ref{Lagrange}) in the action. Since $\Sigma$ is topologically non-trivial, we can have large gauge transformations
\begin{align*}
\de A = \dd \epsilon =  \dd \overline{\epsilon} + \gamma,
\end{align*}
where $\overline{\epsilon} \in C^{\infty}(\Sigma)$ and $\gamma$ is closed but not exact. Note that since $\de \theta = \epsilon$, and $\theta$ is multivalued with period $2\pi$, we have that $\epsilon$ is also multivalued with period $2\pi$. That is, $\gamma \in H^1(\Sigma,\ZZ)$ and so $\oint_{C} \gamma \in 2 \pi \ZZ$ for some closed contour $C$ in $\Sigma$. Then the relevant term in the variation of the action becomes
\begin{align*}
\de S &= \frac{1}{2\pi} \de\left(  \int_{\Sigma} \dd \hat{\theta} \wedge A \right) \\[1em]
&= \frac{1}{2\pi} \int_{\Sigma} \dd \hat{\theta} \wedge \de A \\[1em]
&= \frac{1}{2\pi} \int_{\Sigma} \dd \hat{\theta} \wedge \dd \overline{\epsilon} + \frac{1}{2\pi} \int_{\Sigma} \dd \hat{\theta} \wedge \gamma \\[1em]
&= \frac{1}{2\pi} \int_{\Sigma} \dd \, \left( -\overline{\epsilon} \dd \hat{\theta} \right) + \frac{1}{2\pi} \int_{\Sigma} \dd \hat{\theta} \wedge \gamma \\[1em]
&=   \frac{1}{2\pi}\int_{\Sigma} \dd \hat{\theta} \wedge \gamma.
\end{align*}
Using the Riemann bilinear identity, this becomes
\begin{align*}
\de S &= \frac{1}{2\pi} \left( \oint_{a} \dd \hat{\theta} \oint_{b} \gamma - \oint_{a} \gamma \oint_{b} \dd \hat{\theta} \right).
\end{align*}
In order for this not to contribute to the path integral, we require $\de S \in 2 \pi \ZZ$ for each such gauge transformation $\gamma$. It follows that 
\begin{align*}
\oint_{C} \dd \hat{\theta} \in 2 \pi \ZZ.
\end{align*}
That is, $\hat{\theta}$ must be multivalued with period $2 \pi$.

Starting with a compact fiber, higher genus considerations lead us to the conclusion that the fibers in the dual space must also be compact. If we start with a non-compact fiber, then by applying the same argument to the dual space, we conclude that the dual fiber must also be non-compact (if it were compact, then that same argument says that \emph{its} dual, i.e. the original fiber, would have to be compact). Thus in string theory, the compactness or noncompactness of the fibers is preserved under T-duality. 

\subsection{Topological T-duality for circle bundles}
\label{subsec:TopCircBund}

In the early days of T-duality, it was noticed that there are examples of T-duality which not only change the geometry of the background, but also change the topology \cite{AABL}. A simple example is given by $S^3$ with the round metric, and no $H$-Flux. Performing a T-duality along the isometry generating the Hopf fibration (see Appendix \ref{app:Hopf}) gives the T-dual which is $S^2 \times S^1$, with one unit of $H$-flux. The relation between the flux and topology of T-dual spacetimes is known as topological T-duality. This was first described in \cite{BEM03,BEM04}, and studied thoroughly since \cite{BHM04,BHM05,BS05,MR04a,MR04b,MR05}.

From a geometric perspective, the Buscher rules intertwine the B-field with the metric - they mix the geometry with the gauge field. From a topological perspective, the Buscher rules intertwine a topological invariant of the spacetime with the cohomology class of the field strength -  they mix the topology with the flux. Topological T-duality is the study of the topological aspects of T-duality, forgetting the intricacies of the geometry. It is necessarily not the entire picture, but more of a topological shadow of what underlies the traditional notion of T-duality.

The starting point of topological T-duality is a string background which admits a description in terms of an (oriented) $S^1$ bundle $E$ over a base $M$:
\begin{center}
	\begin{tikzpicture}
	\node (F) {$S^1$};
	\node (E) [right=1cm of F] {$E$};
	\node (M) [below=1cm of E] {$M$};
	\draw[right hook->] (F) to node  {$ $} (E);
	\draw[->] (E) to node [right] {$\pi$} (M);
	\end{tikzpicture}
\end{center}
also written as $\pi : E \to M$. The fibers have a continuous circle action on them (rotation), and we assume that this action preserves the fibers, and is both transitive,\footnote{A group action $\sG \hookrightarrow X$ is \emph{transitive} if for all $x,y \in X$ there is a $\sg \in \sG$ such that $\sg \cdot x = y$} and free.\footnote{A group action $\sG \hookrightarrow X$ is \emph{free} if there are no fixed points. That is, if $\sg \in \sG$ and $\sg \cdot x = x$ for some $x \in X$, then $\sg$ is the identity.} Such a structure defines a \emph{principal $U(1)$-bundle} over $M$. The total space $E$ is a fiber bundle, and so is locally homeomorphic to $S^1 \times M$ (although we allow the possibility that it may not be a cartesian product globally). We will also assume that $M$ is connected.

The relevant object to describe the Kalb-Ramond field is the $H$-flux - a closed, degree 3 differential form on $E$. The $H$-flux  is the curvature of the $B$-field - locally it satisfies  $H = \dd B$.  We assume that $H$ has integral periods. Since $H$ is closed, it therefore determines a class $H$ in $H^3 (E,\ZZ)$. If $B$ is a globally-defined differential 2-form, then $H = \dd B$ globally, and $H$ is trivial in $H^3(E,\ZZ)$. The topological properties of the T-duality map don't depend on the details of the $B$ field, but only on the class $H$. 

The topology of the bundle $E$ is characterised by its isomorphism class. The set of isomorphism classes of principal $\sG$ bundles over a base $M$ are in bijective correspondence with homotopy classes of maps from $M$ to the classifying space: 
\begin{align*}
\textrm{Prin}_{\sG}(M) &\simeq [M,\mathsf{BG}].
\end{align*}
Here, $\mathsf{BG}$ is the base space of the \emph{universal bundle} $\pi: \mathsf{EG} \to \mathsf{BG}$,\footnote{$\mathsf{BG}$ is only defined up to homotopy equvialence} which has the property that any principle $\sG$-bundle is a pullback of the principal bundle
\begin{center}
	\begin{tikzpicture}
	\node (fEG) {$f^{\ast} (\mathsf{EG})$};
	\node (EG) [right=2cm of fEG] {$\mathsf{EG}$};
	\node (M) [below=2cm of fEG] {$M$};
	\node (BG) [right=2.45cm of M] {$\mathsf{BG}$};
	\draw[->] (fEG) to node  {$ $} (EG);
	\draw[->] (fEG) to node [left] {$\pi'$} (M);
	\draw[->] (EG) to node {$\pi$} (BG);
	\draw[->] (M) to node {$f$} (BG);
	\end{tikzpicture}
\end{center}
and any two bundles are isomorphic iff the maps, $f$, defining them are homotopic. A model for $\mathsf{B}U(1)$ is the space $\mathbb{CP}^{\infty}$, which is an Eilenberg-MacLane space $K(\ZZ,2)$, and so has the property that $[M,K(\ZZ,2)] \simeq H^2(M,\ZZ)$.\footnote{Technically speaking, $M$ must have the structure of a CW complex for this to be true. This will be the case in all the examples we are interested in.} Thus principle $U(1)$-bundles over a space $M$ are classified by $H^2(M,\ZZ)$. The element in $H^2(M,\ZZ)$ characterising the bundle $\pi: E \to M$ is realised by the first Chern class $c_1(L_E)$ of the associated complex line bundle $L_E = E \times_{U(1)} \CC$. It can be computed by calculating the (suitably normalised) curvature of a principal $U(1)$-connection for $E$.

In summary, topological T-duality begins the following topological data: A principal $U(1)$ bundle $\pi: E \to M$, together with a pair of cohomology classes $(F,H)$. The class $F \in H^2(M,\ZZ)$ is the first Chern class, and determines the isomorphism class of the bundle, whilst the class  $H \in H^3 (E,\ZZ)$ is the cohomology class of the curvature of the $B$-field. We will see that T-duality intermixes the $F$ and the $H$. 

Let us pause for a moment to talk about the relationship between this data and the data we have in the usual Buscher approach to T-duality. The normal Buscher procedure begins with a target spacetime $E$, together with a metric $g$, Kalb-Ramond field $B$, and a Killing vector field $k$ satisfying 
\begin{align*}
\Lie_k g = \Lie_k B =0.
\end{align*}
The flow of the vector field $k$ generate the $U(1)$ isometries of $P$. The orbits of points are circles, which we identify as the fibers of the bundle $\pi : E \to M$. In other words, $M$ is defined to be the quotient $E/S^1$, i.e. orbit space. The $H$-flux is the field strength of the $B$-field, and must be closed due to the Bianchi identity. In order for $H$ to give a well-defined contribution to the path integral, we must have 
\begin{align*}
\int_{\Sigma_3} X^{\ast} H \in 2 \pi \ZZ
\end{align*}
for some $\Sigma_3$ with $\Sigma = \pr \Sigma_3$, and for all maps $X: \Sigma \to E$. That is, $H$ is a closed 3-form with integral periods, and therefore determines a class in $H^3(E,\ZZ)$.

Choosing a connection $A = \dd \theta + A_i \dd X^i$ for the fibers allows us to decompose the metric and $B$-field:\footnote{c.f. Section \ref{sec:curv}.}
\begin{align}
\label{topg}
\dd s^2 &= \bar{g}_{ij} \dd X^i \dd X^j + (\dd \theta + A_i \dd X^i)^2 \\[1em]
B &= \tfrac{1}{2}\overline{B}_{ij} \dd X^i \wedge \dd X^j +  B_i \dd X^i \wedge \dd \theta \notag \\
\label{topB}
&= \tfrac{1}{2}(\overline{B}_{ij} - B_i A_j + A_i B_j ) \dd X^i \wedge \dd X^j +  B_i \dd X^i \wedge ( \dd \theta + A_i \dd X^i). 
\end{align}
The connection, $A$, is the coordinate description of a principal $U(1)$-connection on a principal $U(1)$-bundle $\pi:E \to M$. The (cohomology class of the) curvature $F = \dd A$ of this connection is the element $H^2(M,\ZZ)$ classifying the isomorphism class of the bundle. 

As we saw in Section \ref{sec:curv}, the Buscher rules correspond to the interchange of $A_{i}$ and $B_{i}$, giving us the dual metric and $B$-field:
\begin{align}
\label{topgdual}
\widehat{\dd s^2} &= \bar{g}_{ij} \dd X^i \dd X^j + (\dd \hat{\theta} + B_i \dd X^i)^2 \\[1em]
\widehat{B} &= \left( \overline{B}_{ij} - B_i A_j + A_i B_j \right) \dd X^i \wedge \dd X^j + 2 A_i \dd X^i \wedge \dd \hat{\theta} \notag \\
\label{topBdual}
&= \overline{B}_{ij} \dd X^i \wedge \dd X^j + 2 A_i \dd X^i \wedge (\dd \hat{\theta} + B_j \dd X^j).
\end{align}
We interpret the quantity $\widehat{A}=\dd \hat{\theta} + B_i \dd X^i$ as a connection on a new circle bundle $\hat{\pi}: \widehat{E} \to M$ over the same base:
\tikzset{node distance=1in, auto}
\begin{center}
	\begin{tikzpicture}
	\node (E) {$E$};
	\node (empty) [right of =E]{};
	\node (hatE) [right of=empty] {$\widehat{E}$};
	\node (M) [below of=empty] {$M$};
	\node (A4) {};
	\draw[->] (E) to node [swap] {$\pi$} (M);
	\draw[->] (hatE) to node {$\hat{\pi}$} (M);
	\end{tikzpicture}
\end{center}
The curvature $\widehat{F} = \dd \widehat{A}$ of this connection is the element $H^2(M, \ZZ)$ classifying the isomorphism class of the bundle $\hat{\pi} : \widehat{E} \to M$.
Since differential forms on $E$ and $\widehat{E}$ live in different bundles, we cannot compare them directly. We can, however, form the correspondence space:
\begin{align}
E\times_M \widehat{E} = \{ (a,b) \in E \times \widehat{E} : \pi(a) = \hat{\pi}(b) \},
\end{align}
which is simultaneously the pullback of the bundle $\pi:E \to M$ along $\hat{\pi}$, and the pullback of the bundle $\hat{\pi}: \widehat{E} \to M$ along $\pi$:
\begin{center}
	\begin{tikzpicture}
	\node (A1) {};
	\node (Corr) [right of =A1] {$E \times_M \widehat{E}$};
	\node (A2) [right of =Corr] {};
	\node (E) [below of=A1] {$E$};
	\node (A3) [below of =Corr]{};
	\node (hatE) [right of=A3] {$\widehat{E}$};
	\node (M) [below of=A3] {$M$};
	\node (A4) {};
	\draw[->] (Corr) to node [swap] {$\hat{p} = id. \times \hat{\pi}$} (E);
	\draw[->] (Corr) to node {$p = \pi \times id.$} (hatE);
	\draw[->] (E) to node [swap] {$\pi$} (M);
	\draw[->] (hatE) to node {$\hat{\pi}$} (M);
	\end{tikzpicture}
\end{center}
Then comparing the expressions for $B$ and $\widehat{B}$ given by (\ref{topB}) and (\ref{topBdual}),\footnote{We have omitted pullbacks to avoid unnecessary clutter.} we find:
\begin{align}
\label{BminusB}
\widehat{B} - B &=  A \wedge \widehat{A} -  \dd \theta \wedge \dd \hat{\theta}. 
\end{align}
Taking the exterior derivative of both sides and rearranging, we get 
\begin{align}
\label{corrTdual}
H -  \widehat{F} \wedge A = \widehat{H} -  F \wedge \widehat{A}.
\end{align}
Equation (\ref{corrTdual}) holds on the correspondence space $E\times_M \widehat{E}$, but the left hand side of the equation is the pullback of a form on $E$. Similarly, the right hand side is the pullback of a form on $\widehat{E}$, and so we conclude they are both equal to some form $H_3$ on the base $M$. That is, we have
\begin{subequations}
\label{Hdimensional}
\begin{align}
H &= H_3 -  \widehat{F} \wedge A \\[1em]
\label{dualH}
\widehat{H} &= H_3 -  F \wedge \widehat{A}. 
\end{align}
\end{subequations}
This is the key result motivating the definition of topological T-duality. Interpreted geometrically, this result says that under T-duality, the ``legs" of the $H$-flux along the circle fiber get interchanged with the curvature of the connection, and the component of the flux living on the base is unchanged. Integrating these equations over the $S^1$ fibers, we obtain:
\begin{subequations}
\label{tdualequations}
\begin{align}
 \widehat{F} &=  \int_{S^1} H \\[1em]
 F &=   \int_{\widehat{S}^1} \widehat{H}.
\end{align}
\end{subequations}

\subsubsection{Description using Gysin sequences}
\label{subsubsec:TopGysin}
As described in Appendix \ref{sec:Gysin}, a sphere bundle has an associated exact sequence in cohomology, known as the Gysin sequence.\footnote{Recall that an exact sequence is one in which the image of each map is equal to the kernel of the following map.} For an $S^1$ bundle, $\pi: E \to M$, the sequence is 
\begin{equation*}
	\begin{tikzpicture}[baseline=(current  bounding  box.center)]
	\node (ldots) {$\cdots$};
	\node (H3M) [right=1cm of ldots] {$H^n(M,\ZZ)$};
	\node (H3E) [right=1cm of H3M] {$H^n(E,\ZZ)$};
	\node (H2M) [right=1cm of H3E] {$H^{n-1}(M,\ZZ)$};
	\node (H4M) [right=1cm of H2M] {$H^{n+1}(M,\ZZ)$};
	\node (rdots) [right=1cm of H4M] {$\cdots$};
	\node (0) [right of =H2M] {};
	\draw[->] (ldots) to node {} (H3M);
	\draw[->] (H3M) to node {\footnotesize $\pi^{\ast}$} (H3E);
	\draw[->] (H3E) to node {\footnotesize $\pi_{\ast}$} (H2M);
	\draw[->] (H2M) to node {\footnotesize $F \cup $}  (H4M);
	\draw[->] (H4M) to node {}  (rdots);
	\end{tikzpicture}
\end{equation*}
We now see that the content of topological T-duality fits nicely into the $n=3$ segment of this sequence:\footnote{Note that for the remainder of this section we will suppress the $\ZZ$ in $H^k(M,\ZZ)$.}
\begin{equation}
\label{Gysin1}
	\begin{tikzpicture}[baseline=(current  bounding  box.center)]
	\node (ldots) {$\cdots$};
	\node (H3M) [right=1cm of ldots] {$H^3(M)$};
	\node (H3E) [right=1cm of H3M] {$H^3(E)$};
	\node (H2M) [right=1cm of H3E] {$H^2(M)$};
	\node (H4M) [right=1cm of H2M] {$H^4(M)$};
	\node (rdots) [right=1cm of H4M] {$\cdots$};
	\node (0) [right of =H2M] {};
	\draw[->] (ldots) to node {} (H3M);
	\draw[->] (H3M) to node {\footnotesize$ \pi^{\ast}$} (H3E);
	\draw[->] (H3E) to node {\footnotesize$ \pi_{\ast}$} (H2M);
	\draw[->] (H2M) to node { \footnotesize$F \cup $}  (H4M);
	\draw[->] (H4M) to node {}  (rdots);
	\end{tikzpicture}
\end{equation}
The $H$-flux is an element of $H^3(E)$, so we can look at its image under the pushforward map $\pi_{\ast}$. The image, $\widehat{F} = \pi_{\ast}H$, is an element of $H^2(M)$, and therefore defines a new $S^1$-bundle $\hat{\pi}:\widehat{E} \to M$. Since the Gysin sequence is exact, the composition of two maps is identically zero, and so it follows that $ F \cup \widehat{F} = 0$. 
\begin{equation*}
	\begin{tikzpicture}[baseline=(current  bounding  box.center)]
	\node (ldots) {$\cdots$};
	\node (H3M) [right=1cm of ldots] {$H^3(M)$};
	\node (H3E) [right=1cm of H3M] {$H^3(E)$};
	\node (H2M) [right=1cm of H3E] {$H^2(M)$};
	\node (H4M) [right=1cm of H2M] {$H^4(M)$};
	\node (H) [below=0.7cm of H3E] {$H$};
	\node (hatF) [below=0.65cm of H2M] {$\widehat{F}$};
	\node (0) [below=0.7cm of H4M] {$0$};
	\draw[->] (ldots) to node {}  (H3M);
	\draw[->] (H3M) to node {\footnotesize $\pi^{\ast}$} (H3E);
	\draw[->] (H3E) to node {\footnotesize $\pi_{\ast}$} (H2M);
	\draw[->] (H2M) to node {\footnotesize $F \cup $}  (H4M);
	\draw[|->] (H) to node {}  (hatF);
	\draw[dashed,|->] (hatF) to node {}  (0);
	\end{tikzpicture}
\end{equation*}
The class $\widehat{F}$ defines the topology of the dual bundle, but how do we get the dual flux? To find that, consider the dual bundle $\hat{\pi}: \widehat{E} \to M$ defined (up to isomorphism) by the element $\widehat{F} = \pi_{\ast} H$. Since it, too, is a principal $S^1$ bundle, we can look at the Gysin sequence associated to it:
\begin{equation}
\label{Gysin2}
	\begin{tikzpicture}[baseline=(current  bounding  box.center)]
	\node (ldots) {$\cdots$};
	\node (H3M) [right=1cm of ldots] {$H^3(M)$};
	\node (H3E) [right=1cm of H3M] {$H^3(\widehat{E})$};
	\node (H2M) [right=1cm of H3E] {$H^2(M)$};
	\node (H4M) [right=1cm of H2M] {$H^4(M)$};
	\node (rdots) [right=1cm of H4M] {$\cdots$};
	\node (0) [right of =H2M] {};
	\draw[->] (ldots) to node {} (H3M);
	\draw[->] (H3M) to node {\footnotesize $\hat{\pi}^{\ast}$} (H3E);
	\draw[->] (H3E) to node {\footnotesize $\hat{\pi}_{\ast}$} (H2M);
	\draw[->] (H2M) to node {\footnotesize $\widehat{F} \cup $}  (H4M);
	\draw[->] (H4M) to node {}  (rdots);
	\end{tikzpicture}
\end{equation}
Since $ F \cup \widehat{F} = \widehat{F} \cup F = 0$ by our earlier argument, the image of $F$ under the map $\widehat{F} \cup$ in the new Gysin sequence is zero. The new Gysin sequence is also exact, so $F$ must be the image of an element in $H^3(\widehat{E})$:
\begin{equation*}
	\begin{tikzpicture}[baseline=(current  bounding  box.center)]
	\node (ldots) {$\cdots$};
	\node (H3M) [right=1cm of ldots] {$H^3(M)$};
	\node (H3E) [right=1cm of H3M] {$H^3(\widehat{E})$};
	\node (H2M) [right=1cm of H3E] {$H^2(M)$};
	\node (H4M) [right=1cm of H2M] {$H^4(M)$};
	\node (H) [below=0.6cm of H3E] {$\widehat{H}$};
	\node (hatF) [below=0.7cm of H2M] {$F$};
	\node (0) [below=0.7cm of H4M] {$0$};
	\draw[->] (ldots) to node {}  (H3M);
	\draw[->] (H3M) to node {\footnotesize $\hat{\pi}^{\ast}$} (H3E);
	\draw[->] (H3E) to node {\footnotesize $\hat{\pi}_{\ast}$} (H2M);
	\draw[->] (H2M) to node {\footnotesize $\widehat{F} \cup $}  (H4M);
	\draw[dashed,|->] (H) to node {}  (hatF);
	\draw[|->] (hatF) to node {}  (0);
	\end{tikzpicture}
\end{equation*}
That is, $F = \hat{\pi}_{\ast} \widehat{H}$. We have thus arrived at the topological description of T-duality. Given a circle bundle $\pi:E\to M$ and a $H$-flux, together described by a pair $(F,H)$, we construct a new circle bundle and a new $H$-flux described by the pair $(\widehat{F},\widehat{H})$, satisfying 
\begin{subequations}
\label{FHpushforward}
\begin{align}
F &= \hat{\pi}_{\ast} \widehat{H} \\[1em]
\widehat{F} &= \pi_{\ast} H.
\end{align}
\end{subequations}
Note that this is simply the cohomological version of (\ref{tdualequations}). It is clear that there is some ambiguity here in the choice of $\widehat{H}$; if we add to $\widehat{H}$ a term which is in the kernel of $\hat{\pi}_{\ast}$, then the additional term won't change (\ref{FHpushforward}). To consider the origin of this ambiguity, we will need to examine the following double complex: 
\begin{equation}
\label{DoubleComplex}
\begin{tikzpicture}[baseline=(current  bounding  box.center)]
\node (A1) {$0$};
\node (A2) [right of=A1] {$H^1(M)$};
\node (A3) [right of=A2] {$H^1(E)$};
\node (A4) [right of=A3] {$H^0(M)$};
\node (A5) [right of=A4] {$\cdots$}; 
\node (B1) [below=1.2cm of A1] {$H^1(M)$};
\node (B2) [right of=B1] {$H^3(M)$};
\node (B3) [right of=B2] {$H^3(E)$};
\node (B4) [right of=B3] {$H^2(M)$};
\node (B5) [right of=B4] {$\cdots$};
\node (C1) [below=1.2cm of B1] {$H^1(\widehat{E})$};
\node (C2) [right of=C1] {$H^3(\widehat{E})$};
\node (C3) [right of=C2] {$H^3(E \times \widehat{E})$};
\node (C4) [right of=C3] {$H^2(\widehat{E})$}; 
\node (C5) [right of=C4] {$\cdots$};
\node (D1) [below=1.2cm of C1] {$H^0 (M)$};
\node (D2) [right of=D1] {$H^2(M)$};
\node (D3) [right of=D2] {$H^2(E)$};
\node (D4) [right of=D3] {$H^1(M)$};
\node (D5) [right of=D4] {$\cdots$};
\node (E1) [below=1.2cm of D1] {$\vdots$};
\node (E2) [right of=E1] {$\vdots$};
\node (E3) [right of=E2] {$\vdots$}; 
\node (E4) [right of=E3] {$\vdots$};

\draw[->] (A1) to node {\small $ $} (A2);
\draw[->] (A2) to node {\small \color{blue}$\pi^{\ast}$} (A3);
\draw[->] (A3) to node {\small \color{blue}$\pi_{\ast}$} (A4);
\draw[->] (A4) to node {\small $ $} (A5);
\draw[->] (B1) to node {\small \color{blue}$ F\cup$} (B2);
\draw[->] (B2) to node {\small \color{blue}$\pi^{\ast}$} (B3);
\draw[->] (B3) to node {\small \color{blue}$\pi_{\ast}$} (B4);
\draw[->] (B4) to node {\small $ $} (B5);
\draw[->] (C1) to node {\small \color{blue}$ \hat{\pi}^{\ast}F \, \cup $} (C2);
\draw[->] (C2) to node {\small \color{blue}${p}^{\ast} $} (C3);
\draw[->] (C3) to node {\small \color{blue}${p}_{\ast}$} (C4);
\draw[->] (C4) to node {\small $ $} (C5);
\draw[->] (D1) to node {\small \color{blue}$ F \cup$} (D2);
\draw[->] (D2) to node {\small \color{blue}$\pi^{\ast}$} (D3);
\draw[->] (D3) to node {\small \color{blue}$\pi_{\ast}$} (D4);
\draw[->] (D4) to node {$$} (D5);

\draw[->] (A1) to node [swap] {\small $ $} (B1);
\draw[->] (B1) to node [swap] {\small \color{red}$\hat{\pi}^{\ast}$} (C1);
\draw[->] (C1) to node [swap] {\small \color{red} $\hat{\pi}_{\ast}$} (D1);
\draw[->] (D1) to node [swap] {$ $} (E1);
\draw[->] (A2) to node [swap] {\small \color{red} $ \widehat{F} \cup$} (B2);
\draw[->] (B2) to node [swap] {\small \color{red} $\hat{\pi}^{\ast} $} (C2);
\draw[->] (C2) to node [swap] {\small \color{red} $\hat{\pi}_{\ast} $} (D2);
\draw[->] (D2) to node [swap] {$$} (E2);
\draw[->] (A3) to node [swap] {\small \color{red} $ \pi^{\ast} \widehat{F}\cup$} (B3);
\draw[->] (B3) to node [swap] {\small \color{red} $\hat{p}^{\ast} $} (C3);
\draw[->] (C3) to node [swap] {\small \color{red} $\hat{p}_{\ast} $} (D3);
\draw[->] (D3) to node [swap] {$ $} (E3);
\draw[->] (A4) to node [swap] {\small \color{red} $ \widehat{F} \cup$} (B4);
\draw[->] (B4) to node [swap] {\small \color{red} $\hat{\pi}^{\ast} $} (C4);
\draw[->] (C4) to node [swap] {\small \color{red} $\hat{\pi}_{\ast} $} (D4);
\draw[->] (D4) to node [swap] {$ $} (E4);

\end{tikzpicture}
\end{equation}
All rows and columns in the double complex are exact, and all squares commute. The first, second, and fourth rows are simply different sections of the Gysin sequence (\ref{Gysin1}) for the bundle $\pi: E \to M$. Similarly, the first, second, and fourth columns are different sections of the Gysin sequence (\ref{Gysin2}) for the bundle $\hat{\pi}: \widehat{E} \to M$. The third row and third column are a little more complicated. Recall the correspondence space relating the original space $E$ and the dual space $\widehat{E}$:
\begin{equation*}
	\begin{tikzpicture}[baseline=(current  bounding  box.center)]
	\node (A1) {};
	\node (Corr) [right of =A1] {$E \times_M \hat{E}$};
	\node (A2) [right of =Corr] {};
	\node (E) [below of=A1] {$E$};
	\node (A3) [below of =Corr]{};
	\node (hatE) [right of=A3] {$\hat{E}$};
	\node (M) [below of=A3] {$M$};
	\node (A4) {};
	\draw[->,red] (Corr) to node [swap] {$\hat{p} = id. \times \hat{\pi}$} (E);
	\draw[->] (Corr) to node {$p = \pi \times id.$} (hatE);
	\draw[->] (E) to node [swap] {$\pi$} (M);
	\draw[->,red] (hatE) to node {$\hat{\pi}$} (M);
	\end{tikzpicture}
\end{equation*}
The bundle $\hat{\pi}:\widehat{E}\to M$ pulls back along the map $\pi$ to the bundle $\hat{p}:E \times_M \widehat{E} \to E$. This bundle is also principal circle bundle, now over $E$, whose first Chern class is simply the pullback of the first Chern class of $\hat{\pi}:\widehat{E} \to M$. That is, the first Chern class of the bundle $\hat{p} : E \times_M \widehat{E}\to E$ is $\pi^{\ast}\widehat{F}$. We now see that the third column of the double complex (\ref{DoubleComplex}) is just a section of the Gysin sequence associated to the bundle $\hat{p}: E \times_M \widehat{E} \to E$.

Similarly, the bundle $\pi:E\to M$ pulls back along the map $\hat{\pi}$ to the bundle $\hat{p}:E \times_M \widehat{E} \to E$:
\begin{equation*}
	\begin{tikzpicture}[baseline=(current  bounding  box.center)]
	\node (A1) {};
	\node (Corr) [right of =A1] {$E \times_M \hat{E}$};
	\node (A2) [right of =Corr] {};
	\node (E) [below of=A1] {$E$};
	\node (A3) [below of =Corr]{};
	\node (hatE) [right of=A3] {$\hat{E}$};
	\node (M) [below of=A3] {$M$};
	\node (A4) {};
	\draw[->] (Corr) to node [swap] {$\hat{p} = id. \times \hat{\pi}$} (E);
	\draw[->,blue] (Corr) to node {$p = \pi \times id.$} (hatE);
	\draw[->,blue] (E) to node [swap] {$\pi$} (M);
	\draw[->] (hatE) to node {$\hat{\pi}$} (M);
	\end{tikzpicture}
\end{equation*}
The first Chern class of the bundle $p:E\times_M \widehat{E} \to \widehat{E}$ is again given by a pullback,  $\hat{\pi}^{\ast}F$, and the third row of the double complex (\ref{DoubleComplex}) then corresponds to the Gysin sequence of the principal bundle $p:E\times_M \widehat{E} \to \widehat{E}$.
We will now prove the following lemma, which will help us interpret the ambiguity in our choice of $\widehat{H}$.
\begin{lemma}
	\label{pH}
Let $(F,H)$ be a pair corresponding to a circle bundle $\pi:E\to M$ with an $H$-flux. Let $\widehat{F} =\pi_{\ast}H$ define a dual circle bundle $\hat{\pi}:\widehat{E} \to M$. Then
\begin{align*}
\hat{p}^{\ast} H = p^{\ast}\widehat{H}
\end{align*}
for some $\widehat{H} \in H^3(\widehat{E})$
\end{lemma}
\begin{proof}
This proof requires a bit of diagram chasing in the double complex \ref{DoubleComplex}. We begin with $H \in H^3(E)$, and look at its image $\widehat{F}$ under $\pi_{\ast}$, moving horizontally along the second row of the diagram. Then, in the third column, we notice that since $\widehat{F} =  \widehat{F} \cup 1$, we must have $\hat{\pi}^{\ast}(\widehat{F}) = 0$ by exactness. 
\begin{equation*}
	\begin{tikzpicture}[baseline=(current  bounding  box.center)]
	\node (A1) {$ $};
	\node (A2) [right of=A1] {$ $};
	\node (A3) [right of=A2] {$ $};
	\node (A4) [right of=A3] {$1$};
	\node (A5) [right of=A4] {$\cdots$}; 
	\node (B1) [below=1.2cm of A1] {$ $};
	\node (B2) [right of=B1] {$ $};
	\node (B3) [right of=B2] {$H$};
	\node (B4) [right of=B3] {$\widehat{F}$};
	\node (B5) [right of=B4] {$\cdots$};
	\node (C1) [below=1.2cm of B1] {$ $};
	\node (C2) [right of=C1] {$ $};
	\node (C3) [right of=C2] {$ $};
	\node (C4) [right of=C3] {$0$}; 
	\node (C5) [right of=C4] {$\cdots$};
	\node (D1) [below=1.2cm of C1] {$ $};
	\node (D2) [right of=D1] {$ $};
	\node (D3) [right of=D2] {$ $};
	\node (D4) [right of=D3] {$ $};
	\node (D5) [right of=D4] {$\cdots$};
	\node (E1) [below=1.2cm of D1] {$\vdots$};
	\node (E2) [right of=E1] {$\vdots$};
	\node (E3) [right of=E2] {$\vdots$}; 
	\node (E4) [right of=E3] {$\vdots$};

	\draw[->,lightgray] (A1) to node {\small $ $} (A2);
	\draw[->,lightgray] (A2) to node {$ $} (A3);
	\draw[->,lightgray] (A3) to node {$ $} (A4);
	\draw[->,lightgray] (A4) to node {\small $ $} (A5);
	\draw[->,lightgray] (B1) to node {$ $} (B2);
	\draw[->,lightgray] (B2) to node {$ $} (B3);
	\draw[|->] (B3) to node {\small \color{blue}$\pi_{\ast}$} (B4);
	\draw[->,lightgray] (B4) to node {\small $ $} (B5);
	\draw[->,lightgray] (C1) to node {$ $} (C2);
	\draw[->,lightgray] (C2) to node {$ $} (C3);
	\draw[->,lightgray] (C3) to node {$ $} (C4);
	\draw[->,lightgray] (C4) to node {\small $ $} (C5);
	\draw[->,lightgray] (D1) to node {$ $} (D2);
	\draw[->,lightgray] (D2) to node {$ $} (D3);
	\draw[->,lightgray] (D3) to node {$ $} (D4);
	\draw[->,lightgray] (D4) to node {$ $} (D5);

	\draw[->,lightgray] (A1) to node [swap] {\small $ $} (B1);
	\draw[->,lightgray] (B1) to node [swap] {$ $} (C1);
	\draw[->,lightgray] (C1) to node [swap] {$ $} (D1);
	\draw[->,lightgray] (D1) to node [swap] {$ $} (E1);
	\draw[->,lightgray] (A2) to node [swap] {$ $} (B2);
	\draw[->,lightgray] (B2) to node [swap] {$ $} (C2);
	\draw[->,lightgray] (C2) to node [swap] {$ $} (D2);
	\draw[->,lightgray] (D2) to node [swap] {$ $} (E2);
	\draw[->,lightgray] (A3) to node [swap] {$ $} (B3);
	\draw[->,lightgray] (B3) to node [swap] {$ $} (C3);
	\draw[->,lightgray] (C3) to node [swap] {$ $} (D3);
	\draw[->,lightgray] (D3) to node [swap] {$ $} (E3);
	\draw[|->] (A4) to node [swap] {\small \color{red} $ \widehat{F} \cup$} (B4);
	\draw[|->] (B4) to node [swap] {\small \color{red} $\hat{\pi}^{\ast} $} (C4);
	\draw[->,lightgray] (C4) to node [swap] {$ $} (D4);
	\draw[->,lightgray] (D4) to node [swap] {$ $} (E4);
	
	\end{tikzpicture}
\end{equation*}
Now, since the squares in this diagram are commutative, we have $\hat{\pi}^{\ast} \pi_{\ast}H = p_{\ast} \hat{p}^{\ast}H = 0$. Thus $\hat{p}^{\ast}H$ is in the kernel of $p_{\ast}$, and so by exactness of the third row it must be in the image of $p^{\ast}$. That is, there is some $\widehat{H} \in H^3(\widehat{E})$ such that $\hat{p}^{\ast}H = p^{\ast} \widehat{H}$.

\begin{equation*}
	\begin{tikzpicture}[baseline=(current  bounding  box.center)]
	\node (A1) {$ $};
	\node (A2) [right of=A1] {$ $};
	\node (A3) [right of=A2] {$ $};
	\node (A4) [right of=A3] {$ $};
	\node (A5) [right of=A4] {$\cdots$}; 
	\node (B1) [below=1.2cm of A1] {$ $};
	\node (B2) [right of=B1] {$ $};
	\node (B3) [right of=B2] {$H$};
	\node (B4) [right of=B3] {$ $};
	\node (B5) [right of=B4] {$\cdots$};
	\node (C1) [below=1.2cm of B1] {$ $};
	\node (C2) [right of=C1] {$\widehat{H}$};
	\node (C3) [right of=C2] {$ $};
	\node (C4) [right of=C3] {$0$}; 
	\node (C5) [right of=C4] {$\cdots$};
	\node (D1) [below=1.2cm of C1] {$ $};
	\node (D2) [right of=D1] {$  $};
	\node (D3) [right of=D2] {$ $};
	\node (D4) [right of=D3] {$ $};
	\node (D5) [right of=D4] {$\cdots$};
	\node (E1) [below=1.2cm of D1] {$\vdots$};
	\node (E2) [right of=E1] {$\vdots$};
	\node (E3) [right of=E2] {$\vdots$}; 
	\node (E4) [right of=E3] {$\vdots$};

	\draw[->,lightgray] (A1) to node {\small $ $} (A2);
	\draw[->,lightgray] (A2) to node {$ $} (A3);
	\draw[->,lightgray] (A3) to node {$ $} (A4);
	\draw[->,lightgray] (A4) to node {\small $ $} (A5);
	\draw[->,lightgray] (B1) to node {$ $} (B2);
	\draw[->,lightgray] (B2) to node {$ $} (B3);
	\draw[->,lightgray] (B3) to node {$ $} (B4);
	\draw[->,lightgray] (B4) to node {\small $ $} (B5);
	\draw[->,lightgray] (C1) to node {$ $} (C2);
	\draw[|->] (C2) to node {\small \color{blue} $p^{\ast} $} (C3);
	\draw[|->] (C3) to node {\small \color{blue} $p_{\ast}$} (C4);
	\draw[->,lightgray] (C4) to node {\small $ $} (C5);
	\draw[->,lightgray] (D1) to node {$ $} (D2);
	\draw[->,lightgray] (D2) to node {$ $} (D3);
	\draw[->,lightgray] (D3) to node {$ $} (D4);
	\draw[->,lightgray] (D4) to node {$ $} (D5);

	\draw[->,lightgray] (A1) to node [swap] {\small $ $} (B1);
	\draw[->,lightgray] (B1) to node [swap] {$ $} (C1);
	\draw[->,lightgray] (C1) to node [swap] {$ $} (D1);
	\draw[->,lightgray] (D1) to node [swap] {$ $} (E1);
	\draw[->,lightgray] (A2) to node [swap] {$ $} (B2);
	\draw[->,lightgray] (B2) to node [swap] {$ $} (C2);
	\draw[->,lightgray] (C2) to node [swap] {$ $} (D2);
	\draw[->,lightgray] (D2) to node [swap] {$ $} (E2);
	\draw[->,lightgray] (A3) to node [swap] {$ $} (B3);
	\draw[|->] (B3) to node [swap] {\small \color{red} $\hat{p}^{\ast}$} (C3);
	\draw[->,lightgray] (C3) to node [swap] {$ $} (D3);
	\draw[->,lightgray] (D3) to node [swap] {$ $} (E3);
	\draw[->,lightgray] (A4) to node [swap] {$ $} (B4);
	\draw[->,lightgray] (B4) to node [swap] {$ $} (C4);
	\draw[->,lightgray] (C4) to node [swap] {$ $} (D4);
	\draw[->,lightgray] (D4) to node [swap] {$ $} (E4);
	
	\end{tikzpicture}
\end{equation*}
\end{proof}
Let us now discuss the uniqueness of the dual $H$-flux. Suppose we had two fluxes, $\widehat{H}$ and $\widehat{H}'$, each satisfying $\hat{\pi}_{\ast} \widehat{H} = \hat{\pi}_{\ast}\widehat{H}' = F$. Then their difference $d= \widehat{H} - \widehat{H}'$ is in the kernel  of $\hat{\pi}_{\ast}$:
\begin{align*}
\hat{\pi}_{\ast} d = \hat{\pi}_{\ast}\widehat{H} - \hat{\pi}_{\ast}\widehat{H}' = F - F = 0.
\end{align*}
Since $\hat{\pi}_{\ast}d = 0$, exactness of the second column tells us that the difference must be in the image of $\hat{\pi}^{\ast}$. That is, $d = \hat{\pi}^{\ast} \alpha$, for some $\alpha \in H^3(M)$.
\begin{equation*}
\begin{tikzpicture}[baseline=(current  bounding  box.center)]
\node (A1) {$ $};
\node (A2) [right of=A1] {$ $};
\node (A3) [right of=A2] {$ $};
\node (A4) [right of=A3] {$ $};
\node (A5) [right of=A4] {$\cdots$}; 
\node (B1) [below=1.2cm of A1] {$ $};
\node (B2) [right of=B1] {$ \alpha $};
\node (B3) [right of=B2] {$ $};
\node (B4) [right of=B3] {$ $};
\node (B5) [right of=B4] {$\cdots$};
\node (C1) [below=1.2cm of B1] {$ $};
\node (C2) [right of=C1] {$d$};
\node (C3) [right of=C2] {$ $};
\node (C4) [right of=C3] {$ $}; 
\node (C5) [right of=C4] {$\cdots$};
\node (D1) [below=1.2cm of C1] {$ $};
\node (D2) [right of=D1] {$ 0 $};
\node (D3) [right of=D2] {$ $};
\node (D4) [right of=D3] {$ $};
\node (D5) [right of=D4] {$\cdots$};
\node (E1) [below=1.2cm of D1] {$\vdots$};
\node (E2) [right of=E1] {$\vdots$};
\node (E3) [right of=E2] {$\vdots$}; 
\node (E4) [right of=E3] {$\vdots$};

\draw[->,lightgray] (A1) to node {\small $ $} (A2);
\draw[->,lightgray] (A2) to node {$ $} (A3);
\draw[->,lightgray] (A3) to node {$ $} (A4);
\draw[->,lightgray] (A4) to node {\small $ $} (A5);
\draw[->,lightgray] (B1) to node {$ $} (B2);
\draw[->,lightgray] (B2) to node {$ $} (B3);
\draw[->,lightgray] (B3) to node {$ $} (B4);
\draw[->,lightgray] (B4) to node {\small $ $} (B5);
\draw[->,lightgray] (C1) to node {$ $} (C2);
\draw[->,lightgray] (C2) to node {$ $} (C3);
\draw[->,lightgray] (C3) to node {$ $} (C4);
\draw[->,lightgray] (C4) to node {\small $ $} (C5);
\draw[->,lightgray] (D1) to node {$ $} (D2);
\draw[->,lightgray] (D2) to node {$ $} (D3);
\draw[->,lightgray] (D3) to node {$ $} (D4);
\draw[->,lightgray] (D4) to node {$ $} (D5);

\draw[->,lightgray] (A1) to node [swap] {\small $ $} (B1);
\draw[->,lightgray] (B1) to node [swap] {$ $} (C1);
\draw[->,lightgray] (C1) to node [swap] {$ $} (D1);
\draw[->,lightgray] (D1) to node [swap] {$ $} (E1);
\draw[->,lightgray] (A2) to node [swap] {$ $} (B2);
\draw[|->] (B2) to node [swap] {\small \color{red} $\hat{\pi}^{\ast} $} (C2);
\draw[|->] (C2) to node [swap] {\small \color{red} $\hat{\pi}_{\ast}$} (D2);
\draw[->,lightgray] (D2) to node [swap] {$ $} (E2);
\draw[->,lightgray] (A3) to node [swap] {$ $} (B3);
\draw[->,lightgray] (B3) to node [swap] {$ $} (C3);
\draw[->,lightgray] (C3) to node [swap] {$ $} (D3);
\draw[->,lightgray] (D3) to node [swap] {$ $} (E3);
\draw[->,lightgray] (A4) to node [swap] {$ $} (B4);
\draw[->,lightgray] (B4) to node [swap] {$ $} (C4);
\draw[->,lightgray] (C4) to node [swap] {$ $} (D4);
\draw[->,lightgray] (D4) to node [swap] {$ $} (E4);

\end{tikzpicture}
\end{equation*}
This tells us that if there is some ambiguity, it must come from a form on the base. From a physical perspective, however, we expect that this is unchanged under T-duality. That is, T-duality should not affect the part of the flux which `lives' on the base manifold $M$. How can we encode this assumption in the language of this double complex? We claim that it is encoded by the cohomological version of (\ref{BminusB}). That is, on the correspondence space, we have:
\begin{align*}
\widehat{H} - H = \ddd (A \wedge \widehat{A}).
\end{align*}
Explicitly writing the projections we omitted earlier, this means that in cohomology
\begin{align*}
p^{\ast}\widehat{H} - \hat{p}^{\ast}H = 0.
\end{align*}
That is, $p^{\ast}\widehat{H} = \hat{p}^{\ast}H$ in $H^3(E \times_M \widehat{E},\ZZ)$. To see why this enforces the assumption that the flux living on the base is unchanged under duality, consider what happens if we have a flux $H$, and a dual flux $\widehat{H}$, satisfying
\begin{align}
\label{pullbackfluxes}
\hat{p}^{\ast} H - p^{\ast}\widehat{H} = 0.
\end{align}
Now change the dual flux by the pullback of a form on the base, $\widehat{H} \rightarrow \widehat{H}' = \widehat{H}+ \hat{\pi}^{\ast} \alpha$. In order for (\ref{pullbackfluxes}) to remain true, how must we change the original flux $H \rightarrow H'$? We calculate
\begin{align*}
\hat{p}^{\ast} H' - p^{\ast}\widehat{H}' &=  \hat{p}^{\ast} H' - p^{\ast}\widehat{H} - p^{\ast} \hat{\pi}^{\ast} \alpha \\[1em]
&=  \hat{p}^{\ast} H' - \hat{p}^{\ast} H - p^{\ast} \hat{\pi}^{\ast} \alpha  \\[1em]
&=\hat{p}^{\ast} H' - \hat{p}^{\ast} H - \hat{p}^{\ast} \pi^{\ast} \alpha \\[1em]
&= \hat{p}^{\ast} (H' - H - \pi^{\ast} \alpha).
\end{align*}
It follows that (\ref{pullbackfluxes}) is satisfied provided that $H' = H + \pi^{\ast} \alpha$, justifying our claim that (\ref{pullbackfluxes}) is the cohomological version of the physical assumption that the flux on the base is unchanged.\footnote{The astute reader, looking at the third column of (\ref{DoubleComplex}), may wish to point out that $H' - H - \pi^{\ast} \alpha$ need not be zero, since it could be in the image of $\pi^{\ast} \widehat{F} \cup$. This will be discussed in the proof of Theorem (\ref{theorem:topT-dual}).}
Before, we were able to show the existence of a dual $H$-flux $\widehat{H}$ which satisfied $F=\hat{\pi}_{\ast}\widehat{H} $. With our assumption that the flux on the base is unchanged after duality, we are now in a position to prove the following theorem, which asserts uniqueness of the found dual $H$-flux.
\begin{theorem}[Bouwknegt, Evslin, Mathai \cite{BEM03}, Bunke, Schick \cite{BS05}]
	\label{theorem:topT-dual}
	Let $(F,H)$ be a pair corresponding to a circle bundle $\pi:E\to M$ with an $H$-flux. Then there is a pair, $(\widehat{F},\widehat{H})$, describing a dual circle bundle $\hat{\pi}:\widehat{E}\to M$ with a dual $H$-flux, satisfying
	\begin{subequations}
	\label{piH}
	\begin{align}
	F &= \hat{\pi}_{\ast} \widehat{H} \\[1em]
	\widehat{F} &= \pi_{\ast} H.
	\end{align}
	\end{subequations}
If the the flux and the dual flux also satisfy
	\begin{align}
\label{pHHp}
\hat{p}^{\ast} H = p^{\ast}\widehat{H},
\end{align}
then the dual flux is unique up to a bundle automorphism.
\end{theorem}
\begin{proof}
We have already shown existence of a $\widehat{H}$ satisfying (\ref{piH}), so we just need to show that if $\widehat{H}$ also satisfies (\ref{pHHp}), then it is unique. This will also consist of a bit of diagram chasing. As before, we suppose there are two such $H$-fluxes, and let $d$ be their difference $d = \widehat{H} - \widehat{H}'$. We have already seen that $\hat{\pi}_{\ast} d = 0$, but now we also see that $p^{\ast}d= 0$. This follows from (\ref{pHHp}), since
\begin{align*}
p^{\ast}d = p^{\ast} \widehat{H} - p^{\ast} \widehat{H}' = \hat{p}^{\ast}H - \hat{p}^{\ast}H = 0.
\end{align*}
Exactness then says that $d =  \hat{\pi}^{\ast}F \cup \beta$, for some $\beta \in H^1(\widehat{E})$. The pushforward of this element, $n = \hat{\pi}_{\ast} \beta$, maps to zero under the map $ F\cup$, since the squares of the double complex commute.
\begin{equation*}
\begin{tikzpicture}[baseline=(current  bounding  box.center)]
\node (A1) {$ $};
\node (A2) [right of=A1] {$ $};
\node (A3) [right of=A2] {$ $};
\node (A4) [right of=A3] {$ $};
\node (A5) [right of=A4] {$\cdots$}; 
\node (B1) [below=1.2cm of A1] {$ $};
\node (B2) [right of=B1] {$ $};
\node (B3) [right of=B2] {$ $};
\node (B4) [right of=B3] {$ $};
\node (B5) [right of=B4] {$\cdots$};
\node (C1) [below=1.2cm of B1] {$\beta $};
\node (C2) [right of=C1] {$d$};
\node (C3) [right of=C2] {$ 0 $};
\node (C4) [right of=C3] {$ $}; 
\node (C5) [right of=C4] {$\cdots$};
\node (D1) [below=1.2cm of C1] {$ n $};
\node (D2) [right of=D1] {$ 0 $};
\node (D3) [right of=D2] {$ $};
\node (D4) [right of=D3] {$ $};
\node (D5) [right of=D4] {$\cdots$};
\node (E1) [below=1.2cm of D1] {$\vdots$};
\node (E2) [right of=E1] {$\vdots$};
\node (E3) [right of=E2] {$\vdots$}; 
\node (E4) [right of=E3] {$\vdots$};

\draw[->,lightgray] (A1) to node {\small $ $} (A2);
\draw[->,lightgray] (A2) to node {$ $} (A3);
\draw[->,lightgray] (A3) to node {$ $} (A4);
\draw[->,lightgray] (A4) to node {\small $ $} (A5);
\draw[->,lightgray] (B1) to node {$ $} (B2);
\draw[->,lightgray] (B2) to node {$ $} (B3);
\draw[->,lightgray] (B3) to node {$ $} (B4);
\draw[->,lightgray] (B4) to node {\small $ $} (B5);
\draw[|->] (C1) to node {\small \color{blue} $\hat{\pi}^{\ast} F  \cup  $} (C2);
\draw[|->] (C2) to node {\small \color{blue} $p^{\ast} $} (C3);
\draw[->,lightgray] (C3) to node {$ $} (C4);
\draw[->,lightgray] (C4) to node {\small $ $} (C5);
\draw[|->] (D1) to node {\small \color{blue} $F \cup  $} (D2);
\draw[->,lightgray] (D2) to node {$ $} (D3);
\draw[->,lightgray] (D3) to node {$ $} (D4);
\draw[->,lightgray] (D4) to node {$ $} (D5);

\draw[->,lightgray] (A1) to node [swap] {\small $ $} (B1);
\draw[->,lightgray] (B1) to node [swap] {$ $} (C1);
\draw[|->] (C1) to node [swap] {\small \color{red} $\hat{\pi}_{\ast} $} (D1);
\draw[->,lightgray] (D1) to node [swap] {$ $} (E1);
\draw[->,lightgray] (A2) to node [swap] {$ $} (B2);
\draw[->,lightgray] (B2) to node [swap] {$ $} (C2);
\draw[|->] (C2) to node [swap] {\small \color{red} $\hat{\pi}_{\ast}$} (D2);
\draw[->,lightgray] (D2) to node [swap] {$ $} (E2);
\draw[->,lightgray] (A3) to node [swap] {$ $} (B3);
\draw[->,lightgray] (B3) to node [swap] {$ $} (C3);
\draw[->,lightgray] (C3) to node [swap] {$ $} (D3);
\draw[->,lightgray] (D3) to node [swap] {$ $} (E3);
\draw[->,lightgray] (A4) to node [swap] {$ $} (B4);
\draw[->,lightgray] (B4) to node [swap] {$ $} (C4);
\draw[->,lightgray] (C4) to node [swap] {$ $} (D4);
\draw[->,lightgray] (D4) to node [swap] {$ $} (E4);

\end{tikzpicture}
\end{equation*}
Thus $ F \cup n= 0$, and so either $n = 0$, or $F = 0$.\footnote{Recall that $n \in H^0(M) = \ZZ$, since $M$ is connected.} If $F$ = 0, then $d = 0 \cup \beta  = 0$, and the dual flux is unique. If $n = 0$, then $\beta  = \hat{\pi}^{\ast} \gamma$ by exactness of the first column, and it follows that $d = \hat{\pi}^{\ast}(F \cup \gamma)$ by commutativity of the square. 
\begin{equation*}
\begin{tikzpicture}[baseline=(current  bounding  box.center)]
\node (A1) {$ $};
\node (A2) [right of=A1] {$ $};
\node (A3) [right of=A2] {$ $};
\node (A4) [right of=A3] {$ $};
\node (A5) [right of=A4] {$\cdots$}; 
\node (B1) [below=1.2cm of A1] {$\gamma $};
\node (B2) [right of=B1] {$ $};
\node (B3) [right of=B2] {$ $};
\node (B4) [right of=B3] {$ $};
\node (B5) [right of=B4] {$\cdots$};
\node (C1) [below=1.2cm of B1] {$\beta $};
\node (C2) [right of=C1] {$d$};
\node (C3) [right of=C2] {$ $};
\node (C4) [right of=C3] {$ $}; 
\node (C5) [right of=C4] {$\cdots$};
\node (D1) [below=1.2cm of C1] {$n=0 $};
\node (D2) [right of=D1] {$ $};
\node (D3) [right of=D2] {$ $};
\node (D4) [right of=D3] {$ $};
\node (D5) [right of=D4] {$\cdots$};
\node (E1) [below=1.2cm of D1] {$\vdots$};
\node (E2) [right of=E1] {$\vdots$};
\node (E3) [right of=E2] {$\vdots$}; 
\node (E4) [right of=E3] {$\vdots$};

\draw[->,lightgray] (A1) to node {\small $ $} (A2);
\draw[->,lightgray] (A2) to node {$ $} (A3);
\draw[->,lightgray] (A3) to node {$ $} (A4);
\draw[->,lightgray] (A4) to node {\small $ $} (A5);
\draw[|->] (B1) to node {\small \color{blue}$ F \cup  $} (B2);
\draw[->,lightgray] (B2) to node {$ $} (B3);
\draw[->,lightgray] (B3) to node {$ $} (B4);
\draw[->,lightgray] (B4) to node {\small $ $} (B5);
\draw[|->] (C1) to node {\small \color{blue} $\hat{\pi}^{\ast} F \cup   $} (C2);
\draw[->,lightgray] (C2) to node {$ $} (C3);
\draw[->,lightgray] (C3) to node {$ $} (C4);
\draw[->,lightgray] (C4) to node {\small $ $} (C5);
\draw[->,lightgray] (D1) to node {$ $} (D2);
\draw[->,lightgray] (D2) to node {$ $} (D3);
\draw[->,lightgray] (D3) to node {$ $} (D4);
\draw[->,lightgray] (D4) to node {$ $} (D5);

\draw[->,lightgray] (A1) to node [swap] {\small $ $} (B1);
\draw[|->] (B1) to node [swap] {\small \color{red}$\hat{\pi}^{\ast} $} (C1);
\draw[|->] (C1) to node [swap] {\small \color{red} $\hat{\pi}_{\ast} $} (D1);
\draw[->,lightgray] (D1) to node [swap] {$ $} (E1);
\draw[->,lightgray] (A2) to node [swap] {$ $} (B2);
\draw[|->] (B2) to node [swap] {\small \color{red}$\hat{\pi}^{\ast} $} (C2);
\draw[->,lightgray] (C2) to node [swap] {$ $} (D2);
\draw[->,lightgray] (D2) to node [swap] {$ $} (E2);
\draw[->,lightgray] (A3) to node [swap] {$ $} (B3);
\draw[->,lightgray] (B3) to node [swap] {$ $} (C3);
\draw[->,lightgray] (C3) to node [swap] {$ $} (D3);
\draw[->,lightgray] (D3) to node [swap] {$ $} (E3);
\draw[->,lightgray] (A4) to node [swap] {$ $} (B4);
\draw[->,lightgray] (B4) to node [swap] {$ $} (C4);
\draw[->,lightgray] (C4) to node [swap] {$ $} (D4);
\draw[->,lightgray] (D4) to node [swap] {$ $} (E4);

\end{tikzpicture}
\end{equation*}
That is, the ambiguity of the dual flux is actually determined by an element $\gamma \in H^1(M)$. We know, however, that since $U(1)$ is a model for $K(\ZZ,1)$ there is a natural isomorphism $[M,U(1)] \cong H^1(M,\ZZ)$. That is, the element $\gamma \in H^1(M)$ corresponds to a map $\varphi: M \to U(1)$. Such a map induces an automorphism, $\overline{\varphi}$, on $E$ where the $U(1)$ acts by rotation on the fibers:
\begin{center}
	\begin{tikzpicture}
	\node (E) {$E$};
	\node (empty) [right=0.5cm of E]{};
	\node (hatE) [right=0.5cm of empty] {$E$};
	\node (M) [below=1.5cm of empty] {$M$};
	\node (A4) {};
	\draw[->] (E) to node [swap] {$\pi$} (M);
	\draw[->] (hatE) to node {$\pi$} (M);
	\draw[->] (E) to node {$\overline{\varphi}$} (hatE);
	\end{tikzpicture}
\end{center}
It can then be shown that $\overline{\varphi}^{\ast}\widehat{H}' = \widehat{H}$ (see \cite{BS05} for details). 
\end{proof}
The proof that two putative fluxes are related by the pullback of an automorphism given in \cite{BS05} is a little abstract, and certainly not appealing to the average physicist. There is, however, a nice way to see how the ambiguity determined by the element $\gamma \in H^1(M)$ manifests itself. Consider what happens if we change the dual connection $\widehat{A}$ by a large gauge transformation. That is, we consider 
\begin{align*}
\widehat{A}' = \widehat{A} + \gamma,
\end{align*}
where we require $\gamma$ to be closed, although not necessarily exact. That is, $\gamma \in H^1(M)$. Since $\dd \gamma = 0$, the dual curvature $\widehat{F}$ (and therefore the isomorphism class of the bundle) is invariant, but the dual $H$-flux satisfies (\ref{dualH}), and so changes as
\begin{align*}
\widehat{H}' &= H_3 - F \wedge \widehat{A}' \\
&= H_3 - F \wedge \widehat{A} - F \wedge \gamma \\
&= \widehat{H} - F \wedge \gamma.
\end{align*}
That is, $\widehat{H} - \widehat{H}' = F \wedge \gamma$. The ambiguity $d = \hat{\pi}^{\ast} (F \cup \gamma)$ appearing in the proof of Theorem \ref{theorem:topT-dual} is simply the cohomological version of this. Note that
 \begin{align}
 H^1(M,\ZZ) \cong \textrm{Hom}(\pi_1(M),\ZZ),
\end{align}
so that the dual flux is unique if $M$ is simply connected (or indeed, if $\pi_1(M)$ is finite).
\subsection{Examples: topological T-duality}
\subsubsection{$S^3$ with no flux}
\label{S3Topology}
We know from Section \ref{S3Buscher} that $S^3$ provides a nice example of T-duality using the Buscher rules. The dual metric appeared to be a product metric on $S^2 \times S^1$, together with a non-zero $B$-field. Drawing conclusions about the topology of the manifold from coordinate descriptions is dangerous, since coordinates are not guaranteed to be globally defined. Luckily, as we shall see, this guess agrees with the topological description of T-duality.

Consider $S^3$ as the Hopf fibration - that is, as a principal $U(1)$-bundle over $S^2$. As noted in Appendix \ref{app:cohomology}, we have $H^2(S^2,\ZZ) \cong \ZZ$, and so isomorphism classes of principal $U(1)$-bundles over $S^2$ are classified by integers. Indeed, the Hopf fibration corresponds to the generator of this group. We can explicitly write the curvature $F$ in coordinates by rewriting the round metric in terms of a metric on the base $S^2$ and a connection:
\begin{align*}
ds^2 = \dd \eta^2 + \frac{1}{4} \sin^2 (2 \eta)\dd \xi_2^2 + \frac{1}{4} \Big( \dd \xi_1 - \cos (2\eta) \dd \xi_2 \Big)^2.
\end{align*}
The connection here is given by $A = \frac{1}{2} \Big( \dd \xi_1 - \cos (2\eta) \dd \xi_2 \Big)$. The curvature is therefore
\begin{align*}
F = \dd A =  \sin (2\eta) \dd \eta \wedge \dd \xi_2.
\end{align*}
The integer associated to this curvature via $H^2(S^2,\ZZ) \cong \ZZ$ is simply 
\begin{align*}
\frac{1}{4 \pi} \int_{S^2} F = 1,
\end{align*}
justifying our claim that the Hopf fibration is the generator of the group of principal $U(1)$-bundles over $S^2$.
The $B$-field is identically zero, and so the flux is trivial: $H = 0$. To obtain the dual bundle, we look at the pushforward of the $H$-flux, which is of course 0. That is, the curvature of the dual bundle is zero:
\begin{align*}
\widehat{F} = 0.
\end{align*}
The existence of a flat connection on a principal bundle is not quite enough to conclude the bundle is trivial. If the base is simply connected,\footnote{A manifold $M$ is simply connected if $\pi_1 (M) = 0.$} however, then the existence of a flat connection implies that the bundle is trivial \cite{KNvol1}. Since $S^2$ is simply connected, we conclude that the dual space is the trivial bundle $S^2 \times S^1$. The dual metric obtained by the Buscher rules is then recognised as the product of the round metric on the base $S^2$, and the standard flat metric on the fiber $S^1$:
\begin{align*}
\widehat{\dd s^2} & = \underbrace{\dd \eta^2 + \frac{1}{4} \sin^2 (2 \eta) \dd \xi_2^2}_{\dd s^2 (S^2)} + \underbrace{4 \dd \widehat{\xi_1}^2}_{\dd s^2 (S^1)},
\end{align*}
where the dual connection is $\widehat{A} = 2 \dd \widehat{\xi}_1$, which satisfies $\widehat{F} = \dd \widehat{A} = 0$ as expected.
The dual $B$-field obtained from the Buscher rules is given by
\begin{align*}
\widehat{B} = \frac{1}{2} \cos (2 \eta) \dd \xi_2 \wedge \dd \widehat{\xi_1},
\end{align*}
from which we obtain
\begin{align*}
\widehat{H} =  \sin(2 \eta) \dd \eta \wedge \dd \xi_2 \wedge \dd \widehat{\xi_1}. 
\end{align*}
This flux is non-trivial in cohomology, which can be determined by integrating $\widehat{H}$ over the manifold $S^2\times S^1$.\footnote{Since this integral is non-zero, it follows from Stokes' theorem that $\widehat{H}$ is not an exact form. }
We can calculate the pushforward of $\widehat{H}$, which on forms is simply integration over the fiber, to obtain
\begin{align*}
\frac{1}{2 \pi} \int_{\widehat{S}^1} \widehat{H} =  \sin (2 \eta) \dd \eta \wedge \dd \xi_2 = F.
\end{align*}
\subsubsection{Lens spaces}
\label{Lenstopology}
A Lens space is a quotient of $S^3$ by a $\ZZ_k$ action. We begin with $S^3$, thought of as the unit sphere in $\RR^4 = \CC^2$. That is,
\begin{align*}
S^3 = \left\{ (z_1,z_2) \in \CC^2 : |z_1|^2 + |z_2|^2 = 1 \right\}.
\end{align*}
This space has a free $U(1)$-action on it:
\begin{align*}
e^{2 \pi i\theta}: (z_1,z_2) \mapsto (e^{2\pi i \theta}z_1, e^{2 \pi i \theta}z_2).
\end{align*} 
In Hopf coordinates, this is just the action of the Killing vector $\pr_{\xi_1}$. Consider now the action of the discrete subgroup $\ZZ_k \subset U(1)$, for some fixed $k \in \NN$. The action looks like 
\begin{align*}
e^{2 \pi i\theta}: (z_1,z_2) \mapsto (e^{\frac{2\pi i}{k}}z_1, e^{\frac{2 \pi i }{k}}z_2).
\end{align*} 
This action is free, and the (smooth) quotient space is the Lens space $L(k,1)$. It inherits a metric from the round metric of $S^3$, given in coordinates by
\begin{align}
\label{Lensmetric}
\dd s^2 = \dd \eta^2 + \frac{1}{4} \sin^2 (2 \eta) \dd \xi_2^2 + \frac{1}{4} \Big( \dd \xi_1 - k \cos (2\eta) \dd \xi_2 \Big)^2.
\end{align}
This space is also a principal $U(1)$-bundle over $S^2$, and a glance at (\ref{Lensmetric}) tells us that a connection is given by
\begin{align*}
A = \frac{1}{2} \Big( \dd \xi_1 - k \cos(2 \eta) \dd \xi_2 \Big),
\end{align*}
from which we can calculate the Chern class:
\begin{align*}
F = \dd A = k \sin (2\eta) \dd \eta \wedge \dd \xi_2.
\end{align*}
The integer associated to this curvature via $H^2(S^2,\ZZ) \cong \ZZ$ is 
\begin{align*}
\frac{1}{4 \pi} \int_{S^2} F = k.
\end{align*}
Let us now consider this space, equipped with $j$ units of flux - that is, we have 
\begin{align*}
\dd s^2 &= \dd \eta^2 + \frac{1}{4} \sin^2 (2 \eta) \dd \xi_2^2 + \frac{1}{4} \Big( \dd \xi_1 - k \cos (2\eta) \dd \xi_2 \Big)^2 \\[1em]
B &= -\frac{j}{2} \cos(2\eta) \dd \xi_2 \wedge \dd \xi_1,
\end{align*}
so that $H = j \sin (2\eta) \dd \eta \wedge \dd \xi_2 \wedge \dd \xi_1$. The dual Chern class is determined by $\widehat{F} =\pi_{\ast}H$, so we have
\begin{align*}
\widehat{F} &= \frac{1}{2\pi} \int_{S^1} H = j \sin (2 \eta) \dd \eta \wedge \dd \xi_2,
\end{align*}
with associated integer
\begin{align*}
\frac{1}{4\pi} \int_{S^2} \widehat{F} = j.
\end{align*}
It follows that for $k \not= j$, the topology of the dual space will be different to the topology of the original space. The dual metric and $B$-field can be calculated from the Buscher rules:
\begin{align*}
\widehat{\dd s}^2 &= \dd \eta^2 + \frac{1}{4} \sin^2 (2 \eta) \dd \xi_2^2 + \frac{1}{4} \Big( \dd \xi_1 - j \cos (2\eta) \dd \xi_2 \Big)^2 \\[1em]
\widehat{B} &= -\frac{k}{2} \cos(2\eta) \dd \xi_2 \wedge \dd \xi_1,
\end{align*}
and we see that the dual space is also a Lens space $L(j,1)$, equipped with $k$ units of flux. That is, the original and T-dual invariants are swapped under T-duality:
\begin{align*}
[F] = k & & [\widehat{F}] = j \\[1em]
[H] = j & & [\widehat{H}] = k.
\end{align*}
This is a particularly clear example of how T-duality acts by interchanging flux with topology. 
\subsubsection{$\mathbb{T}^3$ with $H$-flux}
\label{T3fluxtopology}
We have already discussed the T-dual of $\mathbb{T}^3$ with one unit of flux in Section \ref{T3Buscher}. Let us generalise slightly to $k$ units of flux, with $k \in \ZZ$. The original metric and $B$-field are:
\begin{subequations}
\begin{align*}
\dd s^2 &= \dd x^2 + \dd y^2 + \dd z^2 \\[1em]
B &= -  k x \dd y \wedge \dd z.
\end{align*}
\end{subequations}
The dual metric and $B$-field follow from a straightforward application of the Buscher rules:
\begin{subequations}
\begin{align}
\label{Nilg}
\widehat{g} &= \dd x^2 + \dd y^2 + \Big( \dd \hat{z} - k x \dd y \Big)^2 \\[1em]
\label{NilB}
\widehat{B} &= 0.
\end{align}
\end{subequations}
We are interested in the topology of the dual space. Our original space, $\mathbb{T}^3$, is a trivial $S^1$-bundle over the 2-torus: 
\begin{center}
	\begin{tikzpicture}
	\node (F) {$S^1_z$};
	\node (E) [right=1cm of A1] {$\mathbb{T}^3$};
	\node (M) [below=1cm of E] {$S^1_x \times S^1_y$};

	\draw[right hook->] (F) to node [swap] {$ $} (E);
	\draw[->] (E) to node {$ $} (M);
	\end{tikzpicture}
\end{center}
A connection for the flat metric on $\mathbb{T}^3$ is simply $A = \dd z$, and we easily see that it is flat since $F = \dd A = 0$. Note that although the base here is \emph{not} simply connected, our bundle is still trivial by construction. The topology of the dual bundle is determined by the curvature $\widehat{F}$, which we obtain by calculating the pushforward of $H$. Now, $H$ is given by 
\begin{align*}
H = \dd B = - k \dd x \wedge \dd y \wedge \dd z,
\end{align*}
and the pushforward acts on forms by integration over the fiber, so that 
\begin{align*}
\widehat{F} = \pi_{\ast} H = \int_{S^1_z} -k \dd x \wedge \dd y \wedge \dd z = -k \dd x \wedge \dd y.
\end{align*}
 Alternatively, a quick glance at (\ref{Nilg}) is enough to see that a connection for the dual bundle $\hat{\pi}: \widehat{E} \to M$ is given by $\widehat{A} = \dd \hat{z} - kx \dd y$. The curvature of this connection is 
\begin{align*}
\widehat{F} = -k \dd x \wedge \dd y.
\end{align*}
The curvature $\widehat{F}$ of the dual bundle determines its isomorphism class, and since $\widehat{F}$ is non-trivial in cohomology (for $k \not=0$), the bundle is not a trivial bundle. That is, the dual space is topologically distinct from $\mathbb{T}^3$. Indeed, from Appendix \ref{app:cohomology}, we have that $H^2(S^1 \times S^1 , \ZZ) = \ZZ$, so for each $k \in \ZZ$, the dual space is a topologically distinct manifold. In the physics literature, $k$ is usually taken to be 1 (or -1), and the background (i.e. the metric and $B$-field) is referred to as the $f$-flux background or the twisted torus. In the mathematics literature, the manifold is usually referred to as the 3D Heisenberg Nilmanifold. We will denote the dual space by $\textrm{Nil}(k)$, or simply $\textrm{Nil}$ if $k=1$. We may also, on occasion, use the notation $f_{xy}\!^z$ to refer to the background. This notation is in common usage in the physics literature, and makes it clear that it is obtained from the three torus with flux, $ T_{xyz}$, by T-dualising along the $z$-coordinate. We can study the periodicity of the dual coordinate simply by requiring the metric to be globally defined. The coordinates $x$ and $y$ are coordinates on the base, so they should still have the same periodicity as the original space:
\begin{align*}
x &\mapsto x' \sim x+1 \\[1em]
y &\mapsto y' \sim y+1 \\[1em]
z &\mapsto z'.
\end{align*}
We now wish to see that the metric does not change under these identifications. The one forms $\dd x$ and $\dd y$ are clearly invariant, so for the metric to be defined under these identifications, we require
\begin{align*}
\dd \hat{z} ' - k x' \dd y' = \dd \hat{z} - kx \dd y.
\end{align*}
This happens precisely when $\hat{z}' = z + ky$. See Appendix \ref{app:fflux} for more discussion of this procedure. 

We can understand the dual space a little more from a group theory perspective. This understanding will be useful in Chapter \ref{chptr:Ch3} when we study non-abelian T-duality. Let us now fix $k =1$, and consider the Heisenberg group, consisting of upper triangular $3 \times 3$ real matrices with 1 along the diagonal:
\begin{align*}
\textrm{Heis} (\RR) = \left\{ \left( \begin{matrix} 
1 &  x & z \\
0 & 1 & y \\
0 & 0 & 1
\end{matrix} \right) : x,y,z \in \RR \right\}
\end{align*}
The group operation for this group is simply matrix multiplication. We can take the quotient of this group by the subgroup consisting of only integer entries:
\begin{align*}
\textrm{Heis} (\ZZ) = \left\{ \left( \begin{matrix} 
1 &  a & c \\
0 & 1 & b \\
0 & 0 & 1
\end{matrix} \right) : a,b,c \in \ZZ \right\}.
\end{align*}
This is not a normal subgroup, so the quotient space won't be a group, but it turns out to be precisely the manifold $\textrm{Nil}$. That is, $\textrm{Nil} = \textrm{Heis}(\RR) / \textrm{Heis}(\ZZ)$. The group $\textrm{Heis}(\RR)$ acts on $\textrm{Nil}$ from the left (and right), and the vector fields generating this action are the right-invariant (and left-invariant) vector fields.\footnote{This is not a typo. The left action of a group on itself is generated by the right-invariant vector fields, and vice versa.} In terms of the coordinates $\{x,y,z\}$, the left and right invariant vector fields are 
\begin{align*}
L_a &= \left\{ \pr_x, \, \pr_y + x \pr_z,\,  \pr_z \right\} \\[1em]
R_a &= \left\{ \pr_x+y\pr_z,\, \pr_y,\, \pr_z \right\}.
\end{align*}
The left and right invariant one-forms are dual to these vector fields:
\begin{align*}
\lambda^a &= \left\{ \dd x, \, \dd y , \, \dd z - x \dd y \right\} \\[1em]
\rho^a &= \left\{ \dd x , \, \dd y , \, \dd z -  y \dd x \right\}.
\end{align*}
We can now see that the metric we obtained from the Buscher rules, (\ref{Nilg}), is simply the left-invariant metric for Heisenberg group:\footnote{Note that the Heisenberg group is not a semisimple group (it is nilpotent), and therefore we are not guaranteed a bi-invariant metric. We have $g_L \not = g_R$.}
\begin{align*}
g_L &= \delta_{ij} \lambda^i \lambda^j \\
&= (\lambda^1)^2 + (\lambda^2)^2 + (\lambda^3)^2.
\end{align*}
The metric (\ref{Nilg}) is independent of the coordinates $y$ and $z$, so has the two obvious Killing vectors $\pr_y$ and $\pr_{z}$. On the other hand, because of the identity
\begin{align*}
\Lie_{R_a} \lambda^b = 0,
\end{align*}
it follows that the metric is invariant under the full set of right-invariant vector fields. That is,
\begin{align*}
\Lie_{\pr_x + x \pr_{z}} g_L &=  0 \\[1em]
\Lie_{\pr_y} g_L &= 0 \\[1em]
\Lie_{\pr_z} g_L &= 0.
\end{align*}
The right-invariant vector fields form a non-abelian Lie algebra whose only non-zero commutator is
\begin{align*}
[R_1,R_2] &= - R_3.
\end{align*}
This example will appear again in Section \ref{NATDHeis} when we discuss examples of non-abelian T-duality. 
\subsection{Isomorphism of twisted cohomology and twisted K-theory}
\label{TopKtheory}
In Section \ref{RRfluxes} we introduced the RR fluxes, and argued that they should be classified by (twisted) cohomology, or twisted K-theory. If T-duality is meant to to be a symmetry of the full string theory, then we expect therefore that the twisted cohomologies and twisted K-theories of a space and its T-dual should match, albeit with a shift in degree to account for the transition between type IIA and type IIB. This was shown in \cite{BEM03} by using the convenient Hori formula for the transformation of the RR fluxes \cite{Hori}. Recall that in Section \ref{RRfluxes}, we packaged the RR fluxes into a polyform $G$ which was $d_H$-closed. The Hori formula states that the dual RR fluxes are encoded in the dual polyform:
\begin{align}
\label{Horiformula}
\widehat{G} &= \int_{S^1} e^{\widehat{A} \wedge A} G. 
\end{align}
The operation of producing the new fluxes therefore consists of the following steps: package the RR fluxes into a polyform $G$ and pull back to the correspondence space $E\times_M \widehat{E}$, wedge with $e^{\widehat{A} \wedge A}$, then integrate along the $S^1$ fiber of $E$ to obtain the dual polyform $\widehat{G}$. Since $\dd \, (\widehat{A} \wedge A) = H - \widehat{H}$ from (\ref{BminusB}), we have that 
\begin{align*}
d_{\widehat{H}} \widehat{G} &= \int_{S^1} e^{\widehat{A} \wedge A} d_H G,
\end{align*}
so that the T-duality transformation maps $d_H$-closed forms to $d_{\widehat{H}}$-closed forms. For a more thorough discussion of the isomorphism of the twisted cohomologies, and the lift to twisted K-theories, we refer the reader to \cite{BEM03}.
\subsection{Topological T-duality for torus bundles}
\label{sec:torusbundles}
The Buscher rules have a straightforward generalisation to multiple commuting killing vectors, and so it is natural to consider topological T-duality for higher rank torus bundles. We very quickly run into an issue, however. T-duality for a torus bundle should represent multiple concurrent T-dualities, but we saw in Section \ref{T3Buscher} that it is possible for a chain of dualities to result in a space which is no longer a manifold - indeed the third duality in the chain
\begin{align}
T_{xyz} \stackrel{\pr_z}{\longleftrightarrow} f_{xy} \!^z \stackrel{\pr_y}{\longleftrightarrow} Q_x \!^{yz} \stackrel{\pr_x}{\longleftrightarrow} R^{xyz}.
\end{align}
doesn't even appear to be well-defined, since $\pr_x$ is not an isometry of the space $Q_x \!^{yz}$. It would be nice to have a characterisation of the topological structure for multiple T-dualities, and a criterion for when we can perform the duality and obtain a (geometric) dual space. This was first studied in \cite{BHM03,BHM04}. We include here the relevant details, and comment on the similarities and differences between rank one and higher rank torus bundles. As we shall see, the obstruction to performing multiple T-dualities lies in the structure of the $H$-flux. 

We begin by considering a principal torus bundle
\begin{center}
	\begin{tikzpicture}
	\node (F) {$\mathbb{T}^n$};
	\node (E) [right=1cm of F] {$E$};
	\node (M) [below=1cm of E] {$M$};
	\draw[right hook->] (F) to node  {$ $} (E);
	\draw[->] (E) to node [right] {$\pi$} (M);
	\end{tikzpicture}
\end{center}
Of course, when $n=1$ we should recover the results of Section \ref{subsec:TopCircBund}. 

The isomorphism classes of principal torus bundles over $M$ are classified by $H^2(M,\ZZ^n)$, and we can identify the image of this in de Rham cohomology with $H^2(M,\fg)$, where $\fg$ is the Lie algebra of $\mathbb{T}^n$, thought of as the Lie group $U(1)^n$. We denote by $\fg^{\ast}$ the dual Lie algebra to $\fg$.\footnote{Here, the word \emph{dual} refers to a vector space dual. That is,  $\fg^{\ast}$ is the vector space of linear maps from $\fg$ to $\RR$. } Given an element $X \in \fg$ and an element $\alpha \in \fg^{\ast}$, there is a natural pairing $\alpha(X) \in \RR$.

Physicists will be familiar with using dimensional reduction to decompose forms on a fibration in terms of lower degree forms on the base. For those who are not, recall from Appendix \ref{app:cohomology} the K\"{u}nneth theorem, Theorem \ref{Kunneth}. We can apply the K\"{u}nneth theorem to the bundle $\pi:E \to M$ locally, since a principal torus bundle is locally $M \times \mathbb{T}^n$. Combining this with $H^q(\mathbb{T}^n) = \bigwedge^q \fg^{\ast}$, this gives a (local) decomposition of a $\mathbb{T}^n$-invariant form $\omega \in \Omega^k (E)$ as
\begin{align*}
\omega &= \omega_k \otimes V_0 + \omega_{k-1} \otimes V_1 + \dots + \omega_1 \otimes V_{k-1} + \omega_0 \otimes V_{k}
\end{align*}
Making this decomposition global requires choosing a principal $\mathbb{T}^n$ connection $A \in \Omega^1(E,\fg)$. Choosing such a connection determines an isomorphism 
\begin{align*}
f_A : \bigoplus_{p+q = k} \left( \Omega^{p}(M) \otimes \bigwedge^q \fg^{\ast} \right) \longrightarrow \Omega^k_{inv} (E)
\end{align*}
which acts by
\begin{align*}
f_A \Big( \omega \otimes (\alpha_1 \wedge \dots \wedge \alpha_q) \Big) = \omega \wedge A(\alpha_1) \wedge \dots A(\alpha_q).
\end{align*}
The upshot of all this is that we can think of the $H$-flux as a tuple $(H_3, H_2, H_1, H_0)$, where each $H_i \in \Omega^i (M) \otimes \bigwedge^{3-i}\fg^{\ast}$. Similarly, we can view the curvature $F$ as a tuple $(F_2,0,0)$, with $F_i \in \Omega^i \otimes \bigwedge^{2-i} \fg^{\ast}$. The tuples in these decompositions have a nice geometric interpretation - the differential forms $H_i$ are simply the components of the form $H$ with $i$ legs in the base.  

Let us think about this for the $n=1$ case for a moment, since in that case we have considerable simplification. When we have a principal circle bundle, the dimension of $\fg^{\ast}$ is 1 and so $\bigwedge^q \fg^{\ast} = 0$ for $q \geq 2$. Then we have 
\begin{align*}
H &= (H_3 , H_2, 0, 0) \\[1em]
F &= (F_2, 0, 0).
\end{align*}
T-duality then acts by interchanging $H_2$ and $F_2$. That is, 
\begin{align*}
\widehat{H } &= (H_3 , F_2, 0, 0) \\[1em]
\widehat{F} &= (H_2, 0, 0).
\end{align*}
The form $\widehat{F}$ is then identified with the curvature of a  principal $U(1)$ bundle $\hat{\pi}: \widehat{E} \to M$, and $\widehat{H}$ is interpreted as the dual flux. This is precisely the content of (\ref{Hdimensional}).

More generally, T-duality acts by the interchange $(F_2,F_1,F_0) \longleftrightarrow (H_2,H_1,H_0)$, so that 
\begin{equation}
\label{torusTduality}
\begin{aligned}[c]
H &= (H_3 , H_2, H_1, H_0) \\[1em]
F &= (F_2, F_1, F_0).
\end{aligned}
\qquad \qquad \qquad
\begin{aligned}[c]
\widehat{H } &= (H_3 , F_2, F_1, F_0) \\[1em]
\widehat{F} &= (H_2, H_1, H_0).
\end{aligned}
\end{equation}
We now see an immediate problem if either $H_1$ or $H_0$ are non-zero. Namely, the dual $\widehat{F}$ can no longer be identified with the curvature of a principal $\mathbb{T}^n$ bundle. When $n=1$, these terms vanish identically for dimensional reasons, and so we always have a dual which is a principal bundle. When $n=2$, the term $H_0$ always vanishes for dimensional reasons, but $H_1$ can be non-zero. When $n \geq 3$, we can have both $H_1$ and $H_0$ which are non-zero. Both of these cases will be discussed in the following sections. For the moment, let us refer to $H$-fluxes of the form $H = (H_3,H_2,0,0)$ as \emph{admissible fluxes}. Any $\mathbb{T}^n$ bundle together with an admissible flux determines a unique T-dual given by (\ref{torusTduality}) \cite{BHM04}.

\subsection{The algebraic approach}
\label{Cstar}
\subsubsection{T-folds}
In the previous section, we saw that there was an obstruction to constructing T-duals for higher-rank torus bundles when the $H$-flux had multiple legs along the fiber directions. How do we reconcile this with the na\"{i}ve Buscher rule calculation for the T-dual of the flat torus with flux? That is, the duality chain
\begin{align}
T_{xyz} \stackrel{\pr_z}{\longleftrightarrow} f_{xy} \!^z \stackrel{\pr_y}{\longleftrightarrow} Q_x \!^{yz}
\end{align}
suggests that we should be able to think of $\mathbb{T}^3$ as a trivial $\mathbb{T}^2$ bundle over $S^1$ and T-dualise. On the other hand, the $H$-flux decomposes as $H = (0,0, H_1 , 0)$, and so there is an obstruction to doing so. Which computation is correct? To understand the answer, we first note that the metric and $B$-field of the $Q$-flux background, given by (\ref{Qflux}), are not globally defined tensor fields on a compact manifold. Of course, the $B$-field is usually not globally-defined,\footnote{The $H$-flux, however, should be globally defined.} but the metric definitely should be. We can see this by looking at how the metric changes under the identification $x \sim x+1$. Although the one-forms $\{\dd x, \dd \hat{y} , \dd \hat{z} \}$ are invariant, the metric includes the function 
\begin{align}
\frac{1}{1+x^2},
\end{align}
which is not invariant. The resolution of this paradox is therefore that both calculations are correct! The Buscher rules give a dual space, and the topological argument says that this space cannot be a principal torus bundle. Indeed, the dual space is not even a manifold! This tells us that we need to expand our notion of what it means to be a string background. 

This $Q$-flux is a strange object - walking once around a circle and returning to the same spot, we find that the landscape has changed! Miraculously, however, this is not something a string notices. To a string, this background is the same as the flat torus with $H$-flux. The $Q$-flux background is not a manifold, since we cannot construct open patches for it and glue them together in a consistent way using diffeomorphisms. On the other hand, if we allow more general transformations than just diffeomorphisms, that is if we also allow T-duality transformations, then we \emph{can} glue together these open patches. Such an object goes by the name of a T-fold. 

\subsubsection{Non-commutative backgrounds and $C^{\ast}$-algebras}
\label{noncommtorus}
There is a very nice description of topological T-duality for higher-rank torus bundles using the theory of $C^{\ast}$-algebras, which incorporates both the admissable fluxes, as well as the T-fold background. The description is reasonably technical, and is not content with which most physicists will be familiar. For this reason, we include here only a bare-bones discussion of the construction of T-duality in this setting, and refer the interested reader to the original papers for a more detailed account \cite{MR04a,MR04b,MR05}. Note that non-commutative string backgrounds have appeared previously in the physics literature \cite{SW,CDS}. 

Let $E$ be our spacetime, which we will eventually think of as a principal $\mathbb{T}^n$-bundle over $M$. Let us assume for simplicity that $E$ is compact.\footnote{This assumption is not necessary, but simplifies some aspects of the following discussion.} There is a very natural (unital) commutative \CA-algebra  associated to $E$, namely the algebra of continuous functions on $E$, denoted by $C(E)$.\footnote{A \CA-algebra is an algebra (i.e. a vector space with associative multiplication) on a complete normed space, together with an involution satisfying certain `nice' properties. The standard examples are the complex numbers, $\CC$, or the bounded operators on a Hilbert space, $\mathcal{B}(\mathcal{H})$.} In the theory of \CA-algebras, this is not just a `special' commutative algebra - it is in a sense the \emph{only} commutative \CA-algebra. To be a bit more concrete, the Gelfand-Naimark theorem for commutative \CA-algebras states that every (unital) commutative \CA-algebra is isometrically isomorphic to $C(X)$ for some compact space $X$, and that $X$ is unique up to homeomorphism. Indeed, $X$ is constructed from the algebra as its spectrum. By analogy, more general (i.e. non-commutative) \CA-algebras can be thought of as the algebra of functions on a `non-commutative' space', which doesn't exist as a classical topological space. For this reason, the relationship between topology and \CA-algebras is often referred to as non-commutative topology.

We are now interested in a slight generalisation of the commutative algebras, known as \emph{continuous trace \CA-algebras} over $E$. It is result of Dixmier and Douady that after stabilisation,\footnote{That is, after tensoring with the compact operators, $\mathcal{K}(\mathcal{H})$, on an infinite dimensional separable Hilbert space.} these algebras are locally isomorphic to $C\big( E,\mathcal{K}(\mathcal{H}) \big)$ \cite{DD}. Stable isomorphism classes of continuous trace algebras with spectrum $E$ are classified by the Dixmier-Douady class, $\delta$, in $H^3(E,\ZZ)$. We will denote the continuous trace \CA-algebra associated to a space $E$ and a class $\delta \in H^3(E,\ZZ)$ by $CT(E,\delta)$. Note that $CT(E,0) = C(E)$.

Here is another way to think of $CT(E,\delta)$: Given a stable, continuous-trace algebra with spectrum $E$, there exists an algebra bundle $\mathscr{A}$ over $E$ whose fibers are the compact operators on an infinite-dimensional Hilbert space, $\mathcal{K}(\mathcal{H})$. The structure group of the bundle is $\textrm{Aut}(\mathcal{K}(\mathcal{H})) \simeq PU(\mathcal{H})$, and so bundles of this type are classified by homotopy classes of maps from $E$ to $\mathsf{B}PU(\mathcal{H})$. It turns out that $\mathsf{B}PU(\mathcal{H})$ is an Eilenberg-Maclane $K(\ZZ,3)$ space, and so bundles are classified by $H^3 (E,\ZZ)$, the Dixmier-Douady invariant. The space of sections of this bundle form a stable, continuous-trace \CA-algebra with spectrum $E$ and Dixmier-Douady invariant $\delta$, and this algebra is precisely $CT(E,\delta)$.

Let us assume, as we mentioned earlier, that $E$ is a principal $\mathbb{T}^n$ bundle over $M$ with $H$-flux. This means that $E$ has a free $\mathbb{T}^n$ action, and it is a reasonable question to ask when such an action lifts to an action on $CT(E,H)$. It is shown in \cite{MR04a} that such an action lifts precisely when the $H$-flux has the form $H = (H_3,H_2,H_1,0)$. If that is the case, the T-dual algebra is now given by $CT(E,H) \rtimes \RR^n$, which is not in general a stable, continuous-trace \CA-algebra. 

When $H_1=0$, the T-dual algebra \emph{is} a stable, continuous-trace \CA-algebra whose spectrum $\widehat{E}$ is a principal $\mathbb{T}^n$ bundle and whose Dixmier-Douady invariant $\widehat{H}$ satisfies (\ref{FHpushforward}). This is referred to in the literature as a \emph{classical T-dual}.

When $H_1 \not=0$, the dual algebra is not a stable, continuous-trace \CA-algebra, but it can still be considered as the algebra of sections of a bundle of algebras over $M$. The fiber of this bundle of algebras over a point is $A_{f(z)} \otimes \mathcal{K}(\mathcal{H})$, where $A_{\theta}$ is the irrational rotation algebra, commonly referred to as the noncommutative torus. The non-commutativity parameter $f$ is a representative $f:M \to U(1)$ of the class $\pi_{\ast} H \in H^1(M,\ZZ)$. That is, the T-dual is realised as a bundle of noncommutative tori fibered over $M$. 

The dual of the dual
\begin{align*}
\Big( CT(E,H) \rtimes \RR^n \Big) \rtimes \widehat{\RR}^n
\end{align*}
is not, in general, isomorphic to $CT(E,H)$, but it is Morita equivalent to it (and therefore has the same spectrum and K-theory of $CT(E,H)$). The Morita equivalence of these algebras has been known in the mathematical literature for some time, where it is referred to as Takai-duality \cite{Takai}.\footnote{The convenience of still being able to refer to this as T-duality is fully appreciated by the author.} 

Non-commutative topology is the relation between topology and \CA-algebras inspired by the Gelfand-Naimark theorem. In the context of string theory, however, the topological spaces in which we are interested have additional structure - in particular, they have a Riemannian metric. It is a natural question to ask whether such a structure has an analogue in the theory of \CA-algebras. Connes initiated the study of such structures in \cite{C80},\footnote{See \cite{C80English} for an English translation.} and the theory goes by the name of non-commutative geometry. It is possible to define a suitable notion of a connection in the non-commutative setting, although the objects involved are complicated. We have included a schematic correspondence between structures in the topological setting and the \CA-algebra setting in the following table:
\begin{center}
	\begin{tabular}{ | c | c || c | c | } 
		\hline
		Topology & Notation & \CA-algebra & Notation  \\
		\hline \hline
		Topological space& $X$ & Commutative \CA-algebra & $\mathcal{A} = C(X)$  \\ 
		\hline
		``Non-commutative space'' & $X$ & \begin{tabular}{@{}c@{}}Non-commutative \CA-algebra \\ with spectrum $X$ \end{tabular} &  $\mathcal{A}$   \\ 
		\hline
		Vector bundle over $X$ & $E$ &  Projective module over $\mathcal{A}$ & $\mathcal{M}$   \\ 
		\hline
		Connection on vector bundle & $\nabla$ & \begin{tabular}{@{}c@{}}Linear operator on $\mathcal{M}$ \\ satisfying Leibniz rule \end{tabular}    & $\nabla$  \\
		\hline
		Topological K-theory & $K(X)$ & Algebraic K-theory  & $K(\mathcal{A})$ \\
		\hline
	\end{tabular}
\end{center}
With this correspondence, the definition of a connection in the \CA-algebraic framework gives us a natural notion of curvature:
\begin{align}
\mathcal{R}^{\nabla} (U,V) := \nabla_U \nabla_V - \nabla_V \nabla_U - \nabla_{[U,V]}.
\end{align}
For the case that $X = \RR^n$ and $\mathcal{A} = C^{\infty}(X)$, we obtain the standard definition of a connection on a vector bundle. 

This correspondence provides an interesting possibility! One the one hand, we know that the $Q$-flux background fits into the \CA-algebraic framework, and is well-described by a bundle of non-commutative tori fibered over $S^1$. Such a bundle, by Connes' theory of non-commutative differential geometry, admits an analogue of a connection and curvature. On the other hand, the Buscher rules give us a local description of the $Q$-flux background metric, (\ref{Qfluxmetric}). Such a coordinate description allows for easy calculation of various geometric quantities, such as curvature tensors and scalars. In particular, the scalar curvature of (\ref{Qfluxmetric}) is
\begin{align*}
\mathcal{R} &= -\frac{2(5x^2 - 2)}{(x^2+1)^2}.
\end{align*} 
Of course, since $x$ is a periodic coordinate on the base $S^1$, this scalar curvature is not a well-defined quantity under the identification $x \sim x+1$. Nevertheless, it would be an interesting exercise to compare this quantity to the curvature quantities appearing on the algebraic side. We leave such an undertaking to future work. 

\subsubsection{Non-associative backgrounds}
\label{nonassoctorus}
In the previous section, we saw that as long as $H_0 = 0$, the $\mathbb{T}^n$ action of $E$ lifted to an action on $CT(E,H)$, and therefore defined a T-dual \CA-algebra. When $H_0 \not=0$, the action lifts to a \emph{twisted action}, and so we can still define a T-dual as the \emph{twisted crossed product} $CT(E,H) \rtimes_{\alpha} \RR^n$, where $\alpha$ is the tricharacter associated to $H_0$. T-duality now takes us outside the realms of \CA-algebras - the dual algebra is in general a nonassociative, noncommutative algebra, which we can realise as a bundle of \emph{nonassociative tori} fibered over $M$ \cite{BHM05}. This is the situation for the $R$-flux background, obtained as the putative third T-dual of the 3-torus with $H$-flux.

%% file: chapter3.tex
\chapter{Non-abelian T-duality}
\label{chptr:Ch3} 
Non-abelian T-duality began with the work of de la Ossa and Quevedo \cite{DQ}, which generalised the gauging procedure of Buscher to non-abelian isometries. As with many generalisations in mathematics and physics, we shall see that non-abelian T-duality does not inherit all of the nice properties that abelian T-duality enjoys, and the role that non-abelian T-duality plays in string theory is still unclear. In particular, the extension of the Buscher procedure to higher genus worldsheets is lacking, and quantum properties of the duality are largely unknown.  Despite this, it has been employed successfully as a solution generating technique in supergravity and generalised supergravity. In this Chapter we will introduce non-abelian T-duality via the Buscher procedure, although we note that a discussion of the quantum aspects of non-abelian T-duality is outside the scope of this thesis. 

Note that we shall set $4\pi \alpha' = 1$ to avoid numerical factors in front of the actions.

\section{Geometry}
\label{NabelGeometry}

\subsection{The non-abelian Buscher rules}
\label{subsec:NATDBuscher}
We begin, as with abelian T-duality, with the non-linear sigma model:
\begin{align}
\label{NANLSM}
S &=  \int \ddd^2 \!z \,  (g_{mn} + B_{mn}) \pr X^{m} \bar{\pr} X^{n} \\[1em]
&=  \int_{\Sigma} g_{mn} \dd X^{m} \wedge \star \dd X^{n} + B_{mn} \dd X^{m} \wedge \dd X^{n}.  
\end{align}
We will assume, for the moment, that the target space is a compact, non-abelian Lie group $\sG$. The group $\sG$ acts on itself by right multiplication, and the fundamental vector fields associated to this action are the left-invariant vector fields. The set of one-forms, $\{\lambda^a\}$, dual to this frame are defined by
\begin{align}
\langle L_a, \lambda^b \rangle = \delta_a^b.
\end{align}
These one-forms define a coframe, and are components of the left-invariant Maurer-Cartan form, $\lambda$, expressed in the basis $\{L_a\}$ of $\fg$:
\begin{align}
\lambda = \lambda^a L_a \in \Omega^1(\sG) \otimes  \fg.
\end{align}
The basis $\{L_a\}$ of $\fg$ define the structure constants of the Lie algebra:
\begin{align}
[L_a , L_b] &= \tilde{f}^c \!_{ab} L_c.
\end{align}
Note that $\lambda$ satisfies the Maurer-Cartan equation
\begin{align}
\label{leftMCequation}
\dd \lambda^c = - \frac{1}{2} \tilde{f}^c \!_{ab} \, \lambda^a \wedge \lambda^b,
\end{align}
and can therefore be written as $\lambda =  \sg^{-1} \, \dd \sg$ for an embedding $\sg: \sG \hookrightarrow GL(n)$. Here, we have chosen to identify the  space of left-invariant vector fields with the Lie algebra $\fg$ of $\sG$. Note that we could also have chosen to use the space of right-invariant vector fields. The two Lie algebras are isomorphic, with the differential of the inversion map on $\sG$ providing the required isomorphism. Under this isomorphism, the right-invariant vector fields have structure constants $[R_a,R_b] = f^c\!_{ab} R_c$. These are related to the structure constants by the relation $\tilde{f}^c \!_{ab} = - f^c \!_{ab}$. 
	
 We now consider a metric on $\sG$ defined in terms of this left-invariant coframe:
\begin{align}
\label{rightinvariant}
ds^2 = \delta_{ab} \,  \lambda^a \lambda^b.
\end{align}
The \emph{right}-invariant vector fields, $\{R_a\} = \{R^{\mu}_a \pr_{\mu} \}$, are isometries of the metric, since 
\begin{align}
\label{rightleftLie}
\Lie_{R_a} \lambda^b = 0.
\end{align}
More generally, we can take a $\textrm{dim}(\sG) \times \textrm{dim}(\sG)$, invertible and $\sG$-invariant matrix $E_{ab}$, and form
\begin{align}
E := E_{ab} \, \lambda^a \otimes \lambda^b.
\end{align}
Since every matrix can be written as the sum of a symmetric and antisymmetric matrix
\begin{align*}
E_{ab} &= \frac{1}{2}\left(E_{ab} + E_{ba} \right) + \frac{1}{2}\left(E_{ab} - E_{ba}\right) \\[1em]
&= g_{ab} + B_{ab},
\end{align*}
this tensor decomposes into a symmetric and an antisymmetric part:
\begin{align}
E &= g_{ab} \lambda^a \otimes \lambda^b + B_{ab} \lambda^a \otimes \lambda^b \\[1em]
&= g_{ab} \lambda^a \lambda^b + B_{ab} \lambda^a \wedge \lambda^b. 
\end{align}
We will assume that both the metric and the $B$-field are written in terms of the Maurer-Cartan forms, and we can therefore specify them by giving the matrix $E_{ab}$. The action (\ref{NANLSM}) with such a metric and $B$ field is a generalisation of the Principal Chiral Model (PCM),\footnote{The standard Principal Chiral Model corresponds to the special case with $E_{ab} = \delta_{ab}$. We shall make no distinction between the standard PCM and the generalised form we have used here.} and is often written in different notation as:
\begin{align}
\label{SPCM}
S_{PCM}[\sg] =  \int \ddd^2 z \,  E_{ab} \left( \sg^{-1} \pr \sg \right)^a \left( \sg^{-1} \bar{\pr} \sg \right)^b .
\end{align}
This notation makes it clear that the matrix $E_{ab}$ is a choice of bilinear form on the Lie algebra $\fg$ of $\sG$. It follows from (\ref{rightleftLie}) that $\Lie_{R_a} E = 0$, and so this action should be invariant under the infinitesimal action of the right-invariant vector fields. In fact, the right-invariant vector fields correspond to, i.e. integrate out to, the \emph{left} action of the group on itself, and one would expect that the PCM is also invariant with respect to this group action. It is, and another advantage to using this notation for the action is that this invariance is manifest:
\begin{align}
S_{PCM}[\sh \sg] &=  \int \dd \sigma \dd \tau \, E_{ab} \left((\sh \sg)^{-1} \pr_{\mu} (\sh \sg) \right)^a \left((\sh \sg)^{-1} \pr^{\mu} (\sh \sg) \right)^b \\[1em]
&=  \int \dd \sigma \dd \tau \, E_{ab} \left(\sg^{-1} \sh^{-1} \sh \pr_{\mu} \sg \right)^a \left(\sg^{-1} \sh^{-1} \sh \,  \pr^{\mu} \sg \right)^b \\[1em]
&= S_{PCM}[g].
\end{align}

The Buscher procedure for abelian T-duality relied on the existence of a global symmetry, which we then promoted to a local symmetry by gauging. Equivalence with the original model was enforced by including a Lagrange multiplier constraining the field strength, and then choosing a gauge which reduced to the original model. The dual model was obtained by first integrating out the gauge fields, and then gauge fixing. 

Non-abelian T-duality is based on the same basic procedure. Explicitly, we begin by gauging the rigid (i.e. global) symmetry, $\sg \mapsto \sh \sg$, of the model (\ref{SPCM}). This is equivalent to gauging the infinitesimal symmetry of (\ref{NANLSM}) generated by the right-invariant vector fields $\{R_a\}$:
\begin{align}
\de X^{\mu} = R^{\mu}_a \epsilon^a.
\end{align}
The group action is non-abelian, and the vector fields $\{R_a\}$ form a non-abelian Lie algebra with structure constants $f^c \!_{ab}$. 
As in the abelian case, we promote this global symmetry to a local one by introducing (non-abelian) gauge fields $\mathcal{A}^a$ and minimally coupling them to the action. There is a point of distinction to make here between the abelian gauging procedure and the non-abelian gauging procedure. In the abelian case, we only had a single gauge field, $\mathcal{A}$, which like all gauge fields, was a Lie-algebra valued one-form. In that case, the Lie-algebra was abelian, so $\mathcal{A} \in \Omega^1 (M; \fg) \sim \Omega^1(M)$, and $\mathcal{A}$ was an honest one-form. Here, we need to be more careful with commutation, since we have a Lie-algebra valued one-form and the Lie-algebra is non-abelian.  
Explicitly, the minimal coupling procedure involves the replacement $\dd X^{\mu} \mapsto \mathscr{D} X^{\mu} = \dd X^{\mu} - R^{\mu}_a \mathcal{A}^a$. The minimally coupled action 
\begin{align}
S_{MC} =  \int_{\Sigma} g_{\mu \nu} \mathscr{D} X^{\mu} \wedge \star \mathscr{D} X^{\nu} + B_{\mu \nu} \mathscr{D} X^{\mu} \wedge \mathscr{D} X^{\nu}
\end{align}
is invariant under the following (local) gauge transformations:
\begin{align}
\de X^{\mu} &= R^{\mu}_a \epsilon^a \\[1em]
\de \mathcal{A}^a &= \dd \epsilon^a +f^a \!_{bc} \mathcal{A}^b \epsilon^c.
\end{align}
In addition to the minimally coupled action, we want to add another gauge invariant term to enforce the flatness of the gauge fields. The extra term we add to the action is
\begin{align}
 \int_{\Sigma} \chi_a \mathcal{F}^a,
\end{align}
where 
\begin{align}
\mathcal{F} := \dd \mathcal{A} + \mathcal{A} \wedge \mathcal{A} = \left( \dd \mathcal{A}^a +  \tfrac{1}{2} f^a \!_{bc} \mathcal{A}^b \wedge \mathcal{A}^c \right) R_a = \mathcal{F}^a R_a
\end{align}
is the standard Yang-Mills field strength, and $\chi_a$ are scalar fields on $\Sigma$ with values in $\fg^{\ast}$.\footnote{Note that in Chapter \ref{chptr:Ch2} these Lagrange multipliers were called $\widehat{X}_a$, and in Chapter \ref{chptr:Ch5} they are called $\eta_a$. In each case, they are simply scalar fields on $\Sigma$ with values in $\fg^{\ast}$.}
This term is gauge invariant, provided the Lagrange multiplier transforms as 
\begin{align}
\de \chi_a &= - f^c \!_{ab} \epsilon^b \chi_c.
\end{align}
In summary, the gauged action
\begin{align}
\label{NATDgaugedaction}
S_G = \int_{\Sigma} g_{mn} \mathscr{D} X^{m} \wedge \star \mathscr{D} X^{n} + B_{mn} \mathscr{D} X^{m} \wedge \mathscr{D} X^{n} +  \int_{\Sigma} \chi_a  \mathcal{F}^a
\end{align}
is invariant under the following local gauge transformations:
\begin{align}
\de X^{m} &= R^{m}_a \epsilon^a \\[1em]
\de \mathcal{A}^a &= \dd \epsilon^a +f^a \!_{bc} \mathcal{A}^b \epsilon^c \\[1em]
\de \chi_a &= - f^c \!_{ab} \epsilon^b \chi_c.
\end{align}
The gauged action (\ref{NATDgaugedaction}) can be written in alternative notation as:
\begin{align}
S_G = \int_{\Sigma} \dd {}^2 z \, E_{mn} DX^m \bar{D}X^n.
\end{align}
We will now perform the non-abelian Buscher procedure by integrating out the two sets of fields separately, in order to see how T-duality works in the non-abelian case. Let us consider a slightly more generalised situation than a group with left-invariant metric. We suppose that our target space has a free $\sG$ action leaving $E_{ij} = g_{ij}+B_{ij}$ invariant. The metric and the $B$-field decompose with the Maurer-Cartan forms of $\sG$ as 
\begin{subequations}
\begin{align}
\label{MCmetric}
\dd s^2 &= g_{\mu \nu} \dd X^{\mu} \dd X^{\nu} + 2 g_{\mu n} \dd X^{\mu} \lambda^n + g_{mn} \lambda^m \lambda^n \\[1em]
\label{MCBfield}
B &= B_{\mu \nu} \dd X^{\mu} \wedge \dd X^{\nu} + 2 B_{\mu n} \dd X^{\mu} \wedge \lambda^n + B_{mn} \lambda^m \wedge \lambda^n,
\end{align}
\end{subequations}
where the coordinates $X^{\mu}$ are the spectator coordinates parametrising the base. Introducing coordinates $X^m$ parametrising the $\sG$ fiber, the non-linear sigma model action is
\begin{align}
S = \int_{\Sigma} \ddd^2 z \, \left[ E_{\mu \nu} \pr X^{\mu} \bar{\pr} X^{\nu}  + E_{m \nu} \pr X^{m} \bar{\pr} X^{\nu}+E_{\mu n} \pr X^{\mu} \bar{\pr} X^{n}+E_{mn} \pr X^{m} \bar{\pr} X^{n}\right].
\end{align}
Gauging this by minimally coupling the non-abelian gauge fields $\mathcal{A}^m = A^m \dd z + \bar{A}^m \dd \bar{z}$, via the replacement:
\begin{align}
\pr X^m \to DX^m = \pr X^m - A^m \\[1em]
\bar{\pr} X^m \to  \bar{D}X^m = \bar{\pr}X^m - \bar{A}^m,
\end{align}
and adding in the field strength $\mathcal{F}$ and Lagrange multiplier term, we obtain the gauged action
\begin{align}
\label{nonabeliangaugedaction}
S_G = &\int_{\Sigma} \ddd^2 z \, \bigg[ E_{\mu \nu} \pr X^{\mu} \bar{\pr} X^{\nu}  + E_{m \nu} D X^{m} \bar{\pr} X^{\nu}+E_{\mu n} \pr X^{\mu} \bar{D} X^{n}+E_{mn} D X^{m} \bar{D} X^{n} \notag \\
& \quad \quad \quad \quad +  \chi_a \left( \pr \bar{A}^a - \bar{\pr} A^a + f^a_{bc} A^b \bar{A}^c \right) \bigg]
\end{align}
As with the abelian case, integrating out the Lagrange multipliers enforces the contraints $\mathcal{F}^a = \dd \mathcal{A}^a +  \tfrac{1}{2} f^a \!_{bc} \mathcal{A}^b \wedge \mathcal{A}^c = 0$. The solution to this is $\mathcal{A} = g^{-1} \dd g$ for any $g: \Sigma \to G$, which we can substitute into the action to get $S_G[X,A,\chi] = S_{MC}[X, \sg^{-1} \dd \sg]$. Each choice of $\sg$ corresponds to a different choice of gauge, and fixing the gauge to $\sg = 1$ gives us back the original model (\ref{NANLSM}).

On the other hand, we can integrate out the gauge fields first. The calculation follows the same basic procedure as in Section \ref{subsec:MultipleAbelian}. The only difference is that now there is an additional term $\chi_a f^a\!_{bc} A^b \bar{A}^c$ in the action arising from the non-abelian field strength. Solving the Euler-Lagrange equations for the fields $A^m$ and $\bar{A}^m$ for the appropriate variable gives
\begin{align}
A^m &= E_{qt} (M^{-1})^{tm} \pr X^q + E_{\mu t} (M^{-1})^{tm} \pr X^{\mu} + (M^{-1})^{tm} \pr \chi_t \\[1em]
\bar{A}^n &= (M^{-1})^{ns} E_{sp} \bar{\pr} X^p + (M^{-1})^{ns} E_{s \nu} \bar{\pr} X^\nu - (M^{-1})^{ns} \bar{\pr} \chi_s,
\end{align}
where we have introduced
\begin{align}
M_{mn} = E_{mn} + \chi_a f^a\!_{mn}.
\end{align}

Substituting this into the gauged action (\ref{nonabeliangaugedaction}) and integrating the curvature term by parts gives the following complicated expression:
\begin{align*}
S = &\int_{\Sigma} \ddd^2 z \, \bigg[  \left( E_{\mu \nu} - E_{\mu m} (M^{-1})^{mn} E_{n \nu} \right) \pr X^{\mu} \bar{\pr} X^{\nu}   +  \left( E_{m \nu} - E_{mn} (M^{-1})^{ns} E_{s \nu} \right) \pr X^{m} \bar{\pr} X^{\nu} \\
&\quad \quad  \quad +  \left( E_{\mu n} - E_{\mu t} (M^{-1})^{tm} E_{m n} \right) \pr X^{\mu} \bar{\pr} X^{n}  +  \left( E_{m n} - E_{mp} (M^{-1})^{ps} E_{sn} \right) \pr X^{m} \bar{\pr} X^{n} \notag \\
&\quad \quad  \quad +  \left( -(M^{-1})^{ms} E_{s \nu} \right) \pr \chi_m \bar{\pr} X^{\nu}  +  \left( -(M^{-1})^{ms} E_{sn} \right) \pr \chi_m \bar{\pr} X^{n} \\
&\quad \quad  \quad +  \left( E_{\mu t} (M^{-1})^{tn} \right) \pr X^{\mu} \bar{\pr} \chi_{n}  +  \left( E_{mt} (M^{-1})^{tn} \right) \pr X^m \bar{\pr} \chi_{n}  +  \left( (M^{-1})^{mn} \right) \pr \chi_m \bar{\pr} \chi_n  \bigg] \notag
\end{align*}
What's going on here? This expression is far more complicated than the situation with abelian T-duality. Worse, the expression appears to still contain terms that have components along the original fiber coordinates $\pr X^{m}$ and $\bar{\pr} X^m$. We can see that in the abelian case we didn't have this problem. When the structure constants are zero, $M_{mn} = E_{mn}$, and so the offending terms such as $\left( E_{m n} - E_{mp} (M^{-1})^{ps} E_{sn} \right)$ vanish from this expression. Thankfully, the action still retains the gauge invariance generated by the right-invariant vector fields. We can use this to fix $\pr X^m = \bar{\pr} X^n = 0$. We interpret the resulting action as a non-linear sigma model on the coordinates $\widehat{X} = \{X^{\mu}, \chi_m\}$: 
\begin{align}
S = \int_{\Sigma} \ddd^2 z \, \left[ \widehat{E}_{\mu \nu} \pr X^{\mu} \bar{\pr} X^{\nu}  + \widehat{E}_{m \nu} \pr \chi_{m} \bar{\pr} X^{\nu}+\widehat{E}_{\mu n} \pr X^{\mu} \bar{\pr} \chi_n+ \widehat{E}_{mn} \pr \chi_{m} \bar{\pr} \chi_{n}\right],
\end{align}
where the new fields are given by
\begin{subequations}
\begin{align}
\widehat{E}_{\mu \nu} &= E_{\mu \nu} - E_{\mu m} (M^{-1})^{mn} E_{n \nu} \\[1em]
\widehat{E}_{m \nu} &=  -(M^{-1})_{ms} E^{s \nu} \\[1em]
\widehat{E}_{\mu n} &= E_{\mu t} (M^{-1})^{tn} \\[1em]
\label{NATDGroup}
\widehat{E}_{mn} &= (M^{-1})^{mn}.
\end{align}
\end{subequations}
These are the ``non-abelian Buscher rules" for non-abelian T-duality \cite{DQ}. To obtain the metric and the $B$ field from these, we simply take the symmetric and antisymmetric parts, as per (\ref{gBfromE}). Note that for vanishing structure constants, this reduces to the expression (\ref{multipleBuscher}) for an abelian $U(1)^k$ T-duality.


\subsection{Examples: NATD}
\label{subsec:NATDexamples}
\subsubsection{$S^3$ with no flux}
\label{NATDS3}
When discussing non-abelian T-duality, the first example is usually $S^3$, thought of as the group manifold $SU(2)$. The round metric on $S^3$ is bi-invariant, so is invariant under the left action of $SU(2)$, as well as under the right action of $SU(2)$. In fact, the full isometry group is $SO(4) \simeq \left( SU(2) \times SU(2) \right) / \ZZ_2$. We will perform a dualisation with respect to one of these $SU(2)$ isometries. 

To perform the calculation, we will use the Hopf coordinates, $(\eta, \xi_1, \xi_2)$, introduced in Section \ref{subsec:abelgeoexamples}. Recall that these are related to complex coordinates $(z_1,z_2) \in \CC^2$ by
\begin{align}
z_1 &= e^{\frac{i (\xi_1 + \xi_2)}{2}} \sin \eta \\[1em]
z_2 &= e^{\frac{i (\xi_1 - \xi_2)}{2}} \cos \eta.
\end{align}
A parameterisation of a group element $U \in SU(2)$ is then given by 
\begin{align}
U = \begin{pmatrix}
z_1 & z_2 \\
-\bar{z}_2 & \bar{z}_1,
\end{pmatrix}
\end{align}
and the (left-invariant) Maurer-Cartan forms are given by $\lambda = U^{-1} \dd U$. Choosing as a basis for $\mathfrak{su}(2)$ the set $\{i \sigma_1, i \sigma_2, i \sigma_3\}$, where $\sigma_j$ are the Pauli matrices:
\begin{align}
\sigma_1 &= \begin{pmatrix}
0 & 1 \\
1 & 0
\end{pmatrix} \\[1em]
\sigma_2 &= \begin{pmatrix}
0 & -i  \\
i & 0
\end{pmatrix} \\[1em]
\sigma_3 &= \begin{pmatrix}
1 & 0 \\
0 & -1
\end{pmatrix},
\end{align} 
we have the following expressions for the components of the (left-invariant) Maurer-Cartan forms:
\begin{subequations}
\label{MCformsSU2}
\begin{align}
\lambda^1 &=  \tfrac{1}{2} \sin(2 \eta) \cos(\xi_2) \dd \xi_1 + \sin(\xi_2) \dd \eta \\[1em]
\lambda^2 &=  \tfrac{1}{2} \sin(2 \eta) \sin( \xi_2) \dd \xi_1 - \cos(\xi_2) \dd \eta \\[1em]
\lambda^3 &= \tfrac{1}{2} \dd \xi_2 - \tfrac{1}{2} \cos(2 \eta) \dd \xi_1.
\end{align}
\end{subequations}
The round metric is written in terms of these Maurer-Cartan forms as
\begin{align}
\dd s^2 &= \delta_{ij} \lambda^i \lambda^j  \\[1em]
&= \dd \eta^2 + \frac{1}{4} \Big( \dd \xi_1^2 + \dd \xi_2^2 - 2 \cos(2 \eta ) \dd \xi_1 \xi_2 \Big).
\end{align}
The Lie algebra, $\mathfrak{so}(4)$, of Killing vectors for this metric decomposes into $\mathfrak{su}(2)_L \times \mathfrak{su}(2)_R$, corresponding to the left-invariant and right-invariant vector fields. In our coordinates, the right-invariant fields are:
\begin{align}
R_1 &=  \sin ( \xi_1) \pr_{\eta} +  \frac{2 \cos(2 \eta) \cos( \xi_1)}{\sin(2\eta)} \pr_{\xi_1} + \frac{2 \cos(\xi_1)}{\sin(2\eta)} \pr_{\xi_2}\\[1em]
R_2 &= - \cos( \xi_1) \pr_{\eta} + \frac{2 \cos(2 \eta) \sin( \xi_1)}{\sin(2\eta)} \pr_{\xi_1} +  \frac{2\sin(\xi_1)}{\sin(2\eta) } \pr_{\xi_2}\\[1em]
R_3 &= - 2 \pr_{\xi_1},
\end{align}
and the structure constants associated to the Lie algebra, $\mathfrak{su}(2)_R$, of right-invariant vector fields are determined by the corresponding commutation relations:
\begin{align}
[R_1,R_2] = 2R_3, \quad [R_2, R_3]  = 2 R_1, \quad [R_3,R_1] = 2R_2.
\end{align}
A short calculation lets us verify that 
\begin{align}
\Lie_{R_1} g &= 0 \\[1em]
\Lie_{R_2} g &= 0 \\[1em]
\Lie_{R_3} g &= 0.
\end{align}
Let us now perform a non-abelian T-duality for the right action of $SU(2)$. Note that since the metric is written in terms of the left-invariant Maurer-Cartan forms, we don't actually require the coordinate description any more. The non-abelian Buscher rules tell us that we can obtain the dual metric and $B$-field from
\begin{align}
\widehat{E}_{mn} = \left( \delta_{mn} + f^a \!_{mn} \chi_a \right)^{-1}.
\end{align}
Disentangling the symmetric and antisymmetric components of this expression gives us the dual metric and $B$-field:
\begin{subequations}
\label{NATDS3undeformed}
\begin{align}
\widehat{\dd s}^2 &= \frac{1}{1+\chi^2} (\delta_{ij} + \chi_i \chi_j) \dd \chi^i \dd \chi^j \\[1em]
\widehat{B} &= -\epsilon_{ijk}\frac{\chi_k}{1+\chi^2} \dd \chi^i \wedge \dd \chi^j,
\end{align}
\end{subequations}
where $\chi^2 = \chi_1^2 + \chi_2^2 + \chi_3^2$. This metric still has a residual $\mathfrak{su}(2)$ isometry, which we can identify with the original $\mathfrak{su}(2)_L$ isometry of the round metric. Indeed, instead of the round metric on $S^3$, we could have started with the metric of the squashed $S^3$:
\begin{align}
\label{squashedS3}
\dd s^2 &= e^{2a} (\lambda^1)^2 + e^{2b} (\lambda^2)^2 + e^{2c} (\lambda^3)^2,
\end{align}
where $a$, $b$, and $c$ are real constants which are not all equal. The right-invariant vector fields are still isometries of this metric, but we have now broken the $\mathfrak{su}(2)_L$ isometry. If $a = b$, then the isometry breaks to $\mathfrak{u}(1) \times \mathfrak{su}(2)_R$, and if $a$, $b$, and $c$ are all different then the only isometries are the ones generated by the right-invariant vector fields. We can still T-dualise the metric of the squashed $S^3$, obtaining a dual metric and $B$-field determined by:
\begin{align}
\widehat{E}_{mn} &= \left( \delta_{m1}\delta_{n1} e^{2a} + \delta_{m2}\delta_{n2} e^{2b}  + \delta_{m3}\delta_{n3} e^{2c}  + f^a \!_{mn} \chi_a \right)^{-1}
\end{align}
When $a = b$ the resulting metric has the residual $\mathfrak{u}(1)$ isometry, but if $a$, $b$, and $c$ are all different, then the resulting metric has no isometries at all. We will discuss this loss of isometry more in Section \ref{lossofisometry}.
\subsubsection{Bianchi V}
\label{NATDBianchiV}
The Bianchi V spacetime was first studied in the context of non-abelian T-duality in \cite{GRV}. The model is a three-dimensional model, with coordinates $\{x,y,z\}$, together with a free parameter $t$. There is no $B$-field or dilaton, and the metric has the form
\begin{align}
\label{BianchiVmetric}
\dd s^2 &= t^2 \dd x^2 + t^2 e^{-2x} \dd y^2 + t^2 e^{-2x} \dd z^2.
\end{align}
This metric has the standard form for performing non-abelian T-duality, 
\begin{align*}
E = E_{ij} \lambda^i \lambda^j,
\end{align*}
where 
\begin{align*}
E_{ij} &= \left( 
\begin{matrix}
t^2 & 0 & 0 \\
0 & t^2 & 0 \\
0 & 0 & t^2
\end{matrix}
\right),
\end{align*}
and the left-invariant Maurer-Cartan forms are
\begin{align*}
\lambda^1 &= \dd x \\[1em]
\lambda^2 &= e^{-x} \dd y \\[1em]
\lambda^3 &= e^{-x} \dd z.
\end{align*}
As a three-dimensional metric,\footnote{That is, for fixed $t$.} (\ref{BianchiVmetric}) is curved since it has a Ricci scalar curvature of $\mathcal{R} = -\frac{6}{t^2}$. When we include the timelike direction, however, we note that the metric is simply a different parametrisation of Minkowski spacetime. 
The right-invariant vector fields are isometries of this metric. They are given in these coordinates by
\begin{align*}
v_1 &= \pr_x + y \pr_y + z\pr_z \\[1em]
v_2 &= \pr_y \\[1em]
v_3 &= \pr_z.
\end{align*}
The structure constants for this model are $f^2 \!_{12} = f^3 \!_{13} = -1$, and so the dual model is obtained by inverting the following matrix:
\begin{align*}
M &= \left( 
\begin{matrix}
t^2 & -y & -z \\
y & t^2 & 0 \\
z & 0 & t^2
\end{matrix}
\right) 
\end{align*}
We obtain the following metric and $B$-field:
\begin{align*}
\widehat{g} &= \frac{1}{t^2 (t^4 + y^2 + z^2)} \left( 
\begin{matrix}
t^4 & 0 & 0 \\
0 & (t^4 + z^2) & -yz \\
0 & -yz & (t^4 + y^2)
\end{matrix}
\right)\\[1em]
\widehat{B} &= \frac{1}{t^2 (t^4 + y^2 + z^2)} \left( 
\begin{matrix}
0 & t^2 y & t^2 z \\
-t^2 y &  & 0 \\
-t^2 z & 0 & 0
\end{matrix}
\right).
\end{align*}
This example is of particular historical significance, since when it was first studied it was realised that this non-abelian T-duality does not lead to a conformal dual model. We discuss this further in Section \ref{nonsemisimple}.
\subsubsection{Non-abelian duals of Minkowski}
Non-abelian T-duals of flat space have been systematically studied in a recent series of papers \cite{HP,PHP,HPP}. Four-dimensional Minkowski spacetime, $\mathcal{M}_4$ has a ten-dimensional group of isometries - the Poincar\'{e} group, $\RR^{3,1} \ltimes \textrm{O}(3,1)$. By choosing a four dimensional subgroup, $H$, of this which acts freely and transitively on $M_4$, we can identify $H$ with the spacetime on which it acts. Writing the flat Minkowski metric, $\eta$, in terms of the Maurer-Cartan forms of $H$ then allows us to perform a non-abelian T-duality with respect to $H$. Let us elucidate this procedure by following the calculation of the first example studied in \cite{HP}. Recall that the Poincar\'{e} algebra has generators $\{P_{\mu}, L_{i}, K_{i} \}$, where the $P_{\mu}$ generate the spacetime translations, the $L_i$ generate spatial rotations, and the $K_i$ generate boosts, and $\{\mu\} = \{ 0 , i\} = \{0,1,2,3\}$. These generators satisfy the following commutation relations:
\begin{align*}
\begin{split}
[P_{\mu},P_{\nu}] &= 0 \\[1em]
[L_{i}, L_j] &= \epsilon_{ijk} L_k\\[1em]
[K_i,K_j] &= -\epsilon_{ijk} L_k\\[1em]
[L_i,K_j] &= \epsilon_{ijk} K_k \\[1em]
\end{split}
\begin{split}
[L_i, P_{0}] &= 0 \\[1em]
[L_i,P_j] &= \epsilon_{ijk} P_k \\[1em]
[K_i, P_{0}] &= P_i \\[1em]
[K_i,P_j] &= \delta_{ij}P_0.
\end{split}
\end{align*}
There are many four-dimensional subalgebras for which the associated group action on $M_4$ is free and transitive \cite{PSWZ}. One such subalgebra is spanned by $\{K_3, L_2 + K_1, L_1 - K_2, P_0 - P_3\}$. In the usual cartesian coordinates $\{t,x,y,z\}$ for $M_4$, we have:
\begin{align}
\label{MinkVectorfields}
v_1 &= K_3 = - z \pr_t - t \pr_z \\[1em]
v_2 &= L_2 + K_1 = -x \pr_t - (t+z)\pr_x + x\pr_z \\[1em]
v_3 &= L_1 - K_2 = y \pr_t + (t+z) \pr_y - y\pr_z \\[1em]
v_4 &= P_0 - P_3 = \pr_t - \pr_z.
\end{align}
The non-vanishing commutation relations satisfied by this subalgebra are
\begin{align}
\label{MinkLieAlgebra}
[v_1, v_2] = -v_2, \qquad [v_1,v_3] = -v_3, \qquad [v_1,v_4] = -v_4.
\end{align}
Note that similarity with the previously studied example of the Bianchi V spacetime. We choose coordinates $\{x^1,x^2,x^3,x^4\}$ for the group generated by these vector fields, and take the parametrisation of a group element to be:
\begin{align*}
\sg &= e^{x^1 T_1} e^{x^2 T_2} e^{x^3 T_3} e^{x^4 T_4} ,
\end{align*}
where the $T_i$ are, say, a basis for the adjoint representation of the Lie algebra defined by (\ref{MinkLieAlgebra}). We can calculate the left-invariant vector fields for this parametrisation of the group, obtaining:
\begin{align}
\label{Minkgroupvectorfields}
V_1 &= \pr_1 + x^2 \pr_2 + x^3 \pr_3 + x^4 \pr_4 \\[1em]
V_2 &= \pr_2 \\[1em]
V_3 &= \pr_3 \\[1em]
V_4 &= \pr_4.
\end{align}
Identifying the vector fields in terms of the original spacetime coordinates (\ref{MinkVectorfields}) with the vector fields in terms of the group coordinates (\ref{Minkgroupvectorfields}) allows us to calculate the coordinate transformation between the spacetime coordinates $\{t,x,y,z\}$ and the group coordinates $\{x^1, x^2,x^3,x^4\}$. The explicit coordinate transformation is:
\begin{align*}
t &= \tfrac{1}{2} e^{-x^1} \left( (x^2)^2 + (x^3)^2 + 1  \right) + x^4 \\[1em]
z &=  -\tfrac{1}{2} e^{-x^1} \left( (x^2)^2 + (x^3)^2 - 1  \right) - x^4 \\[1em]
x &= -e^{-x^1} x^2 \\[1em]
y &= e^{-x^1} x^3.
\end{align*}
We can use this to directly calculate the flat metric in the group coordinates. It is given by
\begin{align}
\dd s^2 &= - \dd t^2 + \dd x^2 + \dd y^2 + \dd z^2 \notag \\[1em] 
\label{flatgroupcoordinates}
&= e^{-x^1} \dd x^1 \dd x^4 + e^{-2x^1} \dd x^2 \dd x^2 + e^{-2x^1} \dd x^3 \dd x^3 + e^{-x^1} \dd x^4 \dd x^1.
\end{align}
Alternatively, we can compute the right-invariant Maurer-Cartan forms, 
\begin{align*}
\rho = \rho^i T_a = \dd \sg\,  \sg^{-1}
\end{align*}
where
\begin{align*}
\rho^1 &= \dd x^1 \\[1em]
\rho^2 &= e^{-x^1} \dd x^2 \\[1em]
\rho^3 &= e^{-x^1} \dd x^3 \\[1em]
\rho^4 &= e^{-x^1}\dd x^4.
\end{align*}
Then, we note that the metric (\ref{flatgroupcoordinates}) can be written in terms of the right-invariant Maurer-Cartan forms as
\begin{align*}
\dd s^2 &= E_{ij} \rho^i \rho^j,
\end{align*}
with $E_{ij}$ is given by:
\begin{align*}
E &= \left( 
\begin{matrix}
0 & 0 & 0 & 1 \\
0 & 1 & 0 & 0 \\
0 & 0 & 1 & 0 \\
1 & 0 & 0 & 0
\end{matrix}
\right).
\end{align*}
The non-abelian T-dual can now be computed, and one obtains the dual $E$ as
\begin{align*}
\widehat{E} &= \left( 
\begin{matrix}
0 & 0 & 0 & \frac{1}{1-\tilde{x}^4} \\
0 & 1 & 0 & \frac{\tilde{x}^2}{1-\tilde{x}^4} \\
0 & 0 & 1 & \frac{\tilde{x}^3}{1-\tilde{x}^4} \\
\frac{1}{1+\tilde{x}^4} & -\frac{\tilde{x}^2}{1+\tilde{x}^4} & -\frac{\tilde{x}^3}{1+ \tilde{x}^4} & -\frac{(\tilde{x}^2)^2 + (\tilde{x}^3)^2}{1-(\tilde{x}^4)^2}
\end{matrix}
\right).
\end{align*}
This background becomes, after a change of coordinates, the Brinkmann form of the pp-wave metric, with vanishing $H$-flux.

One can perform a non-abelian T-duality with respect to other subalgebras of the Poincar\'{e} algebra, including lower dimensional subalgebras (taken to be a non-abelian T-duality with spectators). For a catalogue of the results of the non-abelian T-duals for Minkowski space, we refer the interested reader to \cite{HP,PHP,HPP}.
\subsubsection{$f$-flux background}
\label{NATDHeis}
We have already discussed the $f$-flux background in the context of abelian T-duality in Section \ref{T3Buscher} and Section \ref{T3fluxtopology}. There, we noted that there is a natural group structure on the background obtained by taking the quotient of the Heisenberg group by its integer counterpart:
\begin{align*}
\textrm{Nil} = \textrm{Heis}(\RR) / \textrm{Heis}(\ZZ).
\end{align*}
Choosing coordinates $\{x,y,z\}$ for $\textrm{Nil}$, we found that the metric obtained from the Buscher procedure was simply the left-invariant metric on the group:
\begin{align*}
g &= \dd x^2 + \dd y^2 + (\dd z - k x \dd y)^2.
\end{align*}
The right-invariant vector fields, corresponding to the \emph{left} group action, are Killing vectors for this metric. They are given by
\begin{align*}
R_1 &= \pr_x + ky \pr_z \\[1em]
R_2 &= \pr_y \\[1em]
R_3 &= \pr_z. 
\end{align*}
From this, the only non-vanishing structure constant is $f^3 \!_{12} = -k$. We now perform a non-abelian T-duality with respect to the isometry generated by the right-invariant vector fields. There are no spectator fields, and so we only need to compute the inverse of the matrix $M$ implementing the duality. By a slight abuse of notation, we take the dual coordinates to be $\{x,y,z\}$, and so the matrix $M$ is given by:
\begin{align*}
M &= \left( 
\begin{matrix}
1 & -kz & 0 \\
k z & 1 & 0 \\
0 & 0 & 1
\end{matrix}
\right)
\end{align*}
The inverse is 
\begin{align*}
M^{-1} =\left( 
\begin{matrix}
\frac{1}{1+k^2 z^2} & \frac{kz}{1+k^2 z^2} & 0 \\
\frac{-kz}{1+k^2 z^2} & \frac{1}{1+k^2 z^2} & 0 \\
0 & 0 & 1
\end{matrix}
\right)
\end{align*}
Disentangling the symmetric and antisymmetric parts of this gives us the dual metric and $B$-field:
\begin{align*}
\widehat{g} &= \dd z^2 + \frac{1}{1+k^2 z^2} \left( \dd x^2 + \dd y^2 \right) \\[1em]
\widehat{B} &= \frac{2k z}{1+k^2 z^2} \dd x \wedge \dd y.
\end{align*}
For $k=1$, this is $Q^{xy} \!_z$, which is just the $Q$-flux background from Section \ref{T3Buscher}, albeit with a relabelling of coordinates. We shall discuss this non-abelian T-duality in relation to a chain of abelian T-dualities in Section \ref{NATDcommentsTT}.


\subsection{Loss of isometry}
\label{lossofisometry}
When we perform abelian T-duality along a $U(1)$ isometry, the dual metric always has a $U(1)$ isometry corresponding to the dual coordinate, and T-dualising along this direction takes us back to the original metric. We have seen, on the other hand, that non-abelian T-duality does not exhibit such behaviour. When one performs a T-duality with respect to a non-abelian group of isometries $\sG$, the dual model will generically have fewer isometries than the original model. Indeed, the only symmetries which survive in the dual space are the symmetries which commute with the symmetry to be gauged. For the case of a single abelian isometry, a proof may be found in \cite{P14}.
\begin{lemma}
	When a subgroup $\sH \subset \sG$ of a symmetry group is gauged, the remaining symmetry (global symmetries, that is) of the dual space is given by the commutant:
\begin{align*}
\sH' := \{ \sg \in \sG : \sh \sg = \sg \sh \quad \forall \sh \in \sH\}
\end{align*}
\end{lemma}
\begin{proof}
Let us suppose that we have gauged a subgroup $\sH$ of the isometries of a sigma model corresponding to some set of vectors $\{v_{b}\}$, resulting in the gauged action (\ref{nonabeliangaugedaction}). For simplicity, let us assume that we are on a group manifold, so that there are no spectator fields, and the action (\ref{nonabeliangaugedaction}) becomes:
\begin{align*}
S_G &= \int_{\Sigma} g_{\mu \nu} \scD X^{\mu} \wedge \star \scD X^{\nu} + B_{\mu \nu} \scD X^{\mu} \wedge \scD X^{\nu} \\
&\qquad + \int_{\Sigma} \cF^b \chi_b.
\end{align*}
We are interested in whether this action is invariant under the remaining isometries of the original model. That is, suppose $\{Z_a\}$ is a set of vectors, disjoint from $\{v_{b}\}$, which are global symmetries of the original model:
\begin{align*}
\Lie_{Z_{a}}  \, g = \Lie_{Z_a} B = 0.
\end{align*}
The (global) variations generated by this set of vector fields are:
\begin{align*}
\de X^{\mu} &= \epsilon^a Z^i_a \\[1em]
\de A^b &= 0 \\[1em]
\de \chi_b &= 0,
\end{align*}
where $\epsilon^a$ are constants. The variation of the gauged action is therefore
\begin{align*}
\de S_G &= \int_{\Sigma} \Bigg[ \de (g_{\mu \nu}) \scD X^{\mu} \wedge \star \scD X^{\nu} + g_{\mu \nu}\, \de (\scD X^{\mu}) \wedge \star \scD X^{\nu} + g_{\mu \nu} \scD X^{\mu} \wedge \de ( \star \scD X^{\nu}) \notag \\
& \qquad + \de (B_{\mu \nu}) \scD X^{\mu} \wedge \scD X^{\nu} + B_{\mu \nu}\, \de (\scD X^{\mu}) \wedge  \scD X^{\nu} + B_{\mu \nu} \scD X^{\mu} \wedge \de (  \scD X^{\nu}) \Bigg] \\
& + \int_{\Sigma} (\de \cF^b) \chi_b + \cF^b (\de \chi_b) \\[1em]
&= \int_{\Sigma} \epsilon^a \Big[ (\Lie_{Z_a} g)_{\mu \nu} \scD X^{\mu} \wedge \star \scD X^{\nu} + (\Lie_{Z_a} B)_{\mu \nu} \scD X^{\mu} \wedge \scD X^{\nu} \Big] \\
&\qquad + \int_{\Sigma} \epsilon^a \left( g_{\mu \nu} + g_{\nu \mu} \right) \Big[ v^{\sigma}_{b}(\pr_{\sigma} Z^{\mu}_{a}) - Z^{\sigma}_{a} (\pr_{\sigma}v^{\mu}_b) \Big] A^b \wedge \star \scD X^{\nu} \\
&\qquad + \int_{\Sigma} \epsilon^a \left( B_{\mu \nu} - B_{\nu \mu} \right) \Big[ v^{\sigma}_{b}(\pr_{\sigma} Z^{\mu}_{a}) - Z^{\sigma}_{a} (\pr_{\sigma}v^{\mu}_b) \Big] A^b \wedge \scD X^{\nu} \\[1em]
&= 2 \int_{\Sigma} \epsilon^a  g_{\mu \nu} [Z_a,v_b]^{\mu}  \, A^b \wedge \star \scD X^{\nu}\\
&+ 2 \int_{\Sigma} \epsilon^a  B_{\mu \nu} [Z_a,v_b]^{\mu}  \, A^b \wedge \star \scD X^{\nu}.
\end{align*}
It follows that the gauged action is invariant under the symmetry generated by the $\{Z_a\}$ if and only if 
\begin{align*}
[Z_a,v_b] = 0.
\end{align*}
\end{proof}
This feature of non-abelian T-duality makes it very different to its abelian counterpart. Once we dualise with respect to a non-abelian group $\sG$, the dual model generically won't have isometries with which to gauge. That is, we don't know how to perform another T-duality to get us back to the original model. Non-abelian T-duality is therefore not invertible in the traditional sense, and calling it a duality is a slight misnomer (which we are stuck with for historical reasons). Poisson-Lie T-duality, described in Chapter \ref{chptr:Ch4},  is a generalisation that inverts non-abelian T-duality in a suitably defined sense, although it is of a slightly different flavour, and one must know \emph{a priori} that the two spaces are dual. In Chapter \ref{chptr:Ch5} we discuss our attempts to provide a gauging prescription for inverting non-abelian T-duality.

\subsection{Generalisation to coset manifolds}
In our derivation of the non-abelian Buscher rules, we assumed that the group $\sG$ had a free action on the target space - that is, there were no points on the target space which were fixed by a non-identity group element. A more general group action will have group elements which fix points on the manifold. Let $\sG$ act on $M$, and consider a point $x \in M$. The set of group elements which fix $x$ is 
\begin{align}
\sG_x = \left\{ \sg \in \sG \, : \, \sg x = x \right\}.
\end{align}
For each $x \in M$, this is a subgroup of $\sG$. Mathematicians refer to $\sG_x$ as the stabiliser subgroup of $\sG$ with respect to $x$. For physicists, the term \emph{isotropy subgroup} seems to be more common. A free group action has no non-trivial isotropy subgroups. A simple example of a group acting with isotropy is the three-dimensional rotation group $SO(3)$ acting on $S^2$. Thinking of $S^2$ as the unit sphere in $\RR^3$, we have a natural $SO(3)$ group action acting by rotation. Now consider, for instance, the north pole of $S^2$. This point is fixed by rotations around the $z$-axis, so the isotropy subgroup of the north pole is $SO(2) = U(1)$.

If a group $\sG$ acts transitively on a space $M$, then the isotropy subgroups are all conjugate, and we refer to $M$ as a homogeneous space, or in the physics parlance, as a coset space. We can identify the points in $M$ with the quotient space $\sG/ \sG_x$. There is a nice geometric way of understanding this, which we can illustrate with the $S^2$ example. Define a correspondence between points in $S^2$ and elements of the group $SO(3)$ by associating to each point $m \in S^2$ an $SO(3)$ rotation, $\phi_m$, which sends the north pole, $N$, to the point $m$. That is,
\begin{align*}
\phi \cdot N = m.
\end{align*}
This rotation is not unique, however, precisely because $SO(3)$ acts with isotropy. In particular, for any $\psi$ in the isotropy subgroup of $N$, we have 
\begin{align*}
\left( \phi \psi \right) \cdot N =  \phi \cdot \left( \psi \cdot N \right) = \phi \cdot N = m. 
\end{align*}
The points in $S^2$ can therefore be identified with the quotient of $SO(3)$ by the isotropy subgroup $SO(2)$:
\begin{align*}
S^2 = SO(3) / SO(2).
\end{align*}
Note that we are only interested in examples for which the group $\sG$ acts by isometries, rather than just homeomorphisms. That is, we are interested in homogeneous Riemannian manifolds, rather than just homogeneous topological spaces. As pointed out in \cite{DNP}, the symbol $S^2$ used here should implicitly refer to the topological space $S^2$ equipped with the round metric. If we were to distort the metric on $S^2$, for example by taking the induced metric on the embedding
\begin{align}
x^2 + y^2 +  \frac{z^2}{\alpha^2} =1,
\end{align}
then although the space is still topologically $S^2$, the isometry group is reduced. We no longer have a full $SO(3)$ rotation symmetry, since we have `stretched' the sphere along the $z$-direction (see Figure \ref{sphereellipsoid}). The isometry group is now just $SO(2)$, corresponding to rotation around the $z$-axis, however this group does not act transitively (in particular, the isotropy subgroup of the north pole is the full group $SO(2)$). One should keep in mind that, when writing homogeneous spaces as $M = \sG / \sH$, the implicit assumption is that $G$ acts transitively by isometries.    
\begin{figure}[h!]
	\centering
	\begin{subfigure}{0.5\textwidth}
	\centering
	\includegraphics[width =0.6\linewidth]{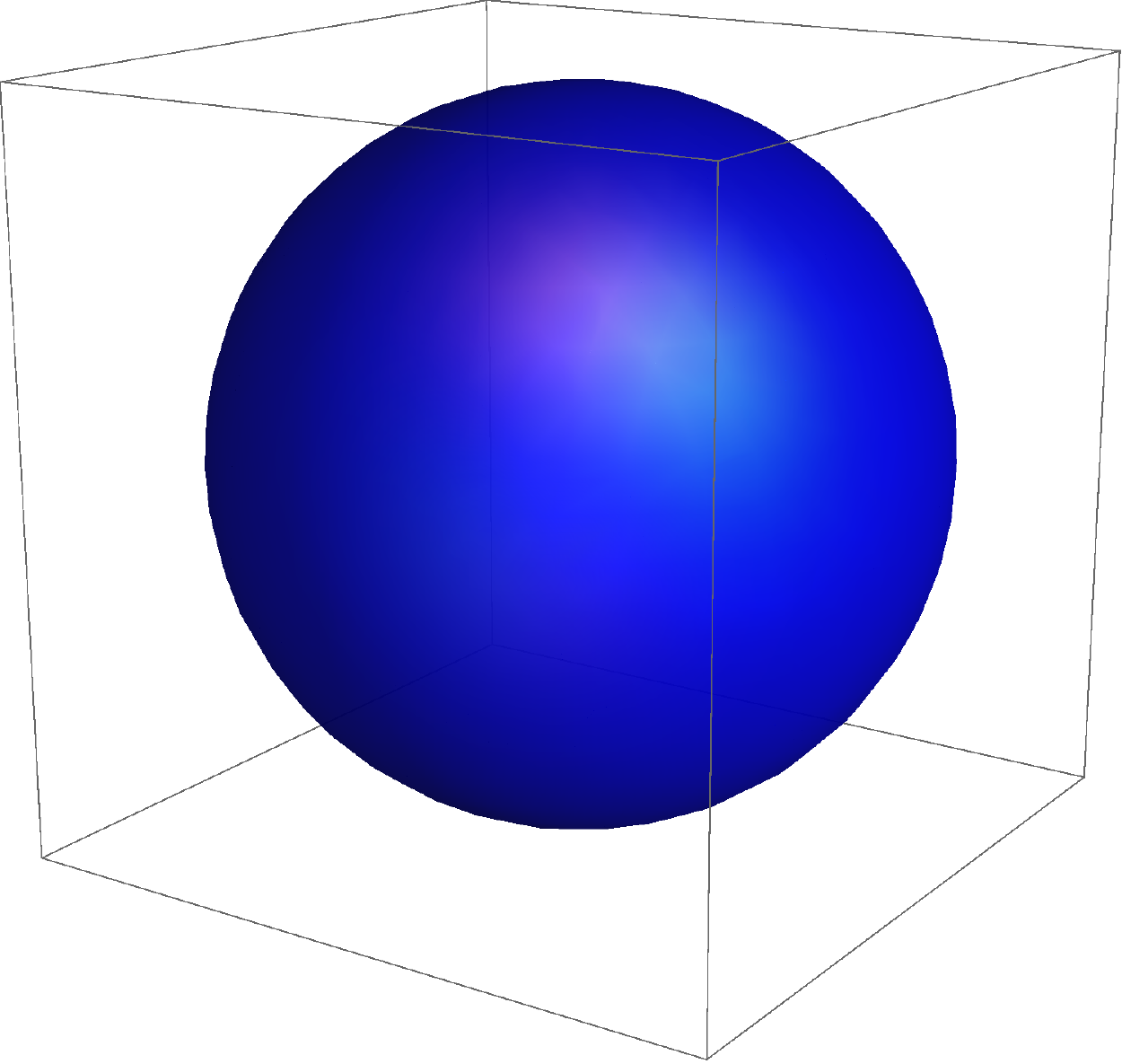}
	\end{subfigure}%
	\begin{subfigure}{0.5\textwidth}
	\centering
	\includegraphics[width=0.6\textwidth]{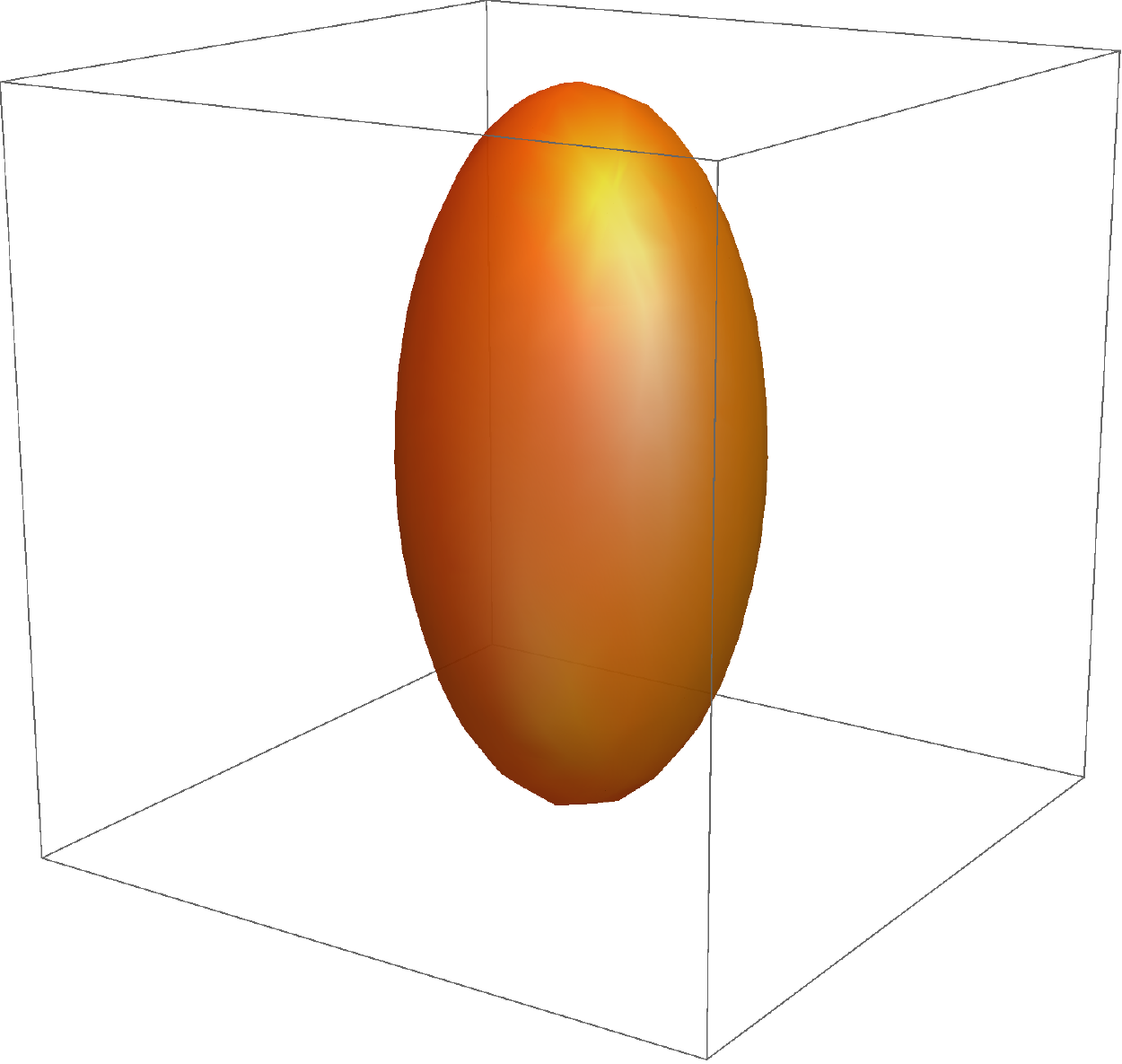}
	\end{subfigure}
	\caption{The isometry group of a round $S^2$ is $SO(3)$. Stretching the sphere in the $z$-direction breaks the isometry group down to $SO(2)$.}
	\label{sphereellipsoid}
\end{figure}

Let us now discuss the generalisation of non-abelian T-duality to coset spaces. This was introduced in \cite{LOST}, where is was used to construct new solutions to type II supergravity. Our starting point is a coset target space $M = \sG / \sH$, where $\sH$ is a subgroup of some (non-abelian) group $\sG$. We aim to perform a non-abelian T-duality with respect to the $\sG$ action on $M$. On the group $\sG$, we choose a splitting of coordinate indices $m = (\alpha, i)$, where the indices $\alpha$ belong to the coset $\sG / \sH$ and the indices $i$ belong to the subgroup $\sH$. For concreteness, let  $\textrm{dim}(\sG / \sH) = n$ and $\textrm{dim} (\sH)=k$, so that $\textrm{dim}(\sG) = n+k$. We now consider a matrix $E$ for the action (\ref{SPCM}) of the form:
\begin{align}
E_{mn} = \left(
\begin{matrix}
E_0 & 0\\
0 & \mu \mathds{1}_k
\end{matrix}
\right),
\end{align}
where $E_0$ is an $n \times n$ $\sG$-invariant matrix, $\mu \in \RR$ is a parameter, and $\mathds{1}_k$ is the $k \times k$ identity matrix.\footnote{Note that we have used $\mu$ as the limiting parameter rather than $\lambda$, as in \cite{LOST}. This is to avoid notational clashes with left-invariant one-forms, which we are denoting by $\lambda$.} This matrix defines a Principal Chiral Model on the group $\sG$, whose T-dual is given by (\ref{NATDGroup}):
\begin{align}
S = \int_{\Sigma} \dd {}^2 z \, \widehat{E}^{mn} \pr \chi_m \bar{\pr} \chi_n,
\end{align}
where
\begin{align}
\widehat{E}^{mn} = (M^{-1})^{mn}.
\end{align}
As before, $M_{mn} = E_{mn} + \chi_a f^{a} \!_{mn} $. Taking the limit $\lambda \to 0$ on the original side gives us a sigma model on the coset space $\sG / \sH$. On the dual side, we have $n+k$ Lagrange multiplier coordinates, and we need to gauge fix $k$ of them. This is done on a case-by-case basis. 

\subsection{Examples: NATD on coset manifolds}
\label{NATDcoset}
The ideas in the previous subsection will become a lot clearer once we study some concrete examples.
\subsubsection{NATD of $S^2$ with respect to $SO(3)$}
\label{NATDS2}
There exist a class of 11-dimensional supergravity solutions labelled $O(n_1,n_2,n_3)$ whose internal spaces are $U(1)$-bundles over $S^2 \times S^2 \times S^2$ \cite{NP,DNP}. Dimensionally reducing over the $U(1)$ gives a type-IIA supergravity solution on $AdS_4 \times S^2 \times S^2 \times S^2$. This class of type-IIA solutions was used as an example in \cite{LOST} to apply their newly-developed non-abelian T-duality for coset geometries. There, they identified each of the $S^2$ as the homogeneous space $SU(2)/U(1)$, and performed non-abelian T-duality with respect to the $SU(2)$ symmetry on each of the $S^2$ factors. Here, we shall only be interested in the NS sector of their analysis.

We recall that the bi-invariant metric on $SU(2)$ can be written in terms of the left-invariant one forms, (\ref{MCformsSU2}), as:
\begin{align*}
\dd s^2 &= \delta_{mn} \lambda^m \lambda^n \\[1em]
&= \dd \eta^2 + \frac{1}{4} \Big( \dd \xi_1^2 + \dd \xi_2^2 - 2 \cos (2\eta) \dd \xi_1 \dd \xi_2 \Big).
\end{align*}
Based on our previous discussion, and in the presence of a vanishing $B$-field, we modify the bi-invariant metric slightly to discuss the coset space $S^2 = SU(2) / U(1)$:
\begin{align}
E_{mn} &= \lambda^1 \lambda^1 + \lambda^2 \lambda^2 + \mu \lambda^3 \lambda^3 \notag \\[1em]
\label{cosetSU2metric}
&= \dd \eta^2 + \frac{1}{4} \sin^2 (2 \eta) \dd \xi_1^2 +  \frac{\mu}{4} \Big( \dd \xi_2 - \cos (2 \eta) \dd \xi_1  \Big)^2
\end{align}
Taking $\mu \to 1$ gives us the normal bi-invariant metric on $SU(2)$, whilst taking the limit $\mu \to 0$ gives us the round metric on the coset space $S^2$. Performing a non-abelian T-duality on the metric (\ref{cosetSU2metric}) with respect to the $SU(2)$ isometry follows simply from (\ref{NATDGroup}). We first construct $M$ by
\begin{align*}
M_{mn} = E_{mn} + \chi_a f^a_{mn}.
\end{align*}
Denoting the dual coordinates by $\chi_a = (x,\rho,z)$,\footnote{This seemingly strange choice of notation is to match the notation of \cite{LOST}} we have:
\begin{align*}
M = 
\left( 
\begin{matrix}
1 & z & - \rho \\
- z & 1 & x \\
\rho & -x & \mu
\end{matrix} 
\right) 
\end{align*} 
from which we immediately have 
\begin{align*}
M^{-1} = 
\frac{1}{\mu z^2 + \rho^2 + x^2 + \mu} 
\left( 
\begin{matrix}
x^2 + \mu &   \rho x -\mu z & xz + \rho \\
 \rho x + \mu z & \rho^2 + \mu & \rho z - x \\
xz - \rho & \rho z +x & z^2 + 1
\end{matrix} 
\right). 
\end{align*}
We now use the residual gauge freedom to set $x = 0$, and take the $\mu \to 0$ limit, obtaining the following model:
\begin{align*}
\widehat{E}_{ij} &= 
\frac{1}{\rho^2} \left( 
\begin{matrix}
\rho^2 & \rho z  \\
\rho z & z^2 + 1 
\end{matrix} 
\right). 
\end{align*}
Extracting the symmetric and antisymmetric components, we get:
\begin{subequations}
\begin{align}
\label{S2NATD}
\widehat{\dd s}^2 &= \frac{\dd z^2}{ \rho^2} + \left( \dd \rho + \frac{z}{\rho} \dd z \right)^2 \\[1em]
\widehat{B} &= 0.
\end{align}
\end{subequations}
This metric has a curvature singularity at $\rho = 0$, as can be seen by computing the Ricci scalar:
\begin{align*}
\widehat{\mathcal{R}} = - \frac{2(\rho^2 + 2z^2 + 2)}{\rho^2}.
\end{align*}
The Ricci scalar here is a negative definite function on the dual manifold. One might be tempted to compare this metric to the T-dual of $S^2$, (\ref{S2dual}), obtained in Section \ref{subsec:S2abelian}. Both metrics are obtained from a T-duality of $S^2$, have negative definite scalar curvatures (with singularities), and vanishing $B$-field. Are they actually different coordinate descriptions of the same metric? This is, in general, a hard question to answer. Here, we have a simple way to see that the answer is no. We first note that the metric (\ref{S2dual}) is independent of the coordinate $\hat{\phi}$, and so $\pr_{\hat{\phi}}$ is a Killing vector for the metric. On the other hand, the metric (\ref{S2NATD}) has no Killing vectors. It follows that the two spaces are not isometric. 

\subsubsection{ A new NATD for the D3 near horizon}
Type IIB supergravity has an $AdS_5 \times S^5$ solution describing the near horizon limit of the D3-brane background. This solution was another one of the coset geometry examples studied in \cite{LOST}, where they computed the non-abelian T-dual with respect to the $SO(6)$ isometry group of the $S^5$. The $SO(6)$ isometry acts with isotropy on $S^5$, and the non-abelian T-duality was performed by realising $S^5$ as the coset space $S^5 = SO(6)/SO(5)$.  

In this section we make the observation that $S^5$ can also be written as the coset space $SU(3)/SU(2)$.\footnote{This observation comes by realising $S^5$ as the unit complex sphere in $\CC^3$.} This observation allows us to perform a non-abelian T-duality with respect to the $SU(3)$ isometry group using the coset construction in \cite{LOST}. 

The metric for the D3 near horizon solution is
\begin{align*}
\dd s^2 = \dd s^2 (AdS_5) + \dd s^2 (S^5).
\end{align*}
There is no dilaton and no $B$-field. The only non-vanishing Ramond flux is
\begin{align*}
F_5 &= 4 \, \Big( \textrm{Vol}(AdS_5) - \textrm{Vol}(S^5)  \Big),
\end{align*} 
which is self-dual, since we are in type IIB. 

We consider the modified metric
\begin{align*}
E_{mn}=
\left(
\begin{array}{*{13}{c}}
\tikzmarkin{a}(0.1,-0.1)(-0.1,0.35)\mu & 0 & 0 & 0 & 0 & 0 & 0 & 0 \\
0 & \mu & 0 & 0 & 0 & 0 & 0 & 0  \\
0 & 0 & \mu\tikzmarkend{a} & 0 & 0 & 0 & 0 & 0\\
0 & 0 & 0 & \tikzmarkin{b}(0.1,-0.1)(-0.1,0.35)1 & 0 & 0 & 0 & 0 \\
0 & 0 & 0 & 0 & 1 & 0 & 0 & 0 \\
0 & 0 & 0 & 0 & 0 & 1 & 0 & 0  \\
0 & 0 & 0 & 0 & 0 & 0 & 1 & 0  \\
0 & 0 & 0 & 0 & 0 & 0 & 0 & 1\tikzmarkend{b}  \\
\end{array}
\right),
\end{align*}
where we have put the subgroup metric in the top left corner for convenience. When $\mu \to 1$ we get the standard bi-invariant metric on $SU(3)$, and when $\mu \to 0$ we get the round metric on $S^5$.

To determine the matrix $M$, we need to choose a basis for the Lie algebra $\mathfrak{su}(3)$. To do so, consider the matrices:
\begin{equation*}
\begin{aligned}[c]
X_1 &= 
\left(
\begin{matrix}
0 & 1 & 0 \\
0 & 0 & 0 \\
0 & 0 & 0
\end{matrix}
\right) 
\end{aligned}
\qquad
\begin{aligned}[c]
X_2 &= 
\left(
\begin{matrix}
0 & 0 & 0 \\
0 & 0 & 1 \\
0 & 0 & 0
\end{matrix}
\right)  
\end{aligned}
\qquad
\begin{aligned}[c]
X_3 &= 
\left(
\begin{matrix}
0 & 0 & 1 \\
0 & 0 & 0 \\
0 & 0 & 0
\end{matrix}
\right)
\end{aligned}
\end{equation*}
\begin{equation*}
\begin{aligned}[c]
Y_1 &= 
\left(
\begin{matrix}
0 & 0 & 0 \\
1 & 0 & 0 \\
0 & 0 & 0
\end{matrix}
\right) 
\end{aligned}
\qquad
\begin{aligned}[c]
Y_2 &= 
\left(
\begin{matrix}
0 & 0 & 0 \\
0 & 0 & 0 \\
0 & 1 & 0
\end{matrix}
\right)  
\end{aligned}
\qquad
\begin{aligned}[c]
Y_3 &= 
\left(
\begin{matrix}
0 & 0 & 0 \\
0 & 0 & 0 \\
1 & 0 & 0
\end{matrix}
\right)
\end{aligned}
\end{equation*}
\begin{equation*}
\begin{aligned}[c]
H_0 &= 
\left(
\begin{matrix}
\tfrac{1}{2} & 0 & 0 \\
0 & \tfrac{1}{2} & 0 \\
0 & 0 & -1
\end{matrix}
\right) 
\end{aligned}
\qquad
\begin{aligned}[c]
H_1 &= 
\left(
\begin{matrix}
1 & 0 & 0 \\
0 & -1 & 0 \\
0 & 0 & 0
\end{matrix}
\right).\\
\end{aligned}
\end{equation*}
Note that $X_3 = [X_1,X_2]$ and $Y_3 = [Y_2,Y_1]$. These matrices form a basis for the Lie algebra of $\mathfrak{su}(3)$:
\begin{align*}
\{Y_1,X_1,H_1,Y_2,Y_3,X_2,X_3,H_0\}.
\end{align*}
The matrices $\{Y_1,H_1,X_1\}$ form an $\mathfrak{su}(2)$ subalgebra, corresponding to the injective map of $\rho: SU(2) \hookrightarrow SU(3)$
\begin{align*}
\rho(A) &= \left( 
\begin{matrix}
A & 0 \\
0 & 1
\end{matrix}
\right).
\end{align*}

Under this $SU(3) \to SU(2)$, the adjoint of decomposes as $\textbf{8} \to \textbf{3} \oplus \textbf{2} \oplus \textbf{2} \oplus \textbf{1}$. We label the triplet as $\{y_1,h_1,x_1\}$, the first doublet as $\{y_2,y_3\}$, the second doublet as $\{x_2,x_3\}$, and the singlet as $\{h_0\}$. We then have 
\begin{align*}
v_a = \{ y_1,h_1,x_1,y_2,y_3,x_2,x_3,h_0\},
\end{align*}
and the matrix $M$ is given by
\begin{align*}
M &=  \left( 
\begin{matrix}
M_{11} & M_{12} \\
M_{21} & M_{22}
\end{matrix}
\right),
\end{align*}
where 
\begin{align*}
M_{11} &=  \left(
\begin{matrix} 
\mu & -h_1 & 2y_1 \\
h_1 & \mu & -2x_1 \\
-2y_1 & 2x_1 & \mu
\end{matrix}
\right),
\end{align*}

\begin{align*}
M_{12} &=  \left(
\begin{matrix} 
-y_3 & 0 & 0 & x_2 & 0 \\
0 & -y_2 & x_3 & 0 & 0 \\
y_2 & -y_3 & -x_2 & x_3 & 0
\end{matrix}
\right) \\[1em]
&= -M_{21}^T,
\end{align*}
and
\begin{align*}
M_{22} &=  \left(
\begin{matrix} 
1 & 0 & -h_0 + \tfrac{1}{2} h_1 & -x_1 & \tfrac{3}{2}y_2 \\
0 & 1 & -y_1 & -h_0-\tfrac{1}{2}h_1 & \tfrac{3}{2}y_3 \\
h_0-\tfrac{1}{2}h_1 & y_1 & 1 & 0 & -\tfrac{3}{2}x_2 \\
x_1 & h_0+\tfrac{1}{2}h_1 & 0 & 1 & -\tfrac{3}{2}x_3 \\
-\tfrac{3}{2}y_2 & -\tfrac{3}{2}y_3 & -\tfrac{3}{2}x_2 & -\tfrac{3}{2}x_3 & 1
\end{matrix}
\right).
\end{align*}
Note that there are five independent invariants under $SU(2)$ that one can construct from these multiplets \cite{P12}, given by
\begin{align*}
\mathcal{I}_1 &= h_0 \\[1em]
\mathcal{I}_2 &= h_1^2 + 4x_1 y_1 \\[1em] 
\mathcal{I}_3 &= x_2 y_2 + x_3 y_3 \\[1em]
\mathcal{I}_4 &= h_1 y_2 y_3 + y_1 y_2^2 - x_1 y_3^2 \\[1em]
\mathcal{I}_5 &= h_1 x_2 x_3 + x_1 x_2^2 - x_3^2 y_1,
\end{align*}
together with a sixth invariant
\begin{align*}
\mathcal{I}_6 &= h_1 (x_2 y_2 - x_3 y_3) - 2 (y_1 y_2 x_3 + x_1 x_2 y_3 ),
\end{align*}
which is not algebraically independent, since
\begin{align*}
\mathcal{I}_2 \mathcal{I}_3^2 - 4 \mathcal{I}_4 \mathcal{I}_5  = \mathcal{I}_6^2.
\end{align*}
We want to gauge fix three of the variables in such a way that the invariants are in one to one correspondence with the remaining variables. A convenient choice, though by no means the only one, is to fix $y_1 = y_2 = 0$, and $y_3 = 1$. The invariants are then given by
\begin{align*}
\mathcal{I}_1 &= h_0 \\[1em]
\mathcal{I}_2 &= h_1^2 \\[1em] 
\mathcal{I}_3 &= x_3  \\[1em]
\mathcal{I}_4 &= -x_1 \\[1em]
\mathcal{I}_5 &= h_1 x_2 x_3 + x_1 x_2^2.
\end{align*} 
The resulting dual fields $\widehat{E} = \widehat{g} + \widehat{B}$ are then found by inverting $M$, setting $\mu \to 0$, and applying this choice of gauge. The dual fields are quite complicated, but there is a nonvanishing $B$-field, and the metric has singularities since the determinant vanishes at certain values of the dual coordinates. 

\subsubsection{A puzzle}
The `rule of thumb' requirement for having a valid choice of gauge fixing is that after gauge fixing, the remaining coordinates should be in one-to-one correspondence with the invariants \cite{LOST}. Whilst our choice of gauge fixing here certainly satisfies this requirement, it appears that not all such choices result in a dual geometry. Consider, for example, the choice of gauge given by $x_2 = y_2 = h_1 =1$. The invariants are then
\begin{align*}
\mathcal{I}_1 &= h_0 \\[1em]
\mathcal{I}_2 &= 4 x_1 y_1 +1 \\[1em] 
\mathcal{I}_3 &= x_3 y_3 + 1  \\[1em]
\mathcal{I}_4 &= -x_1 y_3^2 +y_1 +y_3 \\[1em]
\mathcal{I}_5 &= -x_3^2 y_1 + x_1 + x_3.
\end{align*}
These relations can be inverted to give the coordinates as functions of the invariants. One finds, however that with this choice of gauge the dual metric has
\begin{align*}
\textrm{det}(\widehat{g}) &= \mathcal{O}(\mu),
\end{align*}
so that taking the limit as $\mu \to 0$ we get
\begin{align*}
\lim_{\mu \to 0} \textrm{det}(\widehat{g}) = 0.
\end{align*}

\subsubsection{Dimension of the T-dual space}
A peculiarity of non-abelian T-duality is the gauge fixing procedure we need to perform in order to obtain the dual space. The dimension of the Lie algebra of Killing vectors is typically not the same dimension as the manifold, so without gauge fixing, the dual space would have a dimension which differs from the original space. Given this, one might be concerned that before gauge fixing, the dimension of the dual space could be arbitrarily high, and that we therefore need to gauge fix arbitrarily many Lagrange multipliers. The following lemma, proved for example in \cite{O83}, provides a bound on the dimension of the Lie algebra.
\begin{lemma}
Let $(M,g)$ be a Riemannian manifold of dimension $D$. Then the Lie algebra of Killing vectors on $M$ has dimension at most $\frac{D(D+1)}{2}$.
\end{lemma}
It follows that when gauge fixing, we need to fix at most $\frac{D(D-1)}{2}$ of the Lagrange multipliers. 

\subsection{NATD and non-semisimple groups}
\label{nonsemisimple}
Shortly after the introduction of non-abelian T-duality, examples were constructed where the non-abelian T-dual space was not a valid supergravity background, even when the original space was a valid supergravity background \cite{GRV,GR95}. The non-abelian duals for the Minkowski and Bianchi V spacetimes of Section \ref{subsec:NATDexamples} provide illustrations of this. For the Bianchi V spacetime in particular, one can show that the dual model does not solve the supegravity equations of motion with the prescribed transformation of the dilaton, and indeed that there is \emph{no} transformation of the dilaton which solves the equations of motion \cite{GRV}. 

It was pointed out in \cite{GR} that the Lie algebras of the isometry groups for these examples had structure constants with non-vanishing trace, and were therefore not semisimple Lie algebras. This was also suggested to give rise to an anomaly. In \cite{AAL94}, it was shown that a mixed gauge/gravity anomaly does indeed appear in the dual model, and this was studied further in \cite{EGRSV}.  

Recent work has studied non-abelian T-duality for these nonsemisimple groups \cite{HKO}. The authors show that if the group being gauged is non-semisimple, the dual model is not a solution of the supergravity equations of motion, but it is a solution of \emph{generalised} supergravity. 


\subsection{Singularities in the dual space}
\subsubsection{The $SU(2)$ NATD of the D3 near-horizon}
In Abelian T-duality, singularities arose in the dual space when the $U(1)$ symmetry acted with fixed points, or equivalently, when the norm of the Killing vector vanished. We think of T-duality acting on a background by inverting the $S^1$ fibers, so if the fiber shrinks to a point in the original space, this should correspond to a singularity in the dual space. This kind of singularity can also occur in non-abelian T-duality. Consider, for example, performing an $SU(2)$ non-abelian T-duality on the D3 near-horizon background,\footnote{Recall the D3 near-horizon is a type-IIB $AdS_5 \times S^5$ solution} which was done in \cite{ST11}. We have already seen this background in Section \ref{NATDcoset}, where we performed a coset space non-abelian T-duality by treating the $S^5$ as the coset space $SU(3) / SU(2)$. It was also considered in \cite{LOST}, where they performed a coset space non-abelian T-duality by treating the $S^5$ as the coset space $SO(6)/SO(5)$. In this section we are interested in the $SU(2)$ isometry acting on the $S^5$, which has fixed points. The metric for the $S^5$ is
\begin{align*}
\dd s^2 (S^5) &= 4\left( \dd \theta^2 + \sin^2 \theta \dd \phi^2 \right) +\cos^2 \theta \dd s^2 (S^3).
\end{align*}
This metric exhibits $S^5$ as a degenerate $S^3$ fibration over $S^2$. That is, at $\theta = \pm \frac{\pi}{2}$, the $S^3$ fiber shrinks to zero size. The non-abelian T-dual with respect to the $SU(2)$ action on the $S^3$ is calculated in \cite{ST11}, and the dual metric and $B$-field are given by
\begin{align*}
\widehat{g} &= \dd s^2 \left( AdS_5 \right) + 4\left (\dd \theta^2 + \sin^2 \theta \dd \phi \right) + \frac{\dd r^2}{\cos^2\theta} + \frac{r^2 \cos^2 \theta}{\cos^4 \theta + r^2} \dd \Omega^2_2 \\[1em]
\widehat{B} &= \frac{r^3}{\cos^4 \theta +r^2} \textrm{Vol}(S^2).
\end{align*}
Note that at $\theta = \pm \frac{\pi}{2}$ the metric blows up, corresponding to the shrinking $S^3$ fibers in the original space. Of course, we should check to see that these singularities aren't simply coordinate artifacts. A straightforward calculation confirms that the Ricci scalar diverges at $\theta = \pm \frac{\pi}{2}$, so these are indeed curvature singularities.

\section{Topology?}

What can we say about the topology of the non-abelian T-dual? Let's consider for the moment the case of the non-abelian T-dual of $S^3$. The coordinates are the Lagrange multipliers from the gauging procedure, and are therefore defined on the dual Lie algebra, $\mathfrak{su}(2)^{\ast}$, of the isometry group $SU(2)$. Na\"{i}vely, this suggests that the topology of the space is $\mathbb{R}^3$,  since Lie algebras are linear spaces. Recall, however, that we obtained a similar result for abelian T-duality - the coordinates there were also Lagrange multipliers, \emph{a priori} taking arbitrary real values. From a supergravity perspective, this is the end of the story. It was a string theoretic argument in Section \ref{subsec:FibTop}, which constrained the range of the dual coordinates in terms of the range of the original coordinates. In particular, if the original coordinates were compact with period $2\pi$, then the dual coordinates were also required to be compact with period $2\pi$. 

As we shall see in Section \ref{subsec:NATDFibTop}, this argument no longer holds for non-abelian T-duals. Although we still have the supergravity interpretation, we don't have an understanding of how non-abelian gauging works in higher genus worldsheets, and therefore no information on the range of the dual coordinates. There have been some recent attempts to understand this from a physics perspective, which we discuss briefly in Section \ref{NATDphysics}. In Section \ref{NATDcommentsS3} and Section \ref{NATDcommentsTT} we offer some comments from a more mathematical perspective. 


\subsection{Topology of the fibers - what fails?}
\label{subsec:NATDFibTop}
In Section \ref{subsec:FibTop}, we reviewed the string-theoretic argument which allowed us to conclude that the topology of the fiber does not change when we perform a T-duality. If the original fiber is compact, i.e. the Killing vector has closed orbits, then the dual coordinate is constrained by the holonomies of the gauge fields to be periodic, so the dual fiber is also compact. If the original fiber is non-compact, then applying the same argument to the dual space allows us to conclude that the dual fiber cannot be compact, and must therefore be non-compact. A natural question to ask is what happens to the topology of the fiber in the case of non-abelian T-duality. 

As in the abelian case, the dual coordinates are defined \emph{a priori} on the Lie algebra $\fg$ of the group of isometries $\sG$. For simplicity, let us assume that $\sG = SU(2)$, so that the coordinates are defined on $\mathfrak{su}(2)^{\ast}$. Since $\mathfrak{su}(2)^{\ast}$ is a real vector space, it is topologically $\RR^3$, and therefore non-compact. It is a reasonable question to ask if we can compactify this non-compact space, and if so, if there is a string theoretic argument which tells us how to do this. Unfortunately, a na\"{i}ve application of the same argument in Section \ref{subsec:FibTop} fails. In the non-abelian case, winding modes of the dual coordinates (if that can be given a meaning) don't multiply the holonomies of the gauge fields, which require path-ordering \cite{AAL94}. A general procedure for performing non-abelian T-duality for higher genus worldsheets is currently lacking. 

\subsection{Topology from the physics perspective}
\label{NATDphysics}
Recently, there have been some very interesting attempts to determine global properties of non-abelian T-duality coming from a physics perspective. These attempts utilise the $AdS/CFT$ duality to attempt to determine global properties of the non-abelian T-dual space. 

\subsubsection{$B$-field holonomy}
The general principle is to consider an $AdS\times M$ background, and then perform a non-abelian T-duality for some group of isometries acting trivially on the AdS part of the spacetime. The dual background, $AdS \times \widehat{M}$ then has a CFT dual by the AdS/CFT correspondence, and it is hoped that by studying the properties of this CFT, one can determine global properties of the dual space $\widehat{M}$.\footnote{Once again, the word ``duality" is a heavily overused word in this thesis. To avoid confusion, we will try to avoid dropping the ``T" in T-duality, and will always refer to the dual conformal field theory obtained via the AdS/CFT correspondence as the ``CFT dual".} The observation of \cite{LM,LOR} relied on the boundedness of the action for a string wrapping a non-trivial 2-cycle, leading to a bound on the quantity:
\begin{align}
\label{bfieldholonomy}
b &= \frac{1}{4 \pi^2} \left| \int_{M_2} B \right| ,
\end{align}
where $M_2$ is a suitably chosen 2-cycle in $M$. In particular, we want to require that $b \in [0,1]$. We first note that in the abelian case, this gives the correct prescription for obtaining the periodicity of the dual coordinates. Consider, for example, the abelian Hopf T-duality of $S^3$ with no flux, studied in Section \ref{S3Buscher}. The dual metric and the dual $B$-field is given by:
\begin{subequations}
	\begin{align*}
	\widehat{\dd s^2} & = \dd \eta^2 + \frac{1}{4}\sin^2 (2 \eta) \dd \xi_2^2 + 4 \dd \widehat{\xi_1}^2 \\[1em]
	\widehat{B} &=  \frac{1}{2} \cos(2 \eta) \dd \xi_2 \wedge \dd \widehat{\xi_1}.
	\end{align*}
\end{subequations}
Here, the original coordinates have ranges $(\eta, \xi_2) \in [0,\frac{\pi}{2}] \times [0,4\pi]$, and we are interested in the range of the dual coordinate $\widehat{\xi}_1$. By fixing $\eta = \eta_0$, we obtain a submanifold over which we can integrate the $B$-field. Setting the as-yet-undetermined periodicity of the dual coordinate to be $\alpha$, the result is:
\begin{align*}
b &= \frac{1}{4 \pi^2} \left| \int_{M_2} \widehat{B} \right| \\[1em]
&=  \frac{1}{8 \pi^2} \left| \int_{\xi_2 = 0}^{\xi_2 = 4\pi} \int_{\widehat{\xi}_1 = 0}^{\widehat{\xi}_1 = \alpha} \cos (2 \eta_0 ) \dd \xi_2 \wedge \dd \widehat{\xi}_1 \right| \\[1em]
&= \frac{\alpha}{2 \pi} \left| \cos(2\eta_0) \right|.
\end{align*}
It follows that $b$ lies in the range $[0,1]$, independent of $\eta_0$, provided that the periodicity of the dual coordinate is $\alpha = 2\pi$. This agrees with the analysis of Section \ref{subsec:FibTop}, where we found that 
\begin{align}
\int \dd \theta \wedge \dd \hat{\theta} = (2 \pi)^2.
\end{align}
What does this tell us when we are considering non-abelian T-duality? For the non-abelian T-dual of $S^3$, studied in Section \ref{subsec:NATDexamples}, we can apply the same analysis. Switching to polar coordinates,\footnote{See Section \ref{NATDcommentsS3}.} the dual metric and $B$-field are given by:
\begin{subequations}
	\label{NATDS3deformed}
\begin{align}
\label{NATDS3metric}
\dd s^2 &= \dd r^2 + \frac{ r^2}{1 +r^2} \left( \dd \theta^2 + \sin^2 \theta \dd \phi^2 \right) \\[1em]
B &= - \frac{r^3 \sin \theta}{1+r^2} \dd \theta \wedge \dd \phi.
\end{align}
\end{subequations}
Computing the quantity (\ref{bfieldholonomy}) requires us to choose a 2-manifold over which to integrate. Symmetry suggests that we should choose a 2-sphere of constant radius $r = r_0$. Of course, we don't have global information about the topology of the dual manifold, so we can't be sure that there are non-trivial 2-cycles at finite $r$. We assume that there are, and obtain:
\begin{align*}
b &= \frac{1}{4 \pi^2} \left| \int_{M_2} \widehat{B} \, \right| \\[1em]
&= \frac{1}{4 \pi^2} \left| \int_{\theta = 0}^{\theta = \pi} \int_{\phi=0}^{\phi = 2\pi} \frac{r_0^3}{1+r_0^2} \sin \theta \dd \theta \wedge \dd \phi \right| \\[1em]
&= \frac{1}{\pi} \left| \frac{r_0^3}{1+r_0^2} \right|.
\end{align*}
For this to be the range $b \in [0,1]$, we require that the radial coordinate should be constrained to $r \in [0, R]$, for some finite $R$.\footnote{$R$ is the unique real root of the cubic $f(r) = r^3 - \pi r^2 - \pi$. It is possible to obtain an exact expression for this, but it is rather complicated so we have omitted it.} 
Of course, the $B$-field is not uniquely defined - it is a gauge field, and therefore physically equivalent configurations are related by gauge transformations. It is easy to see that under a gauge transformation of the form
\begin{align}
B \mapsto B + \dd C,
\end{align}
the expression (\ref{bfieldholonomy}) is invariant, since $\int_{M_2} \dd C = 0$ by Stokes' theorem. If, however, we perform a large gauge transformation of the form
\begin{align*}
B \mapsto B + n\pi Vol(S^2),
\end{align*}
with $n$ an integer, we find that
\begin{align}
b &= \left| \frac{r_0^3}{e^{2\sigma}+r_0^2} + n \pi  \right|.
\end{align}
We refer the interested reader to \cite{LOR} for a more detailed analysis of this situation. The takeaway of their analysis is that physical arguments suggest we should consider a cutoff for the radial coordinate. Such a cutoff would seem to introduce a boundary to the manifold, which would normally require the addition of localised sources so that the equations of motion are still satisfied. This is a somewhat unsatisfactory conclusion, since the geometry is perfectly smooth there.\footnote{Indeed, the geometry is perfectly smooth for all $r \in [0,\infty)$.} Similar arguments are used for the non-abelian T-duals of $AdS_5 \times S^5$, $AdS_5 \times T^{1,1}$, and $AdS_5 \times Y^{p,q}$ in \cite{MNPRW}. There, they find that radial coordinate $\rho$ should lie in the range $\rho \in [0,\pi]$. Large gauge transformations shift the radial function so that it is `quantised' in the range $\rho \in [n\pi,(n+1) \pi]$. There have been attempts to explain this from a physical perspective as a type of Seiberg duality. 

Interestingly, it is noted in \cite{LMMN,LN,MNTZ} that the $r \to \infty$ limit of the fields in the non-abelian T-dual of $S^3$ coincides with the fields obtained in the abelian T-dual of $S^3$ along the Hopf fibration. 
\begin{align}
\lim_{r \to \infty} \left( 
\begin{matrix}
g \\
B \\
e^{\Phi} F
\end{matrix}
\right)_{NATD} &=  \left( 
\begin{matrix}
g \\
B \\
e^{\Phi} F
\end{matrix}
\right)_{ATD}
\end{align}
where the RR fluxes are collected with the dilaton into the polyform $e^{\Phi}F$. At large $r$, the metric (\ref{NATDS3metric}) has the same form of the product metric on $\RR \times S^2$, and periodic identification of the $r$ coordinate gives a topology of $S^1 \times S^2$, precisely matching the topology of the abelian T-dual of $S^3$.
Whether this limit has any additional meaning or is simply a coincidence is yet to be determined. 
\subsection{Topology from the maths perspective}
\subsubsection{Comments on topology: The NATD of $S^3$}
\label{NATDcommentsS3}
We consider in this section the non-abelian T-dual of $S^3$ with no $B$-field, discussed in \ref{NATDS3}. Recall that the dual data is defined on the Lagrange multiplier coordinates $\{\chi_i\} = \{x,y,z\}$, and we have the following metric, $B$-field, and dilaton:
\begin{align*}
\dd s^2 &= \frac{1}{1+\chi^2} (\delta_{ij} + \chi_i \chi_j) \dd \chi^i \dd \chi^j \\[1em]
B &= - \epsilon_{ijk} \frac{\chi_k}{1+\chi^2} \dd \chi^i \wedge \dd \chi^j \\[1em]
\phi &= \log \left( \frac{1}{1+\chi^2} \right),
\end{align*}
where $\chi^2 = x^2 + y^2 + z^2$. We can switch to polar coordinates, using the transformations:
\begin{subequations}
\begin{align}
\label{cartesiantopolar}
x &= r\,  \sin \theta \, \cos \phi \\[1em]
y &= r\,  \sin \theta \, \sin \phi \\[1em]
z &= r\, \cos \theta.
\end{align}
\end{subequations}
In these coordinates, we obtain:
\begin{subequations}
\label{NATDS3polar}
\begin{align}
\label{NATDS3polarmetric}
\dd s^2 &= \dd r^2 + \frac{r^2}{1+r^2} \left( \dd \theta^2 + \sin^2 \theta \dd \phi^2 \right) \\[1em]
B &= -\frac{r^3 \sin \theta}{1+r^2} \dd \theta \wedge \dd \phi \\[1em]
\phi &= \log \left( \frac{1}{1+r^2} \right).
\end{align}
\end{subequations}
This metric is, \emph{a priori}, defined on the dual Lie algebra $\mathfrak{su}(2)^{\ast}$, which, since it is a vector space, is topologically $\RR^3$. Recall that a similar thing happened in the abelian case - the dual metric was defined in terms of the dual coordinate $\hat{\theta}$, which was a Lagrange multiplier living in the dual Lie algebra, $\mathfrak{u}(1)^{\ast} \simeq \RR$, of $U(1)$.  In that case, a string theoretic argument constrained the range of the dual coordinate. Although the corresponding string theoretic argument is lacking in the non-abelian case, one could ask whether there is a `correct' compactification of the dual manifold which is compatible with the dual metric. This is a difficult question since, in principle, there are often many manifolds which are compatible with a given coordinate description of a metric. For this particular example, however, we are fortunate. The Ricci scalar of the metric (\ref{NATDS3polarmetric}) is 
\begin{align}
\label{NATDS3Ricci}
\mathcal{R} &= \frac{2 \left(r^4 + 3r^2 + 9 \right)}{(1+r^2)^2}.
\end{align}
This is manifestly positive, so (\ref{NATDS3polarmetric}) is a metric of positive scalar curvature. It follows from Theorem \ref{SchoenYau}, that the only closed, orientable manifolds with which this metric is compatible  are spherical 3-manifolds, copies of $S^1 \times S^2$, and connect sums of these. 
Guided by the results of Theorem \ref{SchoenYau}, we could try to compactify the dual manifold to give it the topology of $S^3$. This seems like a natural choice, and would be close in spirit to abelian T-duality, where the fibers of the original and the dual spaces are required to have the same topology. 
Let us attempt to do this na\"{i}vely, and see what goes wrong. 

We begin by noting that $S^3$ is the one-point compactification of $\RR^3$, so we can try to compactify the dual space to $S^3$ simply by identifying the `point at infinity'. The Ricci scalar (\ref{NATDS3Ricci}), as well as the Kretschmann scalar
\begin{align}
\mathcal{K} &= \frac{4 \left(r^8 + 6r^6 + 15r^4 +18r^2 + 27 \right)}{(1+r^2)^4},
\end{align}
are both finite as $r \to \infty$. To see the problem, consider two points at a fixed radius $r$. For concreteness, take the two points to be $x_1 = \{r,\theta = \frac{\pi}{2}, \phi = 0 \}$ and $x_2 = \{r,\theta = \frac{\pi}{2}, \phi = \pi \}$. The geodesic distance between these two points approaches a finite, non-zero value as $r \to \infty$, so the points can't be identified with each other in the limit. 

There is, in fact, an easy way to see that the dual manifold is necessarily non-compact if we take $r \in [0, \infty)$, and not just a compact manifold in disguise (as in the abelian case). To see why, recall the definition of the diameter of a Riemannian manifold:
\begin{align}
\textrm{diam}(M) = \sup_{p,q \in M} d(p,q).
\end{align}
If $M$ is compact, then this is a continuous function on a compact set, and therefore attains a maximum. That is, the diameter of a compact set is finite. On the other hand, it is easy to see that with the metric (\ref{NATDS3polarmetric}), the diameter is not finite. Take $N \in \NN$, and consider the distance between the points $p = (r,\theta,\phi)  = (1, \tfrac{\pi}{2}, 0)$ and $q = (N+1,  \tfrac{\pi}{2},  0)$. We have
\begin{align}
d(p,q) &= N,
\end{align}
and so it follows that the diameter cannot be finite, and therefore $M$ cannot be compact.

\subsubsection{Analytic continuation and speculation}
For arbitrary $x,y,z \in \mathbb{R}$, the $r$ coordinate is constrained to $r \in [0, \infty)$. We notice, however, that the metric (\ref{NATDS3polarmetric}) depends only on $r$ quadratically. Indeed, even though taking $r<0$ corresponds to a complex transformation (\ref{cartesiantopolar}), the line element (\ref{NATDS3polarmetric}) remains real.\footnote{This also happens, for instance, in the Kruskal-Szekeres coordinates for the Kerr black hole.} The $B$-field, dilaton, as well as the Ricci and Kretschmann scalars are all perfectly well-defined for all $r \in \mathbb{R}$. Motivated by this observation, we make the coordinate change $r = \tan t$. The new metric, $B$-field and dilaton are given by
\begin{align}
\label{NATDS3tanmetric}
\dd s^2 &= \frac{1}{\cos^4 t} \dd t^2 + \sin^2 t \left( \dd \theta^2 + \sin^2 \theta \dd \phi^2 \right) \\[1em]
B &= - \frac{\sin^3 t \, \sin \theta}{\cos t} \dd \theta \wedge \dd \phi \\[1em]
\phi &= 2 \log \left( \cos t \right).
\end{align}
For $0<r<\infty$, the new coordinate $t$ ranges from $0$ to $\frac{\pi}{2}$, however $r$ is an analytic function of $t$, so we can analytically continue it, at least until we get to the nearest pole. The poles of $\tan t$ are at $t = \frac{\pi}{2} + n \pi$, so we can extend the range of $t$ to be $(-\frac{\pi}{2}, \frac{\pi}{2})$. Note that this corresponds to taking $r \in \mathbb{R}$. The metric still degenerates at $r = t = 0$, and now we have a (coordinate) singularity at $t = \pm \frac{\pi}{2}$, corresponding to $r = \pm \infty$. Nevertheless, on the intervals $t \in (-\frac{\pi}{2}, 0)$ and $(0, \frac{\pi}{2})$, all of our fields are well-defined and smooth. The scalar curvature in these coordinates is $$\mathcal{R} = 14 \cos^4 t + 2 \cos^2 t + 2,$$ which we plot as a function of $t$ in Figure \ref{fig:RicciScalartan}. 

\begin{figure}[h]
		\centering
		\includegraphics[width =0.9\linewidth]{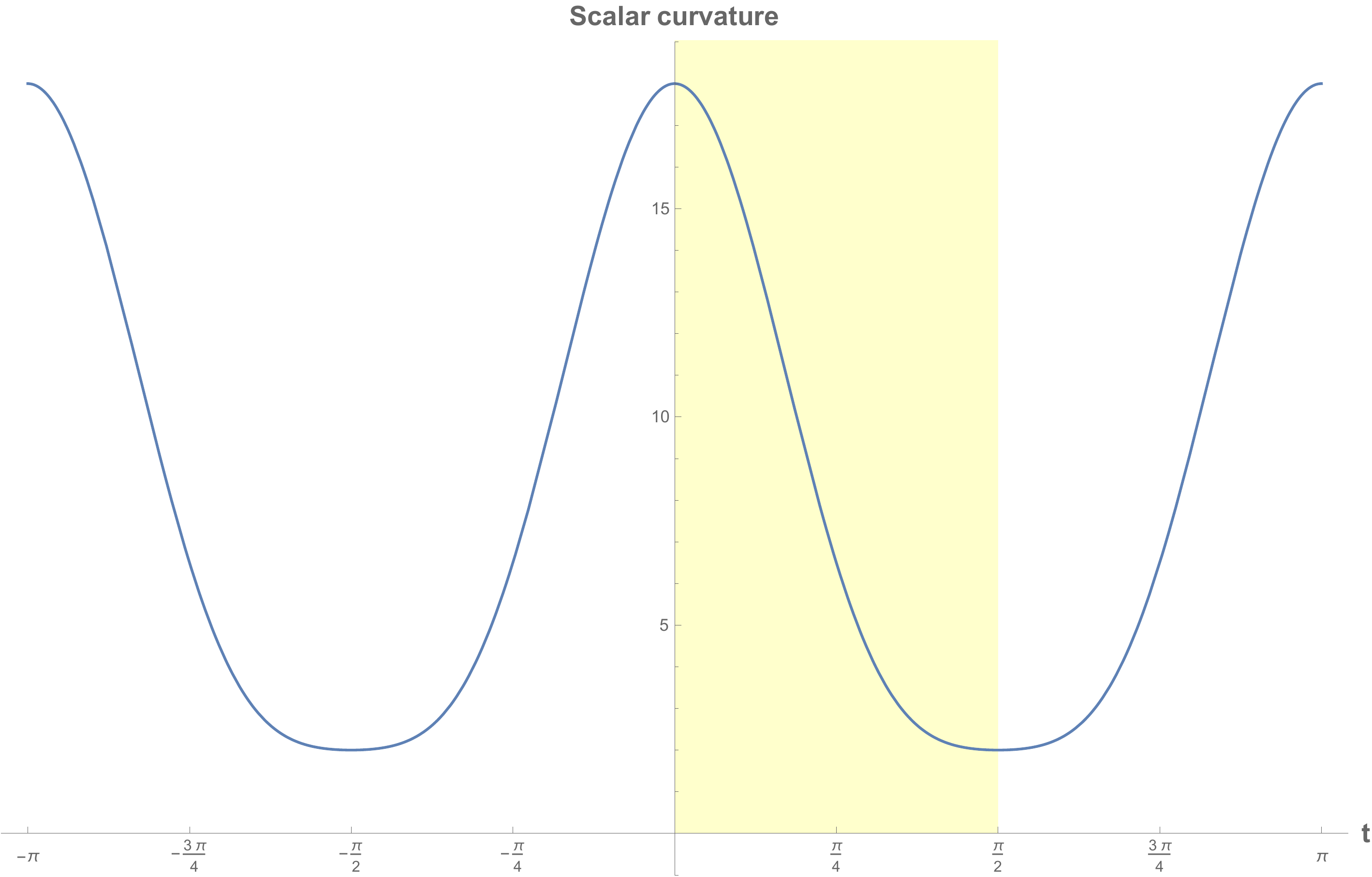}
			\caption{A plot of the scalar curvature of the metric (\ref{NATDS3tanmetric}) as a function of $t$. The region $t \in (0,\tfrac{\pi}{2})$ is highlighted in yellow, corresponding to the original coordinate range of $r \in [0,\infty)$.}
			\label{fig:RicciScalartan}
\end{figure}

It is tempting to consider $t$ a periodic coordinate - the metric is invariant under the substitution $t \to t + \pi$, as are the $B$-field, the dilaton, and the scalar invariants. This would have the effect of identifying $r = -\infty$ with $r = \infty$. 


\subsubsection{Comments on topology: The NATD of the Twisted Torus}
\label{NATDcommentsTT}
The twisted torus was discussed in Section \ref{T3Buscher}, where it occurred as the T-dual of the three torus with $H$-flux. More specifically, it was the first in a series of dualites:
\begin{align}
T_{xyz} \stackrel{\pr_z}{\longleftrightarrow} f_{xy} \!^z \stackrel{\pr_y}{\longleftrightarrow} Q_x \!^{yz} \stackrel{\pr_x}{\longleftrightarrow} R^{xyz}.
\end{align}
We also discussed this example in \ref{NATDHeis}, where we noticed that in addition to the abelian isometry taking us back to the three-torus with $H$-flux, the twisted torus had in fact a non-abelian group of isometries. This comes from the fact that we can also view the twisted torus as the quotient of the real Heisenberg group by its integer counterpart
\begin{align}
\textrm{Nil} = \textrm{Heis}(\RR) / \textrm{Heis}(\ZZ).
\end{align}
Performing a non-abelian T-duality on an $f$-flux background with respect to this non-abelian isometry group gives a remarkable result -  the T-dual is a $Q$-flux background with the indices interchanged. This was noticed in \cite{BBKW}, following the results of \cite{C16,CDJ15}, and is discussed within the context of non-isometric T-duality in Section \ref{subsec:NITDexamples}.

As noted in \cite{C16,CDJ15}, the single non-abelian duality $f_{xy}\!^z \longrightarrow Q^{xy}\!_z$ can be written as a chain of abelian dualities:
\begin{equation}
\label{chainabelian}
\begin{tikzpicture}[baseline=(current  bounding  box.center)]
\node (T) {$T_{xyz}$};
\node (f1) [right=2cm of T] {$f^x \!_{yz}$};
\node (f2) [below=2cm of T] {$f_{xy}\!^z$};
\node (Q) [below=2cm of f1] {$Q^{xy}\!_z$};

\draw[<->] (T) to node {$\pr_x$} (f1);
\draw[<->] (f2) to node {$\pr_z$} (T);
\draw[->,dashed] (f2) to node [swap] {$NATD$} (Q);
\draw[<->] (f1) to node {$\pr_y$} (Q);
\draw[red,->] (T.135) arc (0:290:3mm) node[pos=0.5,swap]{$\delta B$};
\end{tikzpicture}
\end{equation}
where we must perform a gauge transformation for the $B$-field in order to continue the T-duality chain. In fact, there is an even nicer picture here. Starting with the three torus with flux, the various abelian T-dualities can be incorporated into the following duality cube:
\begin{center}
\begin{tikzcd}[back line/.style={densely dotted}, row sep=3em, column sep=3em]
	& Q_x\!^{yz} \ar[leftrightarrow]{rr}[red]{\pr_x} 
	& & R^{xyz}  \\
	f_{xy}\!^z \ar[leftrightarrow]{ur}[blue]{\pr_y} \ar[crossing over,leftrightarrow]{rr}[near start,blue]{\pr_x} 
	& & Q^{x}\!_y\!^z \ar[leftrightarrow]{ur}[red]{\pr_y} \\
	& f_{x}\!^y \!_z \ar[back line,leftrightarrow]{rr}[blue, near start,swap]{\pr_x}  \ar[back line,leftrightarrow]{uu}[near start,blue,swap]{\pr_z} 
	& & Q^{xy}\!_{z} \ar[leftrightarrow]{uu}[red,swap]{\pr_z} \\
	T_{xyz} \ar[leftrightarrow]{uu}{\pr_z} \ar[leftrightarrow]{rr}{\pr_x} \ar[back line,leftrightarrow]{ur}{\pr_y} & & f^x\!_{yz} \ar[crossing over,leftrightarrow]{uu}[blue, near end, swap]{\pr_z} \ar[leftrightarrow]{ur}[blue,swap]{\pr_y}
\end{tikzcd}
\end{center}
The non-abelian T-duality we have been discussing then corresponds to a map from opposite sides of the cube, sending $f \to Q$. This can be realised as a chain of abelian T-dualities by tracing along the edges of the cube. One such path is the one we have already discussed in (\ref{chainabelian}):
\begin{center}
\begin{tikzcd}[back line/.style={densely dotted}, row sep=3em, column sep=3em]
	& \color{lightgray}{Q_x\!^{yz}} \ar[lightgray,leftrightarrow]{rr}[lightgray]{\pr_x} 
	& & \color{lightgray}{R^{xyz}}  \\
	f_{xy}\!^z \ar[lightgray,leftrightarrow]{ur}[lightgray]{\pr_y} \ar[crossing over,lightgray,leftrightarrow]{rr}[near start,lightgray]{\pr_x} 
	& & \color{lightgray}{Q^{x}\!_y\!^z} \ar[lightgray,leftrightarrow]{ur}[lightgray]{\pr_y} \\
	& \color{lightgray}{f_{x}\!^y \!_z} \ar[lightgray,back line,leftrightarrow]{rr}[lightgray, near start,swap]{\pr_x}  \ar[back line,lightgray,leftrightarrow]{uu}[near start,lightgray,swap]{\pr_z} 
	& & Q^{xy}\!_{z} \ar[lightgray,leftrightarrow]{uu}[lightgray,swap]{\pr_z} \\
	T_{xyz} \ar[leftrightarrow,thick]{uu}{\pr_z} \ar[thick,leftrightarrow]{rr}{\pr_x} \ar[back line,lightgray,leftrightarrow]{ur}{\pr_y} & & f^x\!_{yz} \ar[crossing over,lightgray,leftrightarrow]{uu}[swap,near end,lightgray]{\pr_z} \ar[thick,leftrightarrow]{ur}[blue,swap]{\pr_y}
\end{tikzcd}
\end{center}
 Note that the inverse map isn't well-defined as a single non-abelian T-duality, since the $Q$-flux background does not have a globally defined non-abelian group of isometries with which we can dualise. This is in line with our expectations from non-abelian T-duality - we aren't normally able to invert non-abelian T-duality. On the other hand, the inverse chain of abelian T-dualities is certainly well-defined. This coincidence - that a non-abelian T-duality agrees with a chain of abelian T-dualities - is more than a curiosity. Topological aspects of non-abelian T-duality are still not understood, even for the simplest examples of gauging the left action of a group on itself. The twisted torus, however, provides us with an example where we know \emph{explicitly} the topology of the non-abelian T-dual, since the non-abelian T-dual agrees with a chain of abelian T-dualities, whose topological behaviour is well-understood. In this case, the non-abelian T-dual is no longer a manifold, so ``topological aspects of non-abelian T-duality'' should be suitably interpreted. The appearance of the $Q$-flux background as the non-abelian T-dual of the $f$-flux background suggests that global aspects of non-abelian T-duality might only be understood in the broader context of noncommutative and non-associative geometry and T-folds. It is perhaps through this observation that we can make sense of the `periodicity' of the radial coordinate appearing in the non-abelian T-dual of $S^3$, described in Section \ref{NATDphysics}.

%% file: chapter4.tex
\chapter{Poisson-Lie T-duality}
\label{chptr:Ch4}

Poisson-Lie T-duality, first introduced by Klim\v c\'\i k and \v Severa in \cite{KS95,Kli95}, is a generalisation of non-abelian T-duality. Non-abelian T-duality, at least from the path integral perspective, suffers from an inability to invert the procedure - given a sigma model with isometries, it is straightforward to construct a dual sigma model. The dual model, however, will generically have fewer isometries than the original model. In particular, we aren't guaranteed the existence of a group of isometries we can perform a non-abelian T-duality with respect to which will return us to the original model.\footnote{See, however, \cite{CZ94,L96} for a discussion of the invertability of non-abelian T-duality as a canonical transformation.} This observation led Klim\v c\'\i k and \v Severa to propose that the relevant algebraic structure wasn't the group of isometries (or its associated Lie algebra), but some other structure which appears as a group of isometries under some circumstances. Their suggestion was that the relevant algebraic structure is in fact a Lie bialgebra.

In Section \ref{sec:DD}, we discuss the algebraic structure underlying Poisson-Lie T-duality, the Drinfeld double. Then, in Section \ref{sec:PL}, we will discuss the Poisson-Lie symmetry conditions and the various classes of Poisson-Lie T-duality. Finally, in Section \ref{sec:PLexamples} we will discuss examples of Poisson-Lie T-duality. 
\section{Drinfeld double}
\label{sec:DD}
In the context of this thesis, a Drinfeld double is a Lie group $\sD$ whose Lie algebra $\fd$ can be decomposed into a pair of maximally isotropic subalgebras with respect to a non-degenerate invariant bilinear form on $\fd$. The term Drinfeld double has a more general meaning in the context of quantum groups, where it is used to refer to a construction which takes a Hopf algebra and creates a quasitriangular Hopf algebra. These objects are interesting from a mathematical perspective, since the quasitriangularity property gives the $R$ matrix, a solution to the Yang-Baxter equation, which can also be used to construct invariants for knots. The meaning which we will take in this thesis is the so-called ``semiclassical analogue" of the quantum double. 

Let $\fd$ be the Lie algebra of a Lie group $\sD$, and suppose that $\braket{\cdot , \cdot}$ is a non-degenerate invariant bilinear form on $\fd$. An \emph{isotropic subspace} of $\fd$ is a vector subspace on which the bilinear form vanishes. An isotropic subspace is \emph{maximal} if it is not a proper subspace of another istropic subspace. We now say that $\fd$ is a \emph{Drinfeld double} if it can be decomposed into the direct sum of two maximally isotropic subalgebras. A given decomposition $\fd = \fg \oplus \widetilde{\fg}$ is called a Manin triple, and we will label it by the pair $(\fg, \widetilde{\fg})$. Note that since the bilinear form on $\fd$ is non-degenerate, we can use it to identify $\widetilde{\fg}$ with the dual vector space:
\begin{align}
\label{dualLie}
\widetilde{\fg} = \fg^{\ast}.
\end{align}
The Lie subalgebra structure on $\widetilde{\fg}$ then makes $\fd$ into a Lie bialgebra.\footnote{A Lie bialgebra is a Lie algebra $\fg$ with a compatible Lie algebra structure on the dual vector space $\fg^{\ast}$.} Conversely, every Lie bialgebra defines a Manin triple by identifying the dual Lie algebra $\fg^{\ast}$ with $\widetilde{\fg}$, and defining the commutator between $\fg$ and $\widetilde{\fg}$ to make the bilinear form invariant.\footnote{For details, see \cite{KS07}.} 
Choosing generators $T_a$ for $\fg$, and $\widetilde{T}^a$ for $\widetilde{\fg}$, we have
\begin{align*}
\braket{T_a,T_b} &= 0 \\
\braket{\widetilde{T}^a, \widetilde{T}^b} &= 0 \\
\braket{T_a, \widetilde{T}^b} &= \delta_a^b.
\end{align*}
These generators have the following relations:
\begin{align*}
[T_a, T_b] &= f^c \!_{ab} T_c \\
[\widetilde{T}^a, \widetilde{T}^b] &= \widetilde{f}_c \!^{ab}  \widetilde{T}^{c} \\
[T_a, \widetilde{T}^b] &= \widetilde{f}_a \!^{bc} T_{c} - f^b \!_{ac} \widetilde{T}^{c}.
\end{align*}
The Jacobi identity on $\fd$ imposes non-trivial constraints on the structure constants for $\fg$ and $\widetilde{\fg}$. In particular, they must satisfy
\begin{align}
\label{Jacobistructures}
\widetilde{f}_{a} \!^{mc} f^{b} \!_{dm} - \widetilde{f}_{a} \!^{mb} f^{c} \!_{dm} - \widetilde{f}_{d} \!^{mc} f^{b} \!_{am} + \widetilde{f}_{d} \!^{mb} f^{c} \!_{am} - \widetilde{f}_{m} \!^{bc} f^{m} \!_{da} = 0.
\end{align}
We note that this condition is symmetric in the structure constants and their duals. In addition, the condition (\ref{Jacobistructures}) is always satisfied whenever at least one of $\fg$ or $\widetilde{\fg}$ is abelian. Thus if $\fg$ is a Lie algebra of dimension $n$, we always have at least two Manin triples $(\fg,\RR^n)$ and $(\RR^n, \fg)$.
\section{Poisson-Lie symmetry}
\label{sec:PL}
We now consider, once again, the nonlinear sigma model (\ref{NANLSM}) on a manifold $M$:
\begin{align*}
S = \int \ddd^2 \! z \, E_{ij} \pr X^i \bar{\pr} X^j.
\end{align*}
Suppose that a group $\sG$ acts freely on $M$ from the right. The infinitesimal generators of this group action are the left-invariant vector fields $\{L_a\}$:
\begin{align*}
\de X^i &= L^i_a \epsilon^a.
\end{align*}
The currents associated to this group action are the Noetherian forms: they are one-forms associated to this group action, and are given by
\begin{align}
\label{Noetherforms}
J_a &= L_a^i E_{ij} \bar{\pr} X^j \dd \bar{z} - L_a^i E_{ji} \pr X^j \dd z.
\end{align}
If we compute the variation of the action with respect to this group action for constant $\epsilon = \epsilon^a T_a$, we obtain the non-abelian version of (\ref{variationactionabelian}):
\begin{align*}
\de S &= \int_{\Sigma} \epsilon^a \left( \Lie_{v_a} g \right)_{ij} \dd X^i \wedge \star \dd X^j + \epsilon^a \left( \Lie_{v_a} B \right)_{ij} \dd X^i \wedge \dd X^j. 
\end{align*}
Here, we have defined $v_a = L^i_a \pr_i$. 

On the other hand, if we compute the variation of the action with respect to the group action with a worldsheet dependent parameter $\epsilon = \epsilon^a (z,\bar{z}) T_a$, that is the variation under a \emph{local} gauge transformation, we obtain
\begin{align*}
\de S &= \int_{\Sigma} \epsilon^a \left( \Lie_{v_a} g \right)_{ij} \dd X^i \wedge \star \dd X^j + \epsilon^a \left( \Lie_{v_a} B \right)_{ij} \dd X^i \wedge \dd X^j \\
& \qquad +  \int_{\Sigma} \dd \epsilon^a \wedge J_a.
\end{align*}
An extremal surface is a mapping $X$ such that $\de S$ vanishes. If the Lie derivatives of $g$ and $B$ vanish, then on an extremal surface we have
\begin{align*}
\int_{\Sigma} \dd \epsilon^a \wedge J_a &= \int_{\Sigma} \dd \, \left( \epsilon^a J_a \right) - \int_{\Sigma} \epsilon^a \dd J_a\\[1em]
&= -\int_{\Sigma} \epsilon^a \dd J_a.
\end{align*}
Since this must hold for arbitrary $\epsilon$, we must have that $\dd J_a=0$. That is, $J_a$ is closed on extremal surfaces. We now suppose that $J_a$ is not closed, but rather satisfies the Maurer-Cartan equation:
\begin{align}
\label{PLMC}
\dd J_a = \frac{1}{2} \tilde{f}_{a} \!^{bc} J_b \wedge J_c.
\end{align}
We assume here that the constants $\widetilde{f}_{a} \!^{bc}$ are the structure constants for some Lie algebra $\widetilde{\fg}$. If the one-forms satisfy this equation, then they have the Maurer-Cartan form, and therefore can be written in the form
\begin{align*}
J_a &= \dd \tilde{\sg} \, \tilde{\sg}^{-1},
\end{align*} 
for some map $\tilde{\sg}: \Sigma \to \widetilde{\sG}$, where $\widetilde{\fg}$ is the Lie algebra associated to $\widetilde{\sG}$. The condition (\ref{PLMC}) can be expressed in terms of $E_{ij}$ by:
\begin{align}
\label{LiePL}
(\Lie_{v_a}E)_{ij} &= E_{im} L^m_c \widetilde{f}_a \!^{bc} L^n_b E_{nj}.
\end{align}
The integrability condition on the set of first order differential equations  (\ref{LiePL}) gives a compatibility condition for the pairs of structure constants. The integrability condition is
\begin{align*}
[\Lie_{v_a} , \Lie_{v_b}] E_{ij} &= f^{c} \!_{ab} \Lie_{v_c} E_{ij},
\end{align*}
from which it follows that
\begin{align*}
\widetilde{f}_{a} \!^{mc} f^{b} \!_{dm} - \widetilde{f}_{a} \!^{mb} f^{c} \!_{dm} - \widetilde{f}_{d} \!^{mc} f^{b} \!_{am} + \widetilde{f}_{d} \!^{mb} f^{c} \!_{am} - \widetilde{f}_{m} \!^{bc} f^{m} \!_{da} = 0.
\end{align*}
This is precisely the condition (\ref{Jacobistructures}) that a Lie bialgebra must satisfy. That is, for our sigma model with group action, we have an associated Manin triple $(\fg, \widetilde{\fg})$. Due to the manifest duality of the Lie bialgebra structure, one might expect that there is a dual model which has the associated bialgebra $(\widetilde{\fg},\fg)$. This was precisely the observation made in \cite{KS95}. The dual model should satisfy 
\begin{align}
\label{dualNCCL}
(\Lie_{\tilde{v}_a}\widetilde{E})_{ij} &= \widetilde{E}_{im} \tilde{v}^m_c f_a \!^{bc} \tilde{v}^n_b \widetilde{E}_{nj}.
\end{align}
Before we go on to discuss the relation of the original model and the dual model, let us briefly recapitulate the setup for Poisson-Lie T-duality. The Lie algebra $\fg$ is the algebra associated to the group action on the sigma model. It is determined by the vector fields $\{v_a\}$ generating the group action. This algebra can be abelian, in the case of a $U(1)^N$ or $\RR^N$ action, or non-abelian, corresponding to the action of some non-abelian group $\sG$. There are Noether currents $J_a$ associated to this group action via (\ref{Noetherforms}). These Noether currents are not closed, but satisfy the Maurer-Cartan equation (\ref{PLMC}) for a \emph{different} group $\widetilde{\sG}$.  The dual Lie algebra $\widetilde{\fg}$ associated to this group can be abelian, in which case (\ref{LiePL}) tells us that the group $\sG$ acts on the original model by isometries. If the dual Lie algebra is non-abelian, the group $\sG$ acting on $M$ does not act by isometries since $\Lie_{v_a}E \not = 0$. Duality acts on the double $(\fg, \widetilde{\fg})$ by interchanging the algebras, giving $(\widetilde{\fg},\fg)$. We then have three types of dualities:
\begin{itemize}
\item {\bf Abelian doubles}: In this case, both algebras in $(\fg, \widetilde{\fg}) = \Big(\mathfrak{u}(1)^N , \mathfrak{u}(1)^N\Big)$ are abelian. The group action on the sigma model is abelian, and acts by isometries since the dual group is abelian. The dual model, determined by $(\widetilde{\fg},\fg)$, also has an isometric abelian group action on it. This situation corresponds to abelian T-duality. 
\item {\bf Semi-abelian doubles}: A semi-abelian double is one in which we have a non-abelian Lie algebra $\fg$, with an abelian dual structure $\widetilde{\fg} = \mathfrak{u}(1)^N$. Since the dual group is abelian, the group action on the original model acts by isometries. The dual model corresponds to $(\mathfrak{u}(1)^N, \fg)$. The group action is an abelian group action, but since the dual group is now non-abelian, it doesn't act by isometries. This situation corresponds to non-abelian T-duality. The original model has a non-abelian group of isometries, but those isometries are no longer present in the dual model. We therefore see that although non-abelian T-duality is not invertible from a gauging perspective, it is invertible from the perspective of Poisson-Lie T-duality since duality acts by simply swapping the algebras. 
\item {\bf Non-abelian doubles}: The final type of double is a true non-abelian double, where both $\fg$ and $\widetilde{\fg}$ are non-abelian. These examples correspond to dual pairs of models with non-abelian group actions which don't act by isometries. 
\end{itemize}
In order to describe the relations between the dual models, we find in convenient to introduce the following notation:\footnote{This particular notation comes from \cite{TvU}.}
\begin{subequations}
\begin{align}
\alpha^b_a (\tilde{g}) &= \langle \tilde{g}\, T_a\, \tilde{g}^{-1} , \widetilde{T}^b \rangle \\
\beta_{ab} (\tilde{g}) &= \langle \tilde{g}\, T_a\, \tilde{g}^{-1} , T_b \rangle \\
\mu^{ab} (g) &= \langle g\, \widetilde{T}^a \, g^{-1} , \widetilde{T}^b \rangle \\
\nu^a_b (g) &= \langle g\, \widetilde{T}^a \, g^{-1} , T_b \rangle,
\end{align}
\end{subequations}
so that 
\begin{subequations}
\begin{align}
\tilde{g} \, T_a \, \tilde{g}^{-1} &= \alpha^b_a T_b + \beta_{ab} \widetilde{T}^b \\[1em]
g \, \widetilde{T}^a \, g^{-1} &= \mu^{ab} T_b + \nu^a_{b} \widetilde{T}^b.
\end{align}
\end{subequations}
Consistency requires that 
\begin{subequations}
\begin{align}
\alpha(\tilde{g}^{-1}) &= \alpha^{-1} (\tilde{g}) \\
\beta(\tilde{g}^{-1}) &= \beta^T(\tilde{g}) \\
\mu(g^{-1}) &= \mu^T(g) \\
\nu(g^{-1}) &= \nu^{-1}(g).
\end{align}
\end{subequations}
The original and the dual model are then given by 
\begin{subequations}
\begin{align}
E_{ab}(g) &= \Big( \left(E^0_{ab} \right)^{-1} + \mu^{ac} \nu^b_c \Big)^{-1} \\[1em]
\widetilde{E}^{ab} (\tilde{g}) &= \Big( E^0_{ab}  + \beta_{ac} \alpha^c_b \Big)^{-1},
\end{align}
\end{subequations}
where $E^0_{ab}$ is some constant matrix. In the case of non-abelian T-duality, we have
\begin{subequations}
\begin{align}
\alpha^b_a &= \delta^a_b \\
\beta_{ab} &= f^c_{ab} \chi_c,
\end{align}
\end{subequations}
and so the dual space is given by 
\begin{align}
\widetilde{E}^{ab} &= (E^0_{ab} + f^c_{ab}\chi_c)^{-1},
\end{align}
in accord with Section \ref{subsec:NATDBuscher}. Whilst Poisson-Lie T-duality has been formulated as a canonical transformation on the phase space of the classical theory in \cite{S97,S98}, global issues prevent a full path-integral derivation of the duality (cf. \cite{AKT,TvU}).

\section{Examples}
\label{sec:PLexamples}
\subsubsection{The non-abelian T-dual of $S^3$}
\label{subsec:PLS3}
For the first example, let us consider again the non-abelian T-dual of $S^3 \simeq SU(2)$, to see how it fits into the framework of Poisson-Lie T-duality. Since the pair $(\mathcal{G}, \mathcal{G}')$ consists of a non-abelian Lie algebra $\mathcal{G} = \mathfrak{su}(2)$, and an abelian Lie algebra $\mathcal{G}' = \RR^3$, it is a semi-abelian double as described in Section \ref{sec:DD}. Let us now describe the Drinfeld double $D$ for this case.

Recall the following construction of an (outer) semi-direct product. Let $\sN$ and $\sH$ be two (possibly unrelated) groups, and consider a group homomorphism $\varphi : \sH \to \textrm{Aut}(\sN)$. The outer semi-direct product $\sH \ltimes_{\varphi} \sN$ as a set is simply the cartesian product $\sH \times \sN$. We equip this set with a group structure
\begin{align}
(\sh_1,\sn_1) \ast (\sh_2,\sn_2) = (\sh_1 \sh_2, \sn_1 \varphi(\sh_1)(\sn_2)).
\end{align}
Applying this construction to $S^3$, we take $\sH = SU(2)$, and $\sN$ to be $\fg^{\ast} \simeq \RR^3$, thought of as an abelian group. For $\varphi$, we take the coadjoint action of $SU(2)$. That is, we have a map $\varphi: SU(2) \to \textrm{Aut}(\fg^{\ast})$, with $\varphi(g) = \varphi_g$ defined by
\begin{align}
\varphi_g (\phi) (X) := \phi(\textrm{Ad}(g^{-1}) X),
\end{align}
for $\phi \in \fg^{\ast}$, and $X \in \fg$. The Drinfeld double we then define as $D = SU(2) \ltimes_{\varphi} \fg^{\ast}$, which is also the cotangent bundle of $SU(2)$. The Lie algebra is \begin{align*}
\mathcal{D} = \fg + \fg^{\ast} = \mathfrak{su}(2) + \RR^3,
\end{align*}
and so the associated pair is $(\mathfrak{su}(2),\RR^3)$, which is a semi-abelian double as described in Section \ref{sec:DD}. 

To get a geometric realisation of this double, we equip $S^3$ with the round metric and no $B$-field. We can identify $\fg$ with the right-invariant vector fields of $SU(2)$. The dual Lie algebra is abelian, so we should have
\begin{align*}
(\Lie_{v_a} E)_{ij} &= 0.
\end{align*}
This is immediate, however, since the round metric on $S^3$ is the bi-invariant metric on $SU(2)$, and the right-invariant vector fields are therefore isometries. The dual metric and $B$-field can be obtained from the non-abelian T-duality procedure, as we have done in Section \ref{NATDS3}. The result is
\begin{align}
\label{dualPLSU2}
\widetilde{E} &= \frac{1}{1+x^2+y^2+z^2} \left( 
\begin{matrix}
1+x^2 & -z + xy & y + xz \\
z+xy & 1+y^2 & x-yz \\
-y+xz & x+yz & 1+z^2
\end{matrix}
 \right).
\end{align}
We can obtain the metric and $B$-field (\ref{NATDS3undeformed}) by extracting the symmetric and antisymmetric parts. Since $\fg^{\ast}$ is abelian, the dual vector fields form a commutative Lie algebra. It is a straightforward, albeit tedious exercise to show that the vector fields $\tilde{v}_a = \{\pr_x,\pr_y,\pr_z\}$ and $\widetilde{E}$ given by (\ref{dualPLSU2}) satisfy the dual non-commutative conservation laws (\ref{dualNCCL}):
\begin{align*}
(\Lie_{\tilde{v}_a}\widetilde{E})_{ij} &= \widetilde{E}_{im} \tilde{v}^m_c f_a \!^{bc} \tilde{v}^n_b \widetilde{E}_{nj},
\end{align*}
with the $\mathfrak{su}(2)$ structure constants.

\subsubsection{The Borelian double}
This simple two-dimensional example comes from \cite{Kli95}.\footnote{Although note that we are using a slightly different convention.} Take as the Drinfeld double the group $GL(2;\RR)$, together with the following basis for its Lie algebra
\begin{align*}
T_1 &= \left( \begin{matrix} 1 & 0 \\ 0 & 0  \end{matrix} \right), \qquad T_2 = \left( \begin{matrix} 0 & 1 \\ 0 & 0  \end{matrix} \right) \\
\widetilde{T}^1 &=  \left( \begin{matrix} 0 & 0 \\ 0 & 1  \end{matrix} \right), \qquad  \widetilde{T}^2 = \left( \begin{matrix} 0 & 0 \\ -1 & 0  \end{matrix} \right).
\end{align*}
The group, $B_2$, whose Lie algebra is spanned by $\{T_1, T_2\}$, has an explicit parametrisation:
\begin{align*}
g &= \left( 
\begin{matrix}
e^{\chi} & \theta \\
0 & 1
\end{matrix}
\right),
\end{align*}
and the group $\widetilde{G}_2$, whose Lie algebra is spanned by $\{\widetilde{T}^1,\widetilde{T}^2\}$, has an explicit parametrisation:
\begin{align*}
\tilde{g} &= \left( 
\begin{matrix}
1 & 0 \\
-\tilde{\theta} & e^{\tilde{\chi}}
\end{matrix}
\right).
\end{align*}
A quick calculation gives us that
\begin{align*}
\mu &= \left( \begin{matrix} 0 & \theta \\ -\theta e^{-\chi} & \theta^2 e^{-\chi} \end{matrix} \right) \\
\nu &= \left( \begin{matrix} 1 & 0 \\ \theta e^{-\chi} & e^{-\chi} \end{matrix} \right),
\end{align*}
and 
\begin{align*}
\alpha &= \left( \begin{matrix} 1 & 0 \\ \tilde{\theta} e^{-\tilde{\chi}} &  e^{-\tilde{\chi}} \end{matrix} \right) \\
\beta &= \left( \begin{matrix} 0 & \tilde{\theta} \\ -\tilde{\theta} e^{-\tilde{\chi}} & \tilde{\theta}^2 e^{-\tilde{\chi}}  \end{matrix} \right).
\end{align*}
By choosing $E^0_{ab}$ to be the $2 \times 2$ identity matrix, we obtain the original model
\begin{align}
\label{eq:2by2PL}
E_{ab} &= \frac{1}{2\theta^2 e^{-2\chi} + \theta^2 e^{-\chi} + 1} \left( 
\begin{matrix}
1+\theta^2 e^{-2\chi} & -\theta e^{-\chi} \\
-\theta e^{-\chi} \left( \theta^2 e^{-\chi} - 1 \right) & 1+\theta^2 e^{-\chi}
\end{matrix} \right).
\end{align}
This model is self-dual, and the dual model can be obtained from \eqref{eq:2by2PL} by replacing the coordinates $\{\chi,\theta\}$ with their dual versions, $\{\tilde{\chi}, \tilde{\theta}\}$. Note that this expression for the model and its dual differs from that in \cite{Kli95} by a similarity transformation, as explained in \cite{TvU}.
\subsubsection{A three-dimensional example in various limits }
This final set of examples is from \cite{S98b}, where the author considers Poisson-Lie T-duality for the Drinfeld double whose algebras are given by $\mathfrak{su}(2)$ and $\mathfrak{e}_3$. Explicitly, we have generators $\{T_a\}$ for $\mathfrak{su}(2)$, and generators $\{\widetilde{T}^a\}$ for $\mathfrak{e}_3$, where the commutation relations (for $a,b,c = 1,2,3$ and $i,j = 1,2$) are given by:\footnote{Note that the notation for this example differs from the notation in the rest of this chapter, in order to stay consistent with the results of \cite{S98b}.}
\begin{align*}
[T_a,T_b] &= i \epsilon_{abc} T_c, \\
[\widetilde{T^3},\widetilde{T}^i] &= \widetilde{T^i}, \qquad [\widetilde{T}^i, \widetilde{T}^j] = 0.
\end{align*}
The `mixed' relations are given by
\begin{align*}
[T_i, \widetilde{T}^j] &= i \epsilon_{ij} \widetilde{T}^3 - \delta_{ij} T_3, \qquad [T_3, \widetilde{T}^i] = i \epsilon_{ij} \widetilde{T}^j\\
[\widetilde{T}^3,T_i] &= i \epsilon_{ij} \widetilde{T}^j - T_i.
\end{align*}
By choosing $E^0 = \textrm{diag}(\lambda_1, \lambda_2, \lambda_3)$, the author of \cite{S98b} finds the dual models given by
\begin{align*}
\dd s^2 &= \frac{1}{V} \left( A_a A_b + \frac{\lambda_1 \lambda_2 \lambda_3}{\lambda_a} \delta_{ab} \right) L_a L_b \\[1em]
B &= \frac{1}{V} \epsilon_{abc} \lambda_c A_c L_a \wedge L_b
\end{align*}
and
\begin{align*}
\widetilde{\dd s}^2 &= \frac{1}{\widetilde{V}} \left( \widetilde{A}_a \widetilde{A}_b + \frac{\lambda_a}{\lambda_1 \lambda_2 \lambda_3} \delta_{ab} \right) \widetilde{L}_a \widetilde{L}_b \\[1em]
B &= \frac{1}{\widetilde{V}} \epsilon_{abc} \frac{1}{\lambda_c} \widetilde{A}_c \widetilde{L}_a \wedge \widetilde{L}_b,
\end{align*}
where
\begin{align*}
\vec{A} &=  \left( \cos \psi  \sin \theta, \sin \psi \sin \theta, \cos \theta - 1 \right),\\
\vec{\widetilde{A}} &= \left( y_1 e^{-\chi} , y_2 e^{-\chi}, \sinh \chi e^{-\chi} - \tfrac{1}{2}(y_1^2 + y_2^2)e^{-2\chi} \right) 
\end{align*}
and
\begin{align*}
V &= \lambda_1 \lambda_2 \lambda_3 + \lambda_a A_a^2 \\
\widetilde{V} &= \frac{1}{\lambda_1 \lambda_2 \lambda_3} + \frac{\widetilde{A}_a^2}{\lambda_a}.
\end{align*}
Here, $L_a$ and $\widetilde{L}_a$ are the left-invariant Maurer-Cartan forms for the groups associated to $\mathfrak{su}(2)$ and $\mathfrak{e}_3$ respectively, and are given by
\begin{align*}
L_1 &= \cos \psi \sin\theta \dd \phi - \sin \psi \dd \theta \\
L_2 &= \sin \psi \sin \theta \dd 	\phi + \cos \psi \dd \theta \\
L_3 &= \dd \psi + \cos \theta \dd \phi,
\end{align*}
and
\begin{align*}
\widetilde{L}_1 &= e^{-\chi} \dd y_1 \\
\widetilde{L}_2 &= e^{-\chi} \dd y_2 \\
\widetilde{L}_3 &= \dd \chi.
\end{align*}
These dual models have three interesting limits which have been considered in \cite{S98b}. The first is to take a limit in which the untilded metric becomes a metric on the deformed $S^3$, whilst rescaling the dual coordinates in such a way that the dual group becomes abelian. This then reproduces the results of Section \ref{NATDS3} and Section \ref{subsec:PLS3}, which is simply the non-abelian T-dual of $S^3$ with respect to an $SU(2)$ isometry. 
The second limit considered is to contract $SU(2)$ into an abelian group, and the result is the non-abelian T-dual of the group associated to $\mathfrak{e}_3$. 
The final and most interesting limit is to take one of the $\lambda_a$ to infinity, whilst keeping the other two finite. The resulting models have two-dimensional target spaces, and the untilded metric becomes a deformed metric on $S^2$:
\begin{align*}
\dd s^2 &= \frac{1}{1+a^2 \sin^4 \tfrac{\theta}{2}} \left( \dd \theta^2 + \sin^2 \theta \dd \phi^2 \right), 
\end{align*}
whilst the tilded metric becomes a curved metric on a non-compact space:
\begin{align*}
\widetilde{\dd s}^2 &= \frac{1}{2r(1+az)}\left( \dd z^2 + \left(\dd r + \frac{z-\tfrac{ar}{2}}{1+az} \dd z \right)^2 \right). 
\end{align*}
In the $a \to 0$ limit this becomes the non-abelian T-dual of $S^2$ with respect to the $SU(2)$ (discussed in Section \ref{NATDS2}).

%% file: chapter5.tex
\chapter{Non-isometric T-duality}
\label{chptr:Ch5} 
Non-isometric T-duality is an attempt to generalise non-abelian T-duality to a set of vector fields generating infinitesimal variations which do not necessarily leave the metric and $B$-field invariant, but which nevertheless leave a suitably gauged action invariant. 

In Section \ref{sec:NITDomega} we describe the initial proposal by Chatzistavrakidis, Deser, and Jonke (CDJ) for non-isometric T-duality \cite{CDJ15,C16}. We describe the setup, examples of this putative duality, and finally the proof that it is actually just non-abelian T-duality in disguise \cite{BBKW}.

In Section \ref{sec:NITDomegaphi} we discuss extensions of the original proposal, and its relation to Poisson-Lie T-duality.

\section{$\omega$-deformed gauging}
\label{sec:NITDomega}
\subsection{The setup}

Recall from Chapter \ref{chptr:Ch3} that our prescription for non-abelian T-duality involved a metric $g$ and a $B$-field $B$, both invariant under the action of a non-abelian Lie algebra of vector fields $\{v_a\}$. In this chapter, we will restrict our attention to the case that the flux $H = \dd B$ is an exact form. The action is the standard non-linear sigma model action:
\begin{align*}
S = \int_{\Sigma} g_{ij} \dd X^i \wedge \star \dd X^j + B_{ij} \dd X^i \wedge \dd X^j, 
\end{align*}
and the minimally-coupled action is
\begin{align*}
S_{MC} = \int_{\Sigma} g_{ij} \scD X^i \wedge \star \scD X^j + B_{ij} \scD X^i \wedge \scD X^j,
\end{align*}
where $\scD X^i = \dd X^i - v^i_a \cA^a$. The minimally-coupled action is invariant under the (local) gauge transformations:
\begin{subequations}
\begin{align*}
\de X^{i} &= v_a^i \epsilon^a \\[1em]
\de \cA^a &= \dd \epsilon^a + f^a_{\ bc} \cA^b \epsilon^c.
\end{align*}
\end{subequations}
The modification proposed in \cite{CDJ15,C16} is a twofold generalisation of this basic setup. The first is the promotion of the structure constants $f^c_{\ ab}$ to structure functions $f^c_{\ ab} (X)$. The second is to allow an extended notion of gauging, by considering the following $\omega$-modified gauge transformations:
\begin{subequations}
\label{omegagaugetransforms}
\begin{align}
\de X^{i} &= v_a^i \epsilon^a \\[1em]
\label{omegaAtransformation}
\de \cA^a &= \dd \epsilon^a + f^a_{\ bc} \cA^b \epsilon^c + \omega^a_{bi} \epsilon^b \scD X^i
\end{align} 
\end{subequations}
The quantity $\omega^a_{bi}$ is an $X$-dependent matrix of one-forms. We will discuss the geometric interpretation of this quantity in Section \ref{subsec:NITDgeometric}. Let us examine now the consequences of allowing such modified gauge transformations. The gauge covariant derivative transforms homogeneously under these (local) gauge transformations:\footnote{This calculation uses the identity $v^i_a (\pr_i v^j_b) - v^i_b (\pr_i v^j_a) = f^c_{\ ab} v^j_c$, which is the coordinate expression of $[v_a,v_b] = f^c_{\ ab} v_c$.}
\begin{align*}
\de \scD X^i &= \left[ \epsilon^a (\pr_j v^i_a) - v_a^i \omega^a_{bj} \epsilon^b \right] \scD X^j \\[1em]
&=: M^i_j \scD X^j.
\end{align*}
With this convenient transformation, we now compute the (local) variation of the minimially-coupled action:
\begin{align*}
\de S_{MC} &= \int_{\Sigma} \left[ v^k_a (\pr_k g_{ij}) \epsilon^a + g_{kj} M_i^k + g_{ik} M^k_j \right] \scD X^i \wedge \star \scD X^j  \\
&\quad +  \left[ v^k_a (\pr_k B_{ij}) \epsilon^a + B_{kj} M_i^k + B_{ik} M^k_j \right] \scD X^i \wedge \scD X^j \\[1em]
&= \int_{\Sigma} \epsilon^a \left[ (\Lie_{v_a} g)_{ij} - (g_{kj} v^k_b \omega^b_{ai} + g_{ik} v^k_b \omega^b_{aj}) \right] \scD X^i \wedge \star \scD X^j \\
&\quad + \epsilon^a \left[ (\Lie_{v_a} B)_{ij} - (B_{kj} v^k_b \omega^b_{ai} + B_{ik} v^k_b \omega^b_{aj}) \right] \scD X^i \wedge  \scD X^j.
\end{align*}
It follows that the minimally coupled action is invariant under the $\omega$-modified gauge transformations (\ref{omegagaugetransforms}), provided 
\begin{subequations}
\label{Lieomegagross}
\begin{align}
(\Lie_{v_a} g)_{ij} &= g_{kj} v^k_b \omega^b_{ai} + g_{ik} v^k_b \omega^b_{aj} \\[1em]
(\Lie_{v_a} B)_{ij} &= B_{kj} v^k_b \omega^b_{ai} + B_{ik} v^k_b \omega^b_{aj}.
\end{align}
\end{subequations}
Written in more succinct notation, the conditions (\ref{Lieomegagross}) read:
\begin{subequations}
\label{NITDLiegB} 
\begin{align}
\Lie_{v_a} g &= \omega^b_a \vee \iota_{v_b} g \\[1em]
\Lie_{v_a} B &= \omega^b_a \wedge \iota_{v_b} B,
\end{align}
\end{subequations}
where $\iota_X T = T(X,\cdot)$, or in components $(\iota_X T)_{j} = X^k T_{kj}$.\footnote{Note that in this Chapter we are using the following normalisations for the symmetric and antisymmetric products:
\begin{align*}
(A \vee B)_{ij} &= A_i B_j + A_j B_i \\
(A \wedge B)_{ij} &= A_i B_j - A_j B_i.
\end{align*} }
This calculation says that we can write down a gauged action which is invariant under local gauge transformations, provided the metric and the $B$-field satisfy conditions (\ref{NITDLiegB}). In particular, since $\omega$ is not required to be zero, this says that the vector fields we are gauging do not need to be isometries for the metric and $B$-field! Since the existence of isometries is a `rare' property amongst the space of all (pseudo)-Riemannian metrics,\footnote{A generic Riemannian metric will have no continuous symmetries.} being only able to do T-duality on metrics that had isometries was quite a restrictive requirement. This modification is an attempt to expand the class of backgrounds on which we may perform a T-duality. 

In order to attempt T-duality for this modified gauging, we need to add a term to the minimally coupled action enforcing conditions on the gauge fields $\cA^a$ which give us back the original model. In the abelian/non-abelian case, this was simply the field strength of the gauge fields. The proposal of CDJ was to use the $\omega$-modified field strength:
\begin{align}
\label{fieldstrengthomega}
\cF_{\omega}^a = \dd A^a + \frac{1}{2}f^a_{\ bc} \cA^b \wedge \cA^c - \omega ^a_{bi} \cA^b \wedge \scD X^i.
\end{align}
The total gauged action was then
\begin{align}
\label{NITDgaugedaction}
S = \int_{\Sigma} g_{ij} \scD X^i \wedge \star \scD X^j + B_{ij} \scD X^i \wedge \scD X^j + \int_{\Sigma} \eta_a \cF^a_{\omega}.
\end{align}
We shall postpone to Section \ref{proofNITDNATD} the discussion of the gauge invariance of the final term. For now, let us continue with the proposed procedure of non-isometric T-duality. As usual, from the gauged action we wish to be able to obtain the original ungauged action, as well as the dual model. The equation of motion obtained by integrating out the Lagrange multiplier is:
\begin{align*}
\cF^a_{\omega} = \dd \cA^a + \frac{1}{2}f^c_{\ bc} \cA^b \wedge \cA^c - \omega^a_{bi} \cA^b \wedge \scD X^i = 0.
\end{align*}
If this were simply the standard field strength for a non-abelian gauge field, we could now happily conclude that since $\cF^a = 0$, the gauge fields must be pure gauge $\cA = \sg^{-1} \dd \sg$. Then, by performing a suitable gauge transformation, we could set $\cA$ to any convenient value (choosing the gauge transformation which sets $\cA=0$ is a convenient choice, allowing us to recover the original model). Unfortunately, with the modified gauge fields transforming in a non-standard way, it's not clear that every solution, $\cA$, to the constraint $\cF = 0$ will be in the gauge orbit of $\cA=0$. That is, if the constraint $\cF = 0$ is applied, it is not clear that there always exists a gauge transformation, generated by (\ref{omegaAtransformation}), which we can use to set $\cA$ to zero. 

To obtain a dual model, we can apply the same strategy as in the case of (non)-abelian T-duality. Namely, integrate out the gauge fields $\cA$ from the action. Since the conclusion of this chapter is that this modification cannot generate new examples of dual models, we won't include the details of this calculation here, but will instead refer the interested reader to the original articles \cite{CDJ15,C16}.

\subsection{Examples: non-isometric T-duality}
\label{subsec:NITDexamples}
We include in this section a few examples of the putative duality. The first two examples concern the Heisenberg nilmanifold, and are included in the original papers \cite{C16,CDJ15}. 

\subsubsection{Heisenberg nilmanifold Part 1}
The Heisenberg nilmanifold, also called the twisted torus or the $f$-flux background, is a frequent example in this thesis. It arose in Section \ref{T3Buscher} as the dual space to the flat 3-torus with $H$-flux under a single abelian T-duality. It also popped up in Section \ref{NATDcommentsTT} when we discussed examples of non-abelian T-duality. It also appears in this section as an example of non-isometric T-duality, although as we shall see in Section \ref{proofNITDNATD}, this is simply just the normal non-abelian duality in disguise. 

We begin with the metric as obtained in Section \ref{T3Buscher}. It is
\begin{align}
\label{NilNITD}
\dd s^2 &= \dd x^2 + \dd y^2 + (\dd z^2 - x \dd y)^2,
\end{align}
with a vanishing $B$ field. Since we are thinking of this as the original space on which we want to perform a duality, we have dropped the hat on the $z$ coordinate. We now wish to perform a non-isometric T-duality on this space, which amounts to choosing a set of vector fields and then trying to find a set of $\omega^b_a $ so that the symmetry conditions (\ref{NITDLiegB}) are satisfied. Following \cite{C16}, we choose an abelian set of vector fields $\{\pr_x,\pr_z\}$. These commute, but are not both Killing vectors for the metric:
\begin{align*}
\Lie_{\pr_x} g &= 2x \dd y^2 - 2\dd y \dd z\\[1em]
\Lie_{\pr_z} g &= 0.
\end{align*}
A simple calculation is enough to verify that taking \begin{align*}
\omega^3_1 = -\dd y,
\end{align*}
with all other terms vanishing, solves the symmetry conditions (\ref{NITDLiegB}). In addition, this choice of $\omega^b_a$ also solves the additional constraints coming from the requirement of gauge invariance of the action (see \cite{CDJ15,C16} for details). We can therefore perform a non-isometric T-duality with respect to these vector fields. The result is a dual model which has the metric
\begin{align*}
\widehat{\dd s}^2 &= (\dd \hat{x} - \hat{z} \dd y)^2 + \dd y^2 + \dd \hat{z}^2,
\end{align*}
with vanishing $B$-field. This is simply the metric of another $f$-flux background. In particular, we can obtain it from our original metric by following a chain of abelian T-dualites (with a $B$-field transformation in the middle)\cite{C16}:
\begin{equation*}
\begin{tikzpicture}[baseline=(current  bounding  box.center)]
\node (T) {$T_{xyz}$};
\node (f2) [right=3cm of T] {$f^x \!_{yz}$};
\node (f1) [below=3cm of T] {$f_{xy}\!^{z}$};

\draw[<->] (T) to node[swap] {$\pr_x$} (f1);
\draw[<->] (f2) to node[swap] {$\pr_z$} (T);
\draw[<-,dashed] (f2) to node {$(\pr_x, \pr_z)_{NITD}$} (f1);
\draw[red,->] (T.135) arc (0:290:3mm) node[pos=0.5,swap]{$\delta B$};
\end{tikzpicture}
\end{equation*}
Note that it is not always possible to find $\omega^b_a$ for a given set of vector fields, even when the vector fields commute. An example is given by the pair $\{ \pr_x, \pr_y \}$, where it is straightforward to verify that there is no $\omega^b_a$ which will solve (\ref{NITDLiegB}).
\subsubsection{Heisenberg nilmanifold Part 2}
For this example, we perform a non-isometric T-duality on the nilmanifold again, but this time with the vector fields $\{\pr_x , \pr_y+x \pr_z  ,\pr_z \}$.\footnote{Note that this example is essentially just a relabelling of the example given in \cite{CDJ15}, corresponding to $\{ x^2,x^3\} \longleftrightarrow \{z,y\}$.} The metric is the same as the previous section, (\ref{NilNITD}), and the Lie derivatives are
\begin{align*}
\Lie_{\pr_x} g &= 2x \dd y^2 - 2\dd y \dd z\\[1em]
\Lie_{\pr_y+x\pr_z} g &=  -2x \dd x \dd y +2 \dd x \dd z \\[1em]
\Lie_{\pr_z} g &= 0.
\end{align*}
We have already performed a non-abelian T-duality for this background in Section \ref{NATDHeis}, but here we are performing a non-abelian T-duality with a different set of vector fields which are not isometries. It is not difficult to verify that by taking
\begin{align*}
\omega^3_1 &= - \dd y \\[1em]
\omega^3_2 &=  \dd x,
\end{align*}
we can solve the required constraints. The dual space is a $Q$-flux background:
\begin{align*}
\widehat{\dd s}^2 &= \dd \hat{z}^2 + \frac{1}{1+\hat{z}^2} \left( \dd \hat{x}^2 + \dd \hat{y}^2 \right) \\[1em]
\widehat{B} &= \frac{2 \hat{z}}{1+\hat{z}^2} \dd \hat{x} \wedge \dd \hat{y}.
\end{align*}
This is exactly the same space we obtained when we performed a non-abelian T-duality! As we shall see, this is not a coincidence. 
\subsection{Proof that NITD = NATD}
\label{proofNITDNATD}
The gauge invariance conditions (\ref{NITDLiegB}) are more general than the isometry conditions normally required for T-duality, so non-isometric T-duality appears to be applicable in situations where we are unable to perform non-abelian T-duality. In this section, we shall see that the non-isometric T-duality as just described is actually just the standard non-abelian T-duality in disguise. \\
Recall the definition of the $\omega$-modified curvature:
\begin{align*}
\cF^a_{\omega} = \dd \cA^a + \frac{1}{2}f^a_{\ bc} \cA^b \wedge \cA^c - \omega^a_{bi} \cA^b \wedge \scD X^i.
\end{align*}
This field strength appears in the gauged action (\ref{NITDgaugedaction}) through the term $\int_{\Sigma} \eta_a \cF^a_{\omega}$. In order for this term to be gauge invariant, we require $\cF^a_{\omega}$ to transform homogeneously:
\begin{align*}
\de \cF^a_{\omega} = \alpha \cF^a_{\omega}.
\end{align*}
Computing the variation of $\cF^a_{\omega}$ using the transformation properties of the gauge fields (\ref{omegagaugetransforms}), we get
\begin{align}
\label{eomegacurvaturevariation}
\de \cF^a_{\omega} = \left( f^a_{\ bc} - \omega^a_{ci} v^i_b \right)\epsilon^c  \cF^b_{\omega} + \frac{1}{2} R^a_{bij}\epsilon^b \scD X^i \wedge \scD X^j + \cS^a_{bci} \epsilon^c \scD X^i \wedge \cA^b,
\end{align}
where 
\begin{subequations}
\begin{align}
\label{omegacurvature}
R^a_b &= \dd \omega^a_b + \omega^a_c \wedge \omega^c_b \\[1em]
\label{omegaD}
\cS^a_{bc} &= \dd f^a_{\ bc} + f^e_{\ bc}\, \omega^a_e + 2 f^a_{\ d [b} \, \omega^d_{c]} + 2 \iota_{v_d} \omega^a_{[c} \omega^d_{b]} + 2 \Lie_{v_{[c}} \omega^a_{b]} + \iota_{v_{[b}} R^a_{c]}.
\end{align}
\end{subequations}
For the field strength to transform homogeneously for any choice of $\cA$, we therefore require the vanishing of $R^a_b$ and $\cS^a_{bc}$. One might be tempted to think that there could be a non-zero combination of $R^a_b$ and $\cS^a_{bc}$ which nevertheless allow $\cF^a_{\omega}$ to transform homogeneously. We see that that can't be the case, by noting that the requirement that $R^a_b$ vanish can be seen quite simply by computing the variation of the term $\eta_a \cF^a_{\omega}$, and expanding in powers of $\cA$:
\begin{align}
\notag
\de (\eta_a \cF^a_{\omega}) &= (\de \eta_a) \cF^a_{\omega} + \eta_a (\de \cF^a_{\omega}) \\[1em]
\label{powersofA}
&= \eta_a R^a_b \epsilon^b + \mathcal{O}(\cA) + \mathcal{O}(\cA^2).
\end{align}
All three terms must vanish separately, so we require $R^a_b=0$. 
The condition that $\cS^a_{bc}$ must also vanish can be derived independently by considering the closure of the gauge algebroid, as described in Section \ref{subsubsec:NITDclosure}.

Let us now look at the consequences of these constraints. The constraint $R^a_b=0$ is easy to solve, since it has the Maurer-Cartan form. This tells us that there exists a matrix $K^a_{\ b}$ such that 
\begin{align}
\label{omegamaurercartan}
\omega^a_{\ b} = (K^{-1})^a_{\ c} \dd K^c_{\ b}.
\end{align}
Rather than substituting this blindly into the various quantities we have and doing calculations, let us take this $K$ and judiciously make the following field redefinitions:
\begin{subequations}
\label{fieldredefinitions}
\begin{align}
\tilde{\cA}^a &= K^a_{\ b} \cA^b \\[1em]
\label{vectorfieldredefinition}
\tilde{v}^i_a &= v^i_b (K^{-1})^b_{\ a} \\[1em]
\tilde{\eta}_a &= \eta_b (K^{-1})^b_{\ a}.
\end{align}
\end{subequations}
Note that $\scD X^i$ is invariant under this field redefinition. The action now has the form
\begin{align}
\label{NITDtransaction}
S &= \int_{\Sigma} g_{ij} DX^i \wedge \star \scD X^j +B_{ij} \scD X^i \wedge  \scD X^j \\[1em]
& \qquad+ \int_{\Sigma} \tilde{\eta}_a \left( \dd \tilde{\cA}^a + \frac{1}{2} \tilde{f}^a_{\ bc} \tilde{\cA}^b \wedge \tilde{\cA}^c \right), 
\end{align}
where 
\begin{align*}
\tilde{f}^a_{\ bc} = K^a_{\ d} ((K^{-1})^e_{\ b} (K^{-1})^f_{\ c} f^{d}_{\ ef} + (K^{-1})^e_{\ b} v^i_{e} \pr_i (K^{-1})^d_{\ c} - (K^{-1})^e_{\ c} v^i_e \pr_i (K^{-1})^d_{\ b}.
\end{align*}
Note that the vanishing of $\mathcal{S}^a_{bc}$ now reduces to the simple condition $\dd \tilde{f}^a_{bc} = 0$.
We see that the action (\ref{NITDtransaction}) has the standard form of the gauged action from non-abelian T-duality. Remarkably, the gauge invariance conditions (\ref{NITDLiegB}) become the isometry conditions:
\begin{align*}
K^a_{\ b} \Lie_{\tilde{v}_a} g &= K^a_{\ b} \Lie_{v_c (K^{-1})^c_{\ a}}g = \Lie_{v_b} g - (K^{-1})^a_{\ c} \dd \, K^c_{\ b} \vee \iota_{v_a} g = \Lie_{v_b} - \omega^a_{\ b} \vee \iota_{v_a} g = 0 \\[1em]
K^a_{\ b} \Lie_{\tilde{v}_a} B &= K^a_{\ b} \Lie_{v_c (K^{-1})^c_{\ a}}B = \Lie_{v_b} B - (K^{-1})^a_{\ c} \dd \, K^c_{\ b} \wedge \iota_{v_a} B = \Lie_{v_b} - \omega^a_{\ b} \wedge \iota_{v_a} B = 0.
\end{align*}
That is
\begin{align*}
\Lie_{\tilde{v}_a} g &= 0 \\[1em]
\Lie_{\tilde{v}_a} B &= 0.
\end{align*}
It follows that if there is a non-isometric gauging of a non-linear sigma model with which we are able to perform a T-duality transformation, then there is an equivalent description of the model as an isometrically gauged model. Arguing in the other direction, suppose we start with a non-linear sigma model with isometries and then construct the gauged action. Taking any invertible matrix $K$, we can make the field redefinitions
\begin{align*}
\tilde{\cA}^a &= (K^{-1})^a_{\ b} \cA^b \\[1em]
\tilde{v}^i_a &= v^i_b K^b_{\ a} \\[1em]
\tilde{\eta}_a &= \eta_b K^b_{\ a}.
\end{align*}
The corresponding action will then be invariant under the $\omega$-modified gauge transformations (\ref{NITDLiegB}), with $\omega^a_b = K^a_{\ c} (\dd K^{-1})^c_{\ b}$. It follows that the proposed non-isometric T-duality is simply a field redefinition of the usual non-abelian T-duality.

\subsubsection{A comment on the field strength}
The core of the argument showing that the non-isometric T-duality procedure is equivalent to the standard non-abelian T-duality procedure relies on the specific form (\ref{fieldstrengthomega}) of the $\omega$-modified field strength, and in particular, that $\omega$ satisfied $R^a_b = 0$. With that in mind, one might try to argue that a different field strength could give different constraints which don't reduce to the non-abelian case. 

In \cite{KS15} such a generalised field strength was considered. In order to allow for a curved gauging (i.e. $R^a_b = 0$), the field strength had to have the form 
\begin{align*}
G^a = \cF^a_{\omega} + \frac{1}{2}B^a_{ij} \scD X^i \wedge \scD X^j,
\end{align*}
where $B$ is an $E$-valued two-form on $M$.\footnote{See Section \ref{subsec:NITDgeometric}.} Of course, for $B=0$ we reduce to the $\cF^a_{\omega}$ already considered. If $B\not=0$ however, our analysis in Section \ref{proofNITDNATD} no longer holds. Under the $\omega$-modified gauge transformations (\ref{omegagaugetransforms}), the field strength $G^a$ transforms homogeneously provided
\begin{align*}
R^a_b + \Lie_{v_b}B^a - \omega^c_b \wedge \iota_{v_c}B^a + \iota_{v_b}(\omega^a_c)B^c - \mathcal{T}^a_{bc} B^c = 0.
\end{align*}
Although this is an interesting observation, it can't be used to perform non-isometric T-duality. The inclusion of the field strength into the gauged action in the Buscher procedure is done so that by integrating out the Lagrange multipliers, we get an equation of motion forcing the field strength to vanish. We then conclude that there is a gauge transformation in which the gauge fields vanish. If the field strength $G^a$ vanishes, however, there is no gauge transformation which will set the gauge fields to zero. This is because $\cA^a=0$ is not a solution to $G^a = 0$. Thus if we use the gauge covariant field strength $G^a$ in the gauged action, the gauged model won't be equivalent to the original model. 

\subsection{Geometric interpretation}
\label{subsec:NITDgeometric}

The $\omega$-deformed gauging has a nice geometric description in terms of Lie algebroid gauge theories.\footnote{See Appendix \ref{app:LAGT} for a brief introduction to Lie algebroid gauge theories.} 
The geometric description of the gauging is as follows. We have a Lie algebroid $\pi:E\to M$, together with a choice of local frame $\{e_a\}$. The image of this frame under the anchor map $\rho(e_a) = v_a= v^i_a \pr_i$ are the vector fields with which we will gauge, and they have associated to them the structure functions
\begin{align}
\label{structurefunctionsdefn}
[e_a,e_b] = f^c_{\ ab} e_c.
\end{align}
The gauge fields, $\cA$, are one-forms on the worldsheet with values in the pullback bundle $X^{\ast}E$. The fields $X$ are simply scalars on the worldsheet, and the Lagrange multipliers are scalars on the worldsheet with values in the dual pullback bundle $X^{\ast}E^{\ast}$. 
Under a change of frame $e_a \to e'_a = e_b \Lambda^b_{a}$, the fields $\cA$, $X$ and $\eta$ are tensorial, so their components change as
\begin{align*}
X^i &\to X^i \\[1em]
A^a &\to (\Lambda^{-1})^a_b \cA^b \\[1em]
\eta_a &\to \eta_b \Lambda^b_a.
\end{align*}

\begin{lemma}
	Under the change of frame $e_a \to e'_a = e_b \Lambda^b_{a}$, the structure functions change as
	\begin{align}
	\label{structurefunctionstransform}
	f^a_{\ bc} \to f'^a_{\ bc} = (\Lambda^{-1})^a_d \Lambda^e_b \Lambda^f_c f^d_{\ ef} - (\Lambda^{-1})^a_d \Lambda^e_b \Lambda^d_{c,i} v^i_e + (\Lambda^{-1})^a_d \Lambda^e_c \Lambda^d_{b,i} v^i_e 
	\end{align}	
\end{lemma}
\begin{proof}
	This follows from the defining property of the fundamental vector fields (\ref{structurefunctionsdefn}). After a change of frame, the new structure functions are defined by $[e'_b,e'_c] = f'^a_{\ bc} e'_a$. Expanding the bracket $[e'_b,e'_c] = [e_m \Lambda^m_b , e_n \Lambda^n_c]$ by applying the Leibniz property $$[X,fY] = (\rho(X)f)Y + f[X,Y]$$ and antisymmetry of the bracket, and then re-expressing in terms of $e'$ gives us (\ref{structurefunctionstransform}).
\end{proof}
\begin{lemma}
	Under the change of frame $e_a \to e'_a = e_b \Lambda^b_{a}$, the matrix of one-forms $\omega^a_b$ changes as
	\begin{align*}
	\omega^a_b \to \omega'^a_b = (\Lambda^{-1})^a_m \omega^m_n \Lambda^n_b +(\Lambda^{-1})^a_m \dd \Lambda^m_b.
	\end{align*}
	It follows from Appendix \ref{app:LAGT} that $\omega$ determines the connection coefficients for a connection $\nabla$ on $E$:
	\begin{align*}
	\nabla e_a = e_b \, \omega^b_a
	\end{align*}
\end{lemma}
\begin{proof}
	This follows from insisting that the form of the gauge transformations (\ref{omegagaugetransforms}) does not change under the change of frame.
\end{proof}
The connection $\nabla$ determined by the connection form $\omega^a_b$ pulls back naturally to a connection on $X^{\ast}E$, which we will denote by the same symbol to avoid cluttering notation. We can also define an $E$-connection, ${^E}\nabla$,\footnote{See Appendix \ref{app:LAGT} for details} by 
\begin{align}
\label{Econnection1}
{^E}\nabla(s,t) = \nabla_{\rho(s)} (t).
\end{align}
We will refer to $^E \nabla$ as \emph{the} $E$-connection associated to $\nabla$, with the tacit understanding that there may be many possible $E$-connections we can construct from a given connection $\nabla$ . We now note that we can write the field strength for the gauge field in the following way:\footnote{Again, we have omitted pullbacks to avoid cumbersome notation.}
\begin{align*}
\cF^a &= \dd \cA^a + \omega^a_{b} \wedge \cA^b + \frac{1}{2}\left( f^a_{\ bc} + \omega^a_{bi} v^i_c -\omega^a_{ci} v^i_b \right) \cA^b \wedge \cA^c \\[1em]
&= (\nabla \cA)^a - \frac{1}{2} \mathcal{T}^a_{\ bc} \cA^b \wedge \cA^c,
\end{align*}
where $\mathcal{T}^a_{\ bc} = -(f^a_{\ bc} + \omega^a_{bi} v^i_c -\omega^a_{ci} v^i_b)$ is simply the torsion of the $E$-connection ${^E}\nabla$. Note that although the field strength is written as a sum of terms which are covariant in $E$, the same is not true of the variation of the gauge fields:
\begin{align*}
\de A^a = (\nabla \epsilon)^a - \mathcal{T}^a_{bc} \cA^b \epsilon^c + \omega^a_{bi} v^i_c \cA^b \epsilon^c.
\end{align*}
The requirements that the field strength transform homogeneously, that is the vanishing of (\ref{omegacurvature}) and (\ref{omegaD}), now take on a geometric interpretation. The expression (\ref{omegacurvature}) for $R^a_b$ is precisely the expression for the components of the curvature of the connection $\nabla$, or the associated $E$-connection $^E \nabla$. These vanish if and only if the connection is flat. The expression (\ref{omegaD}) for $\cS^a_{bc}$ can be written in terms of geometric quantities as
\begin{align}
\label{omegaDgeom}
\cS^a_{bc} = \nabla \mathcal{T}^a_{bc}  + \iota_{v_b} R^a_c - \iota_{v_c} R^a_b
\end{align}
Since we require $R^a_b=0$, the condition $\cS^a_{bc} = 0$ reduces to  
\begin{align}
\label{covariantlyconstant}
\nabla \mathcal{T}^a_{bc} = 0.
\end{align}
From the flat connection we deduce (\ref{omegamaurercartan}), and then the `field redefinitions' (\ref{fieldredefinitions}) are simply the change in the components of the fields induced by the change of frame:
\begin{align*}
e_a \to  e_b (K^{-1})^b_{\ a}.
\end{align*}
Once we know the connection must be flat, we can simply choose a frame in which the connection form $\omega^a_b$ is zero. Note that in this frame, the condition (\ref{covariantlyconstant}) reduces to $\dd f^a_{\ bc} = 0$, and so the structure functions are simply structure constants. We thus recover the standard non-abelian gauging.

\subsubsection{A different $E$-connection}
\label{subsubsec:anotherEconnection}
There are many ways to lift a connection $\nabla$ on a Lie algebroid to an $E$-connection. Our definition (\ref{Econnection1}) has the advantage of being rather simple. A different $E$-connection which could have been useful to use is the so-called \emph{dual connection} (\cite{B04}),\footnote{\emph{Dual} is an unfortunately loaded word in this thesis. Here it is used to point out that the operation $\star: {^E}\nabla \mapsto \star({^E}\nabla) $ satisfies $\star ( \star (\nabla)) = \nabla$, where $\star({^E}\nabla)(s,t) = {^E}\nabla(t,s) + [s,t]$.} defined by
\begin{align*}
{^E}\widetilde{\nabla} (s,t) &= {^E}\nabla (t,s) + [s,t] \\[1em]
&= \nabla_{\rho(t)}(s) + [s,t].
\end{align*}
The $E$-torsion of the dual connection is (minus) the $E$-torsion of the $E$-connection we have defined, and the $E$-curvature of the dual connection can be written in terms of our $E$-connection as:
\begin{align*}
(\widetilde{\mathscr{R}}_{mn})^a_b &= \left( \nabla_b \mathcal{T}_{mn} \right)^a + (\mathscr{R}_{mb})^a_n - (\mathscr{R}_{nb})^a_m.
\end{align*}
The quantity appearing in (\ref{omegaD}) now has a much simpler geometric interpretation. We have
\begin{align*}
(\widetilde{\mathscr{R}}_{mn})^a_b &=  \iota_{\rho(e_b)} \mathcal{S}^a_{mn}.
\end{align*}
The gauging constraints enforcing the vanishing of $\cS^a_{bc}$ can therefore be interpreted as the requirement that the dual connection has a vanishing $E$-curvature. 

It is useful to have an expression for connection coefficients of an $E$-connection. For the $E$-connection $^E\nabla$, we have 
\begin{align*}
^E \nabla(e_a,e_b) &= \nabla_{\rho(e_a)} e_b \\[1em]
&= \nabla_{v^i_a \pr_i} e_b \\[1em]
&= v^i_a \nabla_i e_b \\[1em]
&= (v^i_a \omega^c_{bi}) \, e_c.
\end{align*}
For the dual connection $^E \widetilde{\nabla}$, we have
\begin{align*}
^E \widetilde{\nabla}(e_a,e_b) &= \nabla_{\rho(e_b)} e_a + [e_a, e_b] \\[1em]
&= \nabla_{v^i_b \pr_i} e_a + f^c \!_{ab} e_c \\[1em]
&= v^i_b \nabla_i e_b + f^c \!_{ab} e_c \\[1em]
&= \left( v^i_b \omega^c_{ai} + f^c \!_{ab} \right) e_c.
\end{align*}
That is, the connection coefficients for the $E$-connection and the dual connection are 
\begin{align*}
^E\omega^c_{ab} &= v^i_a \omega^c_{bi} \\[1em]
^E\widetilde{\omega}^c_{ab} &= v^i_b \omega^c_{ai} + f^c \!_{ab}. 
\end{align*}
It is an interesting fact that the $E$-torsion of the $E$-connection is simply the difference between the connection coefficients of the $E$-connection and its dual:
\begin{align*}
\mathcal{T}^a_{bc} &= \, ^E\omega^a_{bc} -\, ^E\widetilde{\omega}^a_{bc} = -\widetilde{\mathcal{T}}^a_{bc}.
\end{align*}
\subsubsection{Closure of the gauge algebroid}
\label{subsubsec:NITDclosure}
We can also see the constraints which arise from gauge invariance of the action from a different perspective, and that is by considering the closure of the gauge algebroid on the space of fields. In order for the gauge algebroid to have a representation on the space of fields, we require
\begin{align*}
\left( \delta_{\epsilon_1} \delta_{\epsilon_2} - \delta_{\epsilon_2} \delta_{\epsilon_1} - \delta_{[\epsilon_1,\epsilon_2]} \right)A^a = 0.
\end{align*}
We need to be careful here, however, since $\epsilon \in C^{\infty} (\Sigma, X^{\ast}E)$, and we are not guaranteed to have a bracket structure on the pullback bundle, since the Lie algebroid structure doesn't naturally pull back. We can, however, always pull back a Lie algebroid to a Lie algebroid with trivial anchor map. To see how, note that a basis for sections of $E$ pulls back to a basis of sections of $X^{\ast}E$, so any $\epsilon_i$ can be expressed as
\begin{align*}
\epsilon_i = \epsilon^a_i \, (X^{\ast} e_a).
\end{align*} 
We can now define a $C^{\infty}(\Sigma)$-linear bracket $[\cdot, \cdot]_{\ast}$ on the pullback bundle by simply taking the bracket $[\cdot,\cdot]_E$ on the corresponding sections in $E$, and then pulling it back:
\begin{align*}
[\epsilon_1 , \epsilon_2]_{\ast} &= [\epsilon^b_1 \, X^{\ast}(e_b) , \epsilon^c_2 \, X^{\ast}(e_c)]_{\ast} \\[1em]
&:= \epsilon^b_1 \epsilon^c_2 \, X^{\ast} ([e_b,e_c]_E) \\[1em]
&= f^a_{\ bc} \, \epsilon^b_1 \epsilon^c_2 \, X^{\ast}(e_a).
\end{align*}
Note that this bracket satisfies
\begin{align*}
[X,hY]_{\ast} = h[X,Y]
\end{align*}
for functions $h \in C^{\infty}(\Sigma)$, and therefore defines a Lie algebroid bracket on the vector bundle $X^{\ast}E \to \Sigma$ where the anchor map is zero. In other words, $X^{\ast}E$ is a bundle of Lie algebras over $\Sigma$. With this bracket, we now have a well-defined way of asking if the gauge algebroid closes on the space of fields. That is, does the quantity
\begin{align*}
\left( \delta_{\epsilon_1} \delta_{\epsilon_2} - \delta_{\epsilon_2} \delta_{\epsilon_1} - \delta_{[\epsilon_1,\epsilon_2]_{\ast}} \right)\cA^a
\end{align*}
vanish? This calculation was done in \cite{MS09}, where the authors found that the gauge algebra closes only up to a term proportional to $\scD X^i$:
\begin{align*}
\left( \delta_{\epsilon_1} \delta_{\epsilon_2} - \delta_{\epsilon_2} \delta_{\epsilon_1} - \delta_{[\epsilon_1,\epsilon_2]} \right)\cA^a &= \epsilon^b_a \epsilon^c_2 \, (\cS^a_{bc})_i \scD X^i,
\end{align*}
where $\cS^a_{bc}$ is the quantity given in (\ref{omegaD}), or equivalently (\ref{omegaDgeom}). Thus from a gauging perspective we only require the vanishing of $\cS^a_{bc}$. That is, we could in principle have a non-flat connection $\nabla$, although the dual connection $\widetilde{\nabla}$ must still be flat. Of course, if we want to do T-duality we need a field strength transforming homogeneously, and therefore also require a flat connection. 

\subsection{Examples: non-isometric T-duality revisited}
\subsubsection{The Heisenberg nilmanifold Part 1}
With the results of the previous sections in hand, we now revisit the previous examples of non-isometric T-duality considered in Section \ref{subsec:NITDexamples}. Since the $\omega$ we obtained was flat, we know there should exist a $K$ such that $\omega = K^{-1} \dd K$. Indeed, consider the following matrix:
\begin{align*}
K &= 
\left(
\begin{matrix}
1 & 0 \\
-y & 1 \\
\end{matrix}
\right).
\end{align*}
This matrix satisfies $\omega = K^{-1} \dd K$, and we find that transforming the vector fields $v_a = \{\pr_x , \pr_z\}$ using (\ref{vectorfieldredefinition}), we get $\tilde{v}_a = \{\pr_x + y \pr_z, \pr_z\}$. These vector fields are Killing vectors, and performing a non-abelian T-duality with respect to them gives the $f$-flux background, as expected. 

\subsubsection{The Heisenberg nilmanifold Part 2}
For the second gauging, we note that the matrix
\begin{align*}
K &= 
\left(
\begin{matrix}
1 & 0 & 0 \\
0 & 1 & 0 \\
-y & x & 0
\end{matrix}
\right)
\end{align*}
satisfies $\omega = K^{-1} \dd K$. The change of frame
\begin{align*}
\tilde{v}_a &= v_b (K^{-1})^b_a
\end{align*}
then gives $\tilde{v} = \{ \pr_x + y\pr_z, \pr_y, \pr_z \}$. These are Killing vector fields, and indeed the exact same vector fields we performed a non-abelian T-duality with respect to in Section \ref{NATDHeis}. As noted earlier, the appearance of the $Q$-flux background under a non-isometric T-duality was not a coincidence - it was simply the standard non-abelian T-duality in disguise.

\section{$(\omega, \phi)$-deformed gauging}
\label{sec:NITDomegaphi}

In this section we study a generalisation of the proposed non-isometric T-duality corresponding to a further modification of the gauge transformation rules. Non-isometric T-duality was based on a form of non-isometric gauging in which the standard non-abelian gauge fields transformed in a non-standard way, by including a term proportional to $\scD X^i$ in the variation. The generalisation we study, proposed in \cite{CDJS16a} and discussed more in \cite{CDJS17}, includes an additional term proportional to $\star \scD X^i$. This generalisation is interesting from a mathematical perspective, bringing a nice symmetry to the previous non-isometric gauging. This generalisation has not, to our knowledge, been used previously in the context of T-duality.

\subsection{The setup}
\label{sec:NITDomegaphisetup}
We begin with the setup of Section \ref{sec:NITDomega}: a metric and $B$-field, together with the infinitesimal action of a set of vector fields $\{v_a\}$. The action was the non-linear sigma model 
\begin{align*}
S = \int_{\Sigma} g_{ij} \dd X^i \wedge \star \dd X^j + B_{ij} \dd X^i \wedge \dd X^j, 
\end{align*}
together with the corresponding minimally coupled action
\begin{align*}
\label{omegaphiMC}
S_{MC} = \int_{\Sigma} g_{ij} \scD X^i \wedge \star \scD X^j + B_{ij} \scD X^i \wedge \scD X^j. 
\end{align*}
We allow the structure constants to be structure functions
\begin{align*}
[v_a,v_b] = f^c_{\ ab}(X)  v_c,
\end{align*}
and we consider extended gauge transformations of the form:
\begin{subequations}
\label{NITDomegaphivariation}
\begin{align}
\de X^i &= v^i_a \epsilon^a \\
\de A^a &= \dd \epsilon^a + f^a_{\ bc} \cA^b \epsilon^c + \omega^a_{bi} \epsilon^b \scD X^i + \phi^a_{bi} \epsilon^b \star \scD X^i.
\end{align}
\end{subequations}
With the hindsight of already having had done these calculations, we now introduce some simplifying notation. 
\begin{align*}
\Omega^{\pm} &:= \omega \pm \phi, & d_{\pm} X^i &:= \dd X^i \pm \star \dd X^i, & \cA^a_{\pm} &= \cA^a \pm \star \cA^a
\end{align*}
We also define 
\begin{align*}
\scD_{\pm} X^i &:= \scD X^i \pm \star \scD X^i = d_{\pm} X^i - v^i_a \cA^a_{\pm}
\end{align*}
and 
\begin{align*}
E^{\pm} &:= g\pm B.
\end{align*}
Note that since $B$ is antisymmetric, taking the transpose of $E^+$ gives $E^-$.
The minimally coupled action (\ref{omegaphiMC}) now takes the simpler form:
\begin{align}
\label{omegaphiMCsimple}
S_{MC} = \frac{1}{2}\int_{\Sigma} (E^+)_{ij} \scD_- X^i \wedge \scD_+ X^j.
\end{align}
The variation of the gauge covariant derivatives are
\begin{align*}
\de (\scD_{\pm} X^i) &= (M^{\pm})^i_{\ j} \scD_{\pm} X^j,
\end{align*}
where 
\begin{align*}
(M^{\pm})^i_{\ j} &= \epsilon^a (\pr_j v^i_a) - v^i_a (\Omega^{\pm})^a_{bj} \epsilon^b.
\end{align*}
We now compute the variation of the minimally coupled action, obtaining
\begin{align*}
\de S_{MC} &= \int_{\Sigma} \de (E^+)_{ij} \scD_- X^i \wedge \scD_+ X^j + (E^+)_{ij} \de (\scD_- X^i) \wedge \scD_+ X^j \\
& \quad \quad + (E^+)_{ij} \scD_- X^i \wedge \de (\scD_+ X^j) \\[1em]
&= \int_{\Sigma} \left[ v^k_a (\pr_k E^+)_{ij} \epsilon^a + (E^+)_{kj} (M^-)^k_{\ i} + (E^+)_{ik} (M^+)^k_{\ j} \right] \scD_- X^i \wedge \scD_+ X^j \\[1em]
&= \int_{\Sigma} \left[ (\Lie_{v_a} E^+)_{ij} - (E^+)_{kj} v^k_b (\Omega^-)^b_{ai} - (E^+)_{ik} v^k_b (\Omega^+)^b_{aj} \right] \epsilon^a \scD_- X^i \wedge \scD_+ X^j.
\end{align*}
It follows that the minimally coupled action is invariant under the $(\omega,\phi)$-modified gauge transformations (\ref{NITDomegaphivariation}), provided
\begin{align}
\label{omegaphiKilling}
(\Lie_{v_a} E^+)_{ij} = (E^+)_{mj} v^m_b (\Omega^-)^b_{ai} + (E^+)_{im} v^m_b (\Omega^+)^b_{aj}.
\end{align}
Separating the symmetric and antisymmetric components of this equation lets us rewrite the gauge invariance conditions in terms of $\omega$ and $\phi$:
\begin{subequations}
\label{omegaphi}
\begin{align}
\label{LIEg}
\Lie_{v_a}g &= \omega^b_a \vee \iota_{v_b}g - \phi^b_a \vee \iota_{v_b} B \\[1em]
\label{LIEb}
\Lie_{v_a}B &= \omega^b_a \wedge \iota_{v_b}B - \phi^b_a \wedge \iota_{v_b} g.
\end{align}
\end{subequations}
\subsection{Geometric interpretation}

As before, the $\omega^a_b$ transforms as a connection form, but the $\phi^a_b$ transforms homogeneously. That is, $\phi^a_b$ are the components of a one-form with values in $E \otimes E^{\ast} \cong \textrm{End}(E)$. The $\omega^a_b$ and $\phi^a_b$ therefore determine two connections, 
\begin{align*}
\nabla^{\pm} = \nabla \pm \phi
\end{align*}
on $E$. The connection forms for these connections are simply $\Omega^{\pm}$, so that 
\begin{align*}
\nabla^{\pm} e_a = e_b (\Omega^{\pm})^b_a = e_b (\omega^b_a \pm \phi^b_a).
\end{align*}
We shall see that the gauging provides constraints on the geometric quantities associated to these connections. 
\subsubsection{Closure of the gauge algebroid}
\label{subsubsec:NITDphiclosure}
The gauging described in this section is controlled by two connections $\nabla^{\pm}$ on $E$. For the standard non-isometric gauging, closure of the gauge algebroid enforced the vanishing of the quantity $\cS^a_{bc} = \nabla_i \mathcal{T}^a_{bc} + \iota_{v_b}R^a_c - \iota_{v_c} R^a_b$. For the $(\omega,\phi)$-modified gauge transformations (\ref{NITDomegaphivariation}), closure of the gauge algebra imposes similar constraints on the two connections $\nabla^+$ and $\nabla^-$. To fix notation, let $R^+$ be the curvature of the connection $\nabla^+$, and $R^-$ be the curvature of the connection $\nabla^-$. In components, we have
\begin{align*}
(R^{\pm})^a_b = \dd\, (\Omega^{\pm})^a_b + (\Omega^{\pm})^a_m \wedge (\Omega^{\pm})^m_b.
\end{align*}
We lift the connections $\nabla^{\pm}$ to the associated $E$-connections ${^E}\nabla^{\pm}$ through the anchor map
\begin{align*}
{^E}\nabla^{\pm} (s,t) = \nabla^{\pm}_{\rho(s)} (t).
\end{align*}
The $E$-torsions are then defined as usual:
\begin{align*}
(\mathcal{T}^{\pm})(s,t) = {^E}\nabla^{\pm}_s t - {^E}\nabla^{\pm}_t s - [s,t].
\end{align*}
In components, we have
\begin{align*}
(\mathcal{T}^{\pm})^a_{bc}= - \left( f^a_{\ bc} + (\Omega^{\pm})^a_{bi} v^i_c - (\Omega^{\pm})^a_{ci} v^i_b \right).
\end{align*}
We can now state the result on the closure of the gauge algebroid.
\begin{lemma}
\label{omegaphiclosure}
On the $\cA^a$ fields, the closure of the gauge transformations (\ref{NITDomegaphivariation}) is determined by
\begin{align*}
\left( \delta_{\epsilon_1} \delta_{\epsilon_2} - \delta_{\epsilon_2} \delta_{\epsilon_1} - \delta_{\epsilon_3} \right) \cA^a = \frac{1}{2} \epsilon^b_1 \epsilon^c_2 (\cS^+_i)^a_{bc} \scD_+ X^i +  \frac{1}{2} \epsilon^b_1 \epsilon^c_2 (\cS^-_i)^a_{bc} \scD_- X^i
\end{align*} 
where $\epsilon^a_3 = f^a_{bc} \epsilon^b_1 \epsilon^c_2 $ is defined as in Section \ref{subsubsec:NITDclosure}, and
\begin{align}
\label{Qcurvature}
(\cS^{\pm}_j)^a_{bc} = (\nabla^{\pm}_j\mathcal{T}^{\pm})^a_{bc} + v^i_b (R^{\pm}_{ij})^a_c - v^i_c (R^{\pm}_{ij})^a_b.
\end{align}
Since $\scD_+X^i$ and $\scD_-X^i$ are linearly independent, closure of the gauge algebroid on the space of fields therefore requires the vanishing of $(\cS^+)^a_{bc}$ and $(\cS^-)^a_{bc}$.
\end{lemma}
\begin{proof}
The proof is a straightforward, yet lengthy calculation, which we omit here. 
\end{proof}

\subsection{Examples}

\subsubsection{A non-flat example}
Consider the following metric and $B$-field:
\begin{align*}
g &= \dd x^2 + e^{\lambda x y} \dd y^2 \\[1em]
B &= e^{\lambda x y} \dd x \wedge \dd y,
\end{align*}
where $\lambda \not= 0$. This metric appeared in \cite{KS15}, as an example of a non-flat $\omega$-deformed gauging. Here, we consider the metric with a non-zero $B$-field, and want to gauge the action of the vector field $\pr_y$. This metric has no Killing vectors, and we have 
\begin{align*}
\Lie_{\pr_y} g &= \lambda x e^{\lambda x y } \dd y \vee \dd y \\[1em]
\Lie_{\pr_y} B &= \lambda x e^{\lambda x y} \dd x \wedge \dd y. 
\end{align*}
We can modify the gauging in \cite{KS15} to include the $B$-field, as well as include a non-zero $\phi.$ Consider:
\begin{align*}
\omega &= s(x,y) \dd x + \lambda x \dd y \\[1em]
\phi &= - s(x,y) \dd y,
\end{align*}
for an arbitrary function $s(x,y)$. One readily verifies that the conditions (\ref{omegaphi}) are satisfied. On the other hand, the curvatures do not vanish:
\begin{align*}
R^{\pm} &= \dd \Omega^{\pm} + \Omega^{\pm} \wedge \Omega^{\pm} \\[1em]
&= \left( \lambda \mp \frac{\pr s}{\pr x} - \frac{\pr s}{\pr y} \right) \dd x \wedge \dd y.
\end{align*}
\subsubsection{Gauging Poisson-Lie symmetry}
Recall from Chapter \ref{chptr:Ch4} that Poisson-Lie T-duality is a generalisation of the standard non-abelian T-duality. For this example, we show that we are able to gauge the Poisson-Lie symmetry condition using the $(\omega,\phi)$-modified gauging. To that end, let us suppose that we have a standard non-linear sigma model action satisfying the Poisson-Lie symmetry condition
\begin{align}
\label{PLsymmetrychap5}
(\Lie_{v_a} E)_{ij} = \widetilde{f}^{bc}_a v^m_b v^n_c E_{mj} E_{in},
\end{align}
for some set of dual structure constants $\widetilde{f}^{bc}_a$. We now note that the modified Killing equations for the $(\omega,\phi)$-deformed gauging, (\ref{omegaphiKilling}), have following special form when $\Omega^+$ vanishes (which forces $\omega = -\phi$):
\begin{align}
\label{omegaequalsminusphi}
(\Lie_{v_a}E)_{ij} &= 2 \omega^b_{ai} v^m_b E_{mj}.
\end{align}
If we compare (\ref{omegaequalsminusphi}) with (\ref{PLsymmetrychap5}), we see that by choosing
\begin{subequations}
\label{PLomegaphi}
\begin{align}
\omega^b_{ai} &= \frac{1}{2} \widetilde{f}^{bc}_a v^n_c E_{in} \\[1em]
\phi^b_{ai} &= -\frac{1}{2} \widetilde{f}^{bc}_a v^n_c E_{in},
\end{align}
\end{subequations}
the modified Killing equations are automatically satisfied. This is an interesting result! It suggests that we can provide a gauging derivation of Poisson-Lie T-duality. Before we get carried away, we should check that the gauge algebroid closes on the space of fields. By Lemma \ref{omegaphiclosure}, it is enough to check that $\cS^{\pm}$ vanish. This is the content of the next Lemma.
\begin{lemma}
\label{PLconnection}
Let $E_{ij} = g_{ij} + B_{ij}$ be a metric and $B$-field satisfying the Poisson-Lie symmetry condition (\ref{PLsymmetrychap5}). Then the connections defined by (\ref{PLomegaphi}) satisfy 
\begin{align*}
(\cS^{\pm})^a_{bc} = 0.
\end{align*}
\end{lemma}
\begin{proof}
Since $\Omega^{+} = 0$ and the structure functions are constant, the result for $\cS^+$ follows immediately from (\ref{Qcurvature}). For $\cS^-$, we need to perform a more lengthy calculation, which we do not include here. 
\end{proof}
It follows from Lemma \ref{PLconnection} that for this particular choice of $\omega$ and $\phi$, the gauge algebra closes on the space of fields. This is a non-trivial constraint, and it is interesting that this choice also satisfies the modified Killing equations. We make no claims of uniqueness for a choice of connection satisfying both the modified Killing conditions and for which the gauge algebra closes. 
\subsection{The field strength}
The gauging procedure described in Section \ref{sec:NITDomegaphisetup} is complete if all we care about is gauging a non-linear sigma model. Our primary interest, however, is to see if we can use this gauging procedure to perform T-duality. In order to do this, we need to supplement the minimally coupled action (\ref{omegaphiMC}) with a term $\int_{\Sigma} \eta_a \cF^a_{\Omega}$ which enforces the constraint on the gauge fields allowing us to recover the original model. 

What is the correct field strength? In choosing this, we are guided by a few principles. Clearly the field strength should reduce to the standard non-abelian Yang-Mills field strength when $\omega$ and $\phi$ vanish, and when $\phi$ vanishes it should reduce to the Mayer-Strobl field strength. We also want to insist that the field strength transforms homogeneously under the $(\omega,\phi)$-modified gauge transformations. We require this so that the Lagrange multipliers can be transformed in the appropriate way, leaving the term $\int_{\Sigma} \eta_a \cF^a_{\Omega}$ gauge invariant. To that end, consider the variation of the Yang-Mills field strength:
\begin{align*}
\de \cF^a_{YM} &= \de \left( \dd \cA^a + \frac{1}{2} f^a_{bc} \cA^b \wedge \cA^c  \right) \\[1em]
&= \dd \, (\de \cA^a) + \frac{1}{2} \de (f^a_{bc}) \cA^b \wedge A^c + f^a_{bc} (\de \cA^b) \wedge \cA^c \\[1em]
&= \omega^a_{bi} \dd \epsilon^b \wedge \scD X^i + \phi^a_{bi} \dd \epsilon \wedge \star \scD X^i + \dots
\end{align*}
where we have omitted terms which don't contain $\dd \epsilon$. We can see that in order to cancel the terms in the variation containing $\dd \epsilon$, we will need to add additional terms to the field strength. The Mayer-Strobl field strength contains the additional term $-\omega^a_{bi} \cA^b \wedge \scD X^i$, whose variation is given by
\begin{align*}
\de \left(-\omega^a_{bi} \cA^b \wedge \scD X^i \right) &= - \omega^a_{bi} \dd \epsilon^b \wedge \scD X^i +  \dots
\end{align*}
where, again, we have omitted terms which do not contain $\dd \epsilon$. Similarly, we see that variation of the term $- \phi^a_{bi} \cA^b \wedge \star \scD X^i$ gives 
\begin{align*}
\de \left(-\phi^a_{bi} \cA^b \wedge \star \scD X^i \right) &= - \phi^a_{bi} \dd \epsilon^b \wedge \star \scD X^i +  \dots.
\end{align*}
We therefore propose the following field strength for the gauge fields $\cA$:
\begin{align}
\label{omegaphifieldstrength}
\cF^a_{\Omega} &= \dd \cA^a + \frac{1}{2} f^a_{bc} \cA^b \wedge \cA^c - \omega^a_{bi} \cA^b \wedge \scD X^i - \phi^a_{bi} \cA^b \wedge \star \scD X^i.
\end{align}
The variation of this proposed field strength is a lengthy calculation. We summarise the results in the following lemma.
\begin{lemma}
The variation of the  field strength (\ref{omegaphifieldstrength}) is given by 
\begin{align}
\label{fieldstrengthvariation}
\de \cF^a_{\Omega} &= \left( f^a_{\ bc} - \frac{1}{2} \left( \Omega^+ \right)^a_{ci} v^i_b - \frac{1}{2} \left( \Omega^- \right)^a_{ci} v^i_b \right) \epsilon^c \cF^b_{\Omega} \\
&\quad +\frac{1}{2} \left( S^+_i \right)^a_{bc} \epsilon^c \cA^b \wedge \scD_+ X^i +\frac{1}{2} \left( S^-_i \right)^a_{bc} \epsilon^c \cA^b \wedge \scD_- X^i  \notag \\
 &\quad +\frac{1}{4} \epsilon^c \scD_+X^i \wedge \scD_-X^j \left( 2(R^{+} {\,}^a_c)_{ij}+2(R^{-} {\,}^a_c)_{ij}  -\nabla^{\omega}_i \phi^a_{cj} - \nabla^{\omega}_j \phi^a_{ci} \right)    \notag \\
&\quad + \epsilon^c \cA^m \wedge \scD_+ X^j \left( -\frac{1}{2} \phi^a_{ci} \nabla^{+}_j v^i_m \right) \notag \\
&\quad + \epsilon^c \cA^m \wedge \scD_- X^j \left( \frac{1}{2} \phi^a_{ci} \nabla^{-} v^i_m \right) \notag \\
&\quad + \epsilon^c \phi^a_{ci} \dd \ (\star \scD X^i). \notag
\end{align}
\end{lemma}
\begin{proof}
A long and messy calculation gives us
\begin{align*}
\de \cF^a_{\Omega} &= \left( f^a_{\ bc} - \frac{1}{2} \left( \Omega^+ \right)^a_{ci} v^i_b - \frac{1}{2} \left( \Omega^- \right)^a_{ci} v^i_b \right) \epsilon^c \cF^b_{\Omega} \\
&\quad +\frac{1}{2} \left( S^+_i \right)^a_{bc} \epsilon^c \cA^b \wedge \scD_+ X^i +\frac{1}{2} \left( S^-_i \right)^a_{bc} \epsilon^c \cA^b \wedge \scD_- X^i  \notag \\
&\quad +\epsilon^c \scD X^i \wedge \scD X^j \left( \frac{1}{2}\left( \pr_i \omega^a_{cj} - \pr_j \omega^a_{ci} + \omega^a_{bi} \omega^b_{cj} - \omega^a_{bj} \omega^b_{ci}  -\phi^a_{bi} \phi^b_{cj} + \phi^a_{bj} \phi^b_{ci} \right)  \right) \notag \\
&\quad +\epsilon^c \scD X^i \wedge \star \scD X^j \left( \frac{1}{2}\left( \pr_i \phi^a_{cj} + \pr_j \phi^a_{ci} + \omega^a_{bi} \phi^b_{cj} + \omega^a_{bj} \phi^b_{ci} -\phi^a_{bi} \omega^b_{cj} - \phi^a_{bj} \omega^b_{ci} \right)  \right) \notag \\
&\quad + \epsilon^c \cA^m \wedge \scD_+ X^j \left( \frac{1}{2} \left( v^i_b \phi^a_{ci} \phi^b_{mj} +v^i_b \phi^a_{ci} \omega^b_{mj} - \phi^a_{ci} (\pr_j v^i_m)  \right) \right) \notag \\
&\quad + \epsilon^c \cA^m \wedge \scD_- X^j \left( \frac{1}{2} \left( v^i_b \phi^a_{ci} \phi^b_{mj} -v^i_b \phi^a_{ci} \omega^b_{mj} + \phi^a_{ci} (\pr_j v^i_m)  \right) \right) \notag \\
&\quad + \epsilon^c \phi^a_{ci} \dd \ (\star \scD X^i). \notag
\end{align*}	
which we can write as
\begin{align*}
\de \cF^a_{\Omega} &= \left( f^a_{\ bc} - \frac{1}{2} \left( \Omega^+ \right)^a_{ci} v^i_b - \frac{1}{2} \left( \Omega^- \right)^a_{ci} v^i_b \right) \epsilon^c \cF^b_{\Omega} \\
&\quad +\frac{1}{2} \left( S^+_i \right)^a_{bc} \epsilon^c \cA^b \wedge \scD_+ X^i +\frac{1}{2} \left( S^-_i \right)^a_{bc} \epsilon^c \cA^b \wedge \scD_- X^i  \notag \\
&\quad +\frac{1}{4} \epsilon^c \scD_+X^i \wedge \scD_-X^j \left(\begin{matrix} \pr_i \omega^a_{cj} - \pr_j \omega^a_{ci} + \omega^a_{bi} \omega^b_{cj} - \omega^a_{bj} \omega^b_{ci}  -\phi^a_{bi} \phi^b_{cj} + \phi^a_{bj} \phi^b_{ci} \\ -\pr_i \phi^a_{cj} - \pr_j \phi^a_{ci} - \omega^a_{bi} \phi^b_{cj} - \omega^a_{bj} \phi^b_{ci} +\phi^a_{bi} \omega^b_{cj} + \phi^a_{bj} \omega^b_{ci} \end{matrix} \right)    \notag \\
&\quad + \epsilon^c \cA^m \wedge \scD_+ X^j \left( \frac{1}{2} \left( v^i_b \phi^a_{ci} \phi^b_{mj} +v^i_b \phi^a_{ci} \omega^b_{mj} - \phi^a_{ci} (\pr_j v^i_m)  \right) \right) \notag \\
&\quad + \epsilon^c \cA^m \wedge \scD_- X^j \left( \frac{1}{2} \left( v^i_b \phi^a_{ci} \phi^b_{mj} -v^i_b \phi^a_{ci} \omega^b_{mj} + \phi^a_{ci} (\pr_j v^i_m)  \right) \right) \notag \\
&\quad + \epsilon^c \phi^a_{ci} \dd \ (\star \scD X^i). \notag
\end{align*}
\end{proof}
We make here a few comments about this variation. To start, we note that when $\phi = 0$, this variation reduces to the variation of the $\omega$-modified field strength, given by (\ref{eomegacurvaturevariation}). Since our main motivation for the introduction of this field strength is to determine if we are able to provide a gauging derivation of Poisson-Lie T-duality, we are interested in the variation of $F^a_{\Omega}$ when $\Omega^{\pm}$ satisfy the Poisson-Lie conditions. 
\begin{corollary}
\label{PLNITDcor}
When $\Omega^+ = 0$, and therefore $\omega = - \phi$, the variation of $\cF^a_{\Omega}$ reduces to 
\begin{align*}
\de \cF^a_{\Omega} &= \left( f^a_{\ bc} - \omega^a_{ci} v^i_b \right) \epsilon^c \cF^b_{\Omega} + \frac{1}{2} \left( S^-_i \right)^a_{bc} \epsilon^c \cA^b \wedge \scD_-X^i \\
&\qquad + \frac{1}{2} \epsilon^c  \scD_+X^i \wedge \scD_-X^j (\pr_i \omega^a_{cj})  \\
&\qquad + \frac{1}{2} \epsilon^c \cA^m \wedge \scD_+X^j (\omega^a_{ci} \pr_j v^i_m) \\
&\qquad + \frac{1}{2} \epsilon^c \cA^m \wedge \scD_-X^j (-\omega^a_{ci} \pr_j v^i_m + 2\omega^a_{ci}\omega^b_{mj} v^i_b) \\
&\qquad -\epsilon^c \omega^a_{ci} \dd \ (\star \scD X^i) 
\end{align*}
\end{corollary}

For clarity, let's see what happens with an explicit example when we choose $\omega$ and $\phi$ as in (\ref{PLomegaphi}). Let us take the simple example of the $Q$-flux background. This is an example which can be obtained as the non-abelian dual of the twisted torus, as in Section \ref{NATDHeis}. From the Poisson-Lie perspective, the Lie bialgebra associated to this model is $\left( \RR^3, \textrm{Heis} \right)$, since it is the dual of the twisted torus which has an isometric action of the Heisenberg group. 

\subsubsection{Explicit calculation}

We begin with the metric and $B$-field
\begin{align}
g &= \dd z^2 + \frac{1}{1+z^2} \left( \dd x^2 + \dd y^2 \right) \\
B &= \frac{2z}{1+z^2} \dd x \wedge \dd y,
\end{align} 
and we would like to non-isometrically gauge the vector fields $\{\pr_x, \pr_y, \pr_z\}$. Let us use the $\omega$ and $\phi$ given by 
\begin{align}
\omega^b_{ai} &= \frac{1}{2} \widetilde{f}^{bc}_a v^n_c E_{in} \\[1em]
\phi^b_{ai} &= -\frac{1}{2} \widetilde{f}^{bc}_a v^n_c E_{in}. 
\end{align}
The dual structure constants are given by the Heisenberg algebra:
\begin{align}
f^{12}_3 &= -1 \\
f^{21}_3 &= 1
\end{align}
The non-zero components of $\omega$ are then
\begin{align*}
\omega^1_3 &= -\frac{1}{2} \left( \frac{z}{1+ z^2} \dd x + \frac{1}{1+ z^2} \dd y \right) \\[1em]
\omega^2_3 &= \frac{1}{2} \left( \frac{1}{1+ z^2} \dd x - \frac{z}{1+ z^2} \dd y \right),
\end{align*}
and for $\phi$ we have:
\begin{align*}
\phi^1_3 &= \frac{1}{2}  \left( \frac{z}{1+ z^2} \dd x + \frac{1}{1+ z^2} \dd y \right) \\[1em]
\phi^2_3 &= -\frac{1}{2}  \left( \frac{1}{1+ z^2} \dd x - \frac{z}{1+ z^2} \dd y \right).
\end{align*}
That is, we have 
\begin{align*}
(\Omega^-)^1_3 &= -\left( \frac{z}{1+z^2} \right) \dd x - \left( \frac{1}{1+z^2} \right) \dd y \\[1em]
(\Omega^-)^2_3 &= \left( \frac{1}{1+z^2} \right) \dd x - \left( \frac{z}{1+z^2} \right) \dd y,
\end{align*}
and of course $\Omega^+ = 0$. These satisfy the modified killing equations (\ref{omegaphiKilling}) as well as the Poisson-Lie symmetry conditions (\ref{PLsymmetrychap5}), as can be verified by a straightforward calculation. What are the variations of the field strengths?
The field strengths are given by (\ref{omegaphifieldstrength}):
\begin{align}
\cF^a_{\Omega} &= \dd \cA^a + \frac{1}{2} f^a_{bc} \cA^b \wedge \cA^c - \omega^a_{bi} \cA^b \wedge \scD X^i - \phi^a_{bi} \cA^b \wedge \star \scD X^i.
\end{align}
The structure constants vanish, and $\omega = - \phi$, so we have
\begin{align}
\cF^a_{\Omega} &= \dd \cA^a - \omega^a_{bi} \cA^b \wedge  \scD_- X^i .
\end{align}
That is,
\begin{align*}
\cF^1 &= \dd \cA^1 - \omega^1_{3i} \cA^3 \wedge \scD_- X^i \\
\cF^2 &= \dd \cA^2 - \omega^2_{3i} \cA^3 \wedge \scD_- X^i \\
\cF^3 &= \dd \cA^3
\end{align*}
The variation of $\cF^3$ is easy, and we find it vanishes identically. The variation of $\cF^2$ and $\cF^3$ is a bit harder. This can either be done directly, or with reference to Corollary \ref{PLNITDcor}. Either way, the result is
\begin{align*}
\de \cF^1 &= - \omega^1_{31} \epsilon^3 \cF^1 - \omega^1_{32} \epsilon^3 \cF^2 \\
& \qquad + \epsilon^3 (\pr_3 \omega^1_{3j}) \scD X^3 \wedge \scD_- X^j + \epsilon^3 \omega^1_{3i} \omega^i_{3j} \cA^3 \wedge \scD_- X^j \\
& \qquad - \epsilon^3 \omega^1_{3j}  \dd \, \left( \star DX^j \right).
\end{align*}
and
\begin{align*}
\de \cF^2 &= - \omega^2_{31} \epsilon^3 \cF^1 - \omega^2_{32} \epsilon^3 \cF^2 \\
& \qquad + \epsilon^3 (\pr_3 \omega^2_{3j}) \scD X^3 \wedge \scD_- X^j + \epsilon^3 \omega^2_{3i} \omega^i_{3j} \cA^3 \wedge \scD_- X^j \\
& \qquad - \epsilon^3 \omega^2_{3j}  \dd \, \left( \star DX^j \right).
\end{align*}
This doesn't appear to be covariant, but perhaps it is in a non-obvious way. Let us look at the variation of the full Lagrange multiplier term by expanding in powers of $A$, in a similar way to (\ref{powersofA}):
\begin{align*}
\de \left( \eta_a \mathcal{F}^a_{\Omega} \right) &=  \eta_a \epsilon^3 \Big[ \dd \omega^a_{3j} \wedge \dd X^j - \dd \omega^a_{3j} \wedge \star \dd X^j  - \omega^a_{3j} \dd \star \dd X^j \Big] \\[1em]
 &  \quad + \mathcal{O}(A) + \mathcal{O}(A^2)
\end{align*}
We require this to vanish independent of $A$ and $\eta$, so the terms 
\begin{align}
\label{constraint1}
\mathcal{K}^1 &:= \dd \omega^1_{3j} \wedge \dd X^j - \dd \omega^1_{3j} \wedge \star \dd X^j  - \omega^1_{3j} \dd \star \dd X^j,
\end{align}
and 
\begin{align}
\label{constraint2}
\mathcal{K}^2 &:= \dd \omega^2_{3j} \wedge \dd X^j - \dd \omega^2_{3j} \wedge \star \dd X^j  - \omega^2_{3j} \dd \star \dd X^j
\end{align}
must both vanish.
We will get back to these in a moment, but for now let us recall the definition of the Noether currents that arose in Poisson-Lie T-duality. These were one-forms associated to the action of the group. They are given by
\begin{align*}
J_i &= g_{ij} \star \dd X^j + B_{ij} \dd X^j. 
\end{align*}
In our case, we have
\begin{subequations}
\label{NoetherformsQflux}
\begin{align}
J_1 &= \frac{1}{1+z^2} \star \dd x + \frac{z}{1+z^2} \dd y \\[1em]
J_2 &= \frac{1}{1+z^2} \star \dd y - \frac{z}{1+z^2} \dd x \\[1em]
J_3 &= \star \dd z.
\end{align}
\end{subequations}
These one-forms were required to satisfy the Maurer-Cartan equation:
\begin{align*}
\dd J_a &= \frac{1}{2} \widetilde{f}_a \!^{bc} J_b \wedge J_c.
\end{align*}
In particular, this means that in our case we must have $\dd J_1 = \dd J_2 = 0$. Using this closure together with the expressions (\ref{NoetherformsQflux}) for $J_1$ and $J_2$, and then rearranging for $\dd \star \dd x$ and $\dd \star \dd y$ gives us the following relations:
\begin{align*}
\frac{1}{2} \left( \frac{1}{1+z^2} \right) \dd \star \dd x &= - \frac{1}{2}  \dd \left( \frac{1}{1+z^2} \right) \wedge \star \dd x - \frac{1}{2} \dd \left( \frac{z}{1+z^2} \right) \wedge \dd y \\[1em]
\frac{1}{2} \left( \frac{1}{1+z^2} \right) \dd \star \dd y &= - \frac{1}{2} \dd \left( \frac{1}{1+z^2} \right) \wedge \star \dd y + \frac{1}{2} \dd \left( \frac{z}{1+z^2} \right) \wedge \dd x.
\end{align*}
We are now in a position to check whether the terms given by (\ref{constraint1}) and (\ref{constraint2}) vanish. We compute explicitly 
\begin{align*}
\mathcal{K}^1 &= -\frac{1}{2} \dd \left(  \frac{z}{1+z^2} \right) \wedge \dd x  -\frac{1}{2} \dd \left( \frac{1}{1+z^2} \right) \wedge \dd y \\[1em]
&\quad + \frac{1}{2} \dd \left( \frac{z}{1+z^2} \right) \wedge \star \dd x +  \frac{1}{2} \dd \left( \frac{1}{1+z^2} \right) \wedge \star \dd y \\[1em]
&\quad +\frac{1}{2} \left( \frac{z}{1+z^2} \right) \dd \star \dd x + \frac{1}{2} \left( \frac{1}{1+z^2} \right) \dd \star \dd y \\[2em]
&= \frac{1}{2} \Bigg[ \dd \left( \frac{z}{1+z^2} \right) - z \dd \left( \frac{1}{1+z^2} \right)  \Bigg] \wedge \star \dd x \\[1em]
&\quad - \frac{1}{2} \Bigg[ \dd \left( \frac{1}{1+z^2} \right) + z \dd \left( \frac{z}{1+z^2} \right)  \Bigg] \wedge \dd y \\[2em]
&= \frac{1}{2} \left( \frac{1}{1+z^2} \right) \dd z \wedge \star \dd x + \frac{1}{2} \left( \frac{z}{1+z^2} \right) \dd z \wedge \dd y.
\end{align*}
We can similarly compute $\mathcal{K}^2$:
\begin{align*}
\mathcal{K}^2 &= -\frac{1}{2} \left( \frac{z}{1+z^2} \right) \dd z \wedge \dd x + \frac{1}{2} \left( \frac{1}{1+z^2} \right) \dd z \wedge \star \dd y. 
\end{align*}
It follows that, to order 0 in $A$, we have
\begin{align*}
\de \left( \eta_a \mathcal{F}^a_{\Omega} \right) &= \epsilon^3 \left[ \eta_1 \mathcal{K}^1 + \eta_2 \mathcal{K}^2 \right]  \\[1em]
&= \frac{\epsilon^3}{2} \left( \frac{\dd z}{1+z^2} \right) \wedge \Bigg[ \eta_1 \star \dd x - z \eta_2 \dd x + z\eta_1 \dd y + \eta_2 \star \dd y \Bigg].
\end{align*}
Unfortunately, barring computational errors, we can see no obvious reason why this should vanish, as is required for the term $\int \eta_a \mathcal{F}^a_{\Omega}$ to be gauge invariant. With this in mind, let us conclude this chapter with a few comments on the current status of non-isometric T-duality. 

Non-isometric T-duality, as introduced in \cite{CDJ15,C16}, was shown here and in \cite{BBKW} to be equivalent to the standard notion of non-abelian T-duality. In this chapter, we have suggested that a modified notion of gauging, introduced in \cite{CDJS16a,CDJS17}, has the potential to be a non-trivial generalisation of the usual T-duality. The generalisation requires choosing a pair of connections specified by $\omega \pm \phi$. The connections cannot be arbitrarily chosen - they must be chosen to satisfy three constraints. First, they must satisfy the modified Killing conditions. Requiring that the gauge algebra close on the space of fields provides an additional constraint, and finally, gauge invariance of the Lagrange multiplier/field strength term in the gauged action provides the final constraint. Our main goal in pursuing this work was to provide a gauging approach to Poisson-Lie T-duality, and in pursuing this we identified such a choice which satisfied the first two constraints. The calculations at the end of this chapter, however, show that it fails to satisfy the third constraint. For models satisfying the Poisson-Lie condition, it is therefore a very interesting question whether there exist pairs of connections satisfying all three constraints. If there are, then we could use non-isometric gauging to perform T-duality on these models. In particular, this would give a sigma model procedure of inverting non-abelian T-duality.

We hope to report on progress towards this open question in future publications. 

%% file: chapter6.tex
\chapter{Spherical T-duality}
\label{chptr:Ch6}

Spherical T-duality is a generalisation of the topological aspects of abelian T-duality. First introduced in \cite{BEM14a}, and further developed in \cite{BEM14b,BEM15}, Spherical T-duality is a putative duality based on replacing the $U(1)$ bundles of abelian T-duality with $SU(2)$ bundles. The relevance as a target space duality in string theory/M-theory/supergravity is currently unclear.

\section{Topology}
In Section \ref{subsec:TopCircBund}, we reviewed the topological aspects of abelian T-duality. The primary ingredient for topological T-duality is a pair $(F,H)$ of cohomology classes. To obtain this pair, we started with a principal $U(1)$ bundle $\pi:E \to M$ whose isomorphism class determined an element $F \in H^2(M,\ZZ)$, together with an $H$-flux determining an element $H \in H^3(E,\ZZ)$. We then saw that there was a set of dual data $(\widehat{F},\widehat{H})$ for some dual bundle $\hat{\pi}:\widehat{E} \to M$, satisfying
\begin{subequations}
\label{SphTDintermix}
\begin{align}
\widehat{F} &= \pi_{\ast} H \\[1em]
F &= \hat{\pi}_{\ast} \widehat{H}.
\end{align}
\end{subequations}
The arguments for topological T-duality relied heavily on the fibers of the fibration $\pi:E \to M$ having two specific properties: the fibers were a group (the total space was a principal $\sG$ bundle with a free and transitive group action on the fibers), and the fibers were also spheres (since the fibration was a sphere bundle we could form the associated Gysin sequences in cohomology). There is precisely one other space which is both a sphere and a group, and it is $SU(2) = S^3$. From a mathematical perspective, then, it is a natural question to ask whether there is a similar notion of duality for principal $SU(2)$ bundles. 

Before we delve into the construction of spherical T-duality presented in \cite{BEM14a,BEM14b,BEM15}, let us pause for a moment to see if we can motivate the appearance of the data appearing in the next section. We expect that, as with abelian T-duality, spherical T-duality should involve some topological data $(F,H)$, determined by the topological properties of the bundle and some form of flux, and that this data is intermixed under the duality. Furthermore, we expect that the flux will be defined on the total space, while the topological data will only be defined on the base (since this data should be related to the possible types of bundle over the base). Thus if $F_k \in H^k (M)$, then we expect that $H_{k+3} \in H^{k+3} (E)$, since the fibers are 3-dimensional. Since integrating the flux over the three-dimensional fiber would give us something in $H^k (M)$, we are looking for a relationship along the lines of (\ref{SphTDintermix}):
\begin{align*}
\widehat{F}_k &= \pi_{\ast} H_{k+3} \\[1em]
F_k &= \hat{\pi}_{\ast} \widehat{H}_{k+3},
\end{align*}
for some pair of $SU(2)$ bundles 
\tikzset{node distance=1in, auto}
\begin{center}
	\begin{tikzpicture}
	\node (E) {$E$};
	\node (empty) [right of =E]{};
	\node (hatE) [right of=empty] {$\widehat{E}$};
	\node (M) [below of=empty] {$M$};
	\node (A4) {};
	\draw[->] (E) to node [swap] {$\pi$} (M);
	\draw[->] (hatE) to node {$\hat{\pi}$} (M);
	\end{tikzpicture}
\end{center}
In abelian T-duality, the isomorphism class of the bundle is the (cohomology class of the) curvature of a connection, $F  \in H^2(M,\ZZ)$. This cohomology class is realised by the first Chern class, $c_1 (E) $, of the associated complex line bundle. For an $SU(2)$ bundle, we could try to find characteristic classes for the associated quaternionic line bundle.\footnote{Quaternions arise here because, just as $U(1)$ can be thought of as the unit sphere in $\CC$, so too can $SU(2)$ be thought of as the unit sphere in $\HH$.} The relevant characteristic class is the second Chern class, $c_2(E) \in H^4(M)$.\footnote{For principal $SU(2)$-bundles, which we will be exclusively interested in here, the second Chern class is also known as the Euler class. For a four-dimensional base $M$, the integral $\int_M c_2(E)$ is known as the \emph{instanton number} or \emph{topological charge}.} A de-Rham representative for this class is given by
\begin{align*}
c_2(E) &= \frac{1}{8 \pi^2} \Tr (\cF \wedge \cF),
\end{align*}
where $\cF$ is the Yang-Mills curvature of a principal $\mathfrak{su}(2)$ connection $\cA$:
\begin{align*}
\cF = \dd \cA + \cA \wedge \cA.
\end{align*}
Thus the relevant topological data corresponds to $k=4$, and we will start with a principal $SU(2)$ bundle $\pi:E \to M$, together with a pair $(F_4,H_7)$, with $F_4 \in H^4 (M,\ZZ)$ and $H_7 \in H^7 (M,\ZZ)$. This is also motivated by the description in terms of Gysin sequences, discussed more in Section \ref{subsec:sphericalGysin}.
\subsection{Topological T-duality for sphere bundles}
Let us now follow \cite{BEM14a} and describe how spherical T-duality works. We begin with a principal $SU(2)$ bundle $\pi:E \to M$, 
\begin{center}
	\begin{tikzpicture}
	\node (F) {$S^3$};
	\node (E) [right=1cm of F] {$E$};
	\node (M) [below=1cm of E] {$M$};
	\draw[right hook->] (F) to node  {$ $} (E);
	\draw[->] (E) to node [right] {$\pi$} (M);
	\end{tikzpicture}
\end{center}
together with a flux $H_7 \in H^7 (E)$. We will assume that $M$ is both compact and orientable.

We would like a way of classifying the isomorphism classes of the bundle $\pi:E \to M$, but the topological classification of the isomorphism classes of principal $SU(2)$ bundles over $M$ is highly dependent on the dimension of $M$. If $\textrm{dim}(M) \leq 3$, then the only principal $SU(2)$ bundles are trivial, and so $E = SU(2) \times M$. When $\textrm{dim}(M) = 4$, the isomorphism class of the bundle is in one to one correspondence with the second Chern class that we mentioned in the previous section. Indeed, since $H^4(M,\ZZ) \cong \ZZ$ for a compact oriented four-manifold, there are countably many principal $SU(2)$ bundles over $M$. When the dimension of $M$ exceeds 4, the situation is a bit messier. In that case, the second Chern class no longer classifies isomorphism classes of principal $SU(2)$ bundles. That is, there exist non-isomorphic bundles with equal second Chern classes. In addition, there are elements of $H^4(M,\ZZ)$ which are not realised as the second Chern class of any principal $SU(2)$ bundle. 

\subsubsection{$\textrm{dim} (M) \leq 4$}
Let us, for the moment, consider the case of $\textrm{dim}(M) \leq 4$. In that case, the set of isomorphism classes of principal $SU(2)$ bundles over $M$ is classified by the second Chern class $c_2(E) \in H^4 (M,\ZZ)$. To see this, recall that principal $\sG$ bundles over a space $M$ are classified by homotopy classes of maps from $M$ to the classifying space $\mathsf{BG}$ of $\sG$. That is, we have 
\begin{align*}
\textrm{Prin}_{SU(2)}(M) = [M, \mathsf{B} SU(2)].
\end{align*}
A convenient property that all classifying spaces share is that 
\begin{align}
\label{piclassifying}
\pi_k (\mathsf{BG}) = \pi_{k-1}(\mathsf{G}).
\end{align}
In particular, we have that $\pi_k (\mathsf{B}SU(2)) = \pi_{k-1} (SU(2)) = \pi_{k-1} (S^3)$. It follows from the table in Appendix \ref{app:homotopy} that
\begin{align*}
\pi_4 (\mathsf{B}SU(2)) &= \ZZ \\[1em]
\pi_3 (\mathsf{B}SU(2)) &= 0 \\[1em]
\pi_2 (\mathsf{B}SU(2)) &= 0 \\[1em]
\pi_1 (\mathsf{B}SU(2)) &= 0.
\end{align*}
That is, $\mathsf{B}SU(2)$ is a 3-connected space. Now, we can't claim that $\mathsf{B}SU(2)$ is an Eilenberg-Maclane space $K(\ZZ,4)$, since the higher homotopy groups of $S^3$ don't vanish. We can, however, realise $\mathsf{B}SU(2)$ as a subcomplex of a $K(\ZZ,4)$ space,\footnote{Here we are thinking of $\mathsf{B}SU(2)$ as a CW-complex.} by attaching cells of higher dimension. Then, provided $M$ is a CW-complex of dimension at most $4$, we have 
\begin{align*}
[M,\mathsf{B}SU(2)] = [M,K(\ZZ,4)] = H^4(M,\ZZ).
\end{align*}
It follows that isomorphism classes of principal $SU(2)$ bundles over a four-dimensional (compact, orientable, CW-complex) manifold are classified by $H^4 (M,\ZZ)$. 

Note that we have assumed that $M$ is compact so we can't say anything about, for example, $SU(2)$ bundles over $\RR^4 \setminus \{0\}$. In this particular instance, however, we can use (\ref{piclassifying}) to our advantage. Since $\RR^4 \setminus \{0\}$ is homotopy equivalent to $S^3$, we have that 
\begin{align*}
[\RR^4 \setminus \{0\}, \mathsf{B}SU(2)] = [S^3, \mathsf{B}SU(2)] = \pi_{3} (\mathsf{B}SU(2)) = 0,
\end{align*}
so that any $SU(2)$ bundle over $\RR^4 \setminus \{0\}$ or $S^3$ is trivial.\footnote{Of course, we can see that $SU(2)$ bundles over $S^3$ must be trivial in another way, simply by noting that $H^4 (S^3,\ZZ)$ vanishes for dimensional reasons. It is for this reason that all principal $SU(2)$ bundles over manifolds of dimension less than 4 are trivial, as previously claimed.} A similar argument, together with the fact that $\pi_4 (\mathsf{B} SU(2)) = \ZZ$, shows that the set of principal $SU(2)$ bundles over both $S^4$ and $\RR^5 \setminus \{0\}$ is in one-to-one correspondence with $\ZZ$. 

Let us now define, following \cite{BEM14a}, the spherical T-dual data. As in the abelian case, there is a homomorphism $\pi_{\ast} : H^7 (E,\ZZ) \to H^4 (M,\ZZ)$ which acts on de-Rham representatives by integration over the fiber. The image of the $H$-flux therefore determines an element $\widehat{F}_4 = \pi_{\ast} H_7 \in H^4 (M,\ZZ)$. Since $\textrm{dim}(M) \leq 4$, this determines the isomorphism class of a principal $SU(2)$ bundle $\hat{\pi} : \widehat{E} \to M$, giving us the following picture:
\begin{center}
	\begin{tikzpicture}
	\node (E) {$E$};
	\node (empty) [right of =E]{};
	\node (hatE) [right of=empty] {$\widehat{E}$};
	\node (M) [below of=empty] {$M$};
	\node (A4) {};
	\draw[->] (E) to node [swap] {$\pi$} (M);
	\draw[->] (hatE) to node {$\hat{\pi}$} (M);
	\end{tikzpicture}
\end{center}
We now define a dual $H$-flux as any class $\widehat{H_7} \in H^7 (\widehat{E},\ZZ)$, whose image under $\hat{\pi}_{\ast}$ gives back $F_4$. That is, the dual data is a principal $SU(2)$ bundle $\hat{\pi}: \widehat{E} \to M$, together with a pair $(\widehat{F}_4,\widehat{H}_7)$, satisfying the $SU(2)$ analogue of (\ref{SphTDintermix}):
\begin{subequations}
\label{F4H7}
\begin{align}
\label{dualF4}
 \widehat{F}_4 &=  \pi_{\ast} H_7 \\[1em]
F_4 &=  \hat{\pi}_{\ast} \widehat{H}_7.
\end{align}
\end{subequations}
\subsubsection{$\textrm{dim}(M) > 4$}
When the dimension of $M$ is greater than 4, bundles are no longer classified by their second Chern class, but we may still define spherical T-duality in an analogous way. That is, suppose we have a principal $SU(2)$ bundle $\pi: E \to M$, together with a pair $(F_4,H_7)$, where $F_4 \in H^4(M,\ZZ)$ is the second Chern class of the associated bundle, and $H_7 \in H^7 (E,\ZZ)$. The dual data is then a bundle $\hat{\pi}:\widehat{E} \to M$, together with a pair $(\widehat{F}_4,\widehat{H}_7)$, where $\widehat{F}_4 \in H^4(M,\ZZ)$ is the second Chern class of the associated bundle, and $\widehat{H}_7 \in H^7 (\widehat{E},\ZZ)$, satisfying
\begin{align*}
 \widehat{F}_4 &=  \pi_{\ast} H_7 \\[1em]
F_4 &=  \hat{\pi}_{\ast} \widehat{H}_7.
\end{align*}
The distinction is that above dimension 4 the results are weaker, and existence/uniqueness of T-duals may no longer hold. 

\subsection{Desciption using Gysin sequences}
\label{subsec:sphericalGysin}
The Gysin sequence description of abelian topological T-duality is one of the most obvious places where the question of generalising T-duality comes up. Indeed, \emph{every} sphere bundle has an associated Gysin sequence, so it is a natural question to ask if we can generalise the arguments in Section \ref{subsubsec:TopGysin} to higher dimensional spheres. Recall from Appendix \ref{sec:Gysin} the Gysin sequence for a fiber bundle $\pi: E \to M$ with fiber $S^3 = SU(2)$:
\begin{align*}
\cdots \overset{}{\longrightarrow} H^n (M,\ZZ) \overset{\pi^{\ast}}{\longrightarrow} H^{n} (E,\ZZ) \overset{\pi_{\ast}}{\longrightarrow} H^{n-3} (M,\ZZ) \overset{e \wedge}{\longrightarrow} H^{n+1} (M,\ZZ) \longrightarrow \cdots 
\end{align*}
Looking at the $n=7$ section in the associated sequence for \v{C}ech cohomology, we have:
\begin{align*}
\cdots \overset{}{\longrightarrow} H^7 (M,\ZZ) \overset{\pi^{\ast}}{\longrightarrow} H^{7} (E,\ZZ) \overset{\pi_{\ast}}{\longrightarrow} H^{4} (M,\ZZ) \overset{F_4 \cup}{\longrightarrow} H^{8} (M,\ZZ) \longrightarrow \cdots 
\end{align*}
The description now follows the same procedure as the abelian case in Section \ref{subsubsec:TopGysin}. Beginning with the topological data $\pi:E \to M$ and $(F_4,H_7)$, we note that the flux $H_7$ lives in $H^7 (E,\ZZ)$, and so we can consider it's image $\widehat{F}_4 = \pi_{\ast} H_7$, under $\pi_{\ast}$. Since the Gysin sequence is exact, the composition of two maps is zero, and so it follows that $F_4 \cup \widehat{F}_4 = 0$.
\begin{equation*}
\begin{tikzpicture}[baseline=(current  bounding  box.center)]
\node (ldots) {$\cdots$};
\node (H7M) [right=1cm of ldots] {$H^7(M)$};
\node (H7E) [right=1cm of H3M] {$H^7(E)$};
\node (H4M) [right=1cm of H3E] {$H^4(M)$};
\node (H8M) [right=1cm of H2M] {$H^8(M)$};
\node (H) [below=0.7cm of H7E] {$H_7$};
\node (hatF) [below=0.62cm of H4M] {$\widehat{F}_4$};
\node (0) [below=0.72cm of H8M] {$0$};
\draw[->] (ldots) to node {}  (H7M);
\draw[->] (H7M) to node {\footnotesize $\pi^{\ast}$} (H7E);
\draw[->] (H7E) to node {\footnotesize $\pi_{\ast}$} (H4M);
\draw[->] (H4M) to node {\footnotesize $F_4 \cup $}  (H8M);
\draw[|->] (H) to node {}  (hatF);
\draw[dashed,|->] (hatF) to node {}  (0);
\end{tikzpicture}
\end{equation*}
If a spherical T-dual exists, it is another principal $SU(2)$ bundle with second Chern class $\widehat{F}_4 = \pi_{\ast} H_7$. It therefore determines its own Gysin sequence:
\begin{align*}
\cdots \overset{}{\longrightarrow} H^7 (M,\ZZ) \overset{\hat{\pi}^{\ast}}{\longrightarrow} H^{7} (\widehat{E},\ZZ) \overset{\hat{\pi}_{\ast}}{\longrightarrow} H^{4} (M,\ZZ) \overset{\widehat{F}_4 \cup}{\longrightarrow} H^{8} (M,\ZZ) \longrightarrow \cdots 
\end{align*}
We have already seen that $\widehat{F}_4 \cup F_4 = F_4 \cup \widehat{F}_4 = 0$, so by exactness of the Gysin sequence for $\hat{\pi}: \widehat{E} \to M$, we must have that $F_4$ is in the image of $\hat{\pi}_{\ast}$:
\begin{equation*}
\begin{tikzpicture}[baseline=(current  bounding  box.center)]
\node (ldots) {$\cdots$};
\node (H7M) [right=1cm of ldots] {$H^7(M)$};
\node (H7E) [right=1cm of H3M] {$H^7(E)$};
\node (H4M) [right=1cm of H3E] {$H^4(M)$};
\node (H8M) [right=1cm of H2M] {$H^8(M)$};
\node (hatH) [below=0.62cm of H7E] {$\widehat{H}_7$};
\node (F) [below=0.7cm of H4M] {$F_4$};
\node (0) [below=0.72cm of H8M] {$0.$};
\draw[->] (ldots) to node {}  (H7M);
\draw[->] (H7M) to node {\footnotesize $\hat{\pi}^{\ast}$} (H7E);
\draw[->] (H7E) to node {\footnotesize $\hat{\pi}_{\ast}$} (H4M);
\draw[->] (H4M) to node {\footnotesize $\widehat{F}_4 \cup $}  (H8M);
\draw[dashed,|->] (hatH) to node {}  (F);
\draw[|->] (F) to node {}  (0);
\end{tikzpicture}
\end{equation*}
That is, if the dual bundle exists then there is \emph{some} $\widehat{H}_7$ satisfying (\ref{F4H7}). By assuming, as in the abelian case, that the part of the flux living on the base in unchanged after duality, the authors of \cite{BEM14a} were able to prove uniqueness of this dual flux. This is the $SU(2)$ analogue of Theorem \ref{theorem:topT-dual}. 
\begin{theorem}[Bouwknegt, Evslin, Mathai \cite{BEM14a}]
Let $(F_4,H_7)$ be a pair corresponding to a principal $SU(2)$ bundle $\pi:E \to M$ with a flux. Suppose there exists a principal $SU(2)$ bundle $\hat{\pi}:\widehat{E} \to M$ such that the second Chern class of $\widehat{E}$ is given by $\widehat{F}_4 = \pi_{\ast} H_7$. Then:
\begin{enumerate}
\item (Existence) there exists an $\widehat{H}_7 \in H^7 (\widehat{E},\ZZ)$ satisfying (\ref{dualF4}) and 
\begin{align*}
\hat{p}^{\ast} H_7 = p^{\ast}\widehat{H}_7
\end{align*}
\item (Uniqueness) $\widehat{H}_7$ is unique up to the addition of a term $\hat{\pi}^{\ast} \left( F_4 \cup \gamma \right)$ with $\gamma \in H^3(M,\ZZ)$.
\end{enumerate}
\end{theorem}
We will not prove this theorem here,\footnote{The proof is not overly difficult, and involves some diagram chasing as in Theorem \ref{theorem:topT-dual}. The interested reader can find the proof in the original paper \cite{BEM14a}.} but we will pause briefly to make a few comments on the difference between this theorem and the analogous result for $S^1$ bundles, Theorem \ref{theorem:topT-dual}. Unlike the case for abelian T-duality, we are not guaranteed that a T-dual will always exist, and even if it does exist, we are not guaranteed that it will be unique. Another difference is the ambiguity of the dual flux in the two cases. For abelian T-duality, the ambiguity is an element of $H^1 (M,\ZZ)$, and since $H^1(M,\ZZ)$ is isomorphic to $[M,U(1)]$, this ambiguity is in one-to-one correspondence with automorphisms of the bundle. For the non-abelian case, the ambiguity is an element of $H^3(M,\ZZ)$, which in general is not isomorphic to $[M,SU(2)]$. That is, we can't always think of the ambiguity as coming from a bundle automorphism, and there are bundle automorphisms which do not arise as an ambiguity of the flux.

\subsection{Examples: spherical T-duality}
\label{sec:examplesspherical}
In this section we will discuss different examples of spherical T-duality. 
\subsubsection{The quaternionic Hopf fibration}
The first example we will discuss is, in some sense, the prototypical example of spherical T-duality. Indeed, this example is the higher-dimensional version of the Hopf fibration abelian T-duality described in Section \ref{S3Topology}. It, together with the next example, were discussed in the original paper \cite{BEM14a}. The total space is $S^7$, and the base space is $S^4$, so this bundle is the `quaternionic Hopf fibration' mentioned in Appendix \ref{app:Hopf}. We will assume that there is no flux.\footnote{The more general situation with flux is described in the next example.}

The base space, $S^4$, has cohomology $H^4(S^4,\ZZ) = \ZZ$. Since $S^4$ is a four-dimensional manifold, we know therefore that principal $SU(2)$ bundles over $S^4$ are classified by the integers. The quaternionic Hopf fibration corresponds to the generator of $H^4(S^4,\ZZ)$. 

We can compute the second Chern class of $S^7$ by rewriting the round metric in terms of a metric on the base $S^4$ and a principal $SU(2)$ connection:
\begin{align*}
\dd s^2 &= \dd s^2_{S^4} + A^2.
\end{align*}
Following \cite{BEM14a}, we can use quaternionic coordinates 
\begin{align*}
q_1 &= \cos \theta q \\[1em]
q_2 &= \sin \theta p q,
\end{align*}
where $p$ and $q$ are unit quaternions, identified with elements of $S^3$. The metric on the round unit sphere in $\RR^8 \simeq \HH \times \HH$ is given by
\begin{align*}
\dd s^2 &= |\dd q_1|^2 + |\dd q_2|^2 \\[1em]
&= \dd \theta^2 + \tfrac{1}{4} \sin^2 2\theta \left|  \bar{p} \dd p \right|^2 + \left| \bar{q} \dd q + \sin^2 \theta \bar{q} (\bar{p} \dd p) q \right|^2 , 
\end{align*}
from which we identify
\begin{align*}
\dd s^2_{S^4} &= \dd \theta^2 + \tfrac{1}{4} \sin^2 2\theta \left|  \bar{p} \dd p \right|^2 
\end{align*}
and
\begin{align*}
A &= \bar{q} \dd q + \sin^2 \theta \bar{q} (\bar{p} \dd p) q. 
\end{align*}
The second Chern class is then
\begin{align*}
c_2 (S^7) &= \frac{1}{8 \pi^2} \int_{S^4} \Tr \left( F \wedge F \right) = 1.
\end{align*}
Thus it follows that the second Chern class of $S^7$ is the generator of $H^4 (S^4,\ZZ)$. The second Chern class of the dual bundle is determined by 
\begin{align*}
\pi_{\ast} H_7 = 0,
\end{align*}
from which it follows that the dual bundle is trivial. That is, $\widehat{E} = S^3 \times S^4$. The dual flux is uniquely determined, since $H^3(S^4,\ZZ)$, and is given by the generator of $H^7 (S^3 \times S^4,\ZZ)$. A de-Rham representative for this is $\widehat{H}_7 = \dd V_{S^3 \times S^4}$. The dual is therefore a trivial bundle with one unit of flux, directly paralleling the abelian Hopf T-duality of $S^3$.
\subsubsection{All $SU(2)$ bundles over $S^4$}
The previous example is naturally generalised. We take as total space, $E$, the $SU(2)$ bundle over $S^4$ with second Chern class equal to $k$ times the generator, and take an initial flux which is $j$ times the volume form on $E$. Note that since $H^4 (S^4,\ZZ) = H^7 (E,\ZZ) = \ZZ$, this exhausts all possible examples with base $S^4$. The previous example of $S^7$ with no flux corresponds to $k=1$ and $j = 0$. 

Duality in this situation corresponds, as with the Lens space with flux example of Section \ref{Lenstopology}, to interchanging $k$ and $j$, so that the dual bundle has second Chern class equal to $j$ times the generator, together with $k$ units of flux. That is,
\begin{align*}
F_4 = k & & \widehat{F}_4 = j \\[1em]
H_7 = j & & \widehat{H}_7 = k.
\end{align*}
In particular, $S^7$ with one unit of flux is self-dual. 
\subsubsection{$SU(2)$-bundles over $S^5$: What goes wrong?}
To see an example where $H^4(M,\ZZ)$ does not classify the bundle, consider an $S^5$ base. We have $H^4(S^5,\ZZ) = 0$, but there are in fact two principal $SU(2)$ bundles over $S^5$. To see why, we again use (\ref{piclassifying}). Principal $SU(2)$ bundles over $S^5$ are classified by
\begin{align*}
[S^5, \mathsf{B}SU(2)] = \pi_5 (\mathsf{B}SU(2)) = \pi_4 (SU(2)) = \ZZ_2,
\end{align*}
and so it follows that there are two such bundles. The two bundles are the trivial one, $S^3 \times S^5$, and the non-trivial bundle, which happens to be diffeomorphic to $SU(3)$. Note that both of these bundles have vanishing second Chern class. Let us consider this example and see if we can perform a spherical T-duality.
The bundle is 
\begin{center}
	\begin{tikzpicture}
	\node (F) {$S^3$};
	\node (E) [right=1cm of F] {$SU(3)$};
	\node (M) [below=1cm of E] {$S^5$};
	\draw[right hook->] (F) to node  {$ $} (E);
	\draw[->] (E) to node [right] {$\pi$} (M);
	\end{tikzpicture}
\end{center}
In order to see what the allowed fluxes are, we need to compute $H^7 (SU(3),\ZZ)$. Although the manifold $SU(3)$ is not diffeomorphic to $S^3 \times S^5$, they do have the same cohomology groups. That is,
\begin{align*}
H^k (SU(3),\ZZ) = H^k(S^3 \times S^5,\ZZ).
\end{align*}
With this is mind, we can now use the K\"{u}nneth theorem (see Appendix \ref{app:cohomology})), to obtain
\begin{align*}
H^7 (SU(3),\ZZ) &= H^7 (S^3 \times S^5,\ZZ) \\[1em]
&= \Big(H^0(S^3,\ZZ) \otimes H^7(S^5,\ZZ) \Big) \oplus \Big(H^1(S^3,\ZZ) \otimes H^6(S^5,\ZZ) \Big) \\
& \oplus \Big(H^2(S^3,\ZZ) \otimes H^5(S^5,\ZZ) \Big) \oplus \Big(H^3(S^3,\ZZ) \otimes H^4(S^5,\ZZ) \Big) \\
& \oplus \Big(H^4(S^3,\ZZ) \otimes H^3(S^5,\ZZ) \Big) \oplus \Big(H^5(S^3,\ZZ) \otimes H^2(S^5,\ZZ) \Big) \\
& \oplus \Big(H^6(S^3,\ZZ) \otimes H^1(S^5,\ZZ) \Big) \oplus \Big(H^7(S^3,\ZZ) \otimes H^0(S^5,\ZZ) \Big) \\[1em]
&= 0
\end{align*}
That is, there are no topologically non-trivial fluxes. A putative dual would have to have second Chern class
\begin{align*}
\widehat{F}_4 = \pi_{\ast} H_7 = \pi_{\ast} 0 = 0.
\end{align*}
Of course, we already know that $H^4 (S^5,\ZZ) = 0$, so both bundles over $S^5$ have vanishing second Chern class. Thus both $S^3 \times S^5$ and $SU(3)$ are candidates for a spherical T-dual of $SU(3)$. The allowed fluxes on both spaces are both trivial since $H^7 (SU(3),\ZZ) = H^7 (S^3 \times S^5,\ZZ) = 0$, so we have in fact three different spherical T-dualities:
\begin{align*}
SU(3) &\longleftrightarrow SU(3) \\[1em]
SU(3) &\longleftrightarrow S^3 \times S^5\\[1em]
S^3 \times S^5 &\longleftrightarrow S^3 \times S^5,
\end{align*}
where the flux vanishes in all of the spaces. Thus bundles over $S^5$ provide an explicit realisation of the claim that spherical T-duals may not be unique.
\section{Geometry and physics?}
The physical relevance, if any, of spherical T-duality is currently unclear. In this section, we offer a few comments which suggest that this is an interesting link worth exploring. 
\subsection{A physical duality?}
 The lack of existence and uniqueness results for spherical T-duality might lead one to suspect that it has no physical basis. On the other hand, it should be noted that for 11-dimensional supergravity compactifications, we are primarily interested in (possibly warped) )backgrounds of the form $M_4 \times E_7$. If $E_7$ is an $SU(2)$ bundle, then the base space of the bundle has dimension 4, and it is precisely in that situation in which spherical T-duality has interesting results. Although motivated from a completely mathematical perspective, it is interesting to note that the ingredients appearing in spherical T-duality all have counterparts in 11 dimensional supergravity compactifications. The massless bosonic field content of 11-dimensional supergravity consists of the graviton, together with a 4-form flux and the dilaton. The magnetic dual to the 4-form flux is a 7-form flux, and the objects that couple to these fluxes are the M2 and the M5 branes. Wrapping an $M5$ brane around a 3 cycle, such as an $SU(2)$ fiber, leaves an $M2$ brane on the transverse space. 
 
 We can try to identify how spherical T-duality works by looking at the simplest non-trivial example and trying to find a supergravity background in which that space occurs. The simplest (non-trivial) examples of spherical T-duality are the ones that we have discussed previously, namely the $S^7/\ZZ_k$ spaces with $N$ units of flux. These spaces occur in the $AdS_4 \times S^7/\ZZ_k$ supergravity solutions, describing the near horizon limit of a stack of $N$ $M2$-branes wrapping a $\CC^4/\ZZ_k$ orbifold. The curvature of this supergravity solution needs to be supported by a flux, and so there is a 4-form flux living on the $AdS_4$ space, or equivalently, the magnetic dual 7-form flux living on the $S^7/\ZZ_k$. 
 
These solutions can be discussed in the context of the $AdS/CFT$ correspondence, where it forms one of the few well-studied examples. In this context, the conjecture, originally stated in \cite{ABJM}, says that $M$-theory on $AdS_4 \times S^7/\ZZ_k$ with $N$ units of RR flux living on the $AdS$ space is dynamically equivalent to $\mathcal{N} = 6$ superconformal Chern-Simons matter theory in $2+1$ dimensions with gauge group $U(N) \times U(N)$ and Chern-Simons levels $k$ and $-k$. 
 
Recall from Section \ref{sec:examplesspherical} that $S^7/\ZZ_k$ with $N$ units of flux was spherically T-dual to $S^7/\ZZ_N$ with $k$ units of flux. If spherical T-duality is a symmetry of M-theory, then it makes a corresponding prediction for the ABJM theories - namely that ABJM with gauge group $U(N) \times U(N)$ at level $(k,-k)$ is equivalent to ABJM with gauge group $U(k) \times U(k)$ at level $(N,-N)$. An exploration of this rank/level duality conjecture is underway, and we hope to report on progress soon. 
 \subsubsection{A small puzzle}
 If we are to think of spherical T-duality as a generalisation of the usual notion of abelian T-duality, then we might expect that it also interchanges some generalised notion of momentum and winding modes. For abelian T-duality, wrapping a closed string around the $S^1$  fiber gives rise to an integer winding mode, since $\pi_1 (S^1) = \ZZ$. On the other hand, if we were to wrap a closed (spherical) 5 brane around the $S^3$ fiber, we might expect to get winding modes of the form $\pi_5 (S^3) = \ZZ_2$. Do we see these winding modes in any physical solutions? Of course, there are many closed 5-manifolds other than $S^5$, so the spectrum of winding modes should be much richer than just this. Note that although there is no corresponding notion of $S^3$ momentum, there are wrapping modes for the M2 brane. If spherical T-duality is related to U-duality, as we conjecture, then we might expect the interchange of M2 and M5 wrapping modes. 

\subsection{Geometry from U-duality and exceptional field theory}
Topological T-duality in the abelian case was motivated by Buscher's transformation rules of the metric and $B$-field. Topological T-duality can be thought of as the `topological shadow' of the geometric duality. Spherical T-duality, in contrast, was simply a mathematical generalisation of abelian topological T-duality, and wasn't motivated by any underlying geometric transformation rules. If spherical T-duality is in fact a symmetry of M-theory and/or supergravity, then we expect that there are some underlying geometric transformation rules for which spherical T-duality is the `topological shadow'. Determining such a set of spherical Buscher rules would be a mutually beneficial exercise for both mathematicians and physicists. For mathematicians, finding a physical and/or geometric realisation of the topological duality could help shed light on some of the differences between spherical T-duality and its abelian counterpart, such as lack of existence/uniqueness. Spherical Buscher rules would be of interest to physicists as well, since such a set of rules would provide a putative new duality in M-theory/supergravity, and could potentially be used as a new solution generating technique. In this section we will outline a proposed method to derive such a set of spherical Buscher rules.

\subsubsection{Abelian Buscher rules from generalised geometry}
Before we delve into a proposal to extract Buscher rules for spherical T-duality, let is briefly mention another way we can extract Buscher rules for abelian T-duality, using the convenient language of generalised geometry. 

Generalised geometry, first introduced by Hitchin \cite{H02}, and elucidated in Gualtieri's now-famous thesis \cite{Gphd}, extends the standard notion of geometry by replacing the tangent bundle $TM$ of a $d$-dimensional manifold $M$ with the generalised tangent bundle $TM \oplus T^{\ast}M$.\footnote{We will use the shorthand $T \oplus T^{\ast}$ for $TM \oplus T^{\ast}M$ when the manifold is understood.} There is a natural inner product on sections of this bundle, $\langle \cdot , \cdot \rangle$ given by pairing the vectors with the one forms:
\begin{align*}
\braket{X + \xi, Y + \zeta} = \tfrac{1}{2} \Big( X(\zeta) + Y(\xi)  \Big).
\end{align*}
This inner product has an associated matrix
\begin{align*}
\eta &= 
\left( 
\begin{matrix}
0 & \mathds{1} \\
\mathds{1} & 0
\end{matrix}
\right),
\end{align*}
and is invariant under the special orthogonal group $SO(d,d)$. Integrability questions of structures on this bundle can be elegantly phrased in terms of the Courant bracket: 
\begin{align*}
[X+\xi, Y+ \zeta]_C = [X,Y] + \Lie_X \zeta - \Lie_Y \dd \xi - \tfrac{1}{2} \dd \big( \iota_X \zeta - \iota_Y \xi \big) 
\end{align*}
The bundle $T \oplus T^{\ast}$, together with the natural inner product and the Courant bracket, is an example of a Courant algebroid. 

Why is this interesting to physicists? Well, it turns out that the Courant bracket admits more symmetries than just diffeomorphisms, in stark contrast to the Lie bracket. The symmetries of the Courant bracket turn out to be precisely the symmetries on the non-linear sigma model! That is, the group acting by bundle automorphisms of $T \oplus T^{\ast}$ preserving the Courant algebroid structure is a semi-direct product of the diffeomorphisms of $M$, $\textrm{Diff}(M)$, and $B$-field transformations, $\Omega^2_{cl}(M)$. The Courant bracket can also be twisted by a closed three-form $H$ by defining
\begin{align*}
[X+\xi, Y + \zeta]_H = [X+\xi , Y + \zeta]_C + \iota_Y \iota_X H.
\end{align*}
The generalised geometry of $T \oplus T^{\ast}$ is therefore a natural framework to describe sigma models, and by extension, (abelian) T-duality. 

In differential geometry, we are able to describe geometric structures on a manifold $M$ by defining structures on the tangent bundle. For example, an almost complex manifold is a manifold with a smooth vector bundle isomorphism 
\begin{align*}\mathcal{J}: TM \to TM
\end{align*}
which squares to the identity $\mathcal{J}^2 = -1$. Questions of integrability for this structure then lead us into complex geometry. Extending this to the generalised tangent bundle, a generalised almost complex structure is simply an almost complex structure, $\mathcal{J}$, for $T \oplus T^{\ast}$.\footnote{We also want the inner product on $T \oplus T^{\ast}$ to be preserved by $\mathcal{J}$.} If this is integrable, that is, if $\mathcal{J}$ plays nicely with the Courant bracket, then we have a generalised complex structure. Generalised complex structures are remarkably interesting from a mathematical point of view, since they include as special cases both complex structures and symplectic structures, and indeed interpolate between the two. Other generalised structures can be defined, such as a generalised metric, generalised K\"{a}hler structures, and generalised contact structures. Of these, the generalised metric is most relevant to our current discussion. 

Recall that the generalised tangent bundle has an $O(d,d)$ action leaving the inner product invariant. It is useful (and well-supported by T-duality in string theory) to consider more general bundles $E$ over $M$ with structure group $O(d,d)$, that is, bundles over $M$ where the transition functions are in $O(d,d)$. The maximal compact subgroup of $O(d,d)$ is $O(d) \times O(d)$, and we can consider sub-bundles of $E$ which have this as their structure group. A choice of such a sub-bundle is a reduction of the structure group, and is equivalent to choosing a sub-bundle $E^+$ on which $\eta$ is positive definite. It follows that $\eta$ is negative definite on the orthogonal complement $E^-$, and we have $E = E^+ \oplus E^-$. We can now define a positive definite metric, $\mathbb{G}$, on $E$ by taking
\begin{align*}
\mathbb{G} = \eta_{E^+} - \eta_{E^-}.
\end{align*}
We call $\GG$ the generalised metric, and a choice of generalised metric is in fact equivalent to the reduction of the structure group. The moduli space of such reductions is given by the coset 
\begin{align*}
\frac{O(d,d)}{O(d) \times O(d)}.
\end{align*}
At a point $x \in M$, this coset has $d^2$ independent components, and can be parametrised in terms of a symmetric matrix $g$, and an antisymmetric matrix $B$ as
 \begin{align}
\label{genmetric}
\mathbb{G} &= \left( 
\begin{matrix}
g_{ij} - B_{ik} g^{kl} B_{lj} & B_{il} g^{ln} \\
-g^{mk} B_{kj} & g^{mn}
\end{matrix}
\right) 
\end{align}

There is a beautiful picture involving generalised geometry and T-duality, as described in \cite{CG}.\footnote{For an extension of this concept to Heterotic string theory see \cite{BH13}.} Within this framework, we can view (abelian) T-duality as an isomorphism of the Courant algebroids associated to the original space and its T-dual. In the simplest case with $\textrm{dim}(M) = 1$, the isomorphism just interchanges $TM$ with $T^{\ast}M$. This isomorphism allows us to transport structures on the generalised tangent bundle from one space to another. Thus, although T-duality may not preserve a given geometric structure such as a Sasaki-Einstein structure, (see Section \ref{subsec:SasakiEinstein}), there are generalised geometric structures which are naturally preserved. To rederive the Buscher rules for a given T-duality transformation in this context, we first need to identify the isomorphism of the Courant algebroids associated to that T-duality. This isomorphism then induces a map on the generalised metric (\ref{genmetric}), and from the image of this map we can read of the dual metric and $B$-field.

\subsubsection{Spherical Buscher rules from exceptional geometry}
For the purposes of this thesis, generalised geometry corresponds to replacing the tangent bundle $T$ with the generalised tangent bundle $T \oplus T^{\ast}$, as well as introducing the Courant bracket. A reduction of the structure group to the maximal compact subgroup gives a generalised metric parametrised in terms of the metric and $B$-field, and the group $O(d,d)$ acts on the generalised metric by diffeomorphisms and $B$-field transformations. T-duality corresponds to an isomorphism between two Courant algebroids, and the induced map on the generalised metric allows one to recover the transformation rules of the metric and $B$-field, that is, the Buscher rules. 

While Courant algebroids seem to be the appropriate algebraic structure to discuss the geometry of certain flux compactifications in type II string theory and supergravity, the field content of 11-dimensional supergravity is different, and so it seems we should consider other algebraic structures. Exceptional generalised geometry is a conjectural way to extend the generalised geometry of type II to M-theory/11-dimensional supergravity. It is based on the observation that compactifications of 11-dimensional supergravity to $n$ internal dimensions has a symmetry group related to the exceptional Lie groups $E_n$ \cite{H07,BP,BGPW}. This symmetry is referred to as U-duality. 

We are interested, for the moment, in 11-dimensional supergravity compactifications of the form $\mathcal{M}_4 \times E$. The field content is given by a metric $g$ and a closed 4-form flux $F_4$. Locally, we can define a 3-form potential $C_3$ such that $F_4 = \dd C_3$. We can also consider the magnetic dual to the flux, defined by $F_7 = \star F_4$. This is not closed, but satisfies the Bianchi identity:
\begin{align*}
\dd F_7 + \tfrac{1}{2} F_4 \wedge F_4 = 0.
\end{align*}
That is, 
\begin{align*}
\dd \big( F_7 + \tfrac{1}{2} C_3 \wedge F_4 \big) = 0.
\end{align*}
It follows that, locally, we can write 
\begin{align*}
F_7 &= \dd C_6 - \tfrac{1}{2} C_3 \wedge F_4,
\end{align*}
for some 6-form potential $C_6$. The correct algebraic structure should therefore have the action of some group on it, which acts by diffeomorphisms together with gauge transformations of the gauge potentials $C_3$ and $C_6$. A well-studied candidate is the vector bundle
\begin{align}
\label{Mgeometry}
TE \oplus \Lambda^{2} T^{\ast} E \oplus \Lambda^5 T^{\ast} E.
\end{align}
The 2-forms and 5-forms correspond to 2-brane and 5-brane charges respectively. This bundle appears when considering the $E_6$ exceptional duality group of M-theory \cite{BGPW}. It has a natural $E_6$ action on it, and the reduction of the structure group to the maximal subgroup $SL(6,\RR)$, that is the generalised metric $\mathbb{G}$, is parametrised by a metric $g$, a 3-form $C_3$, and a 6-form $C_6$.\footnote{The generalised metric is essentially a more complicated version of (\ref{genmetric}).} There is a (Dorfman) bracket, $\{\cdot , \cdot\}$, on sections of this bundle, invariant under diffeomorphisms and gauge transformations by closed 3-forms and 6-forms, making it into a Leibniz algebroid \cite{B12}:
\begin{align*}
\{X + \sigma_2 + \sigma_5, Y+ \tau_2 + \tau_5\} &= [X,Y] \\
&\quad + \Lie_X \tau_2 - \iota_Y \dd \sigma_2 \\
&\quad + \Lie_X \tau_5 - \iota_Y \dd \sigma_5 + \dd \sigma_2 \wedge \tau_2.
\end{align*}
We can also twist this with a 4-flux and a 7-flux, giving the twisted bracket:
\begin{align*}
\{X + \sigma_2 + \sigma_5, Y+ \tau_2 + \tau_5\}_{ F_4,F_7 } &= [X,Y] \\
&\quad + \Lie_X \tau_2 - \iota_Y \dd \sigma_2  + \iota_X \iota_Y F_4 \\
&\quad + \Lie_X \tau_5 - \iota_Y \dd \sigma_5 + \dd \sigma_2 \wedge \tau_2 + \iota_X F_4 \wedge \tau_2 + \iota_X \iota_Y F_7,
\end{align*}
provided that 
\begin{align*}
\dd F_4 &= 0 \\[1em]
\dd F_7 + \tfrac{1}{2} F_4 \wedge F_4 &= 0.
\end{align*}
Our prescription now is relatively straightforward, albeit technically difficult. We want to consider a compactification of the form $\mathcal{M}_4 \times E$, where $E$ is a seven dimensional manifold that is also a principal $SU(2)$ bundle over a base $M$:
\begin{center}
	\begin{tikzpicture}
	\node (F) {$S^3$};
	\node (E) [right=1cm of F] {$E$};
	\node (M) [below=1cm of E] {$M$};
	\draw[right hook->] (F) to node  {$ $} (E);
	\draw[->] (E) to node [right] {$\pi$} (M);
	\end{tikzpicture}
\end{center}
Since $\textrm{dim}(M) = 4$, principal $SU(2)$ bundles over $M$ are classified by the second Chern class $F_4 \in H^4(M,\ZZ)$, and we want to assume that there is a 7-flux $H_7 \in H^7(E,\ZZ)$, living on $E$. Note that there is a distinction here between the 4-flux determining the topology of the bundle, and the 4-flux of the supergravity solution. The 4-flux coming from the supergravity side defines the 7-flux through the Hodge dual: $F_7 = \star H_4$. Assuming that $F_7$ lives entirely on $E$ is equivalent to assuming that the dual 4-flux lives on $\mathcal{M}_4$. We also want to assume that we have an $SU(2)$-invariant generalised metric $\mathbb{G}$ parametrised by $SU(2)$-invariant fields $\{g,C_3,C_6\}$. We must then identify how spherical T-duality acts as a Leibniz algebroid isomorphism on (\ref{Mgeometry}). Under this isomorphism, the induced map on the $SU(2)$ invariant generalised metric $\mathbb{G}$ should give a generalised metric for the dual bundle from which the component fields can be read off. The new fields $\{\widehat{g},\widehat{C}_3,\widehat{C}_6\}$, expressed in terms of the original fields are precisely the sought-after Buscher rules. 

Note that there is another vector bundle we could consider, namely
\begin{align*}
T \oplus \Lambda^{2} T^{\ast}  \oplus \Lambda^5 T^{\ast}  \oplus \Lambda^6 T.
\end{align*}
The additional factor in the bundle is to account for Kaluza-Klein monopole charge \cite{H97,H07}. When the dimension of $E$ is 7, as it is in our case,  bundle is related to
\begin{align*}
T \oplus \Lambda^{2} T^{\ast}  \oplus \Lambda^5 T^{\ast}  \oplus \left( \Lambda^7 T^{\ast} \otimes T^{\ast}\right).
\end{align*}
The latter bundle has a a canonical Dorfman bracket on its sections making it a Leibniz algebroid. It also has a natural $E_7 \times \RR^{\ast}$ structure that includes transformations by closed 3-forms and 6-forms \cite{B12}. The bracket is given by 
\begin{align*}
\{X + \sigma_2 + \sigma_5 + u, Y+ \tau_2 + \tau_5\ + w \} &= [X,Y] \\
&\quad + \Lie_X \tau_2 - \iota_Y \dd \sigma_2 \\
&\quad + \Lie_X \tau_5 - \iota_Y \dd \sigma_5 + \dd \sigma_2 \wedge \tau_2 \\
&\quad + \Lie_X w - \dd \sigma_2 \diamond \tau_5 + \dd \sigma_5 \diamond \tau_2,
\end{align*}
where $(\alpha \diamond \beta)(X) = \iota_X \alpha \wedge \beta$. This bracket can be twisted by a 4-flux and a 7-flux satisfying 
\begin{align*}
\dd F_4 &= 0 \\[1em]
\dd F_7 + \tfrac{1}{2} F_4 \wedge F_4 &= 0.
\end{align*}
The twisted bracket is given by
\begin{align*}
&\{X + \sigma_2 + \sigma_5 + u, Y+ \tau_2 + \tau_5\ + w \} \\
&\quad = [X,Y] \\
&\quad + \Lie_X \tau_2 - \iota_Y \dd \sigma_2 +\iota_X \iota_Y F_4 \\
&\quad + \Lie_X \tau_5 - \iota_Y \dd \sigma_5 + \dd \sigma_2 \wedge \tau_2+ \iota_X \iota_Y F_7 + \iota_X F_4 \wedge \tau_2 \\
&\quad + \Lie_X w - \dd \sigma_2 \diamond \tau_5 + \dd \sigma_5 \diamond \tau_2 -(\iota_X F_4) \diamond \tau_5 + (\iota_X F_7) \diamond \tau_2.
\end{align*}
Generalised metrics for the $E_6$ and $E_7$ geometries have been studied in \cite{BGPW}.

\chapter{Conclusions and Outlook}
\label{chap:conc}
In this thesis we have reviewed and discussed mathematical and physical aspects of T-duality and its many generalisations. In particular, we discussed abelian T-duality, non-abelian T-duality, Poisson-Lie T-duality, non-isometric T-duality, and spherical T-duality. This thesis is largely pedagogical, providing the interested reader with an overview of T-duality in string theory, and giving contrasts and comparisons between the standard notion of T-duality and the various generalisations. The bulk of the novel results of this thesis are in Chapter \ref{chptr:Ch5}, where we discuss the recent proposal of non-isometric T-duality. We proved that the original  proposal is equivalent to the standard non-abelian T-duality, although this equivalence is certainly not obvious from the outset. We then discussed generalisations of this proposal, and concluded with some tantalising hints that this generalisation could be related to Poisson-Lie T-duality.

The results of this thesis offer an outlook on what the author considers to be the three main open problems in T-duality. The first of these is the longstanding open problem of the global nature of non-abelian T-duality. Since its introduction, the topological nature of the non-abelian dual space has been unclear. Recent work in the physics community has made some interesting headway in this direction, although we are far from approaching a comprehensive understanding. The observations of Chapter \ref{chptr:Ch3} put this open problem in a new perspective, suggesting that the appropriate framework to think of the global structure of the dual space is as some generalisation of a manifold such as a T-fold (or an appropriate non-abelian modification). It would be interesting to frame this duality in the \CA-algebraic framework of Section \ref{Cstar}.

The second open problem concerns the relation between the generalisation of non-isometric T-duality discussed in Chapter \ref{chptr:Ch5}. Poisson-Lie T-duality has attracted considerable attention from physicists, partly due to its appearance and application within integrable systems. The full nature of Poisson-Lie T-duality is not currently understood. The results of Chapter \ref{chptr:Ch5} strongly suggest that it can be described within a gauging framework, which would certainly help elucidate the true nature of the duality and its status as a duality in string theory. Further work on the properties of this non-isometric gauging, and in particular on the properties of the proposed non-isometric field strength, is currently in progress.

The third and final open problem is the geometric nature of spherical duality. The topological results of spherical T-duality so closely mirror the topological results of abelian T-duality precisely in the dimensions required by M-theory that one would be surprised if M-theory made no use of this putative duality. In Chapter \ref{chptr:Ch6} we reviewed spherical T-duality and outlined a proposal for deriving a set of Buscher rules for spherical T-duality based on exceptional geometry. Detailed calculations of this proposal are currently in progress, and we expect to include them in an upcoming publication.

%% file: appendix.tex
\begin{appendices}

\chapter{An ode to $S^3$}
\label{chptr:App1}

The three sphere is undoubtedly one of the most interesting objects in mathematics. It is a (real) sphere, a complex sphere, a Lie group, the unit quaternions, a sphere bundle, a principal $G$-bundle, and a parallelizable manifold to name a few of its guises. It is also, perhaps, the most prevalent example appearing in all of the various T-dualities, and so features prominently in this thesis. In this appendix, we describe a few of the properties of $S^3$ which are of relevance to our work. References will be sparse since these are widely known results, collated here only for convenience.

\section{Coordinate descriptions}
There are many coordinate descriptions of $S^3$, each of which have their advantages and disadvantages. Note that since $S^3$ is not globally $\RR^3$, there are no global system of coordinates for $S^3$. It follows that every coordinate description of the round metric in this section is defined only in some open patch of $S^3$.
\subsection{Cartesian coordinates}
By far the most straightforward description of $S^3$ is as the unit sphere in $\RR^4$. By definition, we have 
\begin{align}
S^3 = \{ (x_1,x_2,x_3,x_4) \in \RR^4 : x_1^2 + x_2^2 + x_3^2 + x_4^2 = 1 \}.
\end{align}
Alternatively, we may think of $\RR^4 \simeq \CC^2$, and write 
\begin{align}
z_1 &= x_1+ix_2 \\
z_2 &= x_3+ix_4.
\end{align}
It follows straightforwardly that $S^3$ is also the unit sphere in $\CC^2$:
\begin{align}
S^3 = \{ (z_1,z_2) \in \CC^2: |z_1|^2 + |z_2|^2 = 1 \}.
\end{align}
$S^3$ is a 3-manifold, so locally we only need three coordinates to describe it. We can eliminate any one of the coordinates by using the defining equation. For example, we can rewrite $x_4 = \pm \sqrt{1 - x_1^2 - x_2^2 - x_3^2}$. 
The flat metric on $\RR^4$ induces a metric on $S^3$ known as the round metric. To obtain it in a patch with $x_4 \not=0$, we apply the exterior derivative to the defining equation to get 
\begin{align}
x_1 \dd x_1 + x_2 \dd x_2 + x_3 \dd x_3 + x_4 \dd x_4 = 0.
\end{align}
Rearranging for $\dd x_4$, and substituting into the flat metric
\begin{align}
\dd s^2 = \dd x_1^2 + \dd x_2^2 + \dd x_3^2 + \dd x_4^2,
\end{align}
we get the induced metric on $S^3 \setminus \{x_4 = 0\}$:
\begin{align}
\dd s^2_{S^3} = \dd x_1^2 + \dd x_2^2 + \dd x_3^2 + \frac{1}{1-x_1^2-x_2^2-x_3^2}\left( x_1 \dd x_1 + x_2 \dd x_2 + x_3 \dd x_3 \right)^2.
\end{align}
We can also apply the same procedure to get coordinate descriptions of the metric in other patches. 
\subsection{Stereographic coordinates}
The cartesian description of $S^3$ is defined as an embedding in $\RR^4$, but it is often useful to have a three dimensional description of $S^3$. To achieve this, we can use stereographic projection. From a pole $(1,0,0,0)$ on $S^3$ we project onto the equatorial $\RR^3$ hyperplane by 
\begin{align}
X = \frac{x_1}{1-x_4}, \qquad Y = \frac{x_2}{1-x_4}, \qquad Z = \frac{x_3}{1-x_4}.
\end{align}
The inverse of this map is 
\begin{align}
x_1 &= \frac{2X}{X^2+Y^2+Z^2+1}, \qquad x_3 = \frac{2Z}{X^2+Y^2+Z^2+1}, \\[1em]
x_2 &= \frac{2Y}{X^2+Y^2+Z^2+1}, \qquad x_4 = \frac{X^2+Y^2+Z^2-1}{X^2+Y^2+Z^2+1}.
\end{align}
In these coordinates, the round metric is
\begin{align}
\dd s^2_{S^3} = \frac{4}{(1+X^2+Y^2+Z^2)^2} \left( \dd X^2 + \dd Y^2 + \dd Z^2 \right).
\end{align}
Note that we could also have projected any other point on the sphere. 
\subsection{Hyperspherical coordinates}
Hyperspherical coordinates are the generalisation of the usual polar coordinates we use in lower dimensions. There are multiple different choices, one of which is given by the transformation
\begin{align}
x_1 &= r \cos \theta \\
x_2 &= r \sin \theta \, \cos \phi \\
x_3 &= r \sin \theta \, \sin \phi \, \cos \psi \\
x_4 &= r \sin \theta \, \sin \phi \, \sin \psi. 
\end{align}
The range of these coordinates is $r \in [0, \infty)$, $\theta \in [0, 2 \pi]$, $\phi \in [0, \pi]$, and $\psi \in [0,\pi]$.
As in the case for lower dimensions, the surface with $r = 1$ is the unit sphere. The round metric is
\begin{align}
\dd s^2_{S^3} = \dd \theta^2 + \sin^2 \theta \left( \dd \phi^2 + \sin^2 \phi \dd \psi^2 \right). 
\end{align} 
\subsection{Hopf coordinates}
As the name would suggest, Hopf coordinates are most useful when describing $S^3$ as the Hopf fibration (see Section \ref{app:Hopf}). They are given by
\begin{align}
x_1 &= \cos \left(\frac{\xi_1 + \xi_2}{2}\right) \sin \eta \\
x_2 &= \sin\left(\frac{\xi_1 + \xi_2}{2}\right) \sin \eta \\
x_3 &= \cos\left(\frac{\xi_1 - \xi_2}{2}\right) \cos \eta \\
x_4 &= \sin\left(\frac{\xi_1 - \xi_2}{2}\right) \cos \eta .
\end{align}
In complex coordinates, these have a simple expression
\begin{align}
z_1 &= e^{\frac{i(\xi_1 + \xi_2)}{2}} \sin \eta \\
z_2 &= e^{\frac{i(\xi_1 - \xi_2)}{2}} \cos \eta.
\end{align}
The range of these coordinates is $\eta \in [0,\frac{\pi}{2}]$, $\xi_1 \in [0, 2\pi]$, and  $\xi_2 \in [0, \pi]$.
The round metric is
\begin{align}
\dd s^2_{S^3} = \dd \eta^2 + \frac{1}{4} \Big( \dd \xi_1^2 + \dd \xi_2^2 - 2 \cos (2 \eta) \dd \xi_1 \dd \xi_2 \Big).
\end{align}

\section{Group structure}
The sphere $S^3$ is also diffeomorphic to the Lie group $SU(2)$. To see this, recall the definition of $SU(2)$:
\begin{align}
SU(2) = \left\{ U  \in M_2 (\CC): U U^{\dagger} = U^{\dagger} U = I \textrm{ and det}(U) = 1 \right\}.
\end{align}
By writing 
\begin{align}
U = \begin{pmatrix}
z_1 & z_2 \\
z_3 & z_4
\end{pmatrix},
\end{align}
and using the unitary property, we find that for $U \in SU(2)$ we must have $z_3 = - \bar{z}_2$ and $z_4 = \bar{z}_1$. Then the determinant of $U$ is
\begin{align}
\textrm{det}(U) &= \left|
\begin{matrix}
z_1 & z_2 \\
-\bar{z}_2 & \bar{z}_1
\end{matrix}
\right| \\
&= |z_1|^2 + |z_2|^2.
\end{align}
The diffeomorphism mapping $S^3$ to $SU(2)$ is then
\begin{align}
(z_1,z_2) \mapsto \begin{pmatrix}
z_1 & z_2 \\
-\bar{z}_2 & \bar{z}_1.
\end{pmatrix}
\end{align}
There is another way we can see this correspondence, by considering the quaternions $\mathbb{H}$. An arbitrary quaternion is written as
\begin{align}
q = x_1 + x_2 \mathbf{i} + x_3 \mathbf{j} + x_4 \mathbf{k}.
\end{align}
The unit quaternions, as a manifold, have the structure of $S^3$. This follows from 
\begin{align}
|q|^2 = q q^{\ast} = x_1^2 + x_2^2 + x_3^2 + x_4^2. 
\end{align}
On the other hand, the quaternions have a group structure given by multiplication, and the unit quaternions are a subgroup. We can define a group isomorphism from the unit quaternions to $SU(2)$ by sending 
\begin{align}
q = x_1 + x_2 \mathbf{i} + x_3 \mathbf{j} + x_4 \mathbf{k} \mapsto \begin{pmatrix}
x_1 + i x_2 & x_3 + ix_4 \\
-x_3 + ix_4 & x_1-ix_2
\end{pmatrix}.
\end{align}
Note that the only spheres which posses a group structure are $S^1 = U(1)$ and $S^3 = SU(2)$.
\section{The Hopf fibration}
\label{app:Hopf}
The Hopf fibration is a description of $S^3$ as a fiber bundle. More specifically, it is a circle over $S^2$:
\begin{center}
	\begin{tikzpicture}
	\node (F) {$S^1$};
	\node (E) [right=1cm of A1] {$S^3$};
	\node (M) [below=1cm of E] {$S^2$};
	
	\draw[right hook->] (F) to node [swap] {$ $} (E);
	\draw[->] (E) to node {$ $} (M);
	\end{tikzpicture}
\end{center}
Since this is a fibration, this tells us that $S^3$ is locally trivial - i.e. $S^3$ can be locally identified with $S^1 \times S^2$. Of course, this is not true globally, as can be seen by comparing topological invariants such as the cohomology (see \ref{app:top}).

Let us explicitly construct the Hopf projection $\pi: S^3 \to S^2$ in cartesian coordinates. For $(x_1,x_2,x_3,x_4) \in S^3$, we have 
\begin{align}
\label{Hopfcartesian}
\pi(x_1,x_2,x_3,x_4) = \left( 2(x_1 x_3 + x_2 x_4),\, 2(x_2 x_3 - x_1 x_4),\, -x_1^2 - x_2^2 + x_3^2 + x_4^2 \right).
\end{align}
The right hand side of (\ref{Hopfcartesian}) is an element $(x,y,z) \in \RR^3$. A quick calculation shows that $$x^2 + y^2 + z^2 = (x_1^2+x_2^2 + x_3^2 + x_4^2)^2,$$ so that $\pi$ maps a point on $S^3 \subset \RR^4$ to a point on $S^2 \subset \RR^3$. In Hopf coordinates, the image of $\pi$ is
\begin{align}
x &= \sin(2\eta) \cos(\xi_2)\\
y &= \sin(2\eta) \sin(\xi_2)  \\
z &= \cos (2\eta), 
\end{align}
which coincides with the usual polar coordinate description of $S^2$ with inclination $\theta = 2\eta$ and azimuth $\varphi =  \xi_2$. In these coordinates, the $S^2$ base is parameterised by the coordinates $(\eta, \xi_2)$, while the $S^1$ fibers are parameterised by the coordinate $\xi_1$.

\begin{figure}
	\centering
	\includegraphics[scale=0.4]{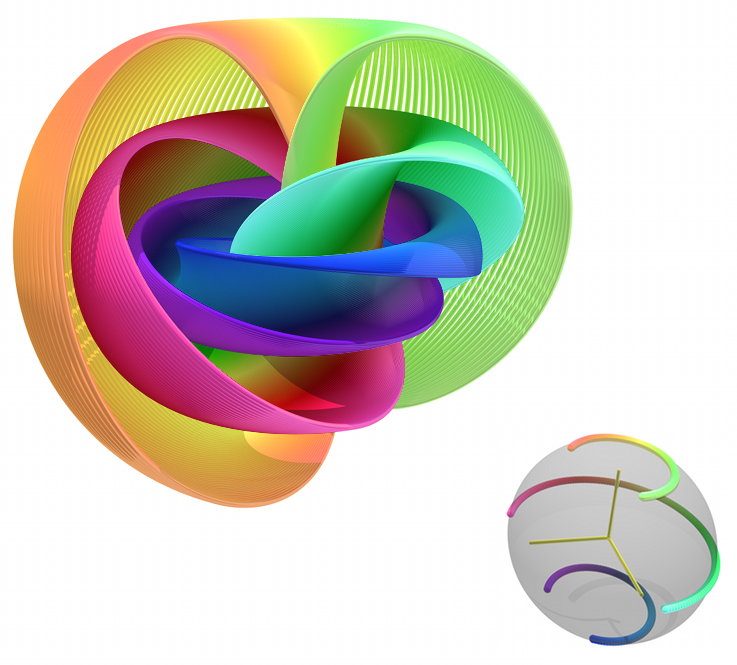}
	\caption{A visualisation of the Hopf projection. Points on the base are coloured in correspondence with the colour of the fiber in the total space. This image utilises the stereographic projection of $S^3$ to $\RR^3$.\protect\footnotemark}
\end{figure}

It is easy to see $S^3$ as a principal $U(1)$-bundle if use complex coordinates $(z_1,z_2)$. There is a natural $U(1)$ group action, acting by $(e^{i \lambda} z_1 , e^{i \lambda} z_2)$, for which it is easy to see that the quotient is $S^2$.

Fibrations of spherical fibers, total space, and base can only occur in certain dimensions. These fibrations are related to the fact that real division algebras can only occur in dimensions 1, 2, 4, and 8, corresponding to the real division algebras $\RR$, $\CC$, $\mathbb{H}$, and $\mathbb{O}$. The related fibrations are:
\begin{align}
S^1 &\hookrightarrow S^1 \rightarrow S^0\\
S^1 &\hookrightarrow S^3 \rightarrow S^2\\
S^3 &\hookrightarrow S^7 \rightarrow S^4\\
S^7 &\hookrightarrow S^{15} \rightarrow S^8.
\end{align}
\section{Geometry of the round metric}
\footnotetext{Used with permission \cite{Jpicture}.}
We have given several coordinate descriptions of $S^3$, and with all of them we have included the coordinate description of the round metric. Although there is an infinite dimensional space of Riemannian metrics we can equip $S^3$ with, the round metric is in some sense a natural one. For one thing, it is the metric induced by the defining embedding and the flat metric in $\RR^4$. Another property of the round metric is that it is a bi-invariant metric on the group $SU(2)$. Indeed, every compact Lie group admits a bi-invariant metric,\footnote{The only connected Lie groups admitting bi-invariant metrics are cartesian products of compact Lie groups and $\RR^m$.} and if the group is simple then bi-invariant metrics are unique up to a constant, so the round metric is \emph{the} bi-invariant metric on $SU(2)$. Bi-invariant metrics are useful because they have a variety of nice properties. For a bi-invariant metric, the geodesics are the integral curves of the left-invariant vector fields. Furthermore, such a metric has constant Ricci curvature. We can verify this calculation for $S^3$ in any of the coordinate systems we have described, and we find that
\begin{align}
\mathcal{R} = 6.
\end{align} 
The round metric also makes $S^3$ into an Einstein manifold. That is, the round metric satisfies 
\begin{align}
R_{ij} = 2 g_{ij}.
\end{align} 
The isometry group of the round metric is $SO(4)$. This is manifest when taking the standard embedding of $S^3$ into $\RR^4$. The Lie algebra $\mathfrak{so}(4)$ corresponding to this group of isometries decomposes into $\mathfrak{su}(2) \times \mathfrak{su}(2)$

\section{Topological properties}
\label{app:top}

\subsection{(Co)Homology groups}
\label{app:cohomology}
$S^3$ has the following \v{C}ech cohomology groups, with integer coefficients:
\begin{align}
H^0 (S^3 , \ZZ) &= \ZZ \\
H^1 (S^3 , \ZZ) &= 0 \\
H^2 (S^3 , \ZZ) &= 0 \\
H^3 (S^3 , \ZZ) &= \ZZ. 
\end{align}
Poincar\'{e} dual to this, we have the homology groups:
\begin{align}
H_0 (S^3 , \ZZ) &= \ZZ \\
H_1 (S^3 , \ZZ) &= 0 \\
H_2 (S^3 , \ZZ) &= 0 \\
H_3 (S^3 , \ZZ) &= \ZZ. 
\end{align}
More generally, for spheres we have:
\begin{align}
H^k (S^n,\ZZ) = H_k (S^n,\ZZ) = \begin{cases}
\ZZ & \textrm{if } k=0\textrm{ or }n \\
0 & \textrm{otherwise.} 
\end{cases}
\end{align}
The following theorem is useful for computing the cohomology groups of product spaces.
\begin{theorem}[K\"{u}nneth formula]
\label{Kunneth}
Let $X$ and $Y$ be topological spaces, and suppose that $H^i(X,\ZZ)$ or $H^i(Y,\ZZ)$ are torsion free for all $i$. Then
\begin{align}
H^k (X \times Y,\ZZ) = \bigoplus_{i+j = k} \left( H^i (X,\ZZ) \otimes H^j(Y,\ZZ) \right).
\end{align} 
\end{theorem}
The K\"{u}nneth formula, together with knowledge of the cohomology groups of the spheres, allows us to compute the cohomology groups of two other spaces we are very interested in: $\mathbb{T}^n$ and $S^1 \times S^2$.
\begin{corollary}
The cohomology of $S^1 \times S^2$ is 
\begin{align}
H^k (S^1 \times S^2,\ZZ) = \begin{cases}
\ZZ & \textrm{if } k=0,1,2 \\
0 & \textrm{for }k>3 
\end{cases}
\end{align}	
\end{corollary}
\begin{corollary}
	The cohomology of $\mathbb{T}^n = (S^1)^n$ is 
	\begin{align}
	H^k (\mathbb{T}^n,\ZZ) = \ZZ^{ n\choose k}.
	\end{align}	
\end{corollary}
\subsection{Homotopy groups}
\label{app:homotopy}
Some homotopy groups for spheres are shown in the following table:
\begin{center}
	\begin{tabular}{|c|c|c|c|c|c|c|c|c|c|c| } 
\hline
		$k$  & 1 & 2& 3 & 4 & 5 & 6 & 7 & 8 & 9 & 10 \\ 
\hline
		$\pi_k(S^1)$ & $\ZZ$  & 0 & 0 & 0 & 0 & 0 & 0 & 0 & 0 & 0  \\
\hline
		$\pi_k(S^2)$ & 0 & $\ZZ$ & $\ZZ$ & $\ZZ_2$ & $\ZZ_2$ & $\ZZ_{12}$ & $\ZZ_{2}$ & $\ZZ_2$ & $\ZZ_3$ & $\ZZ_{15}$  \\
\hline
		$\pi_k(S^3)$ & 0 & 0  & $\ZZ$ & $\ZZ_2$ & $\ZZ_{2}$ & $\ZZ_{12}$ & $\ZZ_2$ & $\ZZ_2$ & $\ZZ_3$ & $\ZZ_{15}$  \\
\hline
	\end{tabular}
\end{center}
The higher homotopy groups of spheres are largely unknown, and do not follow an obvious pattern. The homotopy groups for $S^2$ and $S^3$ are the same for $k\geq 3$. This can be seen by noting that the Hopf fibration
\begin{center} 
	\begin{tikzcd}
		S^1\arrow[hookrightarrow]{r}
		&S^3\arrow{d}{\pi} \\
		&S^2.
	\end{tikzcd}
\end{center}
induces the following long exact sequence of homotopy groups:
$$\cdots \longrightarrow \pi_k(S^1) \longrightarrow \pi_k(S^3) \longrightarrow \pi_k(S^2) \longrightarrow \pi_{k-1}(S^1) \longrightarrow \cdots \longrightarrow \pi_0(S^1) \longrightarrow \pi_0(S^3). $$
Since $\pi_k(S^1) = 0$ for $k\geq 2$, the sequence divides into smaller seqeuences 
\begin{align}
0 \longrightarrow \pi_k(S^3) \longrightarrow  \pi_k(S^2) \longrightarrow 0
\end{align}
for $k\geq 3$, whence it follows that $\pi_k(S^3) \cong \pi_k(S^2)$ for $k \geq 3$.

\subsection{Gysin sequences}
\label{sec:Gysin}
Sphere bundles are fiber bundles where the fiber is a sphere. Associated to every sphere bundle is a particularly useful long exact sequence known as the Gysin sequence.
\begin{theorem}[Gysin]
Let $\pi:E \to M$ be a fiber bundle with fiber $S^k$. Then the following sequence is exact at the level of de Rham cohomology:

\begin{align}
\cdots \overset{}{\longrightarrow} H^n (M) \overset{\pi^{\ast}}{\longrightarrow} H^{n} (E) \overset{\pi_{\ast}}{\longrightarrow} H^{n-k} (M) \overset{e \wedge}{\longrightarrow} H^{n+1} (M) \longrightarrow \cdots 
\end{align}
The map $\pi_{\ast}$ is the pushforward map $\pi_{\ast} : H^n (E) \to H^{n-k}(M)$, which acts by integration of differential forms over the fiber, and $\pi^{\ast}$ is just the pullback induced by the projection $\pi: E \to M$. The map $e \wedge : H^{n-k}(M) \to H^{n+1}(M)$ is the wedge product with the Euler class. 
\end{theorem}
The Gysin sequence is also exact at the level of \v{C}ech cohomology with integral coefficients, though one must replace the wedge product of $e$ with the cup product, and the pushforward map is no longer simply integration over the fiber.

\chapter{Miscellaneous results}
We collect in this Appendix some miscellaneous results which are useful for reference throughout this thesis. 
\section{Adapted coordinates}
\label{Adapt}

Abelian T-duality following the Buscher rules is straightforward provided one works in a system of coordinates adapted to the Killing vector. A system of coordinates is said to be adapted to a Killing vector if the isometry generated by the Killing vector acts as a translation on a single coordinate. If the components of a metric are independent of a particular coordinate $Z$, then the corresponding vector field $\pr_Z$ is a Killing vector field, and the coordinates will be automatically adapted to this Killing vector. More generally, a Killing vector can generate isometries which won't act as translations on the given coordinates. In this subsection we provide a simple method of obtaining a system of coordinates adapted to a Killing vector.

Let $(M,g)$ be a Riemannian manifold of dimension $D$, and let $k$ be a Killing vector of $g$. Suppose we have a set of coordinates $(\{X^{\mu}\})$ for M, where $\mu \in \{1,\dots,D\}$. In these coordinates,the Killing vector has the form $k = k^{\mu}\pr_{\mu}$. We are looking for a coordinate change $\phi = \{Y^{i}(X^{\mu}),\theta(X^{\mu})\}$, for $i \in \{1,\dots,D-1\}$, such that $\phi_{\ast}k = \pr_{\theta}$. To find such a coordinate change, we simply solve the following system of partial differential equations:
\begin{align*}
k^{\mu}\pr_{\mu} Y^1 &= 0 \\
k^{\mu}\pr_{\mu} Y^2 &= 0 \\
&\vdots \\
k^{\mu}\pr_{\mu} Y^{D-1} &= 0 \\
k^{\mu}\pr_{\mu} \theta &= 1.
\end{align*}
That is, we solve the system $\{k(\phi^i) =0,\, k(\theta) = 1\}$. A concrete example will serve to illustrate this, so consider the Heisenberg 3D manifold given by the following metric:
\begin{align*}
ds^2 = dx^2 + (dy - xdz)^2 + dz^2 
\end{align*}
The 4 (local) Killing vectors for this metric are 
\begin{align*}
\left\{ x\pr_z - z\pr_x + \left( \frac{x^2}{2} -\frac{z^2}{2} + 1  \right)\pr_y,\, \pr_x + z\pr_y,\, \pr_y,\, \pr_z \right\}.
\end{align*}
We notice that the components of the metric are independent of the coordinates $y$ and $z$, and so this system of coordinates is already adapted to the corresponding Killing vectors $\pr_y$ and $\pr_z$. Suppose instead that we wanted to find coordinates adapted to the Killing vector $\pr_x + z\pr_y$. To find such coordinates, we solve the system of partial differential equations (where we have relabelled $\{Y^1, Y^2,\theta\}$ as $\{X,Y,\theta\}$ for convenience):
\begin{align*}
\frac{\pr X}{\pr x} &+ z\frac{\pr X}{\pr y} = 0 \\[1em]
\frac{\pr Y}{\pr x} &+ z\frac{\pr Y}{\pr y} = 0 \\[1em]
\frac{\pr \theta}{\pr x} &+ z\frac{\pr \theta}{\pr y} = 1.
\end{align*}
Generally this can be solved with the aid of a computer software package, but in this case it is easy to see that the following set of functions solve the system:
\begin{align*}
X &= y-xz \\
Y &= z \\
\theta &= x
\end{align*}
The inverse transformation is given by:
\begin{align*}
x &= \theta \\
y &= X + \theta Y \\
z &= Y
\end{align*}
It is relatively straightforward now to compute the pushforward of the Killing vector $\pr_x + z\pr_y$ under this coordinate transformation, and indeed one finds that $\phi_{\ast}(\pr_x + z\pr_y) = \pr_{\theta}$, as required. The new metric in these coordinates is 
\begin{align*}
\dd s^2 = \dd X^2 + 2Y \dd X \dd \theta + \dd Y^2 + (1+Y^2) \dd \theta^2.
\end{align*}
This metric is independent of $\theta$, so it is clear that $\pr_{\theta}$ is a Killing vector.

\section{Why the $f$-flux background defines a metric on a compact space.}
\label{app:fflux}
In the previous subsection we considered the Heisenberg 3D manifold, with the metric \begin{align*}
ds^2 = dx^2 + (dy - xdz)^2 + dz^2 
\end{align*}
\emph{A priori}, the coordinates $(x,y,z)$ are real coordinates, and the above line element defines a curved metric on a non-compact space, $\RR^3$. From a physics perspective we are more interested in a compactified form of this manifold, which arises as the abelian T-dual of a three-torus $T^3 = S^1 \times S^1 \times S^1$, equipped with a non-trivial H-flux. 

More concretely, performing an abelian T-duality along the $y$-coordinate of $T^3$, with a B-field given by $B = -x dy \wedge dz$, we get the above Nilmanifold metric with vanishing B-field.\footnote{Note that we are being a little imprecise here - strictly speaking, the $y$-coordinate in the Nilmanifold metric is the dual coordinate, or Lagrange multiplier, from the T-duality procedure, and we should probably call it $\hat{y}$.} 

To compactify, we identify the coordinates periodically, in much the same way that coordinates for the circle can be defined by $\{\theta \in \RR : \theta \sim \theta + 1\}$. Of course, we want to make sure that the metric we have for the Nilmanifold still makes sense after these identifications, so we need to do the identification in a way which preserves the metric. An easy way to do this is to make sure that the basis of one-forms $(dx, dy-xdz, dz)$ is invariant under any identification. A more mathematical way of saying this is that we want the deck transformations defining the compactification to be isometries. For example, the following identification $$(x,y,z) \sim (x,y+1,z)$$ periodically identifies the $y$-coordinate. Thinking of this as a diffeomorphism $$\phi_1: (x,y,z) \mapsto (x,y+1,z)$$the basis of one-forms change under this diffeomorphism via the pullback. The other identifications we need to make are given by 
\begin{align*}
\phi_2 &: (x,y,z) \mapsto (x,y,z+1) \\
\phi_3 &: (x,y,z) \mapsto (x+1,y+z,z)
\end{align*}
To see explicitly that this leaves the basis of one-forms (and therefore the metric) invariant, let's compute, as an example, how the one forms change under the identification $\phi_3$:
\begin{align*}
dx &\mapsto \phi_3^{\ast} (dx) = d(x+1) = dx \\
dy &\mapsto \phi_3^{\ast} (dy-xdz) = d(y+z) - (x+1)dz = dy + dz - xdz - dz = dy -xdz \\
dz &\mapsto \phi_3^{\ast} (dz) = dz \\
\end{align*}
Notice that $\phi_3$ twists the $y$-coordinate with the $z$-coordinate as $x$ is periodically identified. This is why the Nilmanifold is sometimes referred to as the `twisted torus' - it is a non-trivial circle bundle over a two torus.
The three identifications are often written in the succinct notation:
\begin{align*}
(x,y,z) \sim (x,y+1,z) \sim (x,y,z+1) \sim (x+1,y+z,z)
\end{align*}

\chapter{Lie algebroid gauge theory}
\label{app:LAGT}
\section{The essentials}
A Lie algebroid gauge theory replaces the Lie algebra bundle $M \times \fg$ of a gauge theory over $M$ with a Lie algebroid $\pi:E \to M$. This is a vector bundle $E$ over $M$, together with an anchor map $\rho: E \to TM$, and a Lie algebra bracket $[\cdot,\cdot]$ on sections of $E$ satisfying the Leibniz rule
\begin{align}
[s_1,f s_2] = f[s_1,s_2] + \rho(s_1)(f) \, s_2.
\end{align}
Choosing a local frame $\{e_a\}$ for $E$,\footnote{i.e. a local basis of sections.} defines the structure functions
\begin{align}
[e_a,e_b] = f^c_{\ ab} e_c.
\end{align}
The image of the basis $\{e_a\}$ under the anchor map is a set of vector fields on $M$
\begin{align}
\rho(e_a) &= v^i_a \pr_i.
\end{align}
The anchor map is a morphism of the bracket
\begin{align}
\rho([X,Y]) = [\rho(X),\rho(Y)],
\end{align}
where the bracket on the right is the usual commutator of vector fields on $TM$.

Recall that a connection $\nabla$ on a vector bundle $E$ is a differential operator \begin{align}
\nabla:\Gamma(E) \to \Gamma(E) \otimes \Omega^1(M)
\end{align} 
satisfying
\begin{align}
\nabla(sf) &=  (\dd f) s  +  f \nabla(s) ,
\end{align}
for all $f \in C^{\infty}(M)$ and $s \in \Gamma(E)$. Applying this operator to a basis vector from a local frame and expanding in terms of that frame determines the components of the connection form
\begin{align}
\nabla e_a =  e_b \, \omega^b_a.
\end{align}
The components of the connection form, $\omega^a_b$, is a matrix of one-forms on $M$. When we apply the connection $\nabla$ to an arbitrary section, $s$, of $E$, we can expand 
\begin{align}
\nabla(s) &= \nabla(e_a s^a)\\
&= e_a (\dd s^a) + s^a \, \nabla(e_a) \\
&= e_b (\dd s^b + \omega^b_a s^a).
\end{align}
Thus when acting on sections, $\nabla$ is often written as $\dd + \omega$, where it is understood that the exterior derivative and the connection form act on the components of the section expanded in the basis $e_a$. 
Under a change of frame $e_a \to e'_a =  e_b \Lambda^b_a$, the connection form transforms as
\begin{align}
\omega'^a_b &= (\Lambda^{-1})^a_m \omega^m_n \Lambda^n_b +(\Lambda^{-1})^a_m \dd \Lambda^m_b.
\end{align}
Note that, as we have defined it, the connection form is a local object. It is defined as a matrix of one-forms on $M$, or at least, on the open set $U\subset M$ on which the frame $e_a$ defines a trivialisation. Local descriptions of a connection one-form can be patched together into a global connection form, provided they satisfy a patching condition

Every vector bundle admits a connection, but connections are certainly not unique. Given two connections $\nabla_1$ and $\nabla_2$ on $E$, their difference is a one-form on the base with values in the endomorphism bundle
\begin{align}
\nabla_1 - \nabla_2 \in \Omega^1(M,\textrm{End}(E)).
\end{align}
It is easy to see this in coordinates, by looking at $\nabla_1 - \nabla_2$ acting on a section $s$:
\begin{align}
\nabla_1(s) - \nabla_2(s) &= e_b(\dd s^b + (\omega_1)^b_a s^a) - e_b(\dd s^b + (\omega_2)^b_a s^a) \\
&= e_b (\omega_1-\omega_2)^b_a s^a
\end{align} Conversely, given a connection $\nabla$ and an endomorphism-valued one-form $\phi$, we can construct a new connection $\nabla+\phi$. In local coordinates, we have 
\begin{align}
(\nabla + \phi) (s) &= e_b(\dd s^b + \omega^b_a s^a) + e_b \phi^b_a s^a \\
&= e_b(\dd s^b + (\omega + \phi)^b_a s^a)
\end{align} 

The definition of the connection can be extended to act on vector-valued differential forms: 
\begin{align}
\nabla: \Gamma(E) \otimes \Omega^k(M) \to \Gamma(E) \otimes \Omega^{k+1}(M)
\end{align}
by
\begin{align}
\label{connectionform}
\nabla(s \alpha) = s \dd \alpha  + \nabla(s) \wedge \alpha.
\end{align}
The curvature of the connection measures the failure of the covariant derivative to square to zero. Applying the covariant derivative twice to a section, and using the property (\ref{connectionform}) gives us:
\begin{align}
\nabla(\nabla(s)) &= \nabla(\nabla(e_a s^a)) \\
&= \nabla(e_b(\dd s^b + \omega^b_a s^a)) \\
&= e_a (\dd \omega^a_b + \omega^a_c \wedge \omega^c_b) s^b \\
&= e_b R^b_a s^a.
\end{align}
We find that the curvature has the coordinate description $R^a_b = \dd \omega^a_b + \omega^a_c \wedge \omega^c_b$, and is a 2-form on $M$ taking values in End$(E)$. It may also be defined by the expression
\begin{align}
R(X,Y) = \nabla_X \nabla_Y - \nabla_Y \nabla_X - \nabla_{[X,Y]},
\end{align}
for $X,Y \in \Gamma(TM)$.

The connection $\nabla$ on $E$ is a connection on a vector bundle, and so doesn't use the condition that $E$ is also a Lie algebroid. When $E$ is a Lie algebroid, we can use the anchor map to lift any vector bundle connection $\nabla$ on $E$, to a so-called $E$-connection ${^E}\nabla$ on $E$:
\begin{align}
{^E}\nabla : \Gamma(E) \to \Gamma(E) \otimes \Gamma(E^{\ast}).
\end{align}
It is often more convenient to think of ${^E}\nabla$ as a map from $\Gamma(E) \otimes \Gamma(E)$ to $\Gamma(E)$. In this thesis, we will use a specific $E$-connection, given by
\begin{align}
{^E}\nabla_{s_1}(s_2) := \nabla_{\rho(s_1)} (s_2).
\end{align}
More generally, however, an $E$-connection is simply required to satisfy the Leibnitz rule
\begin{align}
{^E}\nabla (s,ft) = f {^E}\nabla (s,t) + (\rho(s)f) t .
\end{align}
Once we have an $E$-connection, we can define the $E$-curvature of ${^E}\nabla$ as
\begin{align}
\mathscr{R}(s,t) = {^E}\nabla_s {^E}\nabla_t - {^E}\nabla_t {^E}\nabla_s - {^E}\nabla_{[s,t]},
\end{align}
and the $E$-torsion of ${^E}\nabla$ as
\begin{align}
\mathcal{T}(s,t) = {^E}\nabla_s t - {^E}\nabla_t s - [s,t],
\end{align}
where the bracket is the Lie algebroid bracket on sections of $E$. The $E$-curvature and the $E$-torsion are maps
\begin{align}
\mathscr{R} &: \Gamma(E) \otimes \Gamma(E) \to \textrm{End}(E) \\
\mathcal{T}&: \Gamma(E) \to \textrm{End}(E).
\end{align}
In coordinates, we have 
\begin{align}
\mathscr{R}_{(e_m,e_n)}(e_b) &= e_a (\mathscr{R}_{mn})^a_b \\
&= e_a \left[ (v^i_m v^j_n - v^i_n v^j_m)( \pr_i \omega^a_{bj} + \omega^a_{ci} \omega^c_{bj}) \right] \\
&= e_a \left( \iota_{v_n} \iota_{v_m} R^a_b \right),
\end{align}
where $R^a_b$ is the curvature of the connection $\nabla$, and
\begin{align}
\mathcal{T}_{e_b}(e_c) &= e_a \mathcal{T}^a_{bc} \\
 &= -e_a \left( f^a_{\ bc} + \omega^a_{bi} v^i_c - \omega^a_{ci} v^i_b \right).
\end{align}
That is, the coordinate expressions for the $E$-curvature and $E$-torsion of the $E$-connection are
\begin{align}
(\mathscr{R}_{mn})^a_b &= \iota_{v_n} \iota_{v_m} R^a_b \\
\mathcal{T}^a_{bc} &= - \left( f^a_{\ bc} + \omega^a_{bi} v^i_c - \omega^a_{ci} v^i_b \right).
\end{align}

\end{appendices}